\documentclass[10pt,letterpaper,reqno]{amsart}
\usepackage[letterpaper,margin=1in]{geometry}
\usepackage{amsmath,amssymb,mathtools}
\usepackage{mathrsfs}
\usepackage{enumitem}
\usepackage{xcolor}
\usepackage{microtype}
\usepackage{hyperref}
\hypersetup{
  colorlinks=true,
  allcolors=blue,
  linktoc=all,
  bookmarksnumbered=true,
  bookmarksopen=true,
  pdfstartview=FitH,
  pdfauthor={Phan Th\`anh Nam, Rongchan Zhu, and Xiangchan Zhu},
  pdftitle={Dynamical classical-field limit of bosonic Gibbs states: Renormalized Hartree NLS correlations in 2D and 3D},
  pdfsubject={Time-dependent correlations for the renormalized Hartree NLS equation from quantum many-body Gibbs states},
  pdfkeywords={quantum Gibbs states, classical Gibbs measures, Hartree NLS equation, Wick quantization, Schatten estimates, Liouville equation}
}
\numberwithin{equation}{section}
\allowdisplaybreaks[2]

\newtheorem{theorem}{Theorem}[section]
\newtheorem{proposition}[theorem]{Proposition}
\newtheorem{lemma}[theorem]{Lemma}
\newtheorem{corollary}[theorem]{Corollary}
\newtheorem{definition}[theorem]{Definition}
\theoremstyle{remark}
\newtheorem{remark}[theorem]{Remark}

\newcommand{\1}{\mathbf 1}
\newcommand{\ii}{\mathrm i}
\newcommand{\dd}{\mathrm d}
\newcommand{\ee}{\mathrm e}
\newcommand{\Tr}{\operatorname{Tr}}
\newcommand{\Dom}{\operatorname{Dom}}
\newcommand{\Op}{\operatorname{Op}}
\newcommand{\supp}{\operatorname{supp}}
\newcommand{\Fock}{\mathcal F}
\newcommand{\gH}{\mathfrak H}
\newcommand{\dG}{\mathrm d\Gamma}
\newcommand{\Acal}{\mathcal A}
\newcommand{\Lcal}{\mathcal L}
\newcommand{\Ucal}{\mathcal U}
\newcommand{\Lbb}{\mathbb L}
\newcommand{\Ubb}{\mathbb U}
\newcommand{\Vcal}{\mathcal V}
\newcommand{\Xcal}{\mathcal X}
\newcommand{\Cscr}{\mathscr C}
\newcommand{\Hcal}{\mathcal H}
\newcommand{\norm}[1]{\left\lVert #1\right\rVert}
\newcommand{\abs}[1]{\left\lvert #1\right\rvert}
\newcommand{\ip}[2]{\left\langle #1,#2\right\rangle}
\newcommand{\Ree}{\operatorname{Re}}
\newcommand{\Imm}{\operatorname{Im}}
\newcommand{\adop}[1]{{\operatorname{ad}_{#1}}}

\title[Renormalized Hartree NLS Correlations in 2D and 3D]{Dynamical classical-field limit of bosonic Gibbs states: Renormalized Hartree NLS correlations in 2D and 3D}

\date{}

\author[P. T. Nam]{Phan Th\`anh Nam}
\address{(P. T. Nam) LMU Munich, Department of Mathematics, Theresienstrasse 39,
80333 Munich, Germany}
\email{nam@math.lmu.de}

\author[R. Zhu]{Rongchan Zhu}
\address{(R. Zhu) Department of Mathematics, Beijing Institute of Technology,
Beijing 100081, China}
\email{zhurongchan@126.com}

\author[X. Zhu]{Xiangchan Zhu}
\address{(X. Zhu) Academy of Mathematics and Systems Science, Chinese Academy of
Sciences, Beijing 100190, China}
\email{zhuxiangchan@126.com}


\begin{document}

\begin{abstract}
In this paper we derive the time-dependent correlation functions of the
renormalized Hartree NLS equation on the torus $\mathbb T^d$, $d=2,3$,
from the corresponding bosonic many-body Gibbs dynamics. In contrast
with the 1D problem studied earlier by Fr\"ohlich, Knowles, Schlein and
Sohinger, in higher dimensions the scaled quantum particle number is
not uniformly bounded and the relevant classical fields require Wick
renormalization. Our proof combines convergence of the quantum
generators in weighted Hilbert spaces with uniqueness for positive
solutions of the limiting Liouville equation that are dominated by the
Gibbs measure.
\end{abstract}

\maketitle
\tableofcontents

\section{Introduction}
\label{sec:introduction}

The classical-field limit of bosonic quantum Gibbs states corresponds to a regime
in which the low-energy quantum modes are effectively described by a classical random field. In
the present setting, the quantum and classical equilibrium objects are
\[
 \Gamma_\lambda
 =\mathcal Z_\lambda^{-1}\ee^{-\lambda\mathbb H_\lambda},
 \qquad
 \dd\nu(u)=\mathfrak Z^{-1}\ee^{-\Vcal(u)}\,\dd\mu_0(u).
\]
Here $\mathbb H_\lambda$ is the renormalized many-body Hamiltonian on
the bosonic Fock space $\Fock$, while $\mu_0$ is the centered
Gaussian free-field measure and $\Vcal$ is the renormalized classical
interaction energy. The constants $\mathcal Z_\lambda$ and
$\mathfrak Z$ are partition functions that normalize the quantum state and the classical measure,
respectively. Nonlinear Gibbs measures such as $\nu$ arise naturally in constructive quantum field
theory and in the probabilistic study of nonlinear Schr\"odinger equations. 

The derivation of nonlinear Gibbs measures from many-body quantum Gibbs
states was initiated in \cite{LNR15} for the 1D Bose gas in a regime just above the critical point of the Bose--Einstein phase transition, where the
occupation number of each low-energy quantum mode is proportional to $\lambda^{-1}$ in the
classical-field limit $\lambda\downarrow0$. Extensions to 2D and 3D systems, with regular interactions and Wick ordering, were obtained
independently in \cite{LNR} using the variational method and in
\cite{FKSS22} using the functional-integral approach. For delta-type interactions, the
Euclidean $\Phi^4_2$ theory was obtained in \cite{FKSS25}, and the $\Phi^4_3$ theory was derived
in \cite{NZZ25}. The latter work goes beyond Wick renormalization by combining the variational approach introduced in \cite{LNR} with the theory of singular SPDEs. 
Further recent developments
include an alternative derivation of $\Phi^4_2$ theory \cite{JN26}, the inhomogeneous 2D Bose gas with a trapping potential
\cite{CKRTG26}, the homogeneous 2D Bose gas with Bessel interactions \cite{NZZ26a}, higher-order models involving Hartree measures with three-body interactions \cite{L26} and general $P(\Phi)_2$ measures \cite{NYZ26}, and the focusing $\Phi^6_1$ measure
\cite{RS25,LNZ26}. 

While the equilibrium classical-field limit is by now
well understood, the purpose of the present paper is to study the corresponding
dynamics. In this direction, a microscopic derivation of the time-dependent correlation functions of
the one-dimensional cubic nonlinear Schr\"odinger equation was already obtained by
Fr\"ohlich, Knowles, Schlein and Sohinger \cite{FKSS}. However, the problem
in 2D and 3D remains challenging since the
classical fields require Wick renormalization and the expected total number of
quantum particles diverges faster than $\lambda^{-1}$ as $\lambda\downarrow0$.
We aim to answer this question for a general class of homogeneous systems on the torus $\mathbb T^d$, $d=2,3$, with regular,
positive-type interactions.

To be precise, since the quantum Gibbs state $\Gamma_\lambda$ is invariant under the many-body Schr\"odinger
evolution, we are interested in the time-dependent correlations  
\begin{equation}\label{eq:intro-Heisenberg-evolution}
 \Tr\!\left(
  \Ubb_\lambda(t_1)A_1\cdots
  \Ubb_\lambda(t_J)A_J\Gamma_\lambda
 \right), \qquad  \Ubb_\lambda(t)A
 :=\ee^{\ii t\mathbb H_\lambda}A
   \ee^{-\ii t\mathbb H_\lambda}
\end{equation}
for observables $A_1,\ldots,A_J$ for which the trace is well defined. 
These correlations contain information about both the equilibrium state
and the many-body Heisenberg evolution. Observables at different times need not
commute, so the order must be preserved.

On the classical side, the limiting Gibbs measure $\nu$ is invariant under the
renormalized Hartree NLS flow: 
\begin{equation}\label{eq:intro-renormalized-Hartree}
 \ii\partial_tu_t
 =(-\Delta+m)u_t+
  \bigl(v*{:}|u_t|^2{:}\bigr)u_t,
 \qquad m>0,
\end{equation}
where \({:}|u|^2{:}\) denotes the Wick product. For a regular, positive-type interaction $v$, the probabilistic
Cauchy theory and invariance of Gibbs measures for the renormalized Hartree NLS equation \eqref{eq:intro-renormalized-Hartree} were established in \cite{BourgainGP}. The Gibbs measure $\nu$ is invariant under the Hartree NLS flow $S_t$, which is defined by $u_t=S_tu_0$ for $\nu$-almost every initial datum $u_0$.

The use of nonlinear Gibbs measures in the study of Schr\"odinger
equations goes back to Lebowitz--Rose--Speer
\cite{LebowitzRoseSpeer88}.  Bourgain developed the probabilistic
Cauchy theory and invariance of Gibbs measures for periodic nonlinear
Schr\"odinger and Hartree NLS (Gross--Pitaevskii) equations in
\cite{Bourgain94,Bourgain2D,BourgainGP}.  This point of view is
dictated by the support of the Gibbs measure: a typical field lies below
the regularity required by the deterministic Cauchy theory, and the
nonlinearity must be defined together with the random initial datum,
usually after Wick renormalization.  Recent progress on Schr\"odinger
equations with rough random initial data was made by Deng--Nahmod--Yue \cite{DNY2D,DNYRandomTensors,DNYHartree}, who
introduced random averaging operators and subsequently developed
the related theory of random tensors for the propagation of randomness
under nonlinear dispersive equations.  In particular, in 3D we may use the Gibbs-invariant
Hartree NLS flow constructed in \cite{DNYHartree}, whose proof relies on the
random-averaging method.  In 2D the Hartree NLS equation
considered here follows from a simpler form of the arguments in
\cite{Bourgain2D,BourgainGP,DNY2D,OhThomannNLS}.

This leads to the following question: does the many-body Heisenberg dynamics converge, in the classical-field limit, to the nonlinear Gibbs dynamics generated by $S_t$? 


Our main dynamical result concerns ordinary Wick polynomials on fixed
finite-dimensional Fourier spaces. If $P_j$ is a polynomial on $E_j$
and $\Op_{\lambda,E_j}^{\rm W}(P_j)$ denotes its Wick quantization (see
\eqref{eq:ordinary-Wick-definition}), we prove, in 2D under a mild
regularity condition on the interaction (see
\eqref{eq:finite-potential-regularity}) and in 3D for smooth
interactions, that
\begin{equation*}
 \Tr\!\left(
  \prod_{j=1}^J
  \Ubb_\lambda(t_j)\bigl(\Op_{\lambda,E_j}^{\rm W}(P_j)\bigr)
  \Gamma_\lambda
 \right)
 \longrightarrow
 \int\prod_{j=1}^J(P_j)_{E_j}(S_{t_j}u)\,\dd\nu(u)
 \qquad(\lambda\downarrow0).
\end{equation*}
Here $(P_j)_{E_j}$ is the corresponding classical polynomial, as defined
in \eqref{eq:classical-Wick-symbol}. In 3D, Theorem~\ref{thm:main-Wick}
also gives two-point convergence under the same mild condition. A key
ingredient of the proof, Theorem~\ref{thm:main-DLU}, establishes
uniqueness for positive Liouville curves dominated by the Gibbs
measure. The proof first establishes an analogous correlation limit for
bounded cylinder observables through anti-Wick quantization; this
quantization and the intermediate result are introduced in
Subsection~\ref{subsec:proof-outline}.

Let us briefly explain the main difficulties and our proof strategy. For the non-interacting system, the free one-particle density matrix $
 \gamma_{0,\lambda}=(\ee^{\lambda h}-1)^{-1}$, with the kinetic operator $h=-\Delta+m$, satisfies
  \[
\Tr \gamma_{0,\lambda} \asymp
\lambda^{-1} \times \begin{cases} 1,&d=1,\\
  1+\abs{\log\lambda},&d=2,\\
  \lambda^{-1/2},&d=3.
 \end{cases}
\]
Thus, unlike the 1D case studied in \cite{FKSS}, in 2D and 3D 
even the first moment of the scaled particle-number operator
\(\lambda\mathcal N\), where \(\mathcal N\) is the number operator, is
not uniformly bounded as \(\lambda\downarrow0\). The same lack of uniform
boundedness holds for the interacting system.  The global
number-sector cutoff used in \cite{FKSS} is therefore unavailable.

In an attempt to modify the approach in \cite{FKSS}, it is natural to try to
replace $\lambda\mathcal N$ by a positive operator controlling the centered
density or the ultraviolet energy. Such an operator would need uniform Gibbs
moments and would also have to commute sufficiently well with the Hamiltonian
and the observables occurring in the ordered product. However, localized
density and energy operators do not have these properties: their commutators
with the Hamiltonian or with a cylinder observable contain terms of the same
order as the quantity that is to be controlled. Moving such a cutoff past the
factors in an ordered correlation therefore produces additional commutator
terms for which no uniform estimate is available. Moreover, the commutator of
the interaction with a bounded cylinder observable is generally unbounded due
to the nonlocality of the potential. Thus neither equilibrium convergence for
bounded observables nor a direct use of the cutoff argument from \cite{FKSS}
is sufficient. We therefore need to develop a  different approach.

Rather than comparing the quantum and classical dynamics directly, our
proof is based on the principle ``convergence of generators plus
uniqueness for the limiting equation''. This approach is inspired by the theory of singular SPDEs, where
this principle underlies the martingale-problem approach; see the
energy-solution works of Gubinelli--Jara \cite{GJ13},
Gon\c{c}alves--Jara \cite{GJ14}, Gubinelli--Perkowski
\cite{GP18,GP20}, and Gr\"afner--Perkowski--Popat \cite{GPP26},
as well as Perkowski's survey \cite{Per26}. Related generator or
resolvent estimates also occur in recent scaling-limit results of
Cannizzaro, Toninelli and collaborators
\cite{CET23,CGT24,CMT25,CKM26}.
Implementing
this strategy for a quantum many-body system, however, requires
overcoming two difficulties that are absent in the standard stochastic
setting. First, quantum observables do not commute, so ordered
correlation functionals are not positive and cannot be interpreted
directly as measures. Second, the Gibbs-weighted Hilbert spaces $L^2(\Gamma_\lambda)$, on
which the Heisenberg dynamics acts unitarily, depend on $\lambda$, so
the quantum and classical evolutions cannot be compared on a common
Hilbert space.

We overcome these difficulties as follows. The varying spaces
$L^2(\Gamma_\lambda)$ are compared with $L^2(\nu)$ through the
anti-Wick quantization of cylinder functions, using the convergence
theory for varying Hilbert spaces of Kuwae--Shioya
\cite{KuwaeShioya03}. To deal with noncommutativity, we replace ordered
correlations by positive functionals dominated by the quantum Gibbs
state, at the cost of an error controlled by a double commutator
estimate from \cite{LNR}. This yields a compactness argument whose
limits are positive solutions of the limiting Liouville equation
dominated by the Gibbs measure, and these limits are then identified by
a uniqueness result for such solutions.

Liouville equations play a central role in the Wigner-measure
approach to bosonic mean-field dynamics developed by Ammari--Nier
\cite{AmmariNier11,AmmariNier15}. Ammari--Liard
\cite{AmmariLiard18} proved uniqueness results for measure-valued
solutions of Hamiltonian Liouville equations, and Rouffort
\cite{Rouffort18} used such uniqueness properties in a general
mean-field theorem. A closely related argument appears in the work
of Ammari--Liard--Rouffort \cite{AmmariLiardRouffort20} on
Gross--Pitaevskii and Hartree hierarchies: their conditional NLS
applications transfer spacetime moment bounds from measure curves
to trajectories through superposition and Fubini's theorem, then
apply Strichartz estimates to obtain uniqueness.

In our setting, the Gibbs measure is supported on fields below the
regularity required by deterministic Cauchy theory, and the available
constructions of Gibbs-invariant Hartree NLS flows
\cite{BourgainGP,DNYHartree} do not yield uniqueness in the trajectory
class arising from the limit. The superposition principle of
Ammari--Farhat--Sohinger \cite{AFS} represents our dominated Liouville
curves by measures on integral paths, but marginal domination alone
does not control differences of the probabilistically defined Wick
nonlinearities. We compare the self-consistent Hartree potentials by
writing two coupled paths as linear evolutions of the same initial
field. The comparison combines deterministic and probabilistic estimates for transported
operators with centered quadratic expressions in that field, defined
on a common Gaussian full-measure set independent of the potentials.
This proves uniqueness for positive Gibbs-dominated Liouville curves.

More precisely, we proceed in five steps.

\emph{Step 1: Gibbs-weighted Hilbert spaces.}
We first formulate the quantum and classical dynamics in Hilbert spaces
weighted by the corresponding Gibbs states. On the quantum side, we set
\[
 \|A\|_{2,\lambda}^2
 :=\Tr(A^*A\Gamma_\lambda),
\]
and denote by $L^2(\Gamma_\lambda)$ the completion of the bounded
operators on $\Fock$ in this norm. This is precisely the GNS Hilbert space associated with the quantum Gibbs
state $\Gamma_\lambda$, with the identity operator as the cyclic vector; see, for example, \cite{BratteliRobinson87}. The Heisenberg evolution
$\Ubb_\lambda(t)$ defined in \eqref{eq:intro-Heisenberg-evolution} acts unitarily on this space, with closed
skew-adjoint generator $\Lbb_\lambda$. On the classical side, the
Hartree NLS flow induces the Koopman group \cite{Koopman31}
\[
 \Ucal(t)F=F\circ S_t
\]
on $L^2(\nu)$. We denote its closed generator by $\Lcal$, while
$\Lcal_0$ denotes the explicit Hamiltonian derivation on bounded
smooth cylinder functions. The equilibrium classical-field limit gives
a natural identification between suitable elements of
$L^2(\Gamma_\lambda)$ and $L^2(\nu)$, which allows us to compare the
two dynamics in varying Hilbert spaces.

\emph{Step 2: Convergence of the generators.}
For every bounded smooth cylinder function $F$, we prove
\[
 \Lbb_\lambda \Op_{\lambda,\gH}^{\rm A}F
 \to_{\mathrm{id}}\Lcal_0F
 \qquad(\lambda\downarrow0),
\]
where $\to_{\mathrm{id}}$ denotes the identified convergence of
Definition~\ref{def:identified}. This gives the infinitesimal relation
between the quantum many-body evolution and the renormalized Hartree NLS
flow without imposing a global particle-number cutoff.  

The main point is that, although the total scaled particle number is not
uniformly bounded, the commutator of a cylinder observable with a
density mode is localized in a finite Fourier space. This allows us to
separate the relevant finite-dimensional part from the high-frequency
density tail. The former has uniform Gibbs moments, while the latter has
uniform higher moments and a second moment that vanishes as the
Fourier cutoff is removed. Summing over the interaction modes then
yields convergence of the quantum generators in both 2D and 3D. Conceptually, this localization argument replaces the cutoff technique used in \cite{FKSS}.

\emph{Step 3: Positivity and compactness.}
Convergence of the generators on cylinder functions does not by itself
identify the limiting dynamics: we do not assume that these functions
form a core for the closed classical generator. We instead construct
positive Liouville limits and identify them by uniqueness. This also
requires treating noncommutativity, since in general
\[
\Tr(AB\Gamma_\lambda)\neq\Tr(BA\Gamma_\lambda).
\]
For the bounded quantized cylinder observables used here, we recover
positivity through the approximation
$$
\Tr(AB\Gamma_\lambda) \approx  \Tr(A \Gamma_\lambda^{1/2}B\Gamma_\lambda^{1/2})
$$
where the error is controlled by a double commutator estimate from \cite[Theorem~7.2]{LNR} (which is of the same type as the Falk--Bruch inequality \cite{FB69}). Moreover, we observe that if $B\ge 0$, then 
\[
0\le  T_B:=\Gamma_\lambda^{1/2}B\Gamma_\lambda^{1/2} \le \|B\| \Gamma_\lambda.
\]
Thus the functional $A\mapsto\Tr(A T_B)$ is positive and dominated
by the Gibbs state. This positivity and Gibbs domination provide the required compactness:
along every sequence $\lambda_n\downarrow0$, the corresponding
time-dependent positive quantum functionals admit a subsequence converging to
a positive weakly continuous measure curve $(\sigma_t)$ dominated by the
classical Gibbs measure:
\[
 0\le\sigma_t\le C_0\nu.
\]
Moreover, the limit is a positive
Gibbs-dominated solution of the classical Liouville equation:
\[
 \int F\,\dd\sigma_t-\int F\,\dd\sigma_0
 =
 \int_0^t\int\Lcal_0F\,\dd\sigma_s\,\dd s.
\]

\emph{Step 4: Uniqueness of the limiting Liouville equation.} 
It remains to prove that every
positive weakly continuous solution of the above Liouville equation
which is dominated by $\nu$ is uniquely determined by its initial
measure $\sigma_0$ and is given by transport under the renormalized Hartree NLS flow:
\[
 \sigma_t=(S_t)_\#\sigma_0.
\]

As explained above, the available constructions of the Gibbs-invariant
Hartree NLS flow do not give uniqueness in the rough class of
trajectories arising here. We therefore  apply the superposition principle of
Ammari--Farhat--Sohinger \cite{AFS} to represent a dominated Liouville curve
by a measure on integral paths. Marginal domination ensures that almost
every path satisfies the Hartree Duhamel equation. Nevertheless, this does
not make a direct comparison of two such paths possible by the usual
nonlinear estimates. At their regularity, the Wick terms are defined
through probabilistic limiting procedures rather than by a locally
Lipschitz map on the Sobolev space. Coupling two superposition
measures over their common initial marginal does not supply  additional joint probabilistic
structure from which a stability estimate for the difference of their
Wick nonlinearities would follow.

Our argument instead compares the associated potentials. For two
prescribed real potentials $A$ and $B$, we write the corresponding paths
as $U_A(t,0)u_0$ and $U_B(t,0)u_0$, where $U_A$ and $U_B$ are the
propagators of the linear equations with potentials $A$ and $B$.
We construct the required centered quadratic expressions in $u_0$ on
a common Gaussian full-measure set chosen independently of both
potentials. Interpolating between $A$ and $B$ and differentiating the
transported term with respect to the interpolation parameter
expresses the difference through commutators of 
multiplication operators. Weighted coefficient-measure estimates for
the propagators and the commutator factors then control the centered
quadratic part, while the cancellation in the commutator controls the
deterministic trace correction. 

The resulting estimate for the nonzero density modes, together with
the finite first Fourier moment of $\widehat v$, implies a closed
Gronwall inequality for the difference of the mean-zero potentials.
Their equality identifies the two paths up to a scalar phase.
A common Fourier cutoff then identifies the zero density mode and
hence the remaining phase. This proves pathwise uniqueness in the
Gibbs-dominated class. Applying it to a coupling of the superposition
measures over their common initial marginal yields uniqueness of
positive Liouville curves dominated by the Gibbs measure $\nu$.

At this point, our uniqueness result identifies the two-point limits and
the dynamics on $L^2(\Gamma_\lambda)$ at each fixed time. Stability under
bounded left multiplication then gives arbitrary ordered multi-time
correlations of bounded cylinder observables in dimensions two and
three, under a mild regularity condition on the interaction potential
(see \eqref{eq:finite-potential-regularity}).


\emph{Step 5: Wick observables.} For technical reasons, Steps 1--4 are
carried out using anti-Wick quantizations of bounded cylinder
functions, which preserve positivity and boundedness. The observables in
Theorem~\ref{thm:main-Wick} differ from these in two respects. First,
they are Wick ordered; however, Wick and anti-Wick quantizations of a
fixed finite-mode polynomial differ only by lower-degree terms carrying
explicit powers of $\lambda$, so the ordering itself is not the main
issue. Second, and more seriously, their symbols are unbounded
polynomials, so the bounded-observable result does not apply directly.

For two-point correlations, this difficulty is handled by the
convergence of the Heisenberg dynamics established above, which yields
the result in both 2D and 3D under the same mild regularity condition on the
interaction. For
higher-order correlation functions of Wick observables
$X_1,\dots,X_J$, however, this argument cannot be iterated by bounded
left multiplications.

Our observation is that, in the spirit of the symmetrization
$\Gamma_\lambda^{1/2}B\,\Gamma_\lambda^{1/2}$ of Step~3, the correlation
is better behaved when the Gibbs weight is distributed evenly between
the factors. In the equally spaced trace
$\operatorname{Tr}\bigl(\Gamma_\lambda^{1/J}X_1\,\Gamma_\lambda^{1/J}X_2
\cdots\Gamma_\lambda^{1/J}X_J\bigr)$, each unbounded factor is
accompanied by a fractional power of $\Gamma_\lambda$, which allows us
to reduce its limit to the bounded-observable case. The ordered trace,
in which all the Gibbs weight sits at one end, does not have this
property. Instead of estimating the difference between the two traces,
we relate them by an  identity: redistributing the Gibbs weight
amounts to an imaginary shift of the times of order $\lambda$, and for
a Gaussian kernel this shift can be absorbed by completing the square.
As a result, a Gaussian time average of width of order $\lambda$ of the
ordered correlation can be expressed through the equally spaced
correlation at real times, and hence converges; see
Lemma~\ref{lem:directional-Gaussian-Wick}.

It remains to remove the time average, which requires control of the
correlation on time scales of order $\lambda$. Each time derivative of
the correlation costs a factor comparable to the scaled particle
number. In 2D, this number grows only logarithmically, so a
first-order time-derivative bound suffices. In 3D, where the scaled
particle number grows like $\lambda^{-1/2}$, the first-order bound is
no longer sufficient; we instead combine higher-order time-derivative
bounds with higher-order Gaussian differences, whose error decays like
a high power of $\lambda$. Controlling these higher derivatives
requires Gibbs estimates for iterated commutators, and this is where
the additional smoothness of the interaction is needed.

%


\medskip
\noindent\textbf{Organization of the paper.} 
Section~\ref{sec:models} introduces the models, states the main results,
and outlines the proof.  Section~\ref{sec:generators} proves the
equilibrium weighted Hilbert-space identification, and
Section~\ref{sec:higher-density-generator} proves convergence of the
generators on \(L^2(\Gamma_\lambda)\).  Section~\ref{sec:liouville-correlations}
constructs the
positive Liouville limits and reduces the bounded correlation results to
their uniqueness.  Section~\ref{sec:DLU} proves this uniqueness for the
renormalized Hartree NLS equation.
Section~\ref{sec:Wick-observables} treats finite-mode ordinary Wick
observables.  Appendix~\ref{app:Toeplitz} recalls the finite-dimensional
Fock and anti-Wick calculus.  Appendix~\ref{app:transported-density-calculus}
develops the commutator calculus for transported quadratic densities,
constructs the prescribed propagators, and proves the comparison
estimate for the density modes. Appendix~\ref{app:Wick-product-estimates}
proves the Gibbs estimates for products of commutators used for the
three-dimensional Wick correlations.

\medskip
\noindent\textbf{Acknowledgments.}
We thank Tadahiro Oh for valuable comments and thank Yuzhao Wang and Jacky Chong for helpful discussions.  We acknowledge the use of AI tools: ChatGPT-5.6~Pro assisted in
collecting related literature and in exploratory discussions, and in
particular suggested representing the relevant operators in terms of
coefficient measures, which we use in the proof of
Proposition~\ref{prop:transported-density-comparison} (see
Appendix~\ref{app:transported-density-calculus}). We also benefited from discussions with ChatGPT-6~Pro, which suggested
using Gaussian time averaging to relate ordered correlations of Wick
observables to traces with equally spaced Gibbs weights, an idea
developed in Lemma~\ref{lem:directional-Gaussian-Wick}. 
All mathematical arguments
were verified and written by the authors.
P. T. Nam was supported by the European Research Council through the
ERC Consolidator Grant RAMBAS (Project No.~10104424).  X. Zhu was supported by the NSFC
(No.~125952811).  R. Zhu and
X. Zhu were supported by the National Key R\&D Program of China
(No.~2022YFA1006300) and the National Natural Science Foundation of
China (NSFC) (No.~12426205).  R. Zhu was also supported by the NSFC
(No.~12271030).  X. Zhu was also supported by the NSFC
(No.~12288201), the Key Laboratory of Random Complex
Structures and Data Science, Chinese Academy of Sciences.

\section{Setting and main results}
\label{sec:models}

Throughout the paper, \(d\in\{2,3\}\), \(m>0\), and \(\lambda>0\).
Products in correlation functions are ordered from left to right.

\subsection{Classical Gibbs measure and renormalized Hartree NLS flow}
\label{subsec:classical-phase-space}

Let \(\mathbb T^d=(\mathbb R/2\pi\mathbb Z)^d\), and set
\[
 \gH=L^2(\mathbb T^d,\dd x),\qquad h=-\Delta+m,\qquad m>0.
\]
All Hilbert-space inner products are conjugate linear in the first
variable and linear in the second.  For Hilbert spaces \(X,Y\), we write
\(\mathcal B(X,Y)\) for the bounded operators from \(X\) to \(Y\),
\(\mathcal B(X):=\mathcal B(X,X)\), and \(\|\cdot\|_{\mathrm{op}}\)
for the operator norm.  We write \(\1_X\) for the identity on \(X\),
omitting the subscript only when the ambient space is clear, and use the
Japanese bracket \(\langle \xi\rangle=(1+|\xi|^2)^{1/2}\).
For \(k\in\mathbb Z^d\), let
\[
 \mathrm e_k(x)=(2\pi)^{-d/2}\ee^{\ii k\cdot x},
 \qquad
 M_kf(x)=\ee^{\ii k\cdot x}f(x).
\]
For a periodic function or distribution \(f\), we use the Fourier
coefficients
\[
 \widehat f(k)
 :=(2\pi)^{-d}\int_{\mathbb T^d}
 f(x)\ee^{-\ii k\cdot x}\,\dd x,
 \qquad
 f(x)=\sum_{k\in\mathbb Z^d}\widehat f(k)\ee^{\ii k\cdot x}.
\]
  We also write 
\[
 f_k:=\langle \mathrm e_k,f\rangle
 =(2\pi)^{d/2}\widehat f(k).
\]

Let
\begin{equation*}
 H^{s-}:=\bigcap_{n\ge1}H^{s-1/n}(\mathbb T^d),
 \qquad
 \Xcal:=
 \begin{cases}
  H^{0-}(\mathbb T^2),&d=2,\\
  H^{-1/2-}(\mathbb T^3),&d=3.
 \end{cases}
\end{equation*}
For a Hilbert space \(K\) and \(1\le p<\infty\), let
\(\mathfrak S^p(K)\) denote the Schatten \(p\)-class on \(K\), with
norm \(\|\cdot\|_{\mathfrak S^p(K)}\); the ambient space is omitted only
when it is clear.
Let \(\mu_0\) be the centered complex Gaussian field with covariance
\(h^{-1}\).  Then \(\mu_0(\Xcal)=1\).

For \(R\ge1\), let
\[
 \begin{aligned}
 \Lambda_R&=\{p\in\mathbb Z^d:|p|\le R\},
 & \Pi_R&=\sum_{p\in\Lambda_R}|\mathrm e_p\rangle\langle \mathrm e_p|
       =\1_{h\le R^2+m}.
 \end{aligned}
\]
For every finite-dimensional subspace \(E\subset\gH\), \(\Pi_E\) denotes the
orthogonal projection onto \(E\).  For any orthogonal projection \(\Pi\),
we write \(\Pi^\perp:=\1_{\gH}-\Pi\).  
For \(k\in\mathbb Z^d\), define the truncated and limiting renormalized terms by
\begin{equation}\label{eq:classical-density-mode}
 \rho_{R,k}(u)
 :=\ip{\Pi_Ru}{M_k\Pi_Ru}-\Tr(h^{-1}\Pi_RM_k\Pi_R),
 \qquad
 \rho_k:=L^q(\mu_0)\!\!-\!\lim_{R\to\infty}\rho_{R,k},
 \quad 1\le q<\infty.
\end{equation}
The limit is well defined in every finite \(L^q(\mu_0)\), is independent
of the cutoff, and satisfies \(\rho_{-k}=\overline{\rho_k}\); see
\cite[Section~3]{LNR}.

Let \(v\) be a periodic distribution, with Fourier coefficients
normalized by \(v(x)=\sum_k\widehat v(k)\ee^{\ii k\cdot x}\).  We
assume throughout that
\begin{equation}\label{eq:finite-potential-regularity}
 \widehat v(k)\ge0,\qquad v\geq 0,\qquad
 \widehat v(-k)=\widehat v(k),\qquad
 \sum_{k\in\mathbb Z^d}\langle k\rangle|\widehat v(k)|<\infty.
\end{equation}
For \(s\ge0\), let \(\mathcal W^s=\mathcal W^s(\mathbb T^d)\) be the
weighted Wiener space of periodic distributions with finite norm
\[
 \|w\|_{\mathcal W^s}
 :=\sum_{k\in\mathbb Z^d}\langle k\rangle^s|\widehat w(k)|.
\]
Thus \(v\) is real, even, of positive type, and belongs to
\(\mathcal W^1\subset C^1(\mathbb T^d)\).

Define the interaction and Gibbs measure by
\begin{equation*}
 \Vcal(u)
 =\frac12\sum_{k\in\mathbb Z^d}\widehat v(k)\abs{\rho_k(u)}^2,
 \qquad
 \mathfrak Z:=\int_{\Xcal}\ee^{-\Vcal(u)}\,\dd\mu_0(u),
 \qquad
 \dd\nu(u)
 =\mathfrak Z^{-1}\ee^{-\Vcal(u)}\,\dd\mu_0(u).
\end{equation*}
This is the positive-type renormalized Gibbs measure constructed in
\cite{LNR}.

For a field \(u\), the Wick square is written as
\[
 {:}|u|^2{:}(x)
 =(2\pi)^{-d}\sum_{k\in\mathbb Z^d}
 \rho_{-k}(u)\ee^{\ii k\cdot x}.
\]
The renormalized Hartree NLS equation is formally
\[
 \ii\partial_tu_t
 =h u_t+\bigl(v*{:}|u_t|^2{:}\bigr)u_t
 =(-\Delta+m)u_t+\bigl(v*{:}|u_t|^2{:}\bigr)u_t.
\]

The renormalized Hartree NLS equation has a Gibbs-invariant Hartree NLS
flow \(S_t\).  More details are given in
Theorem~\ref{thm:canonical-flow}.  The corresponding Koopman group is
given by
\[
 \Ucal(t)F=F\circ S_t,
\]
and its closed generator is denoted by \(\Lcal\) (see
Proposition~\ref{prop:classical-Koopman-generator}).  We also use \(\Lcal_0\) to denote the explicit infinitesimal derivation associated
with the renormalized Hartree NLS equation on bounded smooth cylinder
functions.  To this end, we first introduce
the algebra of smooth cylinder functions
\begin{equation*}
	\Acal_{\rm cyl}
	=\operatorname{span}_{\mathbb C}\!\left(
	\{1\}\cup
	\{f\circ\Pi:\Pi\text{ is a finite Fourier projection},
	f\in C_c^\infty(\Pi\gH)\}
	\right).
\end{equation*}
\phantomsection\label{par:positive-cylinder-density}The algebra \(\Acal_{\rm cyl}\) is dense in
\(L^p(\nu)\) for every \(1\leq p<\infty\), and
\[
\overline{\{F\in\Acal_{\rm cyl}:F\geq0\}}^{\,L^2(\nu)}
=\{f\in L^2(\nu):f\geq0\}.
\] 
If $F=f\circ\Pi$, then $\mathcal L_0F$ is the derivative of $f(\Pi u)$
along the projected Hartree vector field. Proposition
\ref{prop:classical-Koopman-generator} shows that
\[
\Acal_{\rm cyl}\subset\operatorname{Dom}(\Lcal),
\qquad
\Lcal F=\Lcal_0F
\quad\text{for }F\in\Acal_{\rm cyl}.
\]

\subsection{Quantum Gibbs states}
\label{subsec:quantum-setting}

Let
\[
 \Fock=\Fock(\gH)
 =\bigoplus_{n=0}^\infty\gH^{\otimes_s n},
 \qquad \gH^{\otimes_s0}=\mathbb C.
\]
For \(b\in\mathcal B(\gH)\), its second quantization is
\[
 \dG(b)\big|_{\gH^{\otimes_s n}}=\sum_{j=1}^n b_j
 \quad(n\ge1),
 \qquad \dG(b)\big|_{\mathbb C}=0.
\]
For the self-adjoint operator \(h\), \(\dG(h)\) denotes its
usual self-adjoint second quantization.  In particular,
\(\mathcal N=\dG(\1_{\gH})\) is the number operator on Fock space.

As in \cite{LNR,NZZ25}, if \(\Gamma=\bigoplus_{n\ge0}\Gamma_n\) is a number-preserving state satisfying \(\Tr\!\left({\mathcal N}^{r}\Gamma\right)<\infty\), then we can define its 
\(r\)-particle reduced density matrix by 
\begin{equation}\label{eq:RDM-definition}
 \Gamma^{(r)}
 :=\sum_{n\ge r}\binom nr
 \Tr_{r+1,\ldots,n}\Gamma_n.
\end{equation}
In our convention, \(\Tr(\Gamma^{(r)})=\Tr\!\left(\binom{\mathcal N}{r}\Gamma\right)\).

In this language, the free quantum Gibbs state and its one-particle density matrix are
\[
 \mathcal Z_{0,\lambda}:=\Tr\ee^{-\lambda\dG(h)},
 \qquad
 \Gamma_{0,\lambda}
 =\mathcal Z_{0,\lambda}^{-1}\ee^{-\lambda\dG(h)},
 \qquad
 \gamma_{0,\lambda}:=\Gamma_{0,\lambda}^{(1)}
 =(\ee^{\lambda h}-1)^{-1}.
\]

Following the construction in \cite[Section~3.1 and Theorem~3.1]{LNR},  
with \(h=-\Delta+m\), we define the self-adjoint operator \(\mathbb H_\lambda\) on the bosonic Fock space by the closed nonnegative quadratic form
\begin{equation}\label{eq:scaled-Gibbs-Hamiltonian}
 \lambda\mathbb H_\lambda
 =\lambda\dG(h)
 +\frac12\sum_{k\in\mathbb Z^d}\widehat v(k)
 B_{\lambda,k}B_{\lambda,-k}, \quad  B_{\lambda,k}=B_\lambda(M_k),
\end{equation}
where for a bounded one-particle operator \(b\) we denote
\begin{equation*}
 B_\lambda(b)
 =\lambda\dG(b)-\lambda\Tr(b\gamma_{0,\lambda})\1. 
\end{equation*}
We are interested in the Gibbs state 
\begin{equation*}
 \Gamma_\lambda
 =\mathcal Z_\lambda^{-1}\ee^{-\lambda\mathbb H_\lambda}, \quad \mathcal Z_\lambda:=\Tr\ee^{-\lambda\mathbb H_\lambda}
\in(0,\infty).
\end{equation*}
For any time \(t\in\mathbb R\) and observable \(A\in\mathcal B(\Fock)\), we define the Heisenberg dynamics as in \eqref{eq:intro-Heisenberg-evolution}: 
\begin{equation*}
 \Ubb_\lambda(t)A
 :=\ee^{\ii t\mathbb H_\lambda}A\ee^{-\ii t\mathbb H_\lambda}.
\end{equation*}

\subsection{Wick quantization on finite Fourier spaces}
\label{subsec:quantizations}

For \(f\in\gH\) and
\(\psi^{(n)}\in\gH^{\otimes_s n}
=L^2_{\mathrm{sym}}((\mathbb T^d)^n)\),
the creation and annihilation operators are defined sectorwise by
\begin{align*}
 (a(f)\psi^{(n)})(x_1,\ldots,x_{n-1})
 &=\sqrt n\int_{\mathbb T^d}\overline{f(x)}
   \psi^{(n)}(x,x_1,\ldots,x_{n-1})\,\dd x,\\
 (a^*(f)\psi^{(n)})(x_1,\ldots,x_{n+1})
 &=\frac1{\sqrt{n+1}}\sum_{j=1}^{n+1}f(x_j)
   \psi^{(n)}(x_1,\ldots,\widehat{x_j},\ldots,x_{n+1}),
\end{align*}
and satisfy
\[
 [a(f),a(g)]=[a^*(f),a^*(g)]=0,
 \qquad
 [a(f),a^*(g)]=\ip{f}{g}.
\]
We follow the standard bosonic Fock-space conventions of
\cite[Section~6.1]{LNR} and use the scaled creation and annihilation operators
\begin{equation*}
 a_\lambda(f)=\sqrt\lambda\,a(f),
 \qquad
 a_\lambda^*(f)=\sqrt\lambda\,a^*(f).
\end{equation*}

Let \(E\subset\gH\) be a fixed finite Fourier space, with orthogonal
projection \(\Pi_E\) and complex dimension \(r_E\).  Fix an orthonormal
basis \((\psi_j)_{j=1}^{r_E}\) and identify \(E\) with
\(\mathbb C^{r_E}\). We regard functions on
\(E\simeq\mathbb C^{r_E}\) as functions of both \(z\) and
\(\bar z\), without any holomorphicity assumption.   In
polynomial coordinates, \(z_j=\langle \psi_j,u\rangle\) corresponds to
\(a_\lambda(\psi_j)\), whereas \(\bar z_j=\langle u,\psi_j\rangle\)
corresponds to \(a_\lambda^*(\psi_j)\).

Set
\[
 z_E(u)=\bigl(\langle \psi_1,u\rangle,\ldots,
                  \langle \psi_{r_E},u\rangle\bigr)\in\mathbb C^{r_E}.
\]
If
\begin{align}\label{eq:P-pol}
 P(z,\bar z)=\sum_{\alpha,\beta}c_{\alpha\beta}
 \bar z^\alpha z^\beta,
\end{align}
where \(\alpha,\beta\in\mathbb N_0^{r_E}\) and only finitely many
coefficients are nonzero, let
\begin{equation}\label{eq:classical-Wick-symbol}
 (P)_E(u)=P\bigl(z_E(u),\overline{z_E(u)}\bigr).
\end{equation}
\medskip
\noindent\textbf{Wick quantization.}
The Wick quantization of the polynomial  \(P\) in \eqref{eq:P-pol} is
\begin{equation}\label{eq:ordinary-Wick-definition}
 \Op_{\lambda,E}^{\rm W}(P)
 =\sum_{\alpha,\beta}c_{\alpha\beta}
 \prod_{j=1}^{r_E}a_\lambda^*(\psi_j)^{\alpha_j}
 \prod_{j=1}^{r_E}a_\lambda(\psi_j)^{\beta_j},
\end{equation}
with respect to this basis.
Whenever an operator quantized on a finite Fourier space \(E\) is used
on \(\Fock(\gH)\), it is understood to be extended by
\(\1_{\Fock(E^\perp)}\) under the factorization
\(\Fock(\gH)\simeq\Fock(E)\otimes\Fock(E^\perp)\).

\subsection{Main results}
\label{subsec:main-results}

Our main result concerns dynamical correlations of ordinary Wick
polynomials on fixed finite-dimensional Fourier spaces. We obtain the
convergence of two-point correlation functions in both 2D and 3D under
the condition \eqref{eq:finite-potential-regularity} on the interaction
potential. For higher-order correlation functions, the same condition
suffices in 2D, while in 3D we require the interaction to be smooth.

\begin{theorem}
\label{thm:main-Wick}
Let \(d\in\{2,3\}\) and \(m>0\). Let $v$ satisfy  
\eqref{eq:finite-potential-regularity}. 
Let \(S_t\) be the
Hartree NLS flow of Theorem~\ref{thm:canonical-flow}. Then the
following statements hold.
\begin{enumerate}[label=\textup{(\roman*)},leftmargin=2.2em]
\item If $P_1,P_2$ are polynomials on fixed finite
Fourier spaces $E_1,E_2$, respectively, then for all fixed $s,t\in\mathbb R$,
\begin{align*}
 &\Tr\!\left(
   \bigl(\Ubb_\lambda(s)\Op_{\lambda,E_1}^{\rm W}(P_1)\bigr)^*
   \Ubb_\lambda(t)\bigl(\Op_{\lambda,E_2}^{\rm W}(P_2)\bigr)
   \Gamma_\lambda\right)\\
 &\qquad\longrightarrow
 \int\overline{(P_1)_{E_1}(S_su)}(P_2)_{E_2}(S_tu)\,\dd\nu(u)
 \qquad(\lambda\downarrow0).
\end{align*}
\item If $d=2$, or if $d=3$ and $v\in C^\infty(\mathbb T^3)$, then for
every fixed integer $J\ge1$, finite Fourier spaces $E_1,\dots,E_J$,
polynomials $P_j$ on $E_j$, and times $t_1,\dots,t_J\in\mathbb R$, we
have
\[
 \Tr\!\left(\prod_{j=1}^J\Ubb_\lambda(t_j)
        \bigl(\Op_{\lambda,E_j}^{\rm W}(P_j)\bigr)\Gamma_\lambda\right)
 \longrightarrow
 \int\prod_{j=1}^J(P_j)_{E_j}(S_{t_j}u)\,\dd\nu(u)
 \qquad(\lambda\downarrow0).
\]
\end{enumerate}
The convergence in part~\textup{(ii)} is locally uniform in the time
variables. 
\end{theorem}

In \cite{FKSS}, time-dependent correlation limits were proved for the
1D cubic nonlinear Schr\"odinger equation. In that setting
$h^{-1}\in\mathfrak S^1$, whereas here
$h^{-1}\in\mathfrak S^2\setminus\mathfrak S^1$, so the density requires
Wick renormalization. Under the identification $\tau=\lambda^{-1}$, the
finite-rank observables with finite Fourier support used in
\cite{FKSS} coincide with the finite-mode Wick operators considered
here. Thus Theorem~\ref{thm:main-Wick} can be interpreted as an extension of the result of \cite{FKSS} to dimensions two and three.

%
%
%


A key ingredient in the proof of Theorem~\ref{thm:main-Wick} is the
following uniqueness result for positive measure-valued solutions of
the Liouville equation dominated by the Gibbs measure. It is a purely
classical statement, which may be of independent interest.


\begin{theorem}
\label{thm:main-DLU}
Let \(d\in\{2,3\}\). Let $v$ satisfy \eqref{eq:finite-potential-regularity}.  Let
\(1/2<\kappa<1\), and fix \(T>0\) and \(0\le C_0<\infty\).
Let \(t\mapsto\sigma_t\), \(|t|\le T\), be a weakly continuous curve of
finite positive Borel measures on \(H^{-\kappa}(\mathbb T^d)\) such that
\[
 0\le\sigma_t\le C_0\nu,
 \qquad |t|\le T,
\]
and
\[
 \int F\,\dd\sigma_t-\int F\,\dd\sigma_0
 =\int_0^t\int\Lcal_0F\,\dd\sigma_s\,\dd s
 \qquad(F\in\Acal_{\rm cyl},\ |t|\le T).
\]
Let \(S_t\) be the Hartree NLS flow of
Theorem~\ref{thm:canonical-flow}. Then
\begin{equation*}
 \sigma_t=(S_t)_\#\sigma_0,
 \qquad |t|\le T.
\end{equation*}
\end{theorem}

Theorem~\ref{thm:main-DLU} will be restated and proved as
Theorem~\ref{thm:DLU}.


Let us comment on the assumptions on the interaction. The probabilistic
theory of the limiting Hartree NLS equation works for rougher
interactions than those allowed by
\eqref{eq:finite-potential-regularity}; see
\cite{BourgainGP,DNYHartree}. We need
\eqref{eq:finite-potential-regularity} for two steps of our method:
the convergence of the quantum generators
(Theorem~\ref{thm:one-generator}) and the uniqueness of positive
Gibbs-dominated solutions of the Liouville equation
(Theorem~\ref{thm:DLU}). The additional smoothness assumption in
Theorem~\ref{thm:main-Wick}\textup{(ii)} for $d=3$ is used only to
bound higher-order time derivatives of the correlations; see
Proposition~\ref{prop:smooth-raw-derivatives} and
Appendix~\ref{app:Wick-product-estimates}.

We also note that our compactness argument produces a solution of the
Liouville equation, but it does not show directly that the
corresponding trajectories are the solutions constructed in
\cite{DNYHartree} by the random-averaging method; see
Remark~\ref{dny}. Instead, we identify the limit with the Hartree NLS
flow by using the uniqueness of the Liouville equation.

\subsection{Outline of the proof}
\label{subsec:proof-outline}

The main ideas of the proof were already explained in the Introduction. Here we
give a more detailed roadmap, with precise references to the results
proved in the subsequent sections. 

%

First, the proof of Theorem~\ref{thm:main-Wick} treats bounded cylinder
observables, using anti-Wick quantization. Although Wick quantization
is the natural choice for the observables in
Theorem~\ref{thm:main-Wick}, it produces unbounded operators and does not preserve
positivity. Anti-Wick quantization, by contrast, maps nonnegative
symbols to nonnegative operators and bounded symbols to bounded
operators, and its expectations in a quantum state are given exactly by
integrals against the corresponding lower-symbol measures. These
properties fit the compactness argument based on positive measures
described in the Introduction. The two quantizations are related by the Berezin identities below,
which we recall for orientation; the passage from bounded to Wick
observables is carried out in Section~\ref{sec:Wick-observables}. In the following we 
introduce anti-Wick quantization and state the bounded-observable limit,
and then explain how Theorem~\ref{thm:main-DLU} identifies the limiting
dynamics.

\medskip
\noindent\textbf{Coherent states.}
Let \(E\subset\gH\) be a finite Fourier space, with
\(r_E=\dim_{\mathbb C}E\), and use the coordinate convention of
Subsection~\ref{subsec:quantizations}. We write \(\dd u\) for
Lebesgue measure on the underlying real vector space.
The normalized scale-\(\lambda\)
coherent vector and its resolution of the identity are
\begin{equation}\label{eq:coherent-resolution}
 \begin{aligned}
 \xi_{\lambda,E}(u)
 &=\ee^{-\norm{u}^2/(2\lambda)}
  \bigoplus_{n=0}^\infty
 \frac{u^{\otimes n}}{\sqrt{n!}\,\lambda^{n/2}},
 \qquad u\in E,\qquad
 (\pi\lambda)^{-r_E}\int_E
 |\xi_{\lambda,E}(u)\rangle\langle\xi_{\lambda,E}(u)|\,\dd u
 =\1_{\Fock(E)}.
 \end{aligned}
\end{equation}
Moreover,
\begin{equation}\label{eq:coherent-eigenvector}
 a_\lambda(f)\xi_{\lambda,E}(u)
 =\langle f,u\rangle\xi_{\lambda,E}(u),\qquad
 \langle\xi_{\lambda,E}(u), a_\lambda^*(f)\psi\rangle
 =\langle u,f\rangle
   \langle\xi_{\lambda,E}(u),\psi\rangle,
 \qquad f,u\in E,
\end{equation}
(see e.g. \cite[Section~6.1]{LNR15}).  Thus
\((\pi\lambda)^{-r_E}|\xi_{\lambda,E}(u)\rangle
\langle\xi_{\lambda,E}(u)|\,\dd u\) is a normalized positive
operator-valued measure (POVM) on \(E\).
For a positive trace-class operator \(T\) on \(\Fock(E)\), its
scale-\(\lambda\) lower-symbol measure is
\begin{equation*}
 \dd\mu_{E,T}^{\lambda}(u)
 =(\pi\lambda)^{-r_E}
 \langle\xi_{\lambda,E}(u),T\xi_{\lambda,E}(u)\rangle\,\dd u.
\end{equation*}

\medskip
\noindent\textbf{Anti-Wick quantization.}
  For a
bounded Borel function \(f\) on a finite Fourier space \(E\), define
\begin{equation}\label{eq:ordinary-Anti-Wick-definition}
	\Op_{\lambda,E}^{\rm A}(f)
	:=(\pi\lambda)^{-r_E}
	\int_E f(u)
	|\xi_{\lambda,E}(u)\rangle
	\langle\xi_{\lambda,E}(u)|\,\dd u.
\end{equation}
Then
\[
\|\Op_{\lambda,E}^{\rm A}(f)\|_{\mathrm{op}}\le\|f\|_\infty,
\qquad
\Op_{\lambda,E}^{\rm A}(\overline f)
=\Op_{\lambda,E}^{\rm A}(f)^*.
\]
The notation
\(\Op_{\lambda,\gH}^{\rm A}\) denotes the cylinder extension to \(\Fock(\gH)\). 
More precisely, for the cylinder function \(f\circ\Pi\), set
\[
\Op_{\lambda,\gH}^{\rm A}(f\circ\Pi)
=\Op_{\lambda,\Pi\gH}^{\rm A}(f)
\otimes\1_{\Fock(\Pi^\perp\gH)}
\]
under the factorization
\(\Fock(\gH)\simeq\Fock(\Pi\gH)\otimes\Fock(\Pi^\perp\gH)\),
and extend linearly to \(\Acal_{\rm cyl}\).  If \(\Pi_1\ge\Pi_0\), the
coherent POVM factorization gives
\[
\Op_{\lambda,\Pi_1\gH}^{\rm A}(f\circ\Pi_0)
=\Op_{\lambda,\Pi_0\gH}^{\rm A}(f)
\otimes\1_{\Fock((\Pi_1-\Pi_0)\gH)}.
\]
Thus the definition is independent of the representation.  Moreover,
\[
F\ge0\ \Longrightarrow\ \Op_{\lambda,\gH}^{\rm A}F\ge0,
\qquad
\norm{\Op_{\lambda,\gH}^{\rm A}F}_{\mathrm{op}}\le\norm{F}_\infty,
\qquad
\Op_{\lambda,\gH}^{\rm A}\overline F
=(\Op_{\lambda,\gH}^{\rm A}F)^*.
\]
Set
\begin{equation}\label{eq:equilibrium-lower-symbol-notation}
 E_L:=\Pi_L\gH,
 \qquad
 \Gamma_{\lambda,L}:=\Tr_{\Fock(E_L^\perp)}\Gamma_\lambda,
 \qquad
 \mu_{\lambda,L}:=\mu_{E_L,\Gamma_{\lambda,L}}^\lambda.
\end{equation}
If \(F=F\circ\Pi_L\), then the definition of the lower-symbol measure
and the coherent-state resolution give the  identity
\begin{equation}\label{eq:lower-symbol-pairing}
 \Tr(\Op_{\lambda,\gH}^{\rm A}F\,\Gamma_\lambda)
 =\int_{E_L}F(u)\,\dd\mu_{\lambda,L}(u).
\end{equation}

For a polynomial \(P\) as in \eqref{eq:classical-Wick-symbol}, its anti-Wick quantization is the
quadratic form
\begin{equation*}
 \Op_{\lambda,E}^{\rm A}(P)
 =(\pi\lambda)^{-r_E}\int_E(P)_E(u)
 |\xi_{\lambda,E}(u)\rangle
 \langle\xi_{\lambda,E}(u)|\,\dd u.
\end{equation*}
The coherent-state identities
\eqref{eq:coherent-resolution}--\eqref{eq:coherent-eigenvector} give the
anti-normal formula
\begin{equation*}
 \Op_{\lambda,E}^{\rm A}(P)
 =\sum_{\alpha,\beta}c_{\alpha\beta}
 \prod_{j=1}^{r_E}a_\lambda(\psi_j)^{\beta_j}
 \prod_{j=1}^{r_E}a_\lambda^*(\psi_j)^{\alpha_j}.
\end{equation*}

The finite-dimensional Wick--anti-Wick
identities of Berezin \cite{Berezin66,Berezin71} give
\begin{equation*}
 \Op_{\lambda,E}^{\rm A}(P)
 =\Op_{\lambda,E}^{\rm W}
   \bigl(\ee^{\lambda\Delta_E^{\mathbb C}}P\bigr), \quad 
 \Op_{\lambda,E}^{\rm W}(P)
 =\Op_{\lambda,E}^{\rm A}
   \bigl(\ee^{-\lambda\Delta_E^{\mathbb C}}P\bigr), \quad 
 \Delta_E^{\mathbb C}
 =\sum_{j=1}^{r_E}\partial_{z_j}\partial_{\bar z_j}.
\end{equation*}

\medskip
\noindent\textbf{Bounded-observable limit.}
The first part of the proof establishes the following intermediate
result. Its proof is completed at the end of
Subsection~\ref{subsec:pathwise-uniqueness}.
\begin{proposition}
\label{prop:bounded-correlations}
Let \(d\in\{2,3\}\), let \(m>0\), and assume
\eqref{eq:finite-potential-regularity}.  Let \(S_t\) be the
Hartree NLS flow of Theorem~\ref{thm:canonical-flow}.  For every \(J\ge1\),
every \(F_1,\ldots,F_J\in\Acal_{\rm cyl}\), and arbitrary fixed
\(t_1,\ldots,t_J\in\mathbb R\),
\begin{equation*}
 \lim_{\lambda\downarrow0}
 \Tr\!\left(
  \prod_{j=1}^J
  \Ubb_\lambda(t_j)(\Op_{\lambda,\gH}^{\rm A}F_j)
  \Gamma_\lambda
 \right)
 =
 \int\prod_{j=1}^J
 F_j(S_{t_j}u)\,\dd\nu(u).
\end{equation*}
\end{proposition}

\begin{remark}
The proof also gives convergence of the generators on
\(L^2(\Gamma_\lambda)\) and of the identified dynamics.  Since these
statements require the identified Hilbert spaces and the closed generators
on \(L^2(\Gamma_\lambda)\) introduced below, they are
stated in their natural setting in
Theorem~\ref{thm:main-identified-dynamics}.
\end{remark}
\medskip
\noindent\textbf{Proof of the bounded-observable limit.}
We first prove Proposition~\ref{prop:bounded-correlations} in the case \(J=2\).
The two-point limit determines the classical limit of a single
Heisenberg-evolved observable.  The norm and multiplication results
stated below show that this convergence is preserved when another
evolved bounded observable is multiplied from the left.  Applying this
argument repeatedly yields the limit for every finite ordered product. 

The two-point case contains the main difficulty.  Since
\(\Gamma_\lambda\) does not commute with a general observable, the
ordered quantum functional is not positive and cannot be interpreted
directly as a measure.  We therefore compare it with a positive
functional dominated by the Gibbs state.  The proof is then based on
three ingredients: equilibrium convergence, convergence of the quantum
generators, and uniqueness for positive Liouville curves dominated by
\(\nu\).  We denote these ingredients by \textup{(E)}, \textup{(G)},
and \textup{(U)}, respectively.

\begin{description}[leftmargin=3.2em,labelwidth=2.4em,
                    labelsep=0.6em,style=multiline]
\item[\textup{(E)}]
\emph{Equilibrium weighted Hilbert-space identification.}
For \(F,G\in\Acal_{\rm cyl}\),
\begin{align*}
 \Tr(\Op_{\lambda,\gH}^{\rm A}F\Gamma_\lambda)
 &\longrightarrow\int F\,\dd\nu,\\
 \Tr\!\left(
 \bigl(\Op_{\lambda,\gH}^{\rm A}F\bigr)^*
 \Op_{\lambda,\gH}^{\rm A}G\,\Gamma_\lambda
 \right)
 &\longrightarrow \int \overline F G\,\dd\nu,\\
 \|\Op_{\lambda,\gH}^{\rm A}F\Op_{\lambda,\gH}^{\rm A}G
      -\Op_{\lambda,\gH}^{\rm A}(FG)\|_{2,\lambda}
 &\longrightarrow0.
\end{align*}
These facts are proved
in Subsection~\ref{sec:static-convergence}; see
Proposition~\ref{prop:equilibrium-identification}.

\item[\textup{(G)}]
\emph{Convergence of the generators on \(L^2(\Gamma_\lambda)\).}
There is \(\lambda_0>0\) such that, for every
\(F\in\Acal_{\rm cyl}\),
\[
 \Op_{\lambda,\gH}^{\rm A}F\in\Dom(\Lbb_\lambda),\qquad
 \Lbb_\lambda \Op_{\lambda,\gH}^{\rm A}F
 \to_{\mathrm{id}}\Lcal_0F,
 \qquad
 \sup_{0<\lambda\le\lambda_0}
 \|\Lbb_\lambda \Op_{\lambda,\gH}^{\rm A}F\|_{2,\lambda}<\infty.
\]
This is Theorem~\ref{thm:one-generator}, proved in
Section~\ref{sec:higher-density-generator}.

\item[\textup{(U)}]
\emph{Uniqueness of dominated positive Liouville curves.}
Fix \(T>0\), \(1/2<\kappa<1\), and \(C_0<\infty\).  If
\(t\mapsto\sigma_t\), \(|t|\le T\), is a weakly continuous curve of
finite positive Borel measures on \(H^{-\kappa}(\mathbb T^d)\) satisfying
\[
 0\le\sigma_t\le C_0\nu
\]
and
\[
 \int F\,\dd\sigma_t-\int F\,\dd\sigma_s
 =\int_s^t\!\int\Lcal_0F\,\dd\sigma_r\,\dd r,
 \qquad F\in\Acal_{\rm cyl},\quad s,t\in[-T,T],
\]
then
\[
 \sigma_t=(S_t)_\#\sigma_0,
 \qquad |t|\le T.
\]
This is Theorem~\ref{thm:main-DLU}, proved in
Section~\ref{sec:DLU}; see Theorem~\ref{thm:DLU}.
\end{description}

We now explain how these ingredients are combined.
For a real-valued observable \(F\), adding
\(\|F\|_\infty\) makes its anti-Wick quantization positive.  The
trace-norm comparison then reduces the ordered two-point functional,
up to a vanishing error and an equilibrium term, to a positive
functional dominated by the Gibbs state.  The positive Liouville limit
criterion, based on~\textup{(E)}, \textup{(G)}, and invariance of
\(\Gamma_\lambda\), gives a subsequential limit which is a positive
Liouville curve dominated by a constant multiple of \(\nu\).
Statement~\textup{(U)} identifies this curve with the push-forward of
its initial measure under the Hartree NLS flow.  At time zero,
statement~\textup{(E)} identifies the initial measure, and invariance
of \(\nu\) cancels the constant added to \(F\).  This proves the
two-point limit for real-valued \(F\); the general case follows by
linearity.

The two-point limit gives identified weak convergence of the
Heisenberg dynamics.  Unitarity and the norm convergence
in~\textup{(E)} upgrade this to identified strong convergence.
The multiplication result in~\textup{(E)} then shows that left
multiplication by an evolved bounded cylinder observable preserves
this convergence.  Applying this fact successively yields the ordered
multi-time limits in Proposition~\ref{prop:bounded-correlations}.

The positive comparison, the positive Liouville limit criterion, and
the identification using~\textup{(U)} are carried out in
Subsections~\ref{subsec:positive-functionals-compactness} and
\ref{subsec:Liouville-transfer}, respectively.
Section~\ref{sec:DLU} proves~\textup{(U)}.


\medskip
\noindent\textbf{Passage from bounded to Wick observables.}
To prove Theorem~\ref{thm:main-Wick}, we pass from bounded cylinder
observables to ordinary Wick polynomials in
Section~\ref{sec:Wick-observables}. Their two-point limits in both
dimensions follow from equilibrium convergence and the convergence of
the Heisenberg dynamics; see
Proposition~\ref{prop:finite-mode-Wick-Hilbert}. For higher-order
correlations, this argument cannot be iterated because the Wick
operators are unbounded. Moreover, the ordered trace is a priori not
even bounded uniformly in $\lambda$.

We therefore consider the trace with equal powers
$\Gamma_\lambda^{1/J}$ between successive factors. Schatten estimates
allow radial cutoffs to be introduced factor by factor
(Lemma~\ref{lem:static-Wick-Schatten}). The resulting trace of bounded
observables is compared with the corresponding ordered trace by
commutator estimates (Lemma~\ref{lem:bounded-cylinder-Gibbs-commutator})
and identified using Proposition~\ref{prop:bounded-correlations}; see
Lemma~\ref{lem:positive-gap-correlations}. To return to the ordered
trace of the Wick observables, we use an  identity. At a finite
spectral cutoff, the ordered and the equally spaced traces differ only
by Boltzmann factors $\ee^{-\lambda r\cdot\omega}$, where $\omega$
collects the energy differences between successive eigenstates and $r$
is a fixed direction in the time variables. A Gaussian time average
along $r$ of width of order $\lambda$ absorbs these factors by
completing the square. This gives the convergence of the averaged
ordered correlation; see Lemma~\ref{lem:directional-Gaussian-Wick}.

Finally, the time average is removed in
Subsection~\ref{subsec:Gaussian-removal-Wick}. In dimension two, this
follows from a first-order time-derivative bound; see
Lemma~\ref{lem:ordered-trace-number-bounds} and
Theorem~\ref{thm:two-dimensional-Wick}. For smooth interactions in
dimension three, we use higher-order time-derivative bounds based on
the Gibbs estimates of Appendix~\ref{app:Wick-product-estimates}; see
Proposition~\ref{prop:smooth-raw-derivatives} and
Theorem~\ref{thm:smooth-3d-Wick}.

\section{Generators of the classical and quantum dynamics}
\label{sec:generators}

In this section we define the classical and quantum generators on \(L^2(\nu)\) and
\(L^2(\Gamma_\lambda)\), respectively. For a cylinder function, the
classical generator depends only on the finite-dimensional range of its
projection. On the quantum
side, the Heisenberg dynamics extends to a unitary group on the weighted
Hilbert space \(L^2(\Gamma_\lambda)\).

Subsection~\ref{subsec:classical-generator} gives the finite-dimensional
classical differential calculus.
Subsection~\ref{subsec:quantum-weighted-generator} constructs the closed
quantum generator on \(L^2(\Gamma_\lambda)\) and computes the Hamiltonian
commutator on
finite-particle vectors.  Subsection~\ref{sec:static-convergence} proves
the equilibrium mean, product, and inner-product limits, and introduces
convergence in the varying weighted Hilbert spaces. Subsection~\ref{subsec:equilibrium-density-tails} establishes the moment
estimates
used later.  The convergence and multiplication results in
Subsection~\ref{sec:static-convergence} give statement~\textup{(E)} of
Subsection~\ref{subsec:proof-outline}.

\subsection{The classical generator on cylinder functions}
\label{subsec:classical-generator}

In this subsection we define the explicit infinitesimal derivation
\(\Lcal_0\) on cylinder observables and prove that every
\(F\in\Acal_{\rm cyl}\) belongs to \(\operatorname{Dom}(\Lcal)\), with
\(\Lcal F=\Lcal_0F\), where \(\Lcal\) is the closed generator of the
Koopman group associated with the Hartree NLS flow.  We first state the
basic properties of the measurable Gibbs-invariant Hartree NLS flow.  We
then introduce the finite-dimensional differential calculus and define
\(\Lcal_0\) by differentiating a cylinder function along the projected
Hartree vector field.  The finite-dimensional chain rule along the
Hartree NLS flow yields the asserted domain inclusion and generator
identity.  Finally, in preparation for comparison with the quantum
commutator, we introduce the derivation associated with a one-particle
operator, extend it to the non-self-adjoint Fourier multipliers \(M_k\),
and rewrite \(\Lcal_0\) in the symmetric form used below.

Let \(h_p=|p|^2+m\) and
\begin{equation*}
	\mathscr F_p(u)
	:=\sum_{k\in\mathbb Z^d}\widehat v(k)\rho_k(u)u_{p+k}.
\end{equation*}
The series is the limit of its finite partial sums in every finite
\(L^q(\nu)\).  Recall that the Hartree NLS equation is understood as follows:
\begin{equation}\label{eq:renormalized-Hartree-equation}
	\ii\partial_t u_{t,p}
	=h_pu_{t,p}+\mathscr F_p(u_t),
	\qquad p\in\mathbb Z^d.
\end{equation}

We first state the properties of the Gibbs-invariant Hartree NLS flows and
identify their Koopman generators.

\begin{theorem}
	\label{thm:canonical-flow}
	Assume \eqref{eq:finite-potential-regularity}.  For each
	\(d\in\{2,3\}\), there is a Borel set
	\(\Sigma_d\subset\Xcal\), with \(\nu(\Sigma_d)=1\), and a family of
	measurable solution maps \(S_t:\Sigma_d\to\Sigma_d\), \(t\in\mathbb R\),
	for \eqref{eq:renormalized-Hartree-equation}.  We use the measurable
	extension to \(\Xcal\) defined by
	\(
	S_tu=0, u\notin\Sigma_d, t\in\mathbb R.
	\)
	For every \(T<\infty\) and every \(s\) satisfying \(s<0\) when \(d=2\)
	and \(s<-1/2\) when \(d=3\), the path map
	\(u\mapsto S_{\cdot}u\) is Borel from \(\Xcal\) to
	\(C([-T,T];H^s(\mathbb T^d))\).
	\begin{enumerate}[label=\textup{(\roman*)},leftmargin=2.2em]
		\item On \(\Sigma_d\), the maps form a measure-preserving group pointwise: for $ t,\tau\in\mathbb R$
		\[
		S_0u=u,\qquad
		S_{t+\tau}u=S_t(S_{\tau}u),
		\qquad
		(S_t)_\#\nu=\nu.
		\]
		\item For every \(u\in\Sigma_d\)  and every
		\(p\in\mathbb Z^d\),
		\[
		(S_tu)_p
		=\ee^{-\ii t h_p}u_p
		-\ii\int_0^t\ee^{-\ii(t-r)h_p}
		\mathscr F_p(S_ru)\,\dd r.
		\]
	
		\item The flow is gauge equivariant:
		\[
		S_t(\ee^{\ii\theta}u)
		=\ee^{\ii\theta}S_tu,
		\qquad t,\theta\in\mathbb R
		\]
		on a common full-measure subset of \(\Sigma_d\).
	\end{enumerate}
\end{theorem}

 In dimension three, the Hartree NLS flow follows from the construction underlying
\cite[Theorem~1.3]{DNYHartree}, adapted to our renormalization
convention; the stronger Fourier summability assumption
\eqref{eq:finite-potential-regularity} simplifies the estimates.  In dimension
two, one may specialize the Hartree argument of \cite{BourgainGP} or use
the invariant Gibbs-flow arguments of
\cite{Bourgain2D,DNY2D,OhThomannNLS}.  These references give the measurable
flow, its group and invariance properties, the mild equation, and gauge
equivariance.  In the following we identify the Koopman
generator of the Hartree NLS flow on cylinder functions.

To this end, we introduce the following notation. If
\(f\in C^1(\Pi\gH;\mathbb C)\) and
\(z,\eta\in \Pi\gH\), define the real directional derivative of \(f\)
at \(z\) in the fixed vector direction \(\eta\) by
\begin{equation*}
 D_\eta f(z)
 :=\left.\frac{\dd}{\dd s}f(z+s\eta)\right|_{s=0},
 \qquad s\in\mathbb R.
\end{equation*}
Equivalently, \(D_\eta f(z)=Df(z)[\eta]\), where \(Df(z)\) is the real
Fr\'echet differential.  Let
\(r_\Pi=\dim_{\mathbb C}(\Pi\gH)=\operatorname{rank}\Pi\).  In complex
coordinates \(z_j=x_j+\ii y_j\), \(1\le j\le r_\Pi\), set
\begin{equation*}
 \partial_{z_j}
 =\frac12(\partial_{x_j}-\ii\partial_{y_j}),
 \qquad
 \partial_{\bar z_j}
 =\frac12(\partial_{x_j}+\ii\partial_{y_j}).
\end{equation*}
Then
\begin{equation*}
 D_\eta f(z)
 =\sum_{j=1}^{r_\Pi}
 \left(\partial_{z_j}f(z)\eta_j
 +\partial_{\bar z_j}f(z)\overline{\eta_j}\right).
\end{equation*}

For the cylinder function \(f\circ\Pi\in \Acal_{\rm cyl}\), define the projected vector field
\begin{equation*}
 X_\Pi(u)
 :=
 -\ii h\Pi u
 -\ii\sum_{k\in\mathbb Z^d}\widehat v(k)\rho_k(u)\Pi M_{-k}u
 \in \Pi\gH.
\end{equation*}
We then define $\Lcal_0$ on $\Acal_{\rm cyl}$ as follows:
\begin{equation}\label{eq:classical-generator-direct}
 (\Lcal_0(f\circ\Pi))(u):=D_{X_\Pi(u)}f(\Pi u).
\end{equation}
Then \eqref{eq:finite-potential-regularity} and Gaussian hypercontractivity imply that, for $F\in\Acal_{\rm cyl}$ and $q<\infty$,
\begin{equation}\label{eq:classical-generator-Lq}
	\Lcal_0F\in L^q(\nu).
\end{equation}

\begin{proposition}
	\label{prop:classical-Koopman-generator}
	The operators \(\Ucal(t)F=F\circ S_t\) form a strongly continuous
	unitary group on \(L^2(\nu)\).  If \(\Lcal\) denotes its closed
	skew-adjoint generator, then, for every \(F\in\Acal_{\rm cyl}\),
	\begin{equation}\label{eq:canonical-cylinder-identity}
		F(S_tu)-F(u)
		=\int_0^t(\Lcal_0F)(S_ru)\,\dd r.
	\end{equation}
	The identity holds in \(L^2(\nu)\) and for \(\nu\)-almost every
	initial datum. Furthermore,
	\begin{equation}\label{eq:app-classical-integrated-generator}
		\Ucal(t)F-F
		=\int_0^t\Ucal(r)\Lcal_0F\,\dd r
		\qquad\text{in }L^2(\nu),
	\end{equation}
	and hence
	\begin{equation}\label{eq:app-classical-generator-identification}
		\Acal_{\rm cyl}\subset\Dom(\Lcal),
		\qquad
		\Lcal F=\Lcal_0F.
	\end{equation}
\end{proposition}

\begin{proof}
	By the group property and invariance in
	Theorem~\ref{thm:canonical-flow}\textup{(i)}, each \(\Ucal(t)\) is an
	isometry with inverse \(\Ucal(-t)\), hence is unitary.  For a bounded
	continuous cylinder function \(F\), path continuity and continuity of
	the finite-dimensional projection give
	\(F(S_tu)\to F(u)\) for \(\nu\)-almost every \(u\).
	Bounded convergence gives convergence in \(L^2(\nu)\).  The density of
	\(\Acal_{\rm cyl}\) stated in Section~\ref{sec:models}, together with
	unitarity, then gives strong continuity on all of \(L^2(\nu)\).
	We then apply Stone's theorem \cite[Section~VIII.4, Theorem~VIII.8]{ReedSimon1}
	to obtain the closed skew-adjoint generator \(\Lcal\).

	Let \(F=f\circ\Pi\in\mathcal A_{\mathrm{cyl}}\). By
	Theorem~\ref{thm:canonical-flow}\textup{(ii)}, for \(\nu\)-almost every
	\(u\),
	\[
	\Pi S_{\cdot}u\in W_{\mathrm{loc}}^{1,1}(\mathbb R;\Pi\gH),
	\qquad
	\partial_t(\Pi S_tu)=X_\Pi(S_tu)
	\quad\text{for a.e. }t.
	\]
	Therefore, by the finite-dimensional chain rule and
	\eqref{eq:classical-generator-direct},
	\[
	\begin{aligned}
		F(S_tu)-F(u)
		&=\int_0^t
		D_{\partial_r(\Pi S_ru)}f(\Pi S_ru)\,\dd r 
		=\int_0^t
		D_{X_\Pi(S_ru)}f(\Pi S_ru)\,\dd r 
		=\int_0^t
		(\Lcal_0F)(S_ru)\,\dd r,
	\end{aligned}
	\]
	which implies \eqref{eq:canonical-cylinder-identity}.
	
	Finally, \(\Lcal_0F\in L^2(\nu)\) by
	\eqref{eq:classical-generator-Lq}, and invariance of \(\nu\) gives
	\(\|\Ucal(r)\Lcal_0F\|_{L^2(\nu)}=\|\Lcal_0F\|_{L^2(\nu)}\).
	Fubini's theorem therefore shows that the identity obtained above for
	\(\nu\)-almost every \(u\) also holds in \(L^2(\nu)\), namely,
	\eqref{eq:app-classical-integrated-generator}.
	Strong continuity now implies
	\[
	\frac1t\int_0^t\Ucal(r)\Lcal_0F\,\dd r
	\longrightarrow\Lcal_0F
	\qquad\text{in }L^2(\nu)
	\]
	as \(t\to0\).  Hence \(F\in\Dom(\Lcal)\) and
	\(\Lcal F=\Lcal_0F\), and
	\eqref{eq:app-classical-generator-identification} follows.
\end{proof}

\begin{remark}
The operators \(\Lcal_0\) and \(\Lcal\) have different
	domains.  The first is the explicit derivation
	\eqref{eq:classical-generator-definition} on
	\(\Acal_{\rm cyl}\), whereas the second is the closed
	skew-adjoint generator of the Koopman group on \(L^2(\nu)\).  Thus
	\[
	\Lcal_0\subset\Lcal.
	\]
	We do not claim that \(\Acal_{\rm cyl}\) is a core for
	\(\Lcal\). This is analogous to the situation in \cite{ZhuZhu18}, where the
	Dirichlet form associated with the dynamical \(\Phi^4_3\) model
	constructed by singular SPDE methods agrees with the classical gradient
	form on cylinder functions, while it remains open whether the closure of
	the classical gradient form coincides with the Dirichlet form of the
	constructed process.
\end{remark}

For later comparison with the quantum commutator, let \(b=b^*\) be a
self-adjoint one-particle operator such that, for the finite Fourier
projection $\Pi$ under consideration, $\Pi b$ 
extends continuously from $\mathcal X$ to $\Pi\gH$.
  It generates the linear
Schr\"odinger flow \(u\mapsto\ee^{-\ii sb}u\).  We define the
corresponding derivation on observables by
\begin{equation*}
 \mathscr D_bF(u)
 :=
 \left.\frac{\dd}{\dd s}\right|_{s=0}
 F(\ee^{-\ii sb}u)
 =D_{-\ii bu}F(u).
\end{equation*}
Thus, for the cylinder function \(f\circ\Pi\),
\begin{equation}\label{eq:classical-Hamiltonian-derivation-cylinder}
 \mathscr D_b(f\circ\Pi)(u)
 =
 D_{-\ii\Pi bu}f(\Pi u).
\end{equation}
For the non-self-adjoint multiplier \(M_k\), set
\begin{equation}\label{def:fkcs}
 f_{k,\mathrm c}(x):=\cos(k\cdot x),\qquad
 f_{k,\mathrm s}(x):=\sin(k\cdot x).
\end{equation}
Identifying these functions with their multiplication operators, we have
\(M_k=f_{k,\mathrm c}+\ii f_{k,\mathrm s}\), and define
\[
 \mathscr D_{M_k}
 :=\mathscr D_{f_{k,\mathrm c}}+\ii\mathscr D_{f_{k,\mathrm s}}.
\]
For an arbitrary one-body operator \(b\), the resulting
complex-linear derivation is
\begin{equation}\label{eq:complex-Hamiltonian-derivation}
 \mathscr D_b(f\circ\Pi)(u)
 :=-\ii\sum_j\partial_{z_j}f(\Pi u)(\Pi bu)_j
   +\ii\sum_j\partial_{\bar z_j}f(\Pi u)\overline{(\Pi b^*u)_j}.
\end{equation}
For \(b=b^*\), formula
\eqref{eq:complex-Hamiltonian-derivation} is exactly the real
directional derivative \(D_{-\ii\Pi bu}f(\Pi u)\) in
\eqref{eq:classical-Hamiltonian-derivation-cylinder}.  For
non-self-adjoint \(b\), it is instead the complex-linear extension of
that Hamiltonian derivation.
With this notation \eqref{eq:classical-generator-direct} is equivalent
to
\begin{equation}\label{eq:classical-generator-definition}
 \Lcal_0F
 =
 \mathscr D_hF
 +\frac12\sum_{k\in\mathbb Z^d}\widehat v(k)
 \left[
  \rho_k\,\mathscr D_{M_{-k}}F
  +\rho_{-k}\,\mathscr D_{M_k}F
 \right].
\end{equation}
Indeed, insert
\eqref{eq:complex-Hamiltonian-derivation} and relabel \(k\mapsto-k\)
in the second half of the sum.  The \(z\)-derivative becomes
\[
 -\ii\sum_k\widehat v(k)\rho_k(u)
 \sum_j\partial_{z_j}f(\Pi u)(\Pi M_{-k}u)_j,
\]
and its \(\bar z\)-derivative is the conjugate tangent component in
\eqref{eq:classical-generator-direct}.

\subsection{The quantum generator and Hamiltonian commutators}
\label{subsec:quantum-weighted-generator}

In this subsection we realize the Heisenberg dynamics as a strongly
continuous unitary group on \(L^2(\Gamma_\lambda)\) and denote its
closed generator by \(\mathbb L_\lambda\).  We then compute the Hamiltonian
commutator on finite-particle vectors and state the generator identity
 proved in
Section~\ref{sec:higher-density-generator}.

\label{subsec:quantum-weighted-unitary-group}

Define the normal state on \(\mathcal B(\Fock)\) by
\[
 \omega_\lambda(A):=\Tr(A\Gamma_\lambda).
\]
The associated inner product and norm are
\[
 \ip{A}{B}_\lambda=\omega_\lambda(A^*B)
 =\Tr(A^*B\Gamma_\lambda),
 \qquad
 \norm{A}_{2,\lambda}^2=\Tr(A^*A\Gamma_\lambda).
\]
Since \(\ee^{-\lambda\mathbb H_\lambda}\), and hence \(\Gamma_\lambda\), is
injective, this is an inner product.
We denote by \(\Hcal_\lambda:=L^2(\Gamma_\lambda)\) the completion of
\(\mathcal B(\Fock)\) in \(\|\cdot\|_{2,\lambda}\).  The map
\[
 \iota_\lambda(A):=A\Gamma_\lambda^{1/2}
\]
is isometric because
\[
 \norm{\iota_\lambda(A)}_{\mathfrak S^2(\Fock)}^2
 =\Tr(\Gamma_\lambda^{1/2}A^*A\Gamma_\lambda^{1/2})
 =\norm{A}_{2,\lambda}^2.
\]
Hence \(\Hcal_\lambda\) is realized as the closure of
\(\{A\Gamma_\lambda^{1/2}:A\in\mathcal B(\Fock)\}\) in
\(\mathfrak S^2(\Fock)\).

For \(C,A\in\mathcal B(\Fock)\),
\[
 \norm{CA}_{2,\lambda}
 \le \norm{C}_{\mathrm{op}}\norm{A}_{2,\lambda}.
\]
Therefore left multiplication by \(C\) extends uniquely to a bounded
operator on \(\Hcal_\lambda\), still denoted by \(x\mapsto Cx\).
Since \(\Gamma_\lambda=\mathcal Z_\lambda^{-1}
\ee^{-\lambda\mathbb H_\lambda}\), the Heisenberg evolution preserves
the state:
\[
 \omega_\lambda(\Ubb_\lambda(t)A)=\omega_\lambda(A).
\]
Thus, on the dense subspace represented by bounded operators,
\[
 \norm{\Ubb_\lambda(t)A}_{2,\lambda}^2
 =\omega_\lambda(\Ubb_\lambda(t)(A^*A))
 =\omega_\lambda(A^*A),
\]
so \(\Ubb_\lambda(t)\) extends to a unitary operator on
\(\Hcal_\lambda\), with inverse \(\Ubb_\lambda(-t)\).  For $C\in\mathcal B(\Fock)$ and $x\in\Hcal_\lambda$, we have
\begin{equation}\label{eq:weighted-unitary-implementation}
 \Ubb_\lambda(t)\1=\1,
 \qquad
 \Ubb_\lambda(t)(Cx)
 =\bigl(\Ubb_\lambda(t)C\bigr)\Ubb_\lambda(t)x.
\end{equation}
Under the Hilbert--Schmidt identification, the Heisenberg evolution is
represented by unitary conjugation:
\begin{equation}\label{eq:weighted-HS-conjugation}
	\iota_\lambda\bigl(\Ubb_\lambda(t)x\bigr)
	=\ee^{\ii t\mathbb H_\lambda}\iota_\lambda(x)
	\ee^{-\ii t\mathbb H_\lambda}.
\end{equation}
Since unitary conjugation preserves the Hilbert--Schmidt norm, this
formula also shows directly that \(\Ubb_\lambda(t)\) is unitary on
\(L^2(\Gamma_\lambda)\).  By
\cite[Section~7.2]{LNR}, or directly from
\eqref{eq:weighted-HS-conjugation}, \(\Ubb_\lambda(t)\) is strongly continuous.
Stone's theorem \cite[Section~VIII.4, Theorem~VIII.8]{ReedSimon1}
therefore gives a unique closed densely defined skew-adjoint generator
\(\Lbb_\lambda\), with \(\Ubb_\lambda(t)=\ee^{t\Lbb_\lambda}\) and
\[
 \Dom(\Lbb_\lambda)
 =\left\{A\in\Hcal_\lambda:
 \lim_{t\to0}\frac{\Ubb_\lambda(t)A-A}{t}
 \text{ exists in }\Hcal_\lambda\right\},
 \qquad
 \Lbb_\lambda A
 =\lim_{t\to0}\frac{\Ubb_\lambda(t)A-A}{t}.
\]

Formally, we have 
$$\Lbb_\lambda A=\ii[\mathbb H_\lambda,A].$$
Moreover, 
\begin{equation}\label{eq:exact-quantum-generator-formula}
 \ii[\mathbb H_\lambda,A]=\ii[\dG(h),A]
 +\frac{\ii}{2}\sum_k\widehat v(k)
 \left(B_{\lambda,k}\adop{M_{-k}}A
 +\adop{M_k}A\,B_{\lambda,-k}\right).
\end{equation}
Here, for a bounded one-particle operator \(b\), we use the notation
\begin{equation}\label{eq:quantum-density-derivation}
 \adop{b}A
 :=[\dG(b),A]
 =\lambda^{-1}[B_\lambda(b),A].
\end{equation}
  For
\(A=\Op_{\lambda,\gH}^{\rm A}F\),
Section~\ref{sec:higher-density-generator} shows that the right-hand
side converges in \(\Hcal_\lambda\). 
Theorem~\ref{thm:one-generator} below shows that
\begin{equation}\label{eq:closed-generator-identification}
 \Op_{\lambda,\gH}^{\rm A}F\in\Dom(\Lbb_\lambda),
 \qquad
 \Lbb_\lambda \Op_{\lambda,\gH}^{\rm A}F=\ii[\mathbb H_\lambda,\Op_{\lambda,\gH}^{\rm A}F].
\end{equation}

\subsection{Equilibrium limits and identified convergence}
\label{sec:static-convergence}

In this subsection we establish the equilibrium identification between
\(L^2(\Gamma_\lambda)\) and \(L^2(\nu)\).  We first prove the convergence
of equilibrium expectations, products, and inner products.  We then
introduce identified convergence and prove its stability under
multiplication.

\subsubsection{Finite-mode lower symbols and equilibrium limits}
\label{subsec:lower-symbol-equilibrium}

We first recall the equilibrium convergence theorem from
\cite[Theorem~4.2]{LNR} with the fixed-interaction extension supplied by
the argument of \cite[Sections~7--11]{NZZ25}. 

\begin{theorem}
\label{thm:equilibrium-limits}
Let \(d\in\{2,3\}\) and assume
\eqref{eq:finite-potential-regularity}. Then, as $\lambda\downarrow0$,
\[
 \frac{\mathcal Z_\lambda}{\mathcal Z_{0,\lambda}}
 \longrightarrow \mathfrak Z.
\]
Moreover, for every fixed \(r\ge1\), as $\lambda\downarrow0$,
\begin{equation*}
 r!\lambda^r\Gamma_\lambda^{(r)}
 \longrightarrow
 \gamma_{\nu}^{(r)}
 :=\int_{\Xcal}|u^{\otimes r}\rangle\langle u^{\otimes r}|\,
       \dd\nu(u)
 \quad\text{in }\mathfrak S^2(\gH^{\otimes_s r}).
\end{equation*}
\end{theorem}

The estimates in \cite[Lemmas~11.2--11.3]{LNR},  also used in
\cite[Theorem~9.1]{NZZ25}, imply, for every fixed \(r\), 
\begin{equation}\label{eq:LNR-RDM-uniform}
 \sup_{0<\lambda\le\lambda_0}
 \norm{\lambda^r\Gamma_\lambda^{(r)}}_{\mathfrak S^2}<\infty.
\end{equation}
These statements use \eqref{eq:RDM-definition}.
The finite-dimensional lower-symbol argument
in \cite[Sections~5 and~11]{LNR} implies that, for every fixed \(L\), the measure \(\mu_{\lambda,L}\) defined in \eqref{eq:equilibrium-lower-symbol-notation} satisfies, as $\lambda\downarrow0$,
\begin{equation}\label{conver}
 \mu_{\lambda,L}\Longrightarrow(\Pi_L)_\#\nu.
\end{equation}

\begin{lemma}\label{lem:AW-product}
Let \(E\) be fixed and let \(F,G\in C_b^1(E)\).  Then
\begin{equation}\label{eq:AW-product}
 \norm{
 \Op_{\lambda,E}^{\rm A}(F)
 \Op_{\lambda,E}^{\rm A}(G)
 -
 \Op_{\lambda,E}^{\rm A}(FG)
 }_{\mathrm{op}}
 \le C_E\norm{F}_\infty\norm{\nabla G}_\infty\lambda^{1/2},
 \qquad0<\lambda\le1.
\end{equation}
\end{lemma}

\begin{proof}
Recall that \(r_E=\dim_{\mathbb C}E\), and set
\(\dd m_\lambda(z)=(\pi\lambda)^{-r_E}\dd z\), and
\(\xi_z=\xi_{\lambda,E}(z)\).  Using the standard coherent-state formulas in
\cite[Section~5.2]{LNR}, we obtain
\begin{equation}\label{eq:co:xi}
 |\langle \xi_z,\xi_{z'}\rangle_{\Fock(E)}|
 =\ee^{-|z-z'|^2/(2\lambda)},
\qquad
 \int_E|\xi_z\rangle\langle \xi_z|\,\dd m_\lambda(z)=\1_{\Fock(E)}.
\end{equation}
Let
\(\mathcal E_\lambda(F,G)=\Op_{\lambda,E}^{\rm A}(F)
\Op_{\lambda,E}^{\rm A}(G)-\Op_{\lambda,E}^{\rm A}(FG)\).
Inserting \eqref{eq:co:xi}, we obtain
\[
 \mathcal E_\lambda(F,G)=\int_{E\times E}F(z)\bigl(G(z')-G(z)\bigr)
 \langle \xi_z,\xi_{z'}\rangle_{\Fock(E)}|\xi_z\rangle\langle \xi_{z'}|
 \,\dd m_\lambda(z)\dd m_\lambda(z').
\]
Let \(\mathcal I_\lambda\) be the integral operator with kernel
\[
 k_\lambda(z,z')=F(z)\bigl(G(z')-G(z)\bigr)
 \langle \xi_z,\xi_{z'}\rangle_{\Fock(E)}.
\]
Therefore
\(
 \|\mathcal E_\lambda(F,G)\|_{\mathrm{op}}\le\|\mathcal I_\lambda\|_{\mathrm{op}}
\).

We apply Schur's test; see, for example,
\cite[Theorem~5.2]{HalmosSunder78}.   Let
\[
 s_\lambda=\sup_{z\in E}\int_E|k_\lambda(z,z')|
 \,\dd m_\lambda(z'),\qquad
 t_\lambda=\sup_{z'\in E}\int_E|k_\lambda(z,z')|
 \,\dd m_\lambda(z).
\]
For \(f\in L^2(E,m_\lambda)\), the Cauchy--Schwarz inequality with respect to the
measure \(|k_\lambda(z,z')|\,\dd m_\lambda(z')\) implies
\[
 |\mathcal I_\lambda f(z)|^2
 \le s_\lambda
 \int_E|k_\lambda(z,z')|\,|f(z')|^2\,\dd m_\lambda(z').
\]
After integration in \(z\) and an application of Fubini's theorem,
\[
 \|\mathcal I_\lambda f\|_{L^2(E,m_\lambda)}^2
 \le s_\lambda t_\lambda\|f\|_{L^2(E,m_\lambda)}^2.
\]
Thus
\[
 \|\mathcal I_\lambda\|_{\mathrm{op}}
 \le(s_\lambda t_\lambda)^{1/2}.
\]

 \eqref{eq:co:xi} implies
\[
 |k_\lambda(z,z')|
 \le\|F\|_\infty\|\nabla G\|_\infty
 |z-z'|\ee^{-|z-z'|^2/(2\lambda)}.
\]
For fixed \(z\), put first \(r=z'-z\) and then
\(r=\sqrt\lambda\,y\).  Since \(E\simeq\mathbb R^{2r_E}\),
\begin{align*}
 \int_E|k_\lambda(z,z')|\,\dd m_\lambda(z')
 &\le\|F\|_\infty\|\nabla G\|_\infty
  (\pi\lambda)^{-r_E}
  \int_{\mathbb R^{2r_E}}|r|\ee^{-|r|^2/(2\lambda)}\,\dd r\\
 &=C_E\|F\|_\infty\|\nabla G\|_\infty\lambda^{1/2}.
\end{align*}
The same calculation with \(z'\) fixed gives the identical estimate
for \(t_\lambda\).  Schur's bound therefore implies
\eqref{eq:AW-product}.  
\end{proof}

\begin{proposition}
\label{prop:equilibrium-identification}
For \(F,G\in\Acal_{\rm cyl}\), as \(\lambda\downarrow0\), one has
\begin{align}
 \Tr(\Op_{\lambda,\gH}^{\rm A}F\Gamma_\lambda)
 &\longrightarrow \int F\,\dd\nu,                                      \label{eq:static-mean}\\
 \norm{
 \Op_{\lambda,\gH}^{\rm A}F\Op_{\lambda,\gH}^{\rm A}G
 -\Op_{\lambda,\gH}^{\rm A}(FG)}_{2,\lambda}
 &\longrightarrow0,                                                     \label{eq:static-product}\\
 \ip{\Op_{\lambda,\gH}^{\rm A}F}{\Op_{\lambda,\gH}^{\rm A}G}_\lambda
 &\longrightarrow \int\overline F G\,\dd\nu.                            \label{eq:equilibrium-inner-product}
\end{align}
\end{proposition}

\begin{proof}
Choose \(L\) large enough that both \(F\) and \(G\) depend only on the
Fourier modes in \(E:=E_L=\Pi_L\gH\).  There are functions
\(f,g\in C_b^1(E)\) such that
$ F(u)=f(\Pi_Lu), G(u)=g(\Pi_Lu).$
The factorization of the coherent POVM
gives
\[
 \Op_{\lambda,\gH}^{\rm A}F
 =\Op_{\lambda,E}^{\rm A}(f)\otimes\1_{\Fock(E^\perp)},
 \qquad
 \Op_{\lambda,\gH}^{\rm A}G
 =\Op_{\lambda,E}^{\rm A}(g)\otimes\1_{\Fock(E^\perp)}.
\]

By \eqref{eq:lower-symbol-pairing} and \eqref{conver}, we obtain
\[
 \Tr(\Op_{\lambda,\gH}^{\rm A}F\Gamma_\lambda)=\int_E f\,\dd\mu_{\lambda,L}
 \longrightarrow\int_E f\,\dd(\Pi_L)_\#\nu=\int F\,\dd\nu.
\]
Thus \eqref{eq:static-mean} follows.

Since \(\Gamma_\lambda\) is a state,
\[
 \norm{
 \Op_{\lambda,\gH}^{\rm A}F\Op_{\lambda,\gH}^{\rm A}G
 -\Op_{\lambda,\gH}^{\rm A}(FG)}_{2,\lambda}
 \le
 \norm{
 \Op_{\lambda,\gH}^{\rm A}F\Op_{\lambda,\gH}^{\rm A}G
 -\Op_{\lambda,\gH}^{\rm A}(FG)}_{\mathrm{op}}.
\]
By Lemma~\ref{lem:AW-product}, we obtain \eqref{eq:static-product}. 

Using
\((\Op_{\lambda,\gH}^{\rm A}F)^*=\Op_{\lambda,\gH}^{\rm A}\overline F\), we write
\[
 \ip{\Op_{\lambda,\gH}^{\rm A}F}{\Op_{\lambda,\gH}^{\rm A}G}_\lambda-\int\overline FG\,\dd\nu
 =\omega_\lambda\!\left(\Op_{\lambda,\gH}^{\rm A}\overline F\Op_{\lambda,\gH}^{\rm A}G
 -\Op_{\lambda,\gH}^{\rm A}(\overline FG)\right)
 +\omega_\lambda\!\left(\Op_{\lambda,\gH}^{\rm A}(\overline FG)\right)-\int\overline FG\,\dd\nu.
\]
The first term tends to zero by \eqref{eq:static-product} and
the Cauchy--Schwarz inequality in the state \(\Gamma_\lambda\); the second tends to
zero by \eqref{eq:static-mean}.  Thus
\eqref{eq:equilibrium-inner-product} follows.
\end{proof}

\subsubsection{Identified convergence under multiplication}
\label{subsec:varying-Hilbert-spaces}

The Hilbert spaces \(L^2(\Gamma_\lambda)\) vary with \(\lambda\), so
vectors in different spaces cannot be compared by the usual notions of
strong or weak convergence.  We compare them through the anti-Wick
quantizations of cylinder functions and define the corresponding
identified convergence following Kuwae--Shioya (see \cite[Section~2.2]{KuwaeShioya03}).  The inner-product convergence
in \eqref{eq:equilibrium-inner-product}, and in particular the resulting
norm convergence, makes this comparison possible.  We then show that
identified convergence is preserved by left multiplication with bounded
cylinder observables and by finite products of such uniformly bounded
operators.

Recall \(\Hcal_\lambda=L^2(\Gamma_\lambda)\), and let
\[
 \Hcal=L^2(\nu),\qquad
 \Acal_{\rm cyl}\subset\Hcal.
\]
In this abstract varying-Hilbert-space setting,
\(\|\cdot\|_{\Hcal_\lambda}=\|\cdot\|_{2,\lambda}\) and
\(\|\cdot\|_{\Hcal}=\|\cdot\|_{L^2(\nu)}\).
The maps \(\Op_{\lambda,\gH}^{\rm A}\) are defined on \(\Acal_{\rm cyl}\).

Taking \(F=G=g\) in \eqref{eq:equilibrium-inner-product}, we find that
 $\bigl(\Hcal_\lambda,\Hcal,\Acal_{\rm cyl},\Op_{\lambda,\gH}^{\rm A}\bigr)$
satisfies the basic hypothesis in the convergence theory for varying
Hilbert spaces of Kuwae--Shioya (see  \cite[Section~2.2]{KuwaeShioya03}), namely, for $g\in\Acal_{\rm cyl}$
\[
 \norm{\Op_{\lambda,\gH}^{\rm A}g}_{\Hcal_\lambda}
 \longrightarrow \norm{g}_{\Hcal}
 \qquad \lambda\downarrow0.
\]
We use the Kuwae--Shioya strong convergence in an equivalent
\(\varepsilon\)-form.  Weak convergence is defined by bounded testing on
the dense subset \(\Acal_{\rm cyl}\).  Its equivalence with the
Kuwae--Shioya definition is given in
Lemma~\ref{lem:identified-calculus}.

\begin{definition}
\label{def:identified}
For vectors \(x_\lambda\in\Hcal_\lambda\), we write
\(x_\lambda\to_{\mathrm{id}}x\) as \(\lambda\downarrow0\) if, for
every \(\varepsilon>0\), there exists
\(x^\varepsilon\in\Acal_{\rm cyl}\) such that
\[
 \norm{x^\varepsilon-x}_{\Hcal}<\varepsilon,
 \qquad
 \limsup_{\lambda\downarrow0}
 \norm{x_\lambda-\Op_{\lambda,\gH}^{\rm A}x^\varepsilon}_{\Hcal_\lambda}<\varepsilon.
\]
We write \(x_\lambda\rightharpoonup_{\mathrm{id}}x\) as
\(\lambda\downarrow0\) if
\[
 \limsup_{\lambda\downarrow0}\norm{x_\lambda}_{\Hcal_\lambda}<\infty
\]
and, for every $g\in\Acal_{\rm cyl}$,
\[
 \ip{\Op_{\lambda,\gH}^{\rm A}g}{x_\lambda}_\lambda
 \longrightarrow
 \ip{g}{x}_{\Hcal}.
\]
\end{definition}

\begin{lemma}
\label{lem:identified-calculus}
Assume \eqref{eq:equilibrium-inner-product}.  Then:
\begin{enumerate}[label=\textup{(\roman*)},leftmargin=2.2em]
\item
\[
 x_\lambda\to_{\mathrm{id}}x
 \quad\Longleftrightarrow\quad
 \begin{cases}
 x_\lambda\rightharpoonup_{\mathrm{id}}x,\\
 \norm{x_\lambda}_{\Hcal_\lambda}\longrightarrow\norm{x}_{\Hcal};
 \end{cases}
\]
\item if \(x_\lambda\to_{\mathrm{id}}x\) and
\(y_\lambda\to_{\mathrm{id}}y\), then
\[
 \ip{x_\lambda}{y_\lambda}_\lambda
 \longrightarrow
 \ip{x}{y}_{\Hcal};
\]
\item if \(x_\lambda\to_{\mathrm{id}}x\) and
\(y_\lambda\to_{\mathrm{id}}x\), then
\[
 \norm{x_\lambda-y_\lambda}_{\Hcal_\lambda}\longrightarrow0.
\]
\item Let \(T_\lambda\in\mathcal B(\Hcal_\lambda)\) and
\(T\in\mathcal B(\Hcal)\), with
\(
 \sup_{0<\lambda\leq\lambda_0}\norm{T_\lambda}_{\mathrm{op}}<\infty.
\)
If
\(
 T_\lambda \Op_{\lambda,\gH}^{\rm A}g\to_{\mathrm{id}}Tg
\)
as \(\lambda\downarrow0\) for every \(g\in\mathcal A_{\rm cyl}\), 
then
\[
 T_\lambda x_\lambda\to_{\mathrm{id}}Tx
 \quad\text{whenever}\quad x_\lambda\to_{\mathrm{id}}x.
\]
\item Let \(x_\lambda,x_{\lambda,Q}\in\Hcal_\lambda\) and
\(x,x_Q\in\Hcal\).  Assume that, for every fixed \(Q\geq1\),
\[
x_{\lambda,Q}\to_{\mathrm{id}}x_Q
\qquad \lambda\downarrow0,
\]
and that
\(
\norm{x_Q-x}_{\Hcal}\longrightarrow0
\)
as \(Q\to\infty\), as well as
\(
\lim_{Q\to\infty}\limsup_{\lambda\downarrow0}
\norm{x_\lambda-x_{\lambda,Q}}_{\Hcal_\lambda}=0.
\)
Then
\[
x_\lambda\to_{\mathrm{id}}x
\qquad \lambda\downarrow0.
\]
\end{enumerate}
\end{lemma}

\begin{proof}
The strong convergence in Definition~\ref{def:identified} is the
\(\varepsilon\)-form of \cite[Definition~2.4]{KuwaeShioya03}.  It remains to
identify the weak convergence.  We first show that \(x_\lambda\rightharpoonup_{\mathrm{id}}x\) in Definition~\ref{def:identified} is equivalent
to the weak convergence of \cite[Definition~2.5]{KuwaeShioya03}.

Assume first that
\(x_\lambda\rightharpoonup_{\mathrm{id}}x\) as in Definition~\ref{def:identified}.
Let \(y_\lambda\to_{\mathrm{id}}y\).  Given \(\varepsilon>0\), choose
\(g\in\Acal_{\rm cyl}\) such that
\[
 \norm{g-y}_{\Hcal}<\varepsilon,
 \qquad
 \limsup_{\lambda\downarrow0}
 \norm{y_\lambda-\Op_{\lambda,\gH}^{\rm A}g}_{\Hcal_\lambda}<\varepsilon.
\]
Then
\begin{align*}
 &\left|
   \ip{y_\lambda}{x_\lambda}_\lambda-\ip{y}{x}_{\Hcal}
  \right|\\
 &\quad\leq
 \norm{y_\lambda-\Op_{\lambda,\gH}^{\rm A}g}_{\Hcal_\lambda}
 \norm{x_\lambda}_{\Hcal_\lambda}
 +\left|
   \ip{\Op_{\lambda,\gH}^{\rm A}g}{x_\lambda}_\lambda-\ip{g}{x}_{\Hcal}
  \right|
 +\norm{g-y}_{\Hcal}\norm{x}_{\Hcal}.
\end{align*}
Taking the upper limit as \(\lambda\downarrow0\), and then letting
\(\varepsilon\downarrow0\), implies
\[
 \ip{y_\lambda}{x_\lambda}_\lambda
 \longrightarrow
 \ip{y}{x}_{\Hcal}.
\]
Thus \(x_\lambda\) converges weakly in the sense of
\cite[Definition~2.5]{KuwaeShioya03}.

Conversely, suppose that \(x_\lambda\) converges weakly in the
Kuwae--Shioya sense from \cite[Definition~2.5]{KuwaeShioya03}.  For every \(g\in\Acal_{\rm cyl}\), one has
\(\Op_{\lambda,\gH}^{\rm A}g\to_{\mathrm{id}}g\), and hence
\[
 \ip{\Op_{\lambda,\gH}^{\rm A}g}{x_\lambda}_\lambda
 \longrightarrow
 \ip{g}{x}_{\Hcal}.
\]
The uniform norm bound follows from
\cite[Lemma~2.3]{KuwaeShioya03}.  This proves the equivalence of the two
weak formulations.

Part \textup{(i)} is now \cite[Lemma~2.3]{KuwaeShioya03}, part
\textup{(ii)} is \cite[Lemma~2.1\textup{(4)}]{KuwaeShioya03}, and part
\textup{(iii)} is \cite[Lemma~2.1\textup{(6)}]{KuwaeShioya03}.

For part~\textup{(iv)}, let
$C_T=1+\sup_{0<\lambda\leq\lambda_0}\norm{T_\lambda}_{\mathrm{op}}
   +\norm{T}_{\mathrm{op}}.$
Given \(\varepsilon>0\), choose \(g\in\Acal_{\rm cyl}\) such that
\[
 \norm{x-g}_{\Hcal}<\frac{\varepsilon}{3C_T},
 \qquad
 \limsup_{\lambda\downarrow0}
 \norm{x_\lambda-\Op_{\lambda,\gH}^{\rm A}g}_{\Hcal_\lambda}
 <\frac{\varepsilon}{3C_T}.
\]
The assumed convergence on the dense cylinder subspace gives
\(\widetilde g\in\Acal_{\rm cyl}\) such that
\[
 \norm{Tg-\widetilde g}_{\Hcal}<\frac\varepsilon3,
 \qquad
 \limsup_{\lambda\downarrow0}
 \norm{T_\lambda \Op_{\lambda,\gH}^{\rm A}g
       -\Op_{\lambda,\gH}^{\rm A}\widetilde g}_{\Hcal_\lambda}<\frac\varepsilon3.
\]
Consequently,
\[
 \norm{Tx-\widetilde g}_{\Hcal}
 \leq\norm{T}_{\mathrm{op}}\norm{x-g}_{\Hcal}+\norm{Tg-\widetilde g}_{\Hcal}<\varepsilon,
\]
and
\begin{align*}
 &\limsup_{\lambda\downarrow0}
 \norm{T_\lambda x_\lambda-\Op_{\lambda,\gH}^{\rm A}\widetilde g}_{\Hcal_\lambda}\notag\\
 &\leq
 \sup_{0<\lambda\leq\lambda_0}\norm{T_\lambda}_{\mathrm{op}}
 \limsup_{\lambda\downarrow0}
 \norm{x_\lambda-\Op_{\lambda,\gH}^{\rm A}g}_{\Hcal_\lambda}
 +\limsup_{\lambda\downarrow0}
 \norm{T_\lambda \Op_{\lambda,\gH}^{\rm A}g
       -\Op_{\lambda,\gH}^{\rm A}\widetilde g}_{\Hcal_\lambda}<\varepsilon.
\end{align*}
This is the identified convergence in part~\textup{(iv)}.

For part~\textup{(v)}, let \(\varepsilon>0\).  Choose \(Q\) so large
that
\[
\norm{x_Q-x}_{\Hcal}<\frac{\varepsilon}{3},
\qquad
\limsup_{\lambda\downarrow0}
\norm{x_\lambda-x_{\lambda,Q}}_{\Hcal_\lambda}
<\frac{\varepsilon}{3}.
\]
Since \(x_{\lambda,Q}\to_{\mathrm{id}}x_Q\), there exists
\(g\in\Acal_{\rm cyl}\) such that
\[
\norm{g-x_Q}_{\Hcal}<\frac{\varepsilon}{3},
\qquad
\limsup_{\lambda\downarrow0}
\norm{x_{\lambda,Q}
	-\Op_{\lambda,\gH}^{\rm A}g}_{\Hcal_\lambda}
<\frac{\varepsilon}{3}.
\]
Consequently,
\[
\norm{g-x}_{\Hcal}<\varepsilon
\]
and
\[
\limsup_{\lambda\downarrow0}
\norm{x_\lambda-\Op_{\lambda,\gH}^{\rm A}g}_{\Hcal_\lambda}
<\varepsilon.
\]
Definition~\ref{def:identified} therefore implies
\(x_\lambda\to_{\mathrm{id}}x\).
\end{proof}

\begin{theorem}
\label{thm:static-argument}
Let \(F\in\Acal_{\rm cyl}\).  If
\(x_\lambda\to_{\mathrm{id}}x\), then
\begin{equation}\label{eq:static-left-module}
 \Op_{\lambda,\gH}^{\rm A}F\,x_\lambda\to_{\mathrm{id}} Fx.
\end{equation}
Here the product on the left is the extended left action on
\(\Hcal_\lambda\) defined in
Subsection~\ref{subsec:quantum-weighted-generator}, while \(Fx\) is pointwise
multiplication in \(\Hcal\).  Moreover,
\[
 \norm{\Op_{\lambda,\gH}^{\rm A}F\,x_\lambda}_{\Hcal_\lambda}
 \le \norm{F}_\infty\norm{x_\lambda}_{\Hcal_\lambda}.
\]
\end{theorem}

\begin{proof}
 Positivity and unitality of the
anti-Wick quantization, together with the usual bound for multiplication
operators on \(L^2(\nu)\), imply
\[
\bigl\|\operatorname{Op}^{\mathrm A}_{\lambda,\gH}F
\bigr\|_{\mathrm{op}}
\leq \|F\|_\infty,
\qquad
\|F\|_{\mathcal B(L^2(\nu))}
\leq \|F\|_\infty.
\]
For \(g\in\Acal_{\rm cyl}\), \eqref{eq:static-product} yields
\[
 \|\Op_{\lambda,\gH}^{\rm A}F\,\Op_{\lambda,\gH}^{\rm A}g
       -\Op_{\lambda,\gH}^{\rm A}(Fg)\|_{2,\lambda}\longrightarrow0.
\]
Since \(\Op_{\lambda,\gH}^{\rm A}(Fg)\to_{\mathrm{id}}Fg\), the last convergence
and Definition~\ref{def:identified} imply
\[
 \Op_{\lambda,\gH}^{\rm A}F\,\Op_{\lambda,\gH}^{\rm A}g\to_{\mathrm{id}}Fg.
\]
Using Lemma~\ref{lem:identified-calculus}\textup{(iv)}, we obtain
\eqref{eq:static-left-module} for every
\(x_\lambda\to_{\mathrm{id}}x\).  The required norm bound follows from
the first operator-norm estimate.
\end{proof}

\begin{lemma}
\label{lem:identified-products}
Let \(J\ge1\).  For \(1\le j\le J\), let
\(T_{j,\lambda}\in\mathcal B(\Hcal_\lambda)\) and
\(T_j\in\mathcal B(\Hcal)\).  Assume
\[
 \sup_{0<\lambda\le\lambda_0}\norm{T_{j,\lambda}}_{\mathrm{op}}<\infty
\]
and that the operators preserve identified convergence: whenever
\(x_\lambda\to_{\mathrm{id}}x\), one has
\(T_{j,\lambda}x_\lambda\to_{\mathrm{id}}T_jx\).  Then
\[
 T_{1,\lambda}\cdots T_{J,\lambda}x_\lambda
 \to_{\mathrm{id}} T_1\cdots T_Jx
\]
for every \(x_\lambda\to_{\mathrm{id}}x\).  If additionally
\(y_\lambda\to_{\mathrm{id}}y\), then
\[
 \ip{y_\lambda}{T_{1,\lambda}\cdots T_{J,\lambda}x_\lambda}_\lambda
 \longrightarrow
 \ip{y}{T_1\cdots T_Jx}_{\Hcal}.
\]
\end{lemma}

\begin{proof}
Applying the assumed property successively to
\(T_{J,\lambda},T_{J-1,\lambda},\ldots,T_{1,\lambda}\), we obtain
\[
T_{1,\lambda}\cdots T_{J,\lambda}x_\lambda
\to_{\mathrm{id}} T_1\cdots T_Jx.
\]
The pairing convergence follows from
Lemma~\ref{lem:identified-calculus}\textup{(ii)}.
\end{proof}

\subsection{Moment bounds at equilibrium}
\label{subsec:equilibrium-density-tails}
In this subsection we give the moment estimates needed for the
generator convergence in Section~\ref{sec:higher-density-generator}.
We use the self-adjoint real Fourier modes
\(f_{k,\mathrm c}\) and \(f_{k,\mathrm s}\) introduced in \eqref{def:fkcs}.
Choose a set \(\mathbb Z^d_+\) containing one element of
every pair \(\{k,-k\}\), \(k\ne0\), and put
\[
 \mathfrak I
 :=\{0\}\cup\{(k,\mathrm c),(k,\mathrm s):k\in\mathbb Z^d_+\}.
\]
For \(\eta\in\mathfrak I\), define
\[
 (f_\eta,k_\eta,w_\eta)
 =\begin{cases}
   (1,0,\widehat v(0)),&\eta=0,\\
   (f_{k,\mathrm c},k,2\widehat v(k)),&\eta=(k,\mathrm c),\\
   (f_{k,\mathrm s},k,2\widehat v(k)),&\eta=(k,\mathrm s).
  \end{cases}
\]
Thus 
\begin{equation}\label{eq:real-density-square}
 \frac12\sum_{k\in\mathbb Z^d}\widehat v(k)
 B_{\lambda,k}B_{\lambda,-k}
 =
 \frac12\sum_{\eta\in\mathfrak I}w_\eta B_\lambda(f_\eta)^2.
\end{equation}
For \(Q\ge1\), set
\begin{equation}\label{eq:finite-rank-real-mode}
 f_{\eta,Q}:=\Pi_Qf_\eta \Pi_Q.
\end{equation}
The corresponding renormalized term is
\[
 \rho_\eta(u)
 :=\lim_{L\to\infty}
 \left(
  \langle \Pi_Lu,f_\eta \Pi_Lu\rangle
  -\Tr(h^{-1}\Pi_Lf_\eta \Pi_L)
 \right)
 \quad\text{in every finite }L^p(\mu_0).
\]
Thus \(\rho_0\) agrees with the complex zero mode, while
\(\rho_{(k,\mathrm c)}=\Ree\rho_k\) and
\(\rho_{(k,\mathrm s)}=\Imm\rho_k\).
By \cite[(5.57)]{LNR}, for every fixed \(\lambda>0\)
and every integer \(k\ge1\),
\begin{equation}\label{eq:fixed-lambda-number-moments}
 \Tr(\mathcal N^k\Gamma_\lambda)<\infty.
\end{equation}


For \(d\in\{2,3\}\), let
\begin{equation}\label{def:gd}
 g_d(\lambda)
 :=
 \begin{cases}
  1+|\log\lambda|,&d=2,\\
  \lambda^{-1/2},&d=3,
 \end{cases}
 \qquad 0<\lambda\le\lambda_0.
\end{equation}
Then
\[
 \lambda\Tr(\gamma_{0,\lambda})\le C g_d(\lambda),
 \qquad
 \|\lambda\gamma_{0,\lambda}\|_{\mathfrak S^2}\le C.
\]
We also write
\[
 \mathbb W_\lambda
 =\frac12\sum_{\eta\in\mathfrak I}w_\eta B_\lambda(f_\eta)^2\ge0,
 \qquad
 \lambda\mathbb H_\lambda=\lambda\dG(h)+\mathbb W_\lambda.
\]
We use the following fixed-interaction version of the density-tail
estimate as in \cite[Theorems~7.2, 7.3]{NZZ25}. 
For \(Q\) sufficiently large,
\begin{equation}\label{eq:LNR-uniform-label}
	\begin{aligned}
		&\|B_\lambda(f_\eta)-B_\lambda(f_{\eta,Q})\|_{2,\lambda}^{\,2}\le C_v\Big[
		\|\Pi_Q^\perp h^{-1}\|_{\mathfrak S^2}^{1/2}+
		\lambda^2g_d(\lambda)(1+|k_\eta|^2+|k_\eta|Q)
		+\lambda^2g_d(\lambda)^2\Big].
	\end{aligned}
\end{equation}
In particular, for every fixed \(\eta\),
\[
\lim_{Q\to\infty}\limsup_{\lambda\downarrow0}
\|B_\lambda(f_\eta)-B_\lambda(f_{\eta,Q})\|_{2,\lambda}=0.
\]

\begin{lemma}
\label{lem:common-fourth-density-moment}
Let \(b=b^*\) be a bounded one-particle operator satisfying
\begin{equation}\label{eq:common-b-hypotheses}
 \|b\|_{\mathrm{op}}\le C_0,\qquad
 b\Dom(h)\subset\Dom(h),\qquad
 K(b):=\|[[b,h],b]\|_{\mathrm{op}}<\infty.
\end{equation}
Then, for \(0<\lambda\le\lambda_0\),
\begin{equation}\label{eq:common-fourth-density-bound}
 \Tr\bigl(B_\lambda(b)^4\Gamma_\lambda\bigr)
 \le C\left[1+\lambda^4(1+K(b))^2 g_d(\lambda)^6\right],
\end{equation}
where \(C\) depends only on \(m,C_0\).
\end{lemma}

\begin{proof}
We follow the partition-function perturbation argument in the proof of
\cite[Theorem~9.1]{NZZ25}.  Fix \(\delta>0\) so small that
\(4\delta C_0<m\), and, for \(|\tau|\le2\), let
\[
 h_{b,\tau}:=h-\tau\delta b,
 \qquad
 \gamma_{0,\lambda,b,\tau}:=(\ee^{\lambda h_{b,\tau}}-1)^{-1}.
\]
Then \(h_{b,\tau}\ge h/2\).  By this lower bound, the argument used in
\cite[Section~3]{LNR} and \cite[Theorem~9.1]{NZZ25} implies, uniformly
for \(|\tau|\le2\),
\begin{equation}\label{eq:common-perturbed-covariance}
	\|\lambda\gamma_{0,\lambda,b,\tau}\|_{\mathfrak S^2}
	+g_d(\lambda)^{-1}\Tr(\lambda\gamma_{0,\lambda,b,\tau})
	+\|\lambda\gamma_{0,\lambda,b,\tau}
	-\lambda\gamma_{0,\lambda}\|_{\mathfrak S^1}
	\le C.
\end{equation}

Set
\[
 \mathcal Z_{\lambda,b,\tau}
 :=\Tr\ee^{-\lambda\mathbb H_\lambda+\tau\delta B_\lambda(b)},
 \qquad
 \Gamma_{\lambda,b,\tau}
 :=\mathcal Z_{\lambda,b,\tau}^{-1}
   \ee^{-\lambda\mathbb H_\lambda+\tau\delta B_\lambda(b)},
\]
and define \(\mathcal Z_{0,\lambda,b,\tau}\) and
\(\Gamma_{0,\lambda,b,\tau}\) by replacing \(\lambda\mathbb H_\lambda\)
with \(\lambda\dG(h)\).  The centering in \(B_\lambda(b)\) and
\eqref{eq:common-perturbed-covariance} give
\[
 |\partial_\tau\log\mathcal Z_{0,\lambda,b,\tau}|
 =\delta\left|\Tr\bigl(b(\lambda\gamma_{0,\lambda,b,\tau}
                         -\lambda\gamma_{0,\lambda})\bigr)\right|\le C.
\]
The quasi-free Wick rule, \eqref{eq:common-perturbed-covariance}, and
\(\sum_\eta w_\eta<\infty\) imply
\(\sup_{|\tau|\le2}\Tr(\mathbb W_\lambda\Gamma_{0,\lambda,b,\tau})\le C\).
Since \(\mathbb W_\lambda\ge0\), the Gibbs variational principle, as in
\cite{NZZ25}, yields
\[
 \ee^{-C}\le
 \frac{\mathcal Z_{\lambda,b,\tau}}{\mathcal Z_{0,\lambda,b,\tau}}
 \le1.
\]
Convexity of \(\tau\mapsto\log\mathcal Z_{\lambda,b,\tau}\) therefore implies
\begin{equation}\label{eq:common-perturbed-first-moment}
 \sup_{|\tau|\le1}
 \left|\Tr\bigl(\delta B_\lambda(b)\Gamma_{\lambda,b,\tau}\bigr)\right|
 \le C.
\end{equation}
Since \(\mathcal N\) commutes with \(\mathbb H_\lambda\) and
\(B_\lambda(b)\), the above interacting--free comparison remains
valid after adding
\(\sigma\lambda\mathcal N/g_d(\lambda)\) to the exponent.  On the free
side, this amounts to replacing \(h_{b,\tau}\) by
\(
h_{b,\tau}-\frac{\sigma}{g_d(\lambda)}.
\)
Since \(h_{b,\tau}\ge h/2\) and \(g_d(\lambda)\ge1\), this operator is
bounded below by \(h/4\) when \(\sigma>0\) is fixed sufficiently small.
The same free estimate as above therefore implies
\begin{equation}\label{eq:perturbed-number-exponential}
\sup_{|\tau|\le1}
\Tr\!\left[
\ee^{\sigma\lambda\mathcal N/g_d(\lambda)}
\Gamma_{\lambda,b,\tau}\right]
\le C_\sigma.
\end{equation}
Moreover,
\[
\pm\delta B_\lambda(b)
\le
\delta\lambda\|b\|_{\mathrm{op}}\mathcal N
+\delta\bigl|\lambda\Tr(b\gamma_{0,\lambda})\bigr|\1
\le C\bigl(g_d(\lambda)\1+\lambda\mathcal N\bigr).
\]
Since
\[
\bigl(g_d(\lambda)+\lambda\mathcal N\bigr)^6
\le C_\sigma g_d(\lambda)^6
\ee^{\sigma\lambda\mathcal N/g_d(\lambda)},
\]
we obtain
\begin{equation}\label{eq:common-perturbed-sixth-moment}
	\sup_{|\tau|\le1}
	\Tr\!\left[
	\bigl(g_d(\lambda)\1+\lambda\mathcal N\bigr)^6
	\Gamma_{\lambda,b,\tau}\right]
	\le C g_d(\lambda)^6.
\end{equation}

It remains to verify the double-commutator assumption in
\cite[Theorem~3]{DNN25}.  We have
\[
[[B_\lambda(b),\lambda\dG(h)],B_\lambda(b)]
=\lambda^3\dG([[b,h],b]).
\]
For every \(\eta\in\mathfrak I\), a direct commutator calculation gives
\[
\begin{aligned}
	&[[B_\lambda(b),B_\lambda(f_\eta)^2],B_\lambda(b)]=\lambda^2\Bigl(
	\lambda\dG([[b,f_\eta],b])B_\lambda(f_\eta)
	+B_\lambda(f_\eta)\lambda\dG([[b,f_\eta],b])
	-2\bigl(\lambda\dG([b,f_\eta])\bigr)^2
	\Bigr).
\end{aligned}
\]
Thus we obtain
\[
\pm[[B_\lambda(b),\mathbb W_\lambda],B_\lambda(b)]
\le C\lambda^2(\1+\lambda\mathcal N)^2.
\]
Together with
\[
\pm\lambda^3\dG([[b,h],b])
\le \lambda^3K(b)\mathcal N,
\]
this yields
\begin{equation}\label{eq:common-DNN-double-commutator}
	\begin{aligned}
		&\pm[[\delta B_\lambda(b),\lambda\mathbb H_\lambda],
		\delta B_\lambda(b)]\le
		C\delta^2\lambda^2(1+K(b))
		\bigl(g_d(\lambda)\1+\lambda\mathcal N\bigr)^2.
	\end{aligned}
\end{equation}
We know that
\(g_d(\lambda)\1+\lambda\mathcal N\) commutes with
\(\lambda\mathbb H_\lambda\) and \(\delta B_\lambda(b)\).

We now apply \cite[Theorem~3]{DNN25}, using
\(C(g_d(\lambda)\1+\lambda\mathcal N)\) as the commuting control
operator.  By \eqref{eq:common-perturbed-first-moment},
\eqref{eq:common-perturbed-sixth-moment}, and
\eqref{eq:common-DNN-double-commutator},
\[
\begin{aligned}
	\delta^4\Tr\bigl(B_\lambda(b)^4\Gamma_\lambda\bigr)
	&\le C\Bigl[
	1+\delta^4\lambda^4(1+K(b))^2
	\sup_{|\tau|\le1}
	\Tr\!\left[
	\bigl(g_d(\lambda)\1+\lambda\mathcal N\bigr)^6
	\Gamma_{\lambda,b,\tau}\right]
	\Bigr]\\
	&\le C\left[
	1+\delta^4\lambda^4(1+K(b))^2g_d(\lambda)^6
	\right],
\end{aligned}
\]
which implies 
\eqref{eq:common-fourth-density-bound}.
\end{proof}

\section{Convergence of the quantum generators}
\label{sec:higher-density-generator}

In this section we prove statement~\textup{(G)} of
Subsection~\ref{subsec:proof-outline}, namely, the convergence of the
quantum generators on bounded cylinder observables in dimensions two
and three. For each real Fourier mode, we compute the commutator
with a cylinder observable and show that it is supported on a finite
Fourier space. After a sufficiently large Fourier cutoff is introduced,
the remaining high-frequency part of the density acts on the
complementary Fock factor and commutes with this commutator. Uniform
moment bounds and the equilibrium moment estimate then show that
the products in both orders are small in \(L^2(\Gamma_\lambda)\).
Finally, the weighted summability of the Fourier coefficients of \(v\)
allows us to sum over all interaction modes and identify the limit with
the classical generator \(\Lcal_0\).

Formally, the limiting generator can already be read from
\eqref{eq:exact-quantum-generator-formula}.  For
\(A=\Op_{\lambda,\gH}^{\rm A}F\), the semiclassical symbol correspondences are
\[
 \ii[\dG(h),A]\ \leadsto\ \mathscr D_hF,
 \qquad
 B_{\lambda,k}\ \leadsto\ \rho_k,
 \qquad
 \adop{M_k}A\ \leadsto\ -\ii\mathscr D_{M_k}F.
\]
Thus the two insertions formally give
\[
 \ii[\mathbb H_\lambda,\Op_{\lambda,\gH}^{\rm A}F]
 \ \leadsto\
 \mathscr D_hF
 +\frac12\sum_{k\in\mathbb Z^d}\widehat v(k)
 \bigl(\rho_k\mathscr D_{M_{-k}}F
       +\rho_{-k}\mathscr D_{M_k}F\bigr)
 =\Lcal_0F,
\]
which is \eqref{eq:classical-generator-definition}. At a fixed Fourier cutoff, the leading term is obtained by replacing
the creation and annihilation operators with the corresponding
coordinate functions in \(z\) and \(\bar z\).  All correction terms
arising when these operators are rearranged or moved through an
anti-Wick factor contain an explicit factor \(\lambda\).  We prove
this convergence in the varying weighted Hilbert spaces and then
remove the cutoff.  

\subsection{Commutators and their Fourier support}
\label{subsec:exact-density-calculus}

Recall the centered second quantization \(B_\lambda(b)\) from
Subsection~\ref{subsec:quantum-setting} and the derivation
\(\adop{b}\) from \eqref{eq:quantum-density-derivation}.  We first
derive an exact formula for the commutator with an anti-Wick cylinder
observable.  We then identify the finite Fourier space on which this
commutator acts and show that the density tail beyond a sufficiently
large cutoff commutes with it.

\begin{lemma}
\label{lem:finite-mode-commutator}
Let \(\Pi\) be a finite Fourier projection, let
\(r_\Pi=\operatorname{rank}\Pi\), let
\((\psi_j)_{j=1}^{r_\Pi}\) be the normalized Fourier basis of \(\Pi\gH\), and let
\(G\in C_c^\infty(\Pi\gH)\).
Let \(b=b^*\in\mathcal B(\gH)\) and assume that \(b\Pi\gH\) is
contained in a finite Fourier space.  
Then, 
\begin{equation}\label{eq:finite-mode-recursion}
 [\dG(b),\Op_{\lambda,\gH}^{\rm A}G]
 =
 \Op_{\lambda,\gH}^{\rm A}(\ell_bG)+\mathcal R_b(G),
\end{equation}
where
\begin{equation}\label{eq:finite-mode-symbol-derivation}
 \ell_bG
 =
 \sum_{\mu,j}(\Pi b\Pi)_{\mu j}\bar z_\mu\partial_{\bar z_j}G
 -
 \sum_{\mu,j}(\Pi b\Pi)_{\mu j}z_j\partial_{z_\mu}G
\end{equation}
and
\begin{equation*}
 \mathcal R_b(G)
 =
 \sum_j
 \left[
  a_\lambda^*(\Pi^\perp b \psi_j)
  \Op_{\lambda,\gH}^{\rm A}(\partial_{\bar z_j}G)
  -\Op_{\lambda,\gH}^{\rm A}(\partial_{z_j}G)
   a_\lambda(\Pi^\perp b \psi_j)
 \right].
\end{equation*}
\end{lemma}

\begin{proof}
Write \(z=\sum_{j=1}^{r_\Pi}z_j\psi_j\),
\(\xi_z=\xi_{\lambda,\Pi\gH}(z)\), and
\(\dd m_\lambda(z)=(\pi\lambda)^{-r_\Pi}\dd z\) for $\xi_{\lambda,\Pi\gH}$ from \eqref{eq:coherent-resolution}.  Using \cite[Section~5.2]{LNR}, followed by integration
by parts against the Gaussian factor in \(\xi_z\), we obtain
\begin{align}
 a_\lambda(\psi_j)\Op_{\lambda,\gH}^{\rm A}G
 &=\Op_{\lambda,\gH}^{\rm A}(z_jG),&
 \Op_{\lambda,\gH}^{\rm A}G\,a_\lambda^*(\psi_j)
 &=\Op_{\lambda,\gH}^{\rm A}(\bar z_jG),\label{eq:finite-mode-AW-eigen}\\
 a_\lambda^*(\psi_j)\Op_{\lambda,\gH}^{\rm A}G
 &=\Op_{\lambda,\gH}^{\rm A}(\bar z_jG)
   -\lambda \Op_{\lambda,\gH}^{\rm A}(\partial_{z_j}G),&
 \Op_{\lambda,\gH}^{\rm A}G\,a_\lambda(\psi_j)
 &=\Op_{\lambda,\gH}^{\rm A}(z_jG)
   -\lambda \Op_{\lambda,\gH}^{\rm A}(\partial_{\bar z_j}G),
 \label{eq:finite-mode-AW-integration}
\end{align}
which implies
\begin{equation}\label{eq:finite-mode-basic-field-commutators}
 \lambda^{-1}[a_\lambda(\psi_j),\Op_{\lambda,\gH}^{\rm A}G]
 =\Op_{\lambda,\gH}^{\rm A}(\partial_{\bar z_j}G),
 \qquad
 \lambda^{-1}[a_\lambda^*(\psi_j),\Op_{\lambda,\gH}^{\rm A}G]
 =-\Op_{\lambda,\gH}^{\rm A}(\partial_{z_j}G).
\end{equation}

Since \(b=b^*\), decompose the one-particle operator as
\[
 \Pi^\perp b\Pi=\sum_j|\Pi^\perp b \psi_j\rangle\langle \psi_j|,
 \qquad
 \Pi b\Pi^\perp=\sum_j|\psi_j\rangle\langle \Pi^\perp b \psi_j|.
\]
Their scaled second quantizations give
\begin{equation*}
 \lambda\dG(b)
 =
 \lambda\dG(\Pi b\Pi)+\lambda\dG(\Pi^\perp b\Pi^\perp)
 +\sum_j\left[
  a_\lambda^*(\Pi^\perp b \psi_j)a_\lambda(\psi_j)
  +a_\lambda^*(\psi_j)a_\lambda(\Pi^\perp b \psi_j)
 \right].
\end{equation*}
The operator \(\Op_{\lambda,\gH}^{\rm A}G\) acts on \(\Fock(\Pi\gH)\) and
as the identity on \(\Fock(\Pi^\perp\gH)\).  Therefore
\(\lambda\dG(\Pi^\perp b\Pi^\perp)\),
\(a_\lambda(\Pi^\perp b \psi_j)\), and
\(a_\lambda^*(\Pi^\perp b \psi_j)\) commute with \(\Op_{\lambda,\gH}^{\rm A}G\).

Let \(b_{\mu j}=\langle \psi_\mu,\Pi b\Pi\psi_j\rangle\).  Then
\[
 \lambda\dG(\Pi b\Pi)
 =\sum_{\mu,j}b_{\mu j}a_\lambda^*(\psi_\mu)a_\lambda(\psi_j).
\]
Using \eqref{eq:finite-mode-basic-field-commutators} and then
\eqref{eq:finite-mode-AW-integration},
\begin{align}
 \lambda^{-1}
 [a_\lambda^*(\psi_\mu)a_\lambda(\psi_j),\Op_{\lambda,\gH}^{\rm A}G]
 &=a_\lambda^*(\psi_\mu)\Op_{\lambda,\gH}^{\rm A}(\partial_{\bar z_j}G)
 -\Op_{\lambda,\gH}^{\rm A}(\partial_{z_\mu}G)a_\lambda(\psi_j)\notag\\
 &=\Op_{\lambda,\gH}^{\rm A}(\bar z_\mu\partial_{\bar z_j}G
    -z_j\partial_{z_\mu}G).                                  \label{eq:finite-mode-block-commutator}
\end{align}
Here the second-derivative terms cancel because the Wirtinger
derivatives commute.
Summing in \(\mu,j\) gives \(\Op_{\lambda,\gH}^{\rm A}(\ell_bG)\), where
\(\ell_b\) is defined in \eqref{eq:finite-mode-symbol-derivation}.
For the off-diagonal blocks,
\begin{align}
 \lambda^{-1}
 [a_\lambda^*(\Pi^\perp b \psi_j)a_\lambda(\psi_j),\Op_{\lambda,\gH}^{\rm A}G]
 &=a_\lambda^*(\Pi^\perp b \psi_j)\Op_{\lambda,\gH}^{\rm A}(\partial_{\bar z_j}G),\notag\\
 \lambda^{-1}
 [a_\lambda^*(\psi_j)a_\lambda(\Pi^\perp b \psi_j),\Op_{\lambda,\gH}^{\rm A}G]
 &=-\Op_{\lambda,\gH}^{\rm A}(\partial_{z_j}G)a_\lambda(\Pi^\perp b \psi_j).          \label{eq:finite-mode-off-diagonal-commutator}
\end{align}
Combining~\eqref{eq:finite-mode-block-commutator} and
\eqref{eq:finite-mode-off-diagonal-commutator} we derive
\eqref{eq:finite-mode-recursion}.

\end{proof}

We now fix a finite Fourier projection \(\Pi\), a function
\(G\in C_c^\infty(\Pi\gH)\), and, for \(Q\ge1\), set
\begin{align}\label{def:Cl}
 F(u)=G(\Pi u),
 \qquad
 C_{\lambda,\eta}
 :=\adop{f_\eta}\Op_{\lambda,\gH}^{\rm A}F,
 \qquad
 C_{\lambda,Q,\eta}
 :=\adop{f_{\eta,Q}}\Op_{\lambda,\gH}^{\rm A}F.
\end{align}
Throughout the remainder of Sections~\ref{sec:higher-density-generator} and~\ref{sec:liouville-correlations}, \(C_F\) denotes a finite
constant depending only on \(\Pi\), on finitely many seminorms of \(G\),
and on the fixed model parameters; its value may change from line to
line.
Let \(\Lambda_\Pi\subset\mathbb Z^d\) be the finite Fourier support of
\(\Pi\gH\), and define
\begin{equation*}
 E_\eta
 :=\operatorname{span}\{
  \mathrm e_p,\mathrm e_{p+k_\eta},\mathrm e_{p-k_\eta}:
  p\in\Lambda_\Pi
 \}.
\end{equation*}
Let \(\Pi_{E_\eta}\) denote the orthogonal projection onto \(E_\eta\).
For the zero mode, \(k_0=0\), so \(E_0=\Pi\gH\).
For \(Q\ge1\), set
\[
 f_{\eta,Q}^{\mathrm t}:=f_\eta-f_{\eta,Q}.
\]

\begin{lemma}
\label{lem:common-exact-support-separation}
The commutator \(C_{\lambda,\eta}\) is supported on \(E_\eta\):
\begin{equation}\label{eq:common-C-support}
 C_{\lambda,\eta}
 =C_{\lambda,\eta}^{E_\eta}
  \otimes\1_{\Fock(E_\eta^\perp)}.
\end{equation}
If \(Q\) is large enough that
\begin{equation}\label{eq:common-support-Q-condition}
 E_\eta+f_\eta E_\eta\subset \Pi_Q\gH,
\end{equation}
then
\begin{align}
 \adop{f_{\eta,Q}}\Op_{\lambda,\gH}^{\rm A}F
 &=C_{\lambda,\eta},                                      \label{eq:common-compressed-commutator-exact}\\
 [B_\lambda(f_{\eta,Q}^{\mathrm t}),C_{\lambda,\eta}]
 &=[B_\lambda(f_{\eta,Q}^{\mathrm t}),C_{\lambda,\eta}^*]
 =0.                                                       \label{eq:common-tail-strong-commutation}
\end{align}
In fact the operators in
\eqref{eq:common-tail-strong-commutation} act on different factors of
the Fock decomposition
\(\Fock(\gH)\simeq\Fock(E_\eta)\otimes\Fock(E_\eta^\perp)\),
and hence commute.
\end{lemma}

\begin{proof}
Let \((\mathrm e_{p_j})_{j=1}^{r_\Pi}\) be the Fourier basis of
\(\Pi\gH\), and put
\(v_{\eta,j}:=\Pi^\perp f_\eta\mathrm e_{p_j}\).  By Lemma~\ref{lem:finite-mode-commutator}
we have
\begin{equation}\label{eq:common-finite-mode-first-commutator}
 C_{\lambda,\eta}
 =\Op_{\lambda,\gH}^{\rm A}(\ell_{f_\eta}G)
 +\sum_{j=1}^{r_\Pi}\left[
  a_\lambda^*(v_{\eta,j})
       \Op_{\lambda,\gH}^{\rm A}(\partial_{\bar z_j}G)
  -\Op_{\lambda,\gH}^{\rm A}(\partial_{z_j}G)
       a_\lambda(v_{\eta,j})
 \right].
\end{equation}
Multiplication by a real sine or cosine sends \(\mathrm e_p\) to a fixed
linear combination of \(\mathrm e_{p+k_\eta}\) and
\(\mathrm e_{p-k_\eta}\).  Hence \(v_{\eta,j}\in E_\eta\), which
proves \eqref{eq:common-C-support}.

Assume \eqref{eq:common-support-Q-condition}.  For \(u\in E_\eta\),
both \(u\) and \(f_\eta u\) belong to \(\Pi_Q\gH\), and therefore
\[
 f_{\eta,Q}^{\mathrm t}u
 =(f_\eta-\Pi_Qf_\eta \Pi_Q)u=0.
\]
Thus
$
f_{\eta,Q}^{\mathrm t}\Pi_{E_\eta}=0.
$
Since \(f_{\eta,Q}^{\mathrm t}\) is self-adjoint, also
$
\Pi_{E_\eta}f_{\eta,Q}^{\mathrm t}=0.
$
Consequently,
\[
 f_{\eta,Q}^{\mathrm t}
 =\Pi_{E_\eta}^\perp f_{\eta,Q}^{\mathrm t}
  \Pi_{E_\eta}^\perp.
\]
Its second quantization, and hence also its centered second
quantization, acts only on \(\Fock(E_\eta^\perp)\).  This proves the
commutation in \eqref{eq:common-tail-strong-commutation}.
The same factorization shows that the tail commutes with
\(\Op_{\lambda,\gH}^{\rm A}F\), and implies
\eqref{eq:common-compressed-commutator-exact}.  The zero-mode case is
the same.
\end{proof}

\subsection{Moment estimates for commutators and density tails}
\label{subsec:common-DNN-mixed-tail}

In this subsection we derive the uniform moment bounds needed for the convergence of the generator.

\begin{lemma}
\label{lem:common-third-coefficient-moment}
Let \(C_{\lambda,\eta}\) be the operator defined in \eqref{def:Cl}.  Uniformly in
\(\eta\),
\begin{equation*}
 \sup_{0<\lambda\le\lambda_0}\sup_{\eta\in\mathfrak I}
 \Tr\!\left[
  (C_{\lambda,\eta}^*C_{\lambda,\eta})^3\Gamma_\lambda
 \right]\le C_F.
\end{equation*}
The same estimate holds with \(C_{\lambda,\eta}\) replaced by
\(C_{\lambda,Q,\eta}\), uniformly in \(Q\) and \(\eta\).
\end{lemma}

\begin{proof}
Let \(E_\eta'=E_\eta\cap(\Pi\gH)^\perp\), with orthogonal projection
\(\Pi_{E_\eta'}\).
 \eqref{eq:common-finite-mode-first-commutator} is the sum of a
bounded cylinder operator and at most \(2r_\Pi\) terms
\(T_j a_\lambda^\#(v_j)\), where \(T_j\) acts on \(\Fock(\Pi\gH)\), the
field acts on \(\Fock(E_\eta')\), $a_\lambda^\#=a_\lambda$ or $a_\lambda^*$, and
\(\|T_j\|_{\mathrm{op}}\le C_F\), \(\|v_j\|_{\gH}\le1\). Using
Lemma~\ref{lem:app-field-sector-bounds} we obtain
\[
 \|C_{\lambda,\eta}C_{\lambda,\eta}^*C_{\lambda,\eta}\psi\|
 \le C_F\|(\1+\lambda\dG(\Pi_{E_\eta'}))^{3/2}\psi\|.
\]
Write \(C_{\lambda,\eta}=\mathsf V_{\lambda,\eta}|C_{\lambda,\eta}|\).  The identity
\(C_{\lambda,\eta}C_{\lambda,\eta}^*C_{\lambda,\eta}
=\mathsf V_{\lambda,\eta}|C_{\lambda,\eta}|^3\), together with the fact that
\(\mathsf V_{\lambda,\eta}\) is
isometric on the initial space of the polar decomposition, implies
\begin{equation*}
 \|C_{\lambda,\eta}C_{\lambda,\eta}^*C_{\lambda,\eta}\psi\|
 =\||C_{\lambda,\eta}|^3\psi\|.
\end{equation*}
Thus 
\[
 \left\|
 (C_{\lambda,\eta}^*C_{\lambda,\eta})^{3/2}
 (\1+\lambda\dG(\Pi_{E_\eta'}))^{-3/2}
 \right\|_{\mathrm{op}}\le C_F.
\]
Let \(\mathsf N_{\lambda,\eta}:=\1+\lambda\dG(\Pi_{E_\eta'})\).
For every fixed finite Fourier space \(E\), with orthogonal projection
\(\Pi_E\), and every integer \(r\ge1\), we have, with
\(\mathcal N_E=\dG(\Pi_E)\),
\begin{equation}\label{eq:finite-mode-number-moment-bound}
 \Tr\!\left[(\lambda\mathcal N_E)^r\Gamma_\lambda\right]
 \leq\sum_{j=1}^r C_j\lambda^{r-j}
  \Tr\!\left[\Pi_E^{\otimes j}\lambda^j
  \Gamma_\lambda^{(j)}\right]
 \le C_{E,r}.
\end{equation}
Here 
the last estimate follows from \eqref{eq:LNR-RDM-uniform}.  In
particular,
\(\mathsf N_{\lambda,\eta}^{3/2}\Gamma_\lambda^{1/2}
\in\mathfrak S^2(\Fock)\).
Thus we obtain
\begin{align*}
 &\Tr\!\left[(C_{\lambda,\eta}^*C_{\lambda,\eta})^3\Gamma_\lambda\right]
 =\left\|
  (C_{\lambda,\eta}^*C_{\lambda,\eta})^{3/2}\Gamma_\lambda^{1/2}
 \right\|_{\mathfrak S^2(\Fock)}^2=\left\|
  (C_{\lambda,\eta}^*C_{\lambda,\eta})^{3/2}
  \mathsf N_{\lambda,\eta}^{-3/2}
  \mathsf N_{\lambda,\eta}^{3/2}\Gamma_\lambda^{1/2}
 \right\|_{\mathfrak S^2(\Fock)}^2\\
 &\le C_F
 \left\|\mathsf N_{\lambda,\eta}^{3/2}\Gamma_\lambda^{1/2}
 \right\|_{\mathfrak S^2(\Fock)}^2
 \le C_F\Big[1+\sum_{\ell=1}^3
  \|\Pi_{E_\eta'}^{\otimes\ell}\|_{\mathfrak S^2}
  \|\lambda^\ell\Gamma_\lambda^{(\ell)}\|_{\mathfrak S^2}\Big]
 \le C_F.
\end{align*}
For the commutator $C_{\lambda,Q,\eta}$, replace \(E_\eta\) by the Fourier span
of \(\Pi\gH+f_{\eta,Q}\Pi\gH\).  Its dimension and the
field-vector norms satisfy the same bounds.
\end{proof}

We shall repeatedly use the following mixed-product estimate.  Whenever
\(b=b^*\) and \(B_\lambda(b)\) commutes with \(C_{\lambda,\eta}\) and
its adjoint, H\"older's inequality for the joint spectral measure of
\(C_{\lambda,\eta}^*C_{\lambda,\eta}\) and \(B_\lambda(b)^2\), followed
by the Cauchy--Schwarz inequality, gives
\begin{align}
 \|B_\lambda(b)C_{\lambda,\eta}\|_{2,\lambda}^2
 &=\|C_{\lambda,\eta}B_\lambda(b)\|_{2,\lambda}^2\le
 \omega_\lambda\!\left((C_{\lambda,\eta}^*C_{\lambda,\eta})^3\right)^{1/3}
 \omega_\lambda\!\left(B_\lambda(b)^2\right)^{1/3}
 \omega_\lambda\!\left(B_\lambda(b)^4\right)^{1/3}.
 \label{eq:common-mixed-holder-chain}
\end{align}
The same estimate holds with \(C_{\lambda,\eta}\) replaced by
\(C_{\lambda,Q,\eta}\).

\begin{lemma}
\label{lem:interaction-label-bounds}
Set
\[
 \begin{aligned}
 S_{\lambda,\eta}
 &:=B_\lambda(f_\eta)C_{\lambda,\eta}
    +C_{\lambda,\eta}B_\lambda(f_\eta),&
 S_{\lambda,Q,\eta}
 &:=B_\lambda(f_{\eta,Q})C_{\lambda,Q,\eta}
    +C_{\lambda,Q,\eta}B_\lambda(f_{\eta,Q}).
 \end{aligned}
\]
Then, for every \(\eta\in\mathfrak I\),
\begin{align}
 \sup_{0<\lambda\le\lambda_0}
 \norm{S_{\lambda,\eta}}_{2,\lambda}
 &\le C_F\langle k_\eta\rangle,
 \label{eq:full-label-bound}\\
 \sup_{Q\ge1}\limsup_{\lambda\downarrow0}
 \norm{S_{\lambda,Q,\eta}}_{2,\lambda}
 &\le C_F.
 \label{eq:compressed-label-bound}
\end{align}
Moreover, for every fixed \(\eta\in\mathfrak I\),
\begin{equation*}
 \lim_{Q\to\infty}\limsup_{\lambda\downarrow0}
 \norm{S_{\lambda,\eta}-S_{\lambda,Q,\eta}}_{2,\lambda}=0.
\end{equation*}
\end{lemma}

\begin{proof}
For \(\eta\in\mathfrak I\), write
\[
 b_\eta=\Pi_{E_\eta}^\perp f_\eta\Pi_{E_\eta}^\perp,
 \qquad r_\eta=f_\eta-b_\eta.
\]
Then \(B_\lambda(b_\eta)\) commutes with
\(C_{\lambda,\eta}\), while \(r_\eta\) has uniformly bounded finite
rank.  The range of \(r_\eta\), as well as
\(E_\eta+f_\eta E_\eta\), is contained in the Fourier span of
\(\{\mathrm e_{p+r k_\eta}:p\in\Lambda_\Pi,\ |r|\le2\}\), whose dimension
is bounded by \(C_F\).  From
\[
 [M_k,h]\mathrm e_n=(h_n-h_{n+k})\mathrm e_{n+k},\qquad
 |h_n-h_{n+k}|\le2|n||k|+|k|^2,
\]
and \([[f_\eta,h],f_\eta]=2|\nabla f_\eta|^2\), we obtain
\[
 K(b_\eta)\le C_F\langle k_\eta\rangle^2.
\]
Lemma~\ref{lem:common-fourth-density-moment} and 
the Cauchy--Schwarz inequality imply
\begin{equation*}
 \sup_{0<\lambda\le\lambda_0}
 \left[
  \Tr(B_\lambda(b_\eta)^4\Gamma_\lambda)
  +\langle k_\eta\rangle^2
   \Tr(B_\lambda(b_\eta)^2\Gamma_\lambda)
 \right]
 \le C_F\langle k_\eta\rangle^4.
\end{equation*}
Applying \eqref{eq:common-mixed-holder-chain} and
Lemma~\ref{lem:common-third-coefficient-moment}, we obtain
\[
 \|B_\lambda(b_\eta)C_{\lambda,\eta}\|_{2,\lambda}^2
 =\|C_{\lambda,\eta}B_\lambda(b_\eta)\|_{2,\lambda}^2
 \le C_F\langle k_\eta\rangle^2.
\]
The operator \(r_\eta\) is supported on the same finite Fourier space,
and its rank, operator norm, and Hilbert--Schmidt norm are uniformly
bounded.  Lemma~\ref{lem:finite-rank-mixed-word} therefore gives
\[
 \|B_\lambda(r_\eta)C_{\lambda,\eta}\|_{2,\lambda}^2
 +\|C_{\lambda,\eta}B_\lambda(r_\eta)\|_{2,\lambda}^2
 \le C_F\Big(1+\sum_{j=1}^3
  \|\lambda^j\Gamma_\lambda^{(j)}\|_{\mathfrak S^2(\gH^{\otimes_s j})}\Big)
 \le C_F.
\]
Taking square roots and using \(f_\eta=b_\eta+r_\eta\), we obtain
\eqref{eq:full-label-bound}.

For $S_{\lambda,Q,\eta}$, let \(E_{Q,\eta}\) be the Fourier span of
\(\Pi\gH+f_{\eta,Q}\Pi\gH\), with orthogonal projection
\(\Pi_{E_{Q,\eta}}\), and set
\[
 b_{Q,\eta}
 =\Pi_{E_{Q,\eta}}^\perp f_{\eta,Q}\Pi_{E_{Q,\eta}}^\perp,
 \qquad r_{Q,\eta}=f_{\eta,Q}-b_{Q,\eta}.
\]
Then \(B_\lambda(b_{Q,\eta})\) commutes with
\(C_{\lambda,Q,\eta}\). The Fourier span of
\(E_{Q,\eta}+f_{\eta,Q}E_{Q,\eta}\)
 has dimension at most
\(9\operatorname{rank}\Pi\).  It supports \(r_{Q,\eta}\), whose rank,
operator norm, and Hilbert--Schmidt norm are bounded uniformly in
\(Q,\eta\).  Thus Lemmas~\ref{lem:common-third-coefficient-moment} and
\ref{lem:finite-rank-mixed-word} give the same uniform estimates for
\(C_{\lambda,Q,\eta}\) and the terms containing
\(B_\lambda(r_{Q,\eta})\).  Furthermore,
\[
 b_{Q,\eta}=\Pi_Qb_{Q,\eta}\Pi_Q,\qquad
 \|b_{Q,\eta}\|_{\mathrm{op}}\le1.
\]
Since \(\Pi_Q\) commutes with \(h\) and
\(\|h\Pi_Q\|_{\mathrm{op}}\le Q^2+m\),
\[
 \|[b_{Q,\eta},h]\|_{\mathrm{op}}
 \le 2\|b_{Q,\eta}\|_{\mathrm{op}}
       \|h\Pi_Q\|_{\mathrm{op}}
 \le 2(Q^2+m).
\]
Consequently,
\[
 K(b_{Q,\eta})
 \le 2\|[b_{Q,\eta},h]\|_{\mathrm{op}}
       \|b_{Q,\eta}\|_{\mathrm{op}}
 \le 4(Q^2+m).
\]
For each fixed \(Q\),
\[
 (1+K(b_{Q,\eta}))^2\lambda^4 g_d(\lambda)^6
 \longrightarrow0\qquad(\lambda\downarrow0),
\]
for $g_d(\lambda)$ from \eqref{def:gd}.
Taking first \(\limsup_{\lambda\downarrow0}\) for each fixed \(Q\),
Lemma~\ref{lem:common-fourth-density-moment}, the Cauchy--Schwarz inequality,
\eqref{eq:common-mixed-holder-chain}, and the finite-rank estimate imply
\[
 \sup_{Q\ge1}\limsup_{\lambda\downarrow0}
 \norm{S_{\lambda,Q,\eta}}_{2,\lambda}\le C_F.
\]
Thus \eqref{eq:compressed-label-bound} follows.

It remains to prove the approximation result.  For all sufficiently
large \(Q\), Lemma~\ref{lem:common-exact-support-separation} yields
\[
 C_{\lambda,Q,\eta}=C_{\lambda,\eta},\qquad
 [B_\lambda(f_{\eta,Q}^{\mathrm t}),C_{\lambda,\eta}]
 =[B_\lambda(f_{\eta,Q}^{\mathrm t}),C_{\lambda,\eta}^*]=0,
\]
and hence we derive
\[
 S_{\lambda,\eta}-S_{\lambda,Q,\eta}
 =B_\lambda(f_{\eta,Q}^{\mathrm t})C_{\lambda,\eta}
  +C_{\lambda,\eta}B_\lambda(f_{\eta,Q}^{\mathrm t}).
\]
For each fixed \(Q\), we apply Lemma~\ref{lem:common-fourth-density-moment} to the operator \(b=f_{\eta,Q}^{\mathrm t}\), which 
satisfies \eqref{eq:common-b-hypotheses}.
\[
 [f_{\eta,Q},h]=\Pi_Q[f_\eta,h]\Pi_Q.
\]
As above, we obtain
\[
 K(b)\le C\bigl(1+|k_\eta|Q+|k_\eta|^2\bigr).
\]
Since
\(\lambda^4g_2(\lambda)^6\to0\) and
\(\lambda^4g_3(\lambda)^6=\lambda\to0\) as \(\lambda\downarrow0\),
Lemma~\ref{lem:common-fourth-density-moment} implies
\[
 \sup_{Q\ge1}\limsup_{\lambda\downarrow0}
 \omega_\lambda\!\left(B_\lambda(f_{\eta,Q}^{\mathrm t})^4\right)
 \le C.
\]
Moreover, 
\eqref{eq:LNR-uniform-label} yields
\[
 \limsup_{\lambda\downarrow0}
 \omega_\lambda\!\left(B_\lambda(f_{\eta,Q}^{\mathrm t})^2\right)
 \le C_v
 \|\Pi_Q^\perp h^{-1}\|_{\mathfrak S^2}^{1/2}\longrightarrow0
 \qquad Q\to\infty.
\]
The last two bounds, together with
\eqref{eq:common-mixed-holder-chain} and
Lemma~\ref{lem:common-third-coefficient-moment}, prove the required
approximation. 
\end{proof}

\begin{corollary}
\label{cor:full-compressed-interaction-tail}
For every \(0<\lambda\le\lambda_0\), the series
\(\sum_{\eta\in\mathfrak I}w_\eta S_{\lambda,\eta}\) converges
absolutely in \(\Hcal_\lambda\).  Moreover,
\begin{equation}\label{eq:full-compressed-tail}
 \lim_{Q\to\infty}\limsup_{\lambda\downarrow0}
 \norm{
  \sum_{\eta\in\mathfrak I}w_\eta
  (S_{\lambda,\eta}-S_{\lambda,Q,\eta})
 }_{2,\lambda}=0.
\end{equation}
\end{corollary}

\begin{proof}
Let \(\varepsilon>0\), and choose a finite set \(\mathfrak I_N\) so
that
\[
 2C_F\sum_{\eta\notin\mathfrak I_N}
 w_\eta\langle k_\eta\rangle<\frac{\varepsilon}{2}.
\]
Using Lemma~\ref{lem:interaction-label-bounds}, we obtain
\[
 \begin{aligned}
 &\limsup_{\lambda\downarrow0}
  \sum_{\eta\notin\mathfrak I_N}
  w_\eta\norm{S_{\lambda,\eta}}_{2,\lambda}
 +\sup_{Q\ge1}\limsup_{\lambda\downarrow0}
  \sum_{\eta\notin\mathfrak I_N}
  w_\eta\norm{S_{\lambda,Q,\eta}}_{2,\lambda}
 <\frac{\varepsilon}{2}.
 \end{aligned}
\]
For the fixed finite set \(\mathfrak I_N\), the last assertion of
Lemma~\ref{lem:interaction-label-bounds} implies
\[
 \lim_{Q\to\infty}\limsup_{\lambda\downarrow0}
 \sum_{\eta\in\mathfrak I_N}w_\eta
 \norm{S_{\lambda,\eta}-S_{\lambda,Q,\eta}}_{2,\lambda}=0.
\]
The triangle inequality implies \eqref{eq:full-compressed-tail}.
\end{proof}

We now sum the estimates over the Fourier modes for the fixed cylinder
\(F(u)=G(\Pi u)\). Since \(\dG(\Pi^\perp h\Pi^\perp)\) commutes with
\(\Op_{\lambda,\gH}^{\rm A}F\), we obtain
\[
 [\dG(h),\Op_{\lambda,\gH}^{\rm A}F]
 = [\dG(\Pi h\Pi),\Op_{\lambda,\gH}^{\rm A}F].
\]
By Lemma~\ref{lem:finite-mode-commutator} and
\(\ell_{\Pi h\Pi}G=-\ii\mathscr D_hF\), we get
\begin{equation}\label{eq:fixed-Q-free-part}
 \ii[\dG(h),\Op_{\lambda,\gH}^{\rm A}F]
 =\Op_{\lambda,\gH}^{\rm A}(\mathscr D_hF).
\end{equation}
 Define
\begin{equation}\label{eq:candidate-generator}
	\begin{aligned}
 R_\lambda(F)
 &:=
 \ii[\dG(h),\Op_{\lambda,\gH}^{\rm A}F]
 +\frac{\ii}{2}\sum_{\eta\in\mathfrak I}
   w_\eta S_{\lambda,\eta},                                  \\
 R_{\lambda,Q}(F)
 &:=
 \ii[\dG(h),\Op_{\lambda,\gH}^{\rm A}F]
 +\frac{\ii}{2}\sum_{\eta\in\mathfrak I}
   w_\eta S_{\lambda,Q,\eta}.                                \end{aligned}
\end{equation}
 Lemma~\ref{lem:interaction-label-bounds} and
\(\sum_\eta w_\eta\langle k_\eta\rangle<\infty\) show that the 
series in \eqref{eq:candidate-generator} converges absolutely in \(\Hcal_\lambda\). 

Using Corollary~\ref{cor:full-compressed-interaction-tail}, we obtain
\begin{equation}\label{eq:generator-tail}
 \lim_{Q\to\infty}\limsup_{\lambda\downarrow0}
 \norm{R_\lambda(F)-R_{\lambda,Q}(F)}_{2,\lambda}=0.
\end{equation}
Moreover, \eqref{eq:fixed-Q-free-part},
\eqref{eq:full-label-bound}, and the weighted summability of
\((w_\eta)\) imply
\begin{equation}\label{eq:candidate-generator-bound}
 \sup_{0<\lambda\le\lambda_0}
 \norm{R_\lambda(F)}_{2,\lambda}<\infty.
\end{equation}

\subsection{Generator convergence on cylinder observables}

In this subsection we show that the Fourier-truncated classical generators
\(\mathcal L_{0,Q}F\) converge to \(\mathcal L_0F\) as \(Q\to\infty\).
For each fixed Fourier cutoff \(Q\), we then prove the semiclassical
convergence of 
$
R_{\lambda,Q}(F)$ to  $\mathcal L_{0,Q}F$. 
\eqref{eq:generator-tail} allows us to remove the Fourier
cutoff.  Finally, a particle-number truncation argument identifies 
\(R_\lambda(F)\) with the closed generator on
\(\Op_{\lambda,\gH}^{\rm A}F\).

\begin{lemma}
\label{lem:classical-compression}
Recall $f_{\eta,Q}=\Pi_Qf_\eta \Pi_Q$ and define the finite-dimensional centered
quadratic chaos
\[
 \rho_{Q,\eta}(u)
 :=\langle u,f_{\eta,Q}u\rangle-\Tr(h^{-1}f_{\eta,Q}).
\]
Set
\[
 \Lcal_{0,Q}F
 :=\mathscr D_hF
  +\sum_{\eta\in\mathfrak I}
    w_\eta\rho_{Q,\eta}\mathscr D_{f_{\eta,Q}}F.
\]
 Then
\begin{equation}\label{eq:classical-compression}
 \Lcal_{0,Q}F\longrightarrow\Lcal_0F
 \qquad\hbox{in }L^2(\nu),\quad Q\to\infty.
\end{equation}
\end{lemma}

\begin{proof}
By \eqref{eq:classical-generator-definition} and
\eqref{eq:real-density-square},
\[
 \Lcal_0F
 =\mathscr D_hF
  +\sum_{\eta\in\mathfrak I}
    w_\eta\rho_\eta\mathscr D_{f_\eta}F.
\]
\begin{equation}
 \Lcal_0F-\Lcal_{0,Q}F
 =\sum_{\eta\in\mathfrak I}w_\eta
 \bigl(\rho_\eta\mathscr D_{f_\eta}F
       -\rho_{Q,\eta}\mathscr D_{f_{\eta,Q}}F\bigr).
 \label{eq:classical-compression-difference}
\end{equation}

The complex Gaussian Wick isometry and hypercontractivity of the second chaos
\cite[Lemma~5.2]{LNR} imply
\begin{equation*}
 \|\rho_\eta-\rho_{Q,\eta}\|_{L^4(\mu_0)}
 \le C\|h^{-1/2}f_{\eta,Q}^{\mathrm t}h^{-1/2}\|_{\mathfrak S^2(\gH)}.
\end{equation*}
\[
 \|\mathscr D_{f_\eta}F-\mathscr D_{f_{\eta,Q}}F\|_{L^4(\mu_0)}
 \le C_F\|h^{-1/2}f_{\eta,Q}^{\mathrm t}\Pi\|_{\mathfrak S^2(\gH)}.
\]
For fixed \(\eta\), the first right-hand side tends to zero by
the Schatten--H\"older inequality and the strong convergence of \(\Pi_Q\) to $\mathbf{1}$.  The second
vanishes once \(\Pi\gH+f_\eta\Pi\gH\subset \Pi_Q\gH\).  The
same estimates without differences give
\begin{equation}\label{eq:classical-compression-uniform-moments}
 \sup_{\eta,Q}
 \left(
  \|\rho_\eta\|_{L^4(\mu_0)}
  +\|\rho_{Q,\eta}\|_{L^4(\mu_0)}
  +\|\mathscr D_{f_\eta}F\|_{L^4(\mu_0)}
  +\|\mathscr D_{f_{\eta,Q}}F\|_{L^4(\mu_0)}
 \right)\le C_F.
\end{equation}
By H\"older's inequality and
\(\dd\nu/\dd\mu_0\le\mathfrak Z^{-1}\), each summand in
\eqref{eq:classical-compression-difference} tends to zero in
\(L^2(\nu)\) as \(Q\to\infty\).  
\eqref{eq:classical-compression-uniform-moments} bounds it by
\(C_Fw_\eta\).  Since \(\sum_\eta w_\eta<\infty\), dominated
convergence implies
\eqref{eq:classical-compression}.
\end{proof}

\begin{proposition}
\label{prop:fixed-compression}
For every fixed \(Q\ge1\),
\begin{equation}\label{eq:fixed-Q-convergence}
 R_{\lambda,Q}(F)\to_{\mathrm{id}}\Lcal_{0,Q}F
 \qquad \lambda\downarrow0.
\end{equation}
\end{proposition}

\begin{proof}
Fix \(Q\), and write
\[
 \Lambda_Q:=\{p\in\mathbb Z^d:\mathrm e_p\in\Pi_Q\gH\},
 \qquad
 \mathfrak I_Q:=\{\eta\in\mathfrak I:f_{\eta,Q}\ne0\},
 \qquad
 E_Q:=\Pi\gH+\Pi_Q\gH.
\]
If \(f_{\eta,Q}\ne0\), multiplication by \(f_\eta\) connects two
modes in \(\Lambda_Q\).  Hence
\(k_\eta\in\Lambda_Q-\Lambda_Q\), and \(\mathfrak I_Q\) is finite.
All terms in \(R_{\lambda,Q}(F)\) act on \(\Fock(E_Q)\) and as the
identity on \(\Fock(E_Q^\perp)\).  The free term is
\(\Op_{\lambda,\gH}^{\rm A}(\mathscr D_hF)\) by
\eqref{eq:fixed-Q-free-part}.

Let \((\mathrm e_{p_j})_{j=1}^{r_\Pi}\) be the Fourier basis of
\(\Pi\gH\), and set
\[
 v_{\eta,Q,j}:=\Pi^\perp f_{\eta,Q}\mathrm e_{p_j}.
\]
By Lemma~\ref{lem:finite-mode-commutator},
\[
 C_{\lambda,Q,\eta}
 =\Op_{\lambda,\gH}^{\rm A}(\ell_{f_{\eta,Q}}G)
  +\sum_{j=1}^{r_\Pi}\left[
    a_\lambda^*(v_{\eta,Q,j})
      \Op_{\lambda,\gH}^{\rm A}(\partial_{\bar z_j}G)
    -\Op_{\lambda,\gH}^{\rm A}(\partial_{z_j}G)
      a_\lambda(v_{\eta,Q,j})
  \right].
\]
Moreover, by the definition of Wick quantization,
\[
 B_\lambda(f_{\eta,Q})
 =\Op_{\lambda,E_Q}^{\rm W}(b_{\lambda,Q,\eta}),
 \qquad
 b_{\lambda,Q,\eta}(z)
 =\langle z,f_{\eta,Q}z\rangle
  -\lambda\Tr(f_{\eta,Q}\gamma_{0,\lambda}).
\]
With
\(J_{\lambda,Q,\eta}:=
 \Op_{\lambda,\gH}^{\rm A}(\ell_{f_{\eta,Q}}G)\), we write
\begin{align}
 B_\lambda(f_{\eta,Q})C_{\lambda,Q,\eta}
 +C_{\lambda,Q,\eta}B_\lambda(f_{\eta,Q})
 &=B_\lambda(f_{\eta,Q})J_{\lambda,Q,\eta}
   +J_{\lambda,Q,\eta}B_\lambda(f_{\eta,Q})\notag\\
 &\quad+B_\lambda(f_{\eta,Q})
        (C_{\lambda,Q,\eta}-J_{\lambda,Q,\eta})
 +(C_{\lambda,Q,\eta}-J_{\lambda,Q,\eta})
        B_\lambda(f_{\eta,Q}).
 \label{eq:fixed-Q-two-orientations}
\end{align}
The only rearrangements needed in the last line are given by the
scaled canonical commutation relations
\begin{align}
 B_\lambda(f_{\eta,Q})a_\lambda^*(v)
 &=a_\lambda^*(v)B_\lambda(f_{\eta,Q})
   +\lambda a_\lambda^*(f_{\eta,Q}v),
 \label{eq:fixed-Q-order-right-creation}\\
 a_\lambda(v)B_\lambda(f_{\eta,Q})
 &=B_\lambda(f_{\eta,Q})a_\lambda(v)
   +\lambda a_\lambda(f_{\eta,Q}v).
 \label{eq:fixed-Q-order-left-annihilation}
\end{align}
All vectors in these identities belong to \(E_Q\).  After applying
them, every term has the form treated in
Lemma~\ref{lem:finite-mode-mixed-calculus}.  More precisely, we can write
\begin{equation}\label{eq:fixed-Q-normal-template}
 R_{\lambda,Q}(F)
 =\sum_{\ell=1}^{J_Q}
 \Op_{\lambda,E_Q}^{\rm W}(p_{\lambda,Q,\ell}^{\rm L})
 \Bigl(
  \Op_{\lambda,\Pi\gH}^{\rm A}(g_{Q,\ell})
  \otimes\1_{\Fock(E_Q\cap(\Pi\gH)^\perp)}
 \Bigr)
 \Op_{\lambda,E_Q}^{\rm W}(p_{\lambda,Q,\ell}^{\rm R}).
\end{equation}
Here \(J_Q<\infty\), the functions
\(g_{Q,\ell}\in C_c^\infty(\Pi\gH)\) form a finite family, and
\[
 p_{\lambda,Q,\ell}^{\rm L},p_{\lambda,Q,\ell}^{\rm R}
 \in\mathscr P_3(E_Q),
 \qquad
 \deg p_{\lambda,Q,\ell}^{\rm L}
 +\deg p_{\lambda,Q,\ell}^{\rm R}\le3,
\]
for $\mathscr P_3(E_Q)$ introduced in Appendix~\ref{appsubsec:finite-mode-mixed-convergence}. 
The functions \(g_{Q,\ell}\) are independent of \(\lambda\) and have
a common compact support.  The coefficients of the two Wick
polynomials depend on \(\lambda\) only through
\(\lambda\Tr(f_{\eta,Q}\gamma_{0,\lambda})\) and the explicit factors
\(\lambda\) in
\eqref{eq:fixed-Q-order-right-creation}--
\eqref{eq:fixed-Q-order-left-annihilation}.

It remains to identify their limits.  Since \(Q\) is fixed,
\[
 \lambda\Tr(f_{\eta,Q}\gamma_{0,\lambda})
 \longrightarrow\Tr(f_{\eta,Q}h^{-1}),
\]
and hence
\[
 b_{\lambda,Q,\eta}\longrightarrow\rho_{Q,\eta}
 \quad\text{in }\mathscr P_2(E_Q).
\]
For each \(\eta\in\mathfrak I_Q\), define
\(c_{Q,\eta}:E_Q\to\mathbb C\) by
\(c_{Q,\eta}(z):=-\ii\,\mathscr D_{f_{\eta,Q}}F(z)\),
where \(F=G\circ\Pi\). To identify the leading classical function associated with
\(C_{\lambda,Q,\eta}\), we use the coordinate replacements
\[
a_\lambda^*(v)\mapsto \langle z,v\rangle,\qquad
a_\lambda(v)\mapsto \langle v,z\rangle,\qquad
\Op_{\lambda,\gH}^{\rm A}(\phi)\mapsto \phi(\Pi z)
\]
in the previous formula for \(C_{\lambda,Q,\eta}\), and then multiply
the resulting scalar factors.  This implies
\begin{align*}
 c_{Q,\eta}(z)
 &= (\ell_{f_{\eta,Q}}G)(\Pi z)
  +\sum_{j=1}^{r_\Pi}\left[
    \langle z,v_{\eta,Q,j}\rangle\partial_{\bar z_j}G(\Pi z)
    -\partial_{z_j}G(\Pi z)\langle v_{\eta,Q,j},z\rangle
  \right]\\
 &= -\sum_{j=1}^{r_\Pi}
    \partial_{z_j}G(\Pi z)(\Pi f_{\eta,Q}z)_j
  +\sum_{j=1}^{r_\Pi}
    \partial_{\bar z_j}G(\Pi z)
    \overline{(\Pi f_{\eta,Q}z)_j}
 =-\ii\mathscr D_{f_{\eta,Q}}F(z),
\end{align*}
where the last equality follows from
\eqref{eq:complex-Hamiltonian-derivation}.  Thus the terms in
\eqref{eq:fixed-Q-two-orientations} obtained from the first terms on
the right-hand sides of
\eqref{eq:fixed-Q-order-right-creation}--
\eqref{eq:fixed-Q-order-left-annihilation} converge to
\[
 \frac{\ii}{2}
 \bigl(\rho_{Q,\eta}c_{Q,\eta}
       +c_{Q,\eta}\rho_{Q,\eta}\bigr)
 =\rho_{Q,\eta}\mathscr D_{f_{\eta,Q}}F.
\]
Every remaining coefficient contains an explicit factor \(\lambda\)
and therefore tends to zero.  Consequently, for each \(\ell\), the
polynomials in \eqref{eq:fixed-Q-normal-template} converge
coefficientwise in \(\mathscr P_3(E_Q)\), which is defined in Appendix~\ref{appsubsec:finite-mode-mixed-convergence}, and the sum of the limiting
functions is
\[
 \mathscr D_hF
 +\sum_{\eta\in\mathfrak I_Q}
   w_\eta\rho_{Q,\eta}\mathscr D_{f_{\eta,Q}}F
 =\Lcal_{0,Q}F.
\]
Using Lemma~\ref{lem:finite-mode-mixed-calculus}, with
\(E_0=\Pi\gH\), \(E=E_Q\), and \(D=3\), we obtain
\eqref{eq:fixed-Q-convergence}.
\end{proof}
\begin{theorem}
\label{thm:one-generator}
For every \(F\in\Acal_{\rm cyl}\),
\begin{equation}\label{eq:one-generator}
 \Op_{\lambda,\gH}^{\rm A}F\in\Dom(\Lbb_\lambda),\qquad
 \Lbb_\lambda \Op_{\lambda,\gH}^{\rm A}F
 \to_{\mathrm{id}}\Lcal_0F
 \qquad\lambda\downarrow0,
\end{equation}
and
\begin{equation}\label{eq:generator-bound}
 \sup_{0<\lambda\le\lambda_0}
 \norm{\Lbb_\lambda \Op_{\lambda,\gH}^{\rm A}F}_{2,\lambda}<\infty.
\end{equation}
\end{theorem}

\begin{proof}
By linearity, it suffices to take \(F(u)=G(\Pi u)\).  We first prove \eqref{eq:one-generator} and identify the
generator with \(R_\lambda(F)\).  For fixed \(\lambda>0\), set
\[
 \mathsf E_n:=\1_{\{\mathcal N=n\}},
 \qquad
 \mathsf P_N:=\sum_{n=0}^N\mathsf E_n
 =\1_{\{\mathcal N\le N\}}.
\]
Put
\[
 A=\Op_{\lambda,\gH}^{\rm A}F,\qquad R=R_\lambda(F),\qquad
 A_N=\mathsf P_NA\mathsf P_N,\qquad
 R_N=\mathsf P_NR\mathsf P_N.
\]
On sectors with particle number at most \(N\), the bounded commutator
calculation gives \(\ii[\mathbb H_\lambda,A_N]=R_N\). Here the free commutator is
bounded by \eqref{eq:fixed-Q-free-part}, while the particle-number
cutoff makes the other part
bounded. Hence
\[
 \Ubb_\lambda(t)A_N-A_N
 =\int_0^t\Ubb_\lambda(s)R_N\,\dd s
 \quad\text{in }\Hcal_\lambda,
\]
so \(A_N\in\Dom(\Lbb_\lambda)\) and
\(\Lbb_\lambda A_N=R_N\).  Since \(\mathsf P_N\) commutes with
\(\Gamma_\lambda\), Hilbert--Schmidt approximation gives
\(A_N\to A\) and \(R_N\to R\) in \(\Hcal_\lambda\).
Closedness of \(\Lbb_\lambda\) proves
\eqref{eq:closed-generator-identification}.

We apply Lemma~\ref{lem:identified-calculus}\textup{(v)} with
\[
 x_\lambda=R_\lambda(F),\qquad
 x_{\lambda,Q}=R_{\lambda,Q}(F),\qquad
 x_Q=\Lcal_{0,Q}F,\qquad x=\Lcal_0F.
\]
The three hypotheses follow, respectively, from
Proposition~\ref{prop:fixed-compression},
Lemma~\ref{lem:classical-compression}, and
\eqref{eq:generator-tail}.  Hence
\(R_\lambda(F)\to_{\mathrm{id}}\Lcal_0F\).  Together with
\eqref{eq:closed-generator-identification}, we obtain
\eqref{eq:one-generator}.

Finally, \eqref{eq:closed-generator-identification} and
\eqref{eq:candidate-generator-bound} give
\[
 \sup_{0<\lambda\le\lambda_0}
 \|\Lbb_\lambda \Op_{\lambda,\gH}^{\rm A}F\|_{2,\lambda}
 =
 \sup_{0<\lambda\le\lambda_0}
 \|R_\lambda(F)\|_{2,\lambda}<\infty,
\]
which is \eqref{eq:generator-bound}.
\end{proof}

\section{Positive Liouville limits and correlations of bounded observables}
\label{sec:liouville-correlations}

We now combine generator convergence with positivity and dominated
Liouville uniqueness to obtain the bounded dynamical limits described in
Subsection~\ref{subsec:proof-outline}.  Subsection~\ref{subsec:positive-functionals-compactness}
proves that every such dominated positive sequence has a subsequential
limit given by a weakly continuous measure curve dominated by \(C\nu\)
and satisfying the Liouville equation.  Subsection~\ref{subsec:Liouville-transfer}
then assumes statement~\textup{(U)}, identifies the limiting curve with
its Hartree  push-forward, and deduces the two-point limit, identified
strong convergence of the dynamics, and the ordered multi-time limits.
Section~\ref{sec:DLU} proves statement~\textup{(U)} for the renormalized
Hartree NLS equation.

\subsection{Positive Liouville limits dominated by the Gibbs measure}
\label{subsec:positive-functionals-compactness}

We first compare the ordered weighted inner product with a positive trace
functional.  The estimate below replaces \(\Gamma_\lambda A\) or
\(A\Gamma_\lambda\) by
\(\Gamma_\lambda^{1/2}A\Gamma_\lambda^{1/2}\), with a trace-norm error
controlled by the generator.  

\begin{lemma}\label{lem:gibbs-square-root-comparison}
Let \(A=A^*\in\mathcal B(\Fock)\cap\Dom(\Lbb_\lambda)\).  Then
\begin{equation}\label{eq:gibbs-comparison-left}
 \norm{\Gamma_\lambda A-\Gamma_\lambda^{1/2}A\Gamma_\lambda^{1/2}}
      _{\mathfrak S^1(\Fock)}
 \le
 \bigl(\lambda\norm{A}_{2,\lambda}
       \norm{\Lbb_\lambda A}_{2,\lambda}\bigr)^{1/2}.
\end{equation}
\begin{equation}\label{eq:gibbs-comparison-right}
 \norm{A\Gamma_\lambda-\Gamma_\lambda^{1/2}A\Gamma_\lambda^{1/2}}
      _{\mathfrak S^1(\Fock)}
 \le
 \bigl(\lambda\norm{A}_{2,\lambda}
       \norm{\Lbb_\lambda A}_{2,\lambda}\bigr)^{1/2}.
\end{equation}
\end{lemma}

\begin{proof}
We use the Schatten--H\"older inequality to obtain
\[
 \|A\Gamma_\lambda-\Gamma_\lambda^{1/2}A\Gamma_\lambda^{1/2}\|_{\mathfrak S^1(\Fock)}
 =\|\Gamma_\lambda A-\Gamma_\lambda^{1/2}A\Gamma_\lambda^{1/2}\|_{\mathfrak S^1(\Fock)}
 \le \|A\Gamma_\lambda^{1/2}-\Gamma_\lambda^{1/2}A\|_{\mathfrak S^2(\Fock)}.
\]
Applying \cite[Theorem~7.2(ii)]{LNR} with
\(H=\lambda\mathbb H_\lambda\) and \(s=1/2\), and then using
the Cauchy--Schwarz inequality, we obtain
\begin{align*}
 \|A\Gamma_\lambda^{1/2}-\Gamma_\lambda^{1/2}A\|_{\mathfrak S^2(\Fock)}^2
 &=2\left[\Tr(A^2\Gamma_\lambda)
   -\Tr(A\Gamma_\lambda^{1/2}A\Gamma_\lambda^{1/2})\right]\\
 &\le \frac{\lambda}{2}
   \Tr\!\left(\Gamma_\lambda[[A,\mathbb H_\lambda],A]\right)\le \lambda\norm{A}_{2,\lambda}
       \norm{\Lbb_\lambda A}_{2,\lambda}.
\end{align*}
In the last step, we use trace cyclicity and
\(\Lbb_\lambda A=\ii[\mathbb H_\lambda,A]\).
\end{proof}

\begin{corollary}\label{cor:positive-decomp}
Let \(F\in\Acal_{\rm cyl}\) be real-valued, and set
\[
 A_\lambda=\Op_{\lambda,\gH}^{\rm A}F,\qquad
 M=\|F\|_\infty,\qquad
 A_\lambda^+=A_\lambda+M\1,
\qquad
 T_\lambda=\Gamma_\lambda^{1/2}A_\lambda^+
 \Gamma_\lambda^{1/2}.
\]
Then
\[
 0\le A_\lambda^+\le2M\1,
\]
and
\begin{equation}\label{eq:T-domination}
 0\le T_\lambda\le2M\Gamma_\lambda.
\end{equation}
Moreover, for every bounded operator \(C_\lambda\),
\begin{equation}\label{eq:ordered-positive-reduction}
 \left|\Tr\!\left(C_\lambda(\Gamma_\lambda A_\lambda-T_\lambda+M\Gamma_\lambda)\right)\right|
 \le \|C_\lambda\|_{\mathrm{op}}
 \bigl(\lambda M\|\Lbb_\lambda A_\lambda\|_{2,\lambda}\bigr)^{1/2}.
\end{equation}
\end{corollary}

\begin{proof}
Since \(0\le F+M\le 2M\), positivity of the anti-Wick
quantization and the identity
\(\Op_{\lambda,\gH}^{\rm A}1=\1\) imply
\[
0\le A_\lambda^+
=\Op_{\lambda,\gH}^{\rm A}(F+M)
\le 2M\1.
\]
 Thus
\eqref{eq:T-domination} follows.  Moreover,
\(\|A_\lambda\|_{2,\lambda}\le\|A_\lambda\|_{\mathrm{op}}\le M\).
By cyclicity and
Lemma~\ref{lem:gibbs-square-root-comparison},
\[
 \Tr(A_\lambda C_\lambda\Gamma_\lambda)
 =\Tr(C_\lambda\Gamma_\lambda A_\lambda)
 =\Tr(C_\lambda\Gamma_\lambda^{1/2}A_\lambda
       \Gamma_\lambda^{1/2})+\mathcal R_\lambda,
\]
where
\[
 |\mathcal R_\lambda|\le\|C_\lambda\|_{\mathrm{op}}
 \bigl(\lambda M\|\Lbb_\lambda A_\lambda\|_{2,\lambda}\bigr)^{1/2}.
\]
Substituting
\(\Gamma_\lambda^{1/2}A_\lambda\Gamma_\lambda^{1/2}
=T_\lambda-M\Gamma_\lambda\) implies \eqref{eq:ordered-positive-reduction}.
\end{proof}

The next result is a compactness criterion for sequences of positive
trace-class operators \(T_n\) satisfying
\(0\le T_n\le C\Gamma_{\lambda_n}\).  

\begin{theorem}
\label{thm:positive-compactness}
Let \(\lambda_n\downarrow0\), and let
\(T_n\in\mathfrak S^1(\Fock)\) be positive operators satisfying
$ 0\le T_n\le C\Gamma_{\lambda_n}$
with a constant \(C\) independent of \(n\).  Define
$ T_n(t)
 :=\ee^{-\ii t\mathbb H_{\lambda_n}}T_n
   \ee^{\ii t\mathbb H_{\lambda_n}}.$
After passing to a subsequence, there exists a positive measure curve
\((\sigma_t)_{t\in\mathbb R}\) such that, for every
\(G\in\Acal_{\rm cyl}\) and \(T>0\),
\begin{equation}\label{eq:positive-limit}
 \lim_{n\to\infty}
 \sup_{|t|\le T}
 \left|
  \Tr(\Op_{\lambda_n,\gH}^{\rm A}G\,T_n(t))
  -\int G\,\dd\sigma_t
 \right|
 =0.
\end{equation}
More precisely,
\begin{equation}\label{eq:positive-density-form}
 \sigma_t=g_t\nu,
 \qquad
 0\le g_t\le C
 \quad\nu\hbox{-almost everywhere},
 \qquad t\in\mathbb R.
\end{equation}
In particular,
\(0\le\sigma_t\le C\nu\) for \(t\in\mathbb R\),
and \(\sigma_t\) is concentrated on \(\Xcal\).  For every
\(\kappa>(d-2)/2\), the curve is weakly continuous as a curve of finite
Borel measures on \(H^{-\kappa}(\mathbb T^d)\).  It satisfies the
integrated Liouville equation
\begin{equation}\label{eq:positive-Liouville}
 \int G\,\dd\sigma_t
 -\int G\,\dd\sigma_0
 =
 \int_0^t\int\Lcal_0G\,\dd\sigma_s\,\dd s,
 \qquad G\in\Acal_{\rm cyl}.
\end{equation}
\end{theorem}

\begin{proof}
We first prove a
weighted Hilbert-space estimate. 
Since \(\Gamma_{\lambda_n}\) commutes with
\(\mathbb H_{\lambda_n}\), the assumed domination is preserved by unitary conjugation; hence
\begin{equation}\label{eq:positive-domination-time}
 0\le T_n(t)\le C\Gamma_{\lambda_n},
 \quad t\in\mathbb R,  \qquad
 \Tr (T_n(t))\le C.
\end{equation}
The Cauchy--Schwarz inequality implies
\begin{equation}\label{eq:positive-dominated-Hilbert-bound}
 \abs{\Tr(xT_n(t))}
  \le \Tr(T_n(t))^{1/2}\Tr(x^*xT_n(t))^{1/2}\le C\norm{x}_{2,\lambda_n},
 \qquad
 x\in\Hcal_{\lambda_n},\quad t\in\mathbb R.
\end{equation}
Thus $x\mapsto\operatorname{Tr}(xT_n(t))$ extends uniquely to a
bounded linear functional on $\mathcal H_{\lambda_n}$.

\smallskip
\noindent\emph{Step 1: compactness of the scalar functions.}
Choose a countable $\mathbb Q(\ii)$-linear subspace
$\mathscr D\subset\Acal_{\rm cyl}$
which contains $1$, is invariant under complex conjugation,
and is dense in \(L^2(\nu)\). The diagonal limits below are $\mathbb Q(\ii)$-linear on
$\mathscr D$; the uniform $L^2(\nu)$ bound then gives a unique
continuous complex-linear extension to $L^2(\nu)$.
For \(G\in\Acal_{\rm cyl}\), set
$ F_n^G(t):=\Tr(\Op_{\lambda_n,\gH}^{\rm A}G\,T_n(t)).$
 \eqref{eq:positive-dominated-Hilbert-bound} and the norm
convergence following from \eqref{eq:equilibrium-inner-product} imply
\[
 \sup_n\sup_{t\in\mathbb R}|F_n^G(t)|<\infty.
\]
Using
\(
 F_n^G(t)=\Tr\!\left(
  \Ubb_{\lambda_n}(t)(\Op_{\lambda_n,\gH}^{\rm A}G)T_n
 \right),
\)
and the integrated unitary-group formula $$\Ubb_{\lambda}(t)A-\Ubb_{\lambda}(s)A=\int_s^t \Ubb_{\lambda}(r)\Lbb_{\lambda}A\,\dd r,\qquad A\in \textrm{Dom}(\Lbb_\lambda),$$ we obtain
\begin{equation}\label{eq:quantum-positive-integrated}
 F_n^G(t)-F_n^G(s)
 =
 \int_s^t
 \Tr\!\left(
  (\Lbb_{\lambda_n}\Op_{\lambda_n,\gH}^{\rm A}G)T_n(r)
 \right)\,\dd r.
\end{equation}
Hence 
\eqref{eq:generator-bound} implies
\begin{equation}\label{eq:scalar-equi-Lipschitz}
 |F_n^G(t)-F_n^G(s)|
 \le C_G|t-s|,
 \qquad
 C_G:=C\sup_{0<\lambda\le\lambda_0}
 \norm{\Lbb_\lambda\Op_{\lambda,\gH}^{\rm A}G}_{2,\lambda}<\infty.
\end{equation}
Arzel\`a--Ascoli and diagonal extraction over
\(\mathscr D\times\mathbb N\), where \(m\in\mathbb N\) corresponds to
the interval \([-m,m]\), imply a subsequence and functions
\(\ell^G\in C(\mathbb R;\mathbb C)\) such that
\begin{equation}\label{eq:D-local-uniform-limit}
 F_n^G\longrightarrow\ell^G
 \quad\hbox{locally uniformly on }\mathbb R
\end{equation}
for every \(G\in\mathscr D\).  By
\eqref{eq:scalar-equi-Lipschitz}, each \(\ell^G\) is
\(C_G\)-Lipschitz.  We work with this subsequence from now on.

\smallskip
\noindent\emph{Step 2: construction and domination of the limiting
measure curve.}
Fix \(t\). Define the map
\[
 \ell_t:\mathscr D\longrightarrow\mathbb C,
 \qquad
 \ell_t(G):=\ell^G(t).
\]
  By
\eqref{eq:positive-dominated-Hilbert-bound} and the norm convergence
following from \eqref{eq:equilibrium-inner-product},
\[
 |\ell_t(G)|
 \le C\lim_{n\to\infty}
 \norm{\Op_{\lambda_n,\gH}^{\rm A}G}_{2,\lambda_n}
 =C\norm{G}_{L^2(\nu)},
 \qquad G\in\mathscr D.
\]
Thus it extends uniquely to a bounded complex-linear functional on
\(L^2(\nu)\).  Hence there is a unique \(g_t\in L^2(\nu)\) such that
\begin{equation}\label{eq:positive-functional-density}
 \ell_t(x)=\int xg_t\,\dd\nu,
 \qquad x\in L^2(\nu),
 \qquad \norm{g_t}_{L^2(\nu)}\le C.
\end{equation}

We extend \eqref{eq:D-local-uniform-limit} from \(\mathscr D\) to every
cylinder function.  For \(G\in\Acal_{\rm cyl}\), \(H\in\mathscr D\),
and \(T>0\), \eqref{eq:positive-dominated-Hilbert-bound}, the norm
convergence following from \eqref{eq:equilibrium-inner-product}, and
\eqref{eq:positive-functional-density} give
\begin{align*}
 \limsup_{n\to\infty}\sup_{|t|\le T}
 |F_n^G(t)-F_n^H(t)|
 &\le C\norm{G-H}_{L^2(\nu)},\\
 \sup_{|t|\le T}
 \left|\int(G-H)g_t\,\dd\nu\right|
 &\le C\norm{G-H}_{L^2(\nu)}.
\end{align*}
Using the density of \(\mathscr D\) in \(L^2(\nu)\), we obtain, for $G\in\Acal_{\rm cyl}$,
\begin{equation}\label{eq:positive-cylinder-local-uniform}
 \lim_{n\to\infty}\sup_{|t|\le T}
 \left|F_n^G(t)-\int Gg_t\,\dd\nu\right|=0.
\end{equation}

If \(G\in\Acal_{\rm cyl}\) is nonnegative, positivity of the anti-Wick
quantization and \(0\le T_n(t)\le C\Gamma_{\lambda_n}\) yield
\[
 0\le F_n^G(t)
 \le C\Tr(\Op_{\lambda_n,\gH}^{\rm A}G\,\Gamma_{\lambda_n}).
\]
Letting \(n\to\infty\) and using
\eqref{eq:positive-cylinder-local-uniform} and \eqref{eq:static-mean},
we obtain, for $G\in\Acal_{\rm cyl},\ G\ge0$,
\[
 0\le\int Gg_t\,\dd\nu
 \le C\int G\,\dd\nu.
\]
Since
\(\Acal_{\rm cyl}\) is dense in \(L^2(\nu)\), the extended functional is
positive and bounded above by \(G\mapsto C\int G\,\dd\nu\). We then have
\[
 0\le g_t\le C
 \qquad\nu\hbox{-almost everywhere}.
\]
Define \(\sigma_t:=g_t\nu\).  Then \eqref{eq:positive-density-form}
holds, \eqref{eq:positive-cylinder-local-uniform} is
\eqref{eq:positive-limit}, and \(\sigma_t\) is concentrated on \(\Xcal\).
Since \(t\mapsto\ell^G(t)\) is continuous for every \(G\in\mathscr D\),
the same is true for every \(G\in C_b(H^{-\kappa})\): indeed,
approximate \(G\) in \(L^2(\nu)\) by functions in \(\mathscr D\) and
use \(0\le\sigma_t\le C\nu\), uniformly in \(t\).  Thus
\(t\mapsto\sigma_t\) is weakly continuous on \(H^{-\kappa}\) for every
\(\kappa>(d-2)/2\).

\smallskip
\noindent\emph{Step 3: identified convergence and passing to the
Liouville equation.}
We claim that whenever
\[
 x_n\to_{\mathrm{id}}x,
 \qquad
 x_n\in\Hcal_{\lambda_n},\quad x\in\Hcal,
\]
one has, for every \(T>0\),
\begin{equation}\label{eq:positive-identified-uniform}
 \sup_{|t|\le T}
 \left|
  \Tr(x_nT_n(t))-\int x\,\dd\sigma_t
 \right|
 \longrightarrow0.
\end{equation}
Given \(\varepsilon>0\), identified convergence gives
\(G\in\Acal_{\rm cyl}\) such that
\[
\norm{G-x}_{L^2(\nu)}<\varepsilon,
\qquad
\limsup_{n\to\infty}
\norm{x_n-\Op_{\lambda_n,\gH}^{\rm A}G}_{2,\lambda_n}
<\varepsilon.
\]
By \eqref{eq:positive-dominated-Hilbert-bound} and
\(0\le\sigma_t\le C\nu\), for \(|t|\le T\),
\begin{align*}
	\left|\Tr(x_nT_n(t))-\int x\,\dd\sigma_t\right|
	&\le
	C\norm{x_n-\Op_{\lambda_n,\gH}^{\rm A}G}_{2,\lambda_n}
	+\left|F_n^G(t)-\int G\,\dd\sigma_t\right|
	+C\norm{G-x}_{L^2(\nu)}.
\end{align*}
The middle term tends to zero locally uniformly by
\eqref{eq:positive-cylinder-local-uniform}.  Taking the upper limit
as \(n\to\infty\), and then letting \(\varepsilon\downarrow0\), proves
\eqref{eq:positive-identified-uniform}.

 By \eqref{eq:one-generator}, we obtain
\[
 \Lbb_{\lambda_n}\Op_{\lambda_n,\gH}^{\rm A}G
 \to_{\mathrm{id}}\Lcal_0G.
\]
Applying \eqref{eq:positive-identified-uniform} to this convergence
implies
\begin{equation}\label{eq:positive-generator-limit}
 \sup_{|r|\le T}
 \left|
  \Tr\!\left(
   (\Lbb_{\lambda_n}\Op_{\lambda_n,\gH}^{\rm A}G)T_n(r)
  \right)
  -\int\Lcal_0G\,\dd\sigma_r
 \right|
 \longrightarrow0.
\end{equation}
Taking \(s=0\) in \eqref{eq:quantum-positive-integrated}, using
\eqref{eq:positive-limit} and
\eqref{eq:positive-generator-limit} yields
\[
 \int G\,\dd\sigma_t-\int G\,\dd\sigma_0
 =\int_0^t\int\Lcal_0G\,\dd\sigma_r\,\dd r.
\]
Hence \eqref{eq:positive-Liouville} follows.
\end{proof}

\subsection{Identifying the limit by Liouville uniqueness}
\label{subsec:Liouville-transfer}

We now state the consequences of dominated Liouville uniqueness.
Theorem~\ref{thm:bounded-two-point},
Corollary~\ref{cor:strong-identified-dynamics}, and
Theorem~\ref{thm:bounded-multipoint} are conditional on
statement~\textup{(U)}.  Theorem~\ref{thm:main-identified-dynamics} is
different: it is stated directly under the same hypotheses as
Proposition~\ref{prop:bounded-correlations}, and its proof is completed after
Section~\ref{sec:DLU} establishes statement~\textup{(U)} for the
present Hartree NLS equation.  

Fix \(1/2<\kappa<1\).  Statement~\textup{(U)} says that
for every \(T>0\), \(C_0<\infty\), and every weakly continuous curve
\((\sigma_t)_{|t|\le T}\) of finite positive Borel measures on
\(H^{-\kappa}(\mathbb T^d)\) satisfying
\begin{equation}\label{eq:dominated-Liouville-assumption}
 0\le\sigma_t\le C_0\nu,\qquad
 \int G\,\dd\sigma_t-\int G\,\dd\sigma_s
 =
 \int_s^t\!\int\Lcal_0G\,\dd\sigma_r\,\dd r
\end{equation}
for \(G\in\Acal_{\rm cyl}\) and \(s,t\in[-T,T]\), one has
\begin{equation*}\tag{U}\label{eq:statement-U}
 \sigma_t=(S_t)_\#\sigma_0,\qquad |t|\le T.
\end{equation*}

\begin{theorem}
\label{thm:bounded-two-point}
Assume the same hypotheses as in Proposition~\ref{prop:bounded-correlations}
and, in addition, statement~\textup{(U)} above.  Then, for every
\(F,G\in\Acal_{\rm cyl}\) and \(t\in\mathbb R\),
\begin{equation*}
 \Tr\left(
  (\Op_{\lambda,\gH}^{\rm A}F)^*
  \Ubb_\lambda(t)(\Op_{\lambda,\gH}^{\rm A}G)
  \Gamma_\lambda
 \right)
 \longrightarrow
 \int
  \overline{F(u)}G(S_tu)\,\dd\nu(u)
,\qquad\lambda\downarrow0.
\end{equation*}
\end{theorem}

\begin{proof}
	We first assume that \(F\) is real-valued.  Set
	\[
	A_\lambda:=\Op_{\lambda,\gH}^{\rm A}F,
	\qquad
	B_\lambda:=\Op_{\lambda,\gH}^{\rm A}G,
	\qquad
	M:=\norm{F}_\infty,
	\]
	and
	\[
	T_\lambda
	:=\Gamma_\lambda^{1/2}(A_\lambda+M\1)
	\Gamma_\lambda^{1/2}.
	\]
	Then Corollary~\ref{cor:positive-decomp} gives
	\[
	0\le T_\lambda\le 2M\Gamma_\lambda.
	\]
	Moreover, Corollary~\ref{cor:positive-decomp} and the uniform bound in
	Theorem~\ref{thm:one-generator} imply that, for every family
	\((C_\lambda)\) bounded uniformly in operator norm,
	\begin{equation}\label{eq:c-5}
		\Tr(C_\lambda T_\lambda)
		-\Tr(C_\lambda\Gamma_\lambda A_\lambda)
		-M\Tr(C_\lambda\Gamma_\lambda)
		\longrightarrow0
		\qquad(\lambda\downarrow0).
	\end{equation}
	Define the scalar quantity
	\(
	\Phi_\lambda
	:=
	\Tr\left(
	A_\lambda\Ubb_\lambda(t)(B_\lambda)\Gamma_\lambda
	\right).
	\)
	Let \(\lambda_n\downarrow0\) be arbitrary.  Apply
	Theorem~\ref{thm:positive-compactness} to
	\(T_n:=T_{\lambda_n}\).  After passing to a subsequence,
	there is a positive weakly continuous Liouville curve
	\((\sigma_s)_{s\in\mathbb R}\), with
	\(0\le\sigma_s\le2M\nu\), such that
	\[
	T_n(s)
	:=
	\ee^{-\ii s\mathbb H_{\lambda_n}}
	T_n
	\ee^{\ii s\mathbb H_{\lambda_n}},
	\qquad
	\Tr\left(
	\Op_{\lambda_n,\gH}^{\rm A}K\,T_n(s)
	\right)
	\longrightarrow
	\int K\,\dd\sigma_s
	\]
	for every \(K\in\Acal_{\rm cyl}\) and \(s\in\mathbb R\).
	
	We first identify the initial measure.  Applying \eqref{eq:c-5} with
	\(C_{\lambda_n}=\Op_{\lambda_n,\gH}^{\rm A}K\), and using trace
	cyclicity together with
	\eqref{eq:static-mean} and
	\eqref{eq:equilibrium-inner-product}, gives
	\[
	\int K\,\dd\sigma_0
	=
	\int (F+M)K\,\dd\nu,
	\qquad K\in\Acal_{\rm cyl}.
	\]
	Since \(0\le\sigma_0\le2M\nu\) and
	\(\Acal_{\rm cyl}\) is dense in \(L^2(\nu)\), it follows that
	\(
	\sigma_0=(F+M)\nu.
	\)
	
	The curve \((\sigma_s)\) satisfies
	\eqref{eq:positive-Liouville}.  Subtracting this identity at two
	times and applying statement~\textup{(U)} on an interval containing
	\(0\) and \(t\), we obtain
	\[
	\sigma_s=(S_s)_\#\sigma_0
	=(S_s)_\#\bigl((F+M)\nu\bigr).
	\]
	
	Now put
	\(
	C_n:=\Ubb_{\lambda_n}(t)(B_{\lambda_n}).
	\)
	Anti-Wick contractivity and unitary invariance give
	\(
	\norm{C_n}_{\mathrm{op}}
	\le\norm{G}_\infty.
	\)
	Therefore \eqref{eq:c-5}, trace cyclicity, and the invariance of
	\(\Gamma_{\lambda_n}\) under the quantum dynamics yield
	\begin{align*}
		\Phi_{\lambda_n}
		&=\Tr(C_n\Gamma_{\lambda_n}A_{\lambda_n})=\Tr(C_nT_n)
		-M\Tr(C_n\Gamma_{\lambda_n})
		+o_n(1)\\
		&=\Tr(B_{\lambda_n}T_n(t))
		-M\Tr(B_{\lambda_n}\Gamma_{\lambda_n})
		+o_n(1),
	\end{align*}
	where $o_n(1)\to0$ as $n\to\infty$. 
	The defining convergence of \(\sigma_t\) and
	\eqref{eq:static-mean} now imply, along the subsequence selected above,
	\begin{align*}
		\Phi_{\lambda_n}
		&\longrightarrow
		\int G\,\dd\sigma_t-M\int G\,\dd\nu\\
		&=\int(F(u)+M)G(S_tu)\,\dd\nu(u)
		-M\int G\,\dd\nu=\int F(u)G(S_tu)\,\dd\nu(u),
	\end{align*}
	where the last equality follows from \((S_t)_\#\nu=\nu\).
		Thus every sequence \(\lambda_n\downarrow0\) has a subsequence along
	which \(\Phi_{\lambda_n}\) converges to the same limit.  This proves
	the convergence of the full family \(\Phi_\lambda\) as
	\(\lambda\downarrow0\).
	
	Finally, for general \(F\), write
	\(F=F_1+\ii F_2\), where \(F_1,F_2\) are real-valued.  Since
	\[
	(\Op_{\lambda,\gH}^{\rm A}F)^*
	=
	\Op_{\lambda,\gH}^{\rm A}F_1
	-\ii\Op_{\lambda,\gH}^{\rm A}F_2,
	\]
	the result follows from the real-valued case by linearity.

\end{proof}

\begin{corollary}
\label{cor:strong-identified-dynamics}
Assume the hypotheses of Theorem~\ref{thm:bounded-two-point}.  Then, for
every \(t\in\mathbb R\) and \(G\in\Acal_{\rm cyl}\),
\[
 \Ubb_\lambda(t)\Op_{\lambda,\gH}^{\rm A}G
 \to_{\mathrm{id}}\Ucal(t)G.
\]
More generally, for every identified sequence
\(x_\lambda\to_{\mathrm{id}}x\),
\begin{equation}\label{eq:identified-dynamics-strong}
 \Ubb_\lambda(t)x_\lambda\to_{\mathrm{id}}\Ucal(t)x.
\end{equation}
\end{corollary}

\begin{proof}
For every \(F\in\Acal_{\rm cyl}\),
Theorem~\ref{thm:bounded-two-point} gives
\[
 \ip{\Op_{\lambda,\gH}^{\rm A}F}
     {\Ubb_\lambda(t)\Op_{\lambda,\gH}^{\rm A}G}_\lambda
 \longrightarrow
 \int\overline F\,G\circ S_t\,\dd\nu
 =\ip F{\Ucal(t)G}_{L^2(\nu)}.
\]
By unitarity and \eqref{eq:equilibrium-inner-product} with \(F=G\),
\[
 \norm{\Ubb_\lambda(t)\Op_{\lambda,\gH}^{\rm A}G}_{2,\lambda}
 =
 \norm{\Op_{\lambda,\gH}^{\rm A}G}_{2,\lambda}
 \longrightarrow
 \norm{G}_{L^2(\nu)}
 =
 \norm{\Ucal(t)G}_{L^2(\nu)}.
\]
Lemma~\ref{lem:identified-calculus}\textup{(i)} therefore yields
\[
 \Ubb_\lambda(t)\Op_{\lambda,\gH}^{\rm A}G
 \to_{\mathrm{id}}\Ucal(t)G.
\]
Lemma~\ref{lem:identified-calculus}\textup{(iv)}, applied with
\(B_\lambda=\Ubb_\lambda(t)\) and \(B=\Ucal(t)\), both of norm one,
gives \eqref{eq:identified-dynamics-strong}.
\end{proof}

\begin{theorem}
\label{thm:main-identified-dynamics}
Assume the same hypotheses as in Proposition~\ref{prop:bounded-correlations}.  Then
the following statements hold.
\begin{enumerate}[label=\textup{(\roman*)},leftmargin=2.2em]
\item For every \(F\in\Acal_{\rm cyl}\),
\[
 \Op_{\lambda,\gH}^{\rm A}F\in\Dom(\Lbb_\lambda),\qquad
 \Lbb_\lambda \Op_{\lambda,\gH}^{\rm A}F
 \to_{\mathrm{id}}\Lcal_0F,
 \qquad
 \sup_{0<\lambda\le\lambda_0}
 \|\Lbb_\lambda \Op_{\lambda,\gH}^{\rm A}F\|_{2,\lambda}<\infty.
\]
\item For every \(F\in\Acal_{\rm cyl}\) and every fixed
\(t\in\mathbb R\),
\[
 \Ubb_\lambda(t)\Op_{\lambda,\gH}^{\rm A}F
 \to_{\mathrm{id}}\Ucal(t)F.
\]
More generally, \(x_\lambda\to_{\mathrm{id}}x\) implies
\(\Ubb_\lambda(t)x_\lambda\to_{\mathrm{id}}\Ucal(t)x\).
\end{enumerate}
\end{theorem}

The proof is given at the end of
Subsection~\ref{subsec:pathwise-uniqueness}.

\begin{theorem}
\label{thm:bounded-multipoint}
Assume the hypotheses of Theorem~\ref{thm:bounded-two-point}.  Then, for
every \(J\ge1\), \(F_1,\ldots,F_J\in\Acal_{\rm cyl}\), and
\(t_1,\ldots,t_J\in\mathbb R\),
\begin{equation}\label{eq:bounded-multipoint}
 \Tr\left(
  \prod_{j=1}^J
  \Ubb_\lambda(t_j)(\Op_{\lambda,\gH}^{\rm A}F_j)
  \Gamma_\lambda
 \right)
 \longrightarrow
 \int
 \prod_{j=1}^JF_j(S_{t_j}u)
 \,\dd\nu(u)
,\qquad\lambda\downarrow0.
\end{equation}
\end{theorem}

\begin{proof}
Let \(F\in\Acal_{\rm cyl}\), \(t\in\mathbb R\), and suppose that
\(x_\lambda\to_{\mathrm{id}}x\).  Corollary~\ref{cor:strong-identified-dynamics},
first at time \(-t\), and Theorem~\ref{thm:static-argument} give
\[
 \Ubb_\lambda(-t)x_\lambda\to_{\mathrm{id}}\Ucal(-t)x,
 \qquad
 \Op_{\lambda,\gH}^{\rm A}F\,\Ubb_\lambda(-t)x_\lambda
 \to_{\mathrm{id}}F\,\Ucal(-t)x.
\]
Applying Corollary~\ref{cor:strong-identified-dynamics} at time \(t\) and using
\eqref{eq:weighted-unitary-implementation}, we obtain
\[
 \bigl[\Ubb_\lambda(t)(\Op_{\lambda,\gH}^{\rm A}F)\bigr]x_\lambda
 =\Ubb_\lambda(t)\bigl(
   \Op_{\lambda,\gH}^{\rm A}F\,\Ubb_\lambda(-t)x_\lambda
  \bigr)
 \to_{\mathrm{id}}
 \Ucal(t)\bigl(F\,\Ucal(-t)x\bigr)
 =(F\circ S_t)x.
\]
Thus these multiplication maps preserve identified convergence and have operator
norm at most \(\norm{F}_\infty\).  Applying
Lemma~\ref{lem:identified-products} with
\(x_\lambda=y_\lambda=\1\) and \(x=y=1\), we obtain
\eqref{eq:bounded-multipoint}.
\end{proof}

\section{Uniqueness for Liouville curves dominated by the Gibbs measure}
\label{sec:DLU}

By Section~\ref{sec:liouville-correlations}, it remains to prove
statement~\textup{(U)} for the renormalized Hartree NLS equation.  By
Theorem~\ref{thm:positive-compactness}, the limiting
curves are weakly continuous on \(H^{-\kappa}\).   Let
\(t\mapsto\sigma_t\) be a weakly continuous curve of finite positive
measures satisfying, for $G\in\Acal_{\rm cyl}$,
\begin{equation*}
 0\leq\sigma_t\leq C_0\nu,
 \qquad
 \int G\,\dd\sigma_t-\int G\,\dd\sigma_s
 =\int_s^t\!\int\Lcal_0G\,\dd\sigma_r\,\dd r.
\end{equation*}

We first use the superposition
principle to represent such a curve by a measure on rough integral
paths, but it does not place those paths in the random-averaging
solution class used to construct the flow in dimension three (see Remark~\ref{dny} below).  We
then prove uniqueness directly in the class of paths obtained
from the superposition principle and marginal domination.  Instead of comparing two rough nonlinear
Hartree paths directly, we first prescribe an arbitrary real potential
\(A\in L^1([-T,T];\mathcal W^1)\) and consider the linear equation.
For a prescribed potential, the evolution is linear and unique.  The
question is then how the renormalized terms generated by
\(U_A(t,0)u_0\) change when \(A\) is replaced by another prescribed
potential \(B\).
At the regularity \(H^{-\kappa}\), the renormalized terms
 are obtained as almost-sure limits of regularized quadratic
forms.  Consequently, the self-consistent potentials generated by two
rough paths cannot be compared by a deterministic pointwise Lipschitz
estimate.   Proposition~\ref{prop:transported-density-comparison}
gives this comparison for the transported nonzero Fourier modes.
 After
summing over the interaction modes, this estimate gives a Gronwall-type argument for the mean-zero parts of the two potentials.

The proof is organized as follows.
Subsection~\ref{subsec:DLU-superposition-regularity}
uses domination and the superposition principle to obtain paths
satisfying the Hartree Duhamel equation. Subsection~\ref{subsec:DLU-density-comparison}
freezes the potentials and proves the comparison of the nonzero
density modes.  Subsection~\ref{subsec:pathwise-uniqueness} first
identifies the paths modulo a scalar phase, then recovers the zero
mode and the phase, and finally passes from pathwise uniqueness to
uniqueness of the dominated Liouville curve.

\begin{remark}
The condition \(0\leq\sigma_t\leq C_0\nu\) is analogous in spirit to
the restriction used in \cite{RZZ17}.  There, the process obtained
from the Dirichlet form is related, after subtracting the corresponding
Ornstein--Uhlenbeck process, to a shifted equation in a class where
uniqueness is available.  Here domination selects the class arising
from the Gibbs limit and transfers the \(\nu\)-almost-sure definitions
and \(\nu\)-integrability estimates to the superposition measure.  We
therefore prove uniqueness for Liouville curves in this dominated
class, rather than for arbitrary distributional curves.  
\end{remark}

Recall the weighted Wiener spaces \(\mathcal W^s\) from
Subsection~\ref{subsec:classical-phase-space}. For \(s\ge0\), the discrete
Young inequality and $\langle p\rangle^r
 \le C\langle p-k\rangle^r\langle k\rangle^{|r|}$ give
\begin{equation}\label{eq:weighted-Wiener-multiplier}
 \|wf\|_{H^r}\le C_s\|w\|_{\mathcal W^s}\|f\|_{H^r},
 \qquad |r|\le s.
\end{equation}
Fix throughout this section
\begin{equation}\label{eq:path-Sobolev-exponent}
 \frac12<\kappa<1.
\end{equation}
For a Borel map
\(\Phi\), \(\Phi_\#\mu\) denotes the push-forward of \(\mu\).  Let \(\mathscr X_T:=C([-T,T];H^{-\kappa}(\mathbb T^d))\) be the path space, and let
\(e_t(\gamma)=\gamma(t)\) denote evaluation on the path space.

\subsection{From Liouville curves to Hartree trajectories}
\label{subsec:DLU-superposition-regularity}

This subsection converts the Eulerian Liouville equation into a
Lagrangian description.  
For each \(k\in\mathbb Z^d\), fix the Borel representative of
\(\rho_k\) on \(H^{-\kappa}\). On the
Borel set where the following series converges in $\mathcal W^1$ and is real-valued, define $A(u)$ by the series, and set
$A(u) = 0$ on its complement:
\begin{equation}\label{eq:Borel-self-consistent-potential}
 A(u)(x):=\sum_{k\in\mathbb Z^d}
 \widehat v(k)\rho_k(u)\ee^{-\ii k\cdot x}.
\end{equation}
 The moment bounds for the density modes,
Gaussian hypercontractivity, and
\(\sum_k\langle k\rangle|\widehat v(k)|<\infty\) imply
\(A:H^{-\kappa}\to\mathcal W^1\) is Borel and, for every
\(1\le r<\infty\),
\[
 A\in L^r(\mu_0;\mathcal W^1)\cap L^r(\nu;\mathcal W^1).
\]
Moreover, \(A(u)\) is real for \(\mu_0\)- and \(\nu\)-almost every
\(u\).  By \eqref{eq:weighted-Wiener-multiplier}, the map
\(u\mapsto A(u)u\) is Borel from \(H^{-\kappa}\) to
\(H^{-\kappa}\) and belongs to
\(L^1(\mu_0;H^{-\kappa})\cap L^1(\nu;H^{-\kappa})\). 
For each \(p\in\mathbb Z^d\),
\(
\langle \mathrm e_p,A(u)u\rangle
=\mathscr F_p(u)
\)
for \(\nu\)-almost every \(u\), with \(\mathscr F_p(u)\) defined in
Subsection~\ref{subsec:classical-generator}.

\begin{proposition}
\label{prop:common-superposition}
Let \(T>0\) and \(0\le C_0<\infty\).  Let
\(t\mapsto\sigma_t\), \(|t|\le T\), be a weakly continuous curve of
finite positive Borel measures on \(H^{-\kappa}\) such that
\begin{equation*}
 0\le\sigma_t\le C_0\nu,
 \qquad |t|\le T,
\end{equation*}
and, for \(G\in\Acal_{\rm cyl}\) and \(s,t\in[-T,T]\),
\begin{equation}\label{eq:superposition-Liouville-assumption}
 \int G\,\dd\sigma_t-\int G\,\dd\sigma_s
 =\int_s^t\!\int\Lcal_0G\,\dd\sigma_r\,\dd r.
\end{equation}
Then there exists a finite positive Borel measure \(\eta\) on
\(\mathscr X_T\) with the following properties.
\begin{enumerate}[label=\textup{(\roman*)},leftmargin=2.2em]
\item For every \(|t|\le T\),
\begin{equation*}
 (e_t)_\#\eta=\sigma_t.
\end{equation*}
\item For \(\eta\)-almost every \(\gamma\), let
$ A_\gamma(t)(x):=A(\gamma(t))(x).$
Then  $A_\gamma\in L^1([-T,T];\mathcal W^1),$
and, for all \(s,t\in[-T,T]\),
\begin{equation}\label{eq:self-consistent-path-Duhamel}
 \gamma(t)=\ee^{-\ii(t-s)h}\gamma(s)
 -\ii\int_s^t\ee^{-\ii(t-r)h}
 A_\gamma(r)\gamma(r)\,\dd r
 \quad\text{in }H^{-\kappa}.
\end{equation}
\end{enumerate}
\end{proposition}

\begin{proof}
Testing \eqref{eq:superposition-Liouville-assumption} with \(G=1\)
shows that the total mass
\(M:=\sigma_t(H^{-\kappa})\) is independent of \(t\).  If \(M=0\),
the conclusion is immediate, so assume \(M>0\).

Define
\[
 \mathfrak b(x):=-\ii\bigl(hx+A(x)x\bigr), \qquad x\in H^{-\kappa},
\]
and extend it by zero to the Borel complement of $H^{-\kappa}$ in $H^{-\kappa-2}$. 
Since \(\sigma_t\le C_0\nu\),
\begin{equation}\label{eq:superposition-drift-integrability}
 \int_{-T}^T\!\int_{H^{-\kappa-2}}
 \|\mathfrak b(x)\|_{H^{-\kappa-2}}\,
 \dd\sigma_t(x)\,\dd t<\infty.
\end{equation}
By a standard finite-dimensional cutoff approximation,
\eqref{eq:superposition-Liouville-assumption} extends to the cylinder
test class required in \cite[Proposition~2.1]{AFS}.  Applying that
superposition principle to the normalized curve \(M^{-1}\sigma_t\)
and then multiplying the resulting path measure by \(M\), we obtain a
finite positive Borel measure \(\eta\) on
\(C([-T,T];H^{-\kappa-2})\) such that
\((e_t)_\#\eta=\sigma_t\) and, for \(\eta\)-almost every path,
\begin{equation}\label{eq:full-field-integral-curve}
 \gamma(t)=\gamma(s)-\ii\int_s^t
 \bigl(h\gamma(r)+A(\gamma(r))\gamma(r)\bigr)\,\dd r,
 \qquad s,t\in[-T,T].
\end{equation}

We now use marginal domination to improve the path regularity.  If
\(\Psi:H^{-\kappa}\to[0,\infty]\) is Borel and belongs to
\(L^1(\nu)\), Fubini's theorem implies
\begin{equation}\label{eq:marginal-domination-transfer}
 \int\!\int_{-T}^T\Psi(\gamma(t))\,\dd t\,\dd\eta(\gamma)
 \le2TC_0\int\Psi\,\dd\nu.
\end{equation}
By the construction of \(A\), the two Borel functions
\[
 u\longmapsto\|A(u)u\|_{H^{-\kappa}},
 \qquad
 u\longmapsto\|A(u)\|_{\mathcal W^1}
\]
belong to \(L^1(\nu)\).  Applying
\eqref{eq:marginal-domination-transfer} to these two functions shows that, for \(\eta\)-almost every path,
\begin{equation}\label{eq:dominated-path-integrability}
 \begin{aligned}
 \int_{-T}^T\|A(\gamma(t))\gamma(t)\|_{H^{-\kappa}}\,\dd t<\infty,
 \qquad
 \int_{-T}^T\|A(\gamma(t))\|_{\mathcal W^1}\,\dd t<\infty.
 \end{aligned}
\end{equation}
The time-zero marginal also implies
\[
 \eta\{\gamma:\gamma(0)\notin H^{-\kappa}\}
 \le C_0\nu\bigl(H^{-\kappa-2}\setminus H^{-\kappa}\bigr)=0.
\]
Using \eqref{eq:full-field-integral-curve} we obtain \eqref{eq:self-consistent-path-Duhamel}.  Hence
\(\gamma\in C([-T,T];H^{-\kappa})\).
\end{proof}

\begin{remark}\label{dny}
One may ask whether the integral paths obtained above can be
identified directly with the Hartree NLS flow from
Theorem~\ref{thm:canonical-flow}.  In dimension three,
\cite{DNYHartree} constructs this flow as the limit of canonical
finite-dimensional approximations by means of random-averaging
operators.  Theorem~1.3 and Remark~1.4 of \cite{DNYHartree} show that
several canonical approximation procedures give the same limit, but
they do not state uniqueness among all paths satisfying the limiting
Duhamel equation.  The Schr\"odinger evolution has no smoothing effect
that would automatically place a path furnished by the superposition
principle in the random-averaging solution class.  A direct appeal to
\cite{DNYHartree} would therefore require an additional weak--strong
uniqueness statement showing that every dominated Duhamel path belongs
to that class.  The comparison argument below avoids this missing
step.

\end{remark}

\subsection{Comparing nonzero density modes for prescribed potentials}
\label{subsec:DLU-density-comparison}

This subsection proves the comparison estimate for the renormalized
density modes.  We first fix two real potentials \(A\) and \(B\).
The corresponding equations are linear, and
Lemma~\ref{lem:prescribed-potential-propagators} gives their
propagators.  Proposition~\ref{prop:transported-density-comparison}
compares each nonzero density mode along the two linear flows.
Lemma~\ref{lem:sharp-density-identification} then uses marginal
domination to identify these limits with the fixed Borel
representatives \(\rho_q\) used to define the potential associated
with a path. 

\begin{lemma}
\label{lem:prescribed-potential-propagators}
Let \(A\in L^1([-T,T];\mathcal W^1)\) be real.  For every
\(\alpha\in\{-\kappa,0,\kappa\}\), there is a unique strongly
continuous two-parameter evolution family
$U_A(t,s)\in\mathcal B(H^\alpha),$ for $-T\le s,t\le T,$
characterized by
\begin{equation*}
 U_A(t,s)f
 =\ee^{-\ii(t-s)h}f
 -\ii\int_s^t\ee^{-\ii(t-r)h}A(r)
 U_A(r,s)f\,\dd r.
\end{equation*}
  The family obeys
\begin{equation}\label{eq:prescribed-propagator-composition}
 U_A(s,s)=\1,\qquad U_A(t,r)U_A(r,s)=U_A(t,s).
\end{equation}
It satisfies \begin{equation}\label{eq:prescribed-propagator-norm}
 \|U_A(t,s)\|_{\mathrm{op}}
 \leq \ee^{C_\kappa
 \int_{\min\{s,t\}}^{\max\{s,t\}}
 \|A(r)\|_{\mathcal W^1}\,\dd r}.
\end{equation}
For every \(u_s\in H^{-\kappa}\), the equation
\[
 \ii\partial_tu=(h+A(t))u,\qquad u(s)=u_s,
\]
has a unique solution \(U_A(t,s)u_s\) in
$ C([-T,T];H^{-\kappa})
 \cap W^{1,1}([-T,T];H^{-\kappa-2})$.
 When $\alpha=0$, $U_A$ is unitary and
 $U_A(t,s)^*=U_A(s,t).$
\end{lemma}

\begin{proof}
For \(\alpha\in\{-\kappa,0,\kappa\}\), let
\(S_0(t)=\ee^{-\ii th}\). Using
\eqref{eq:weighted-Wiener-multiplier},
we obtain
\[
 \|A(t)f\|_{H^\alpha}
 \le C_{\kappa}\|A(t)\|_{\mathcal W^1}\|f\|_{H^\alpha}.
\]
The standard argument implies the result.
\end{proof}

 To state Proposition~\ref{prop:transported-density-comparison}, we use
 the potential-independent full-measure set \(\Omega_*\) and the control
 function \(\mathcal K\) constructed in
 Appendix~\ref{app:transported-density-calculus}.  On this common
 full-measure set, all the elementary centered quadratic expressions
 entering the proof are defined simultaneously.  Consequently, the
 density-mode limits below are well defined for every pair of prescribed
 potentials.  This allows the same estimate to be applied to two paths
 with the same initial field but different associated potentials.

With \(\mathsf H_\varepsilon:=\ee^{-\varepsilon h/2}\), define, for
\(q\in\mathbb Z^d\), the regularized density by
\[
\rho_{\varepsilon,q}(w)
:=\langle \mathsf H_\varepsilon w,
M_q\mathsf H_\varepsilon w\rangle
-\Tr(h^{-1}\mathsf H_\varepsilon
M_q\mathsf H_\varepsilon).
\]
Whenever the following limit exists, set
\[
\widetilde\rho_q^A(t;u)
:=\lim_{\varepsilon\downarrow0}
\rho_{\varepsilon,q}(U_A(t,0)u).
\]
 We also put
\[
 \mathfrak a_T(A,B)
 :=\int_{-T}^T
   \bigl(\|A(r)\|_{\mathcal W^1}
        +\|B(r)\|_{\mathcal W^1}\bigr)\,\dd r,
 \qquad
 I_t:=[\min\{0,t\},\max\{0,t\}].
\]
For \(q\ne0\), the Fourier shift \(M_q\) has zero diagonal, and hence
\[
 \Tr(h^{-1}\mathsf H_\varepsilon
 M_q\mathsf H_\varepsilon)=0,
 \qquad
 \rho_{\varepsilon,q}(w)
 =\langle\mathsf H_\varepsilon w,
 M_q\mathsf H_\varepsilon w\rangle.
\]
\begin{proposition}
\label{prop:transported-density-comparison}
Let \(u\in\Omega_*\), and let
\(A,B\in L^1([-T,T];\mathcal W^1)\) be real.  For every
\(q\in\mathbb Z^d\setminus\{0\}\) and \(t\in[-T,T]\), the two limits
\(\widetilde\rho_q^A(t;u)\) and \(\widetilde\rho_q^B(t;u)\) exist and
satisfy
\begin{equation}\label{eq:transported-density-comparison}
 \begin{aligned}
 |\widetilde\rho_q^A(t;u)-\widetilde\rho_q^B(t;u)|
 &\leq C_T(1+\mathcal K(u))
       \ee^{C\mathfrak a_T(A,B)}
       \int_{I_t}\|A(r)-B(r)\|_{\mathcal W^1}\,\dd r .
 \end{aligned}
\end{equation}
\end{proposition}

\begin{proof}
The existence of the two limits and the estimate
\eqref{eq:transported-density-comparison}, including the construction
of the common full-measure set and the passage to the limit
\(\varepsilon\downarrow0\), are proved in
Appendix~\ref{app:transported-density-calculus}.

\end{proof}

Here we emphasize that the full-measure set \(\Omega_*\) in
Proposition~\ref{prop:transported-density-comparison} has to be chosen
before the prescribed potentials.  This point is essential.  If, for
each fixed potential \(A\), the renormalized terms were defined only
on a set \(\Omega_A\) of full measure, then
\(\mu_0(\Omega_A)=1\) for every deterministic \(A\) would not imply
that an initial field \(u_0\) belongs to \(\Omega_{A_\gamma}\), because
\(A_\gamma\) is itself determined by the path issued from \(u_0\).
Nor can one take the intersection over the uncountable family of
possible potentials.  To compare two paths coupled to the same rough
initial field, both terms must therefore be defined
at that field on one common set.

The previous estimate is stated for the regularized terms, whereas
the self-consistent potential is defined by the fixed Borel representatives
\(\rho_q\).  We now identify all nonzero modes on one common set of full
measure.

\begin{lemma}
\label{lem:sharp-density-identification}
Let \(\eta\) be a finite measure on
\(\mathscr X_T\) such that
\[
 (e_t)_\#\eta\leq C_0\nu,
 \quad |t|\leq T,
\]
and let \(\gamma\mapsto A_\gamma\) be a family of real potentials.  Assume that,
for \(\eta\)-almost every \(\gamma\),
\[
 u_0:=\gamma(0)\in\Omega_*,
 \qquad A_\gamma\in L^1([-T,T];\mathcal W^1),
 \qquad \gamma(t)=U_{A_\gamma}(t,0)u_0.
\]
Then, for \(\dd t\,\dd\eta\)-almost every \((t,\gamma)\),
\begin{equation}\label{eq:sharp-equals-heat}
 \rho_q(\gamma(t))=\widetilde\rho_q^{A_\gamma}(t;\gamma(0)), \quad q\ne0.
\end{equation}
\end{lemma}

\begin{proof}
Using Wick's rule and Gaussian hypercontractivity, we obtain, for every finite \(r\),
\[
 \|\rho_{R,q}-\rho_q\|_{L^r(\mu_0)}
 +\|\rho_{\varepsilon,q}-\rho_q\|_{L^r(\mu_0)}
 \longrightarrow0
\]
as \(R\to\infty\) and \(\varepsilon\downarrow0\).  The bound
\(\dd\nu/\dd\mu_0\le\mathfrak Z^{-1}\) implies the same convergence in
\(L^r(\nu)\).

Choose \(\varepsilon_j\downarrow0\) such that
\begin{equation}\label{eq:heat-density-subsequence}
 \sum_{0<|q|\le j}
 \|\rho_{\varepsilon_j,q}-\rho_q\|_{L^1(\nu)}
 \leq2^{-2j}.
\end{equation}
Using \eqref{eq:marginal-domination-transfer}, we obtain
\begin{align*}
 &\sum_{j\geq1}\sum_{0<|q|\le j}
 \int_{-T}^T\!\int
 |\rho_{\varepsilon_j,q}(\gamma(t))
       -\rho_q(\gamma(t))|\,\dd\eta\,\dd t
 \leq2TC_0\sum_{j\geq1}\sum_{0<|q|\le j}
 \|\rho_{\varepsilon_j,q}-\rho_q\|_{L^1(\nu)}
 \leq2TC_0\sum_{j\geq1}2^{-2j}<\infty.
\end{align*}
By Fubini's theorem, outside one
\((\dd t\otimes\dd\eta)\)-null set the sequence converges to
\(\rho_q(\gamma(t))\) for every \(q\ne0\).  The full
heat limit exists by
Proposition~\ref{prop:transported-density-comparison}; hence the same
value is
\(\widetilde\rho_q^{A_\gamma}(t;\gamma(0))\).  This proves
\eqref{eq:sharp-equals-heat}.
\end{proof}

\subsection{Uniqueness of Hartree trajectories and Liouville curves}
\label{subsec:pathwise-uniqueness}

This subsection turns the mode comparison into uniqueness.  We first
remove the spatially constant part of each potential.  The nonzero
density modes are unchanged when a path is multiplied by a scalar
phase, so Proposition~\ref{prop:transported-density-comparison} can be
summed against \(\langle q\rangle|\widehat v(q)|\).  Gronwall's lemma
identifies the mean-zero potentials and hence the corresponding paths
after their scalar phases have been removed.  In
Theorem~\ref{thm:crossed-product-DNS}, we then identify the zero modes,
recover the scalar phases, and prove pathwise uniqueness.  Finally,
Theorem~\ref{thm:DLU} gives uniqueness of the dominated Liouville curve
and identifies it with the push-forward by the Hartree NLS flow.

For a real potential \(V\in L^1([-T,T];\mathcal W^1)\), write
\[
 V^\circ(t):=V(t)-\widehat{V(t)}(0),
 \qquad
 \phi_V(t):=\int_0^t\widehat{V(s)}(0)\,\dd s.
\]
The scalar part of the potential can then be removed by the identity
\begin{equation}\label{eq:scalar-phase-propagator}
 U_V(t,0)u=\ee^{-\ii\phi_V(t)}U_{V^\circ}(t,0)u.
\end{equation}

\begin{proposition}
\label{prop:mean-zero-volterra-uniqueness}
Let \(u_0\in\Omega_*\), and let
\(A,B\in L^1([-T,T];\mathcal W^1)\) be real.
Assume that there are measurable modes \(\rho_q^A,\rho_q^B\) such that
\[
 A(t,x)=\sum_q\widehat v(q)\rho_q^A(t)\ee^{-\ii q\cdot x},
 \qquad
 B(t,x)=\sum_q\widehat v(q)\rho_q^B(t)\ee^{-\ii q\cdot x}
\]
define \(A,B\in L^1([-T,T];\mathcal W^1)\) and satisfy, for every
\(q\ne0\),
\[
 \rho_q^A(t)=\widetilde\rho_q^A(t;u_0),
 \qquad
 \rho_q^B(t)=\widetilde\rho_q^B(t;u_0)
 \quad\text{for almost every }t.
\]
Then
\begin{equation}\label{eq:gauge-reduced-propagator-equality}
 A^\circ=B^\circ,\qquad U_{A^\circ}(t,0)u_0=U_{B^\circ}(t,0)u_0,
 \quad |t|\le T.
\end{equation}
\end{proposition}

\begin{proof}
By \eqref{eq:scalar-phase-propagator},
\[
 \ee^{\ii\phi_A(t)}U_A(t,0)u_0=U_{A^\circ}(t,0)u_0,
 \qquad
 \ee^{\ii\phi_B(t)}U_B(t,0)u_0=U_{B^\circ}(t,0)u_0.
\]
It is easy to see that, for every scalar phase
\(\ee^{\ii\theta}\),
\[
 \rho_{\varepsilon,q}(\ee^{\ii\theta}u)
 =\rho_{\varepsilon,q}(u).
\]
Hence, for every \(q\ne0\) and almost every \(t\),
\[
 \widetilde\rho_q^{A^\circ}(t;u_0)=\rho_q^A(t),
 \qquad
 \widetilde\rho_q^{B^\circ}(t;u_0)=\rho_q^B(t).
\]
Using Proposition~\ref{prop:transported-density-comparison}, we obtain, for every \(q\ne0\),
\[
 |\rho_q^A(t)-\rho_q^B(t)|
 \le C\int_{I_t}
 \|A^\circ(r)-B^\circ(r)\|_{\mathcal W^1}\,\dd r.
\]
For almost every \(t\),
\[
 \|A^\circ(t)-B^\circ(t)\|_{\mathcal W^1}
 =\sum_{q\ne0}\langle q\rangle|\widehat v(q)|
   |\rho_q^A(t)-\rho_q^B(t)|.
\]
Multiplying by \(\langle q\rangle|\widehat v(q)|\) and summing over \(q\ne0\),
we obtain, for almost every $t\in[-T,T]$,
\begin{equation*}
 \|A^\circ(t)-B^\circ(t)\|_{\mathcal W^1}
 \le C\|v\|_{\mathcal W^1}
 \int_{I_t}\|A^\circ(r)-B^\circ(r)\|_{\mathcal W^1}\,\dd r.
\end{equation*}
Then Gronwall's lemma implies that 
\[
 A^\circ=B^\circ
 \quad\text{in }L^1([-T,T];\mathcal W^1).
\]
 Thus
Lemma~\ref{lem:prescribed-potential-propagators} implies
\eqref{eq:gauge-reduced-propagator-equality}.
\end{proof}

Now we are ready to prove pathwise uniqueness.

\begin{theorem}
\label{thm:crossed-product-DNS}
Let \(\overline\eta\) be a finite positive Borel
measure on \(\mathscr X_T\times\mathscr X_T\).  Assume that
the following holds for \(\overline\eta\)-almost every pair:
\[
 u_0:=\gamma_1(0)=\gamma_2(0),
\]
and, for \(j\in\{1,2\}\),
\[
 V_j:=A_{\gamma_j}=A\circ\gamma_j\in L^1([-T,T];\mathcal W^1)
 \quad\text{is real},
 \qquad
 \gamma_j(t)=U_{V_j}(t,0)u_0\quad(|t|\le T),
\]
where \(A\) is the Borel potential map defined in
\eqref{eq:Borel-self-consistent-potential}.
Let \(\operatorname{pr}_j(\gamma_1,\gamma_2):=\gamma_j\) be the
coordinate projections.  Assume moreover that
\[
 (e_t\circ\operatorname{pr}_j)_\#\overline\eta
 \le C_j\nu,
 \qquad |t|\le T,\quad j\in\{1,2\}.
\]
Then
\[
 \overline\eta\bigl\{(\gamma_1,\gamma_2):
 \gamma_1(t)=\gamma_2(t)\text{ for every }|t|\le T\bigr\}
 =\overline\eta(\mathscr X_T\times\mathscr X_T).
\]
\end{theorem}

\begin{proof}
Using the marginal bound at \(t=0\), together with \(\nu\ll\mu_0\), we obtain
$u_0\in\Omega_*$
for $\overline\eta$-almost every pair.
Let \(\eta_j=(\operatorname{pr}_j)_\#\overline\eta\).   Applying Lemma~\ref{lem:sharp-density-identification} to
\(\eta_1\) and \(\eta_2\), and then discarding one
\(\overline\eta\)-null set, gives, simultaneously for
\(j\in\{1,2\}\) and every \(q\ne0\),
\[
 \rho_q(\gamma_j(t))
 =\widetilde\rho_q^{V_j}(t;u_0)
 \quad\text{for almost every }t.
\]
Using an argument similar to that in Lemma~\ref{lem:sharp-density-identification}, we can choose a subsequence such that, outside one further \(\overline\eta\)-null set,
\begin{equation*}
 \rho_{R_m,0}(\gamma_j(t))
 \longrightarrow\rho_0(\gamma_j(t))
 \quad\text{for almost every }t,
 \qquad j\in\{1,2\}.
\end{equation*}
By marginal domination and the \(\nu\)-full convergence set in
the definition of \(A\), after discarding another null set, the Fourier
series defining
\(A(\gamma_j(t))\) converges for almost every \(t\) and \(j\in\{1,2\}\).
Using Proposition~\ref{prop:mean-zero-volterra-uniqueness}, we obtain
\[
 V_1^\circ=V_2^\circ
\]
and, for every \(|t|\le T\),
\[
 w(t)
 :=U_{V_1^\circ}(t,0)u_0
 =U_{V_2^\circ}(t,0)u_0.
\]
By \eqref{eq:scalar-phase-propagator}, we obtain
\[
 w(t)=\ee^{\ii\phi_{V_j}(t)}\gamma_j(t),
 \qquad j\in\{1,2\},
\]
which implies
\[
 \rho_{R_m,0}(\gamma_1(t))
 =\rho_{R_m,0}(w(t))
 =\rho_{R_m,0}(\gamma_2(t)).
\]
Passing to the limit, we obtain
\[
 \rho_0(\gamma_1(t))=\rho_0(\gamma_2(t))
 \quad\text{for almost every }t.
\]
Thus
\[
 \widehat{V_1(t)}(0)=\widehat{V_2(t)}(0)
 \quad\text{for almost every }t,
\]
and
\[
 \phi_{V_1}(t)=\phi_{V_2}(t),
 \qquad |t|\le T.
\]
Hence, for $|t|\le T$, we obtain
\[
 \gamma_1(t)
 =\ee^{-\ii\phi_{V_1}(t)}w(t)
 =\ee^{-\ii\phi_{V_2}(t)}w(t)
 =\gamma_2(t).
\]
Thus the result follows.
\end{proof}

We now use the previous pathwise uniqueness result to prove uniqueness
of dominated Liouville curves and to identify the unique curve with the
push-forward by the Hartree NLS flow.

\begin{theorem}
\label{thm:DLU}
Let \(d\in\{2,3\}\), and let $v$ satisfy \eqref{eq:finite-potential-regularity}. Fix \(1/2<\kappa<1\) and \(T>0\), and let
\(S_t\) be the dimension-\(d\) Hartree NLS flow of
Theorem~\ref{thm:canonical-flow}.
Let \((\sigma_t^1)_{|t|\le T}\) and
\((\sigma_t^2)_{|t|\le T}\) be weakly continuous curves of finite
positive Borel measures on \(H^{-\kappa}(\mathbb T^d)\) satisfying
\[
 \int G\,\dd\sigma_t^j-\int G\,\dd\sigma_s^j
 =\int_s^t\!\int\Lcal_0G\,\dd\sigma_r^j\,\dd r
\]
for every \(G\in\Acal_{\rm cyl}\), every \(s,t\in[-T,T]\), and
\(j\in\{1,2\}\).  Assume moreover that, for some finite constants
\(C_1,C_2\),
\[
 0\le\sigma_t^j\le C_j\nu,
 \qquad |t|\le T,
 \quad j\in\{1,2\}.
\]
If \(\sigma_0^1=\sigma_0^2\), then
\[
 \sigma_t^1=\sigma_t^2= (S_t)_\#\sigma_0^1,
 \qquad |t|\le T.
\]
Moreover, 
the same conclusion holds if weak continuity is replaced by the following condition:
there is a countable $\mathbb Q(\ii)$-linear subspace
\(\mathscr D\subset\Acal_{\rm cyl}\), dense in \(L^2(\nu)\), such
that \(t\mapsto\int G\,\dd\sigma_t^j\) is continuous for every
\(G\in\mathscr D\). 
\end{theorem}

\begin{proof}
 Write
\(\sigma_0:=\sigma_0^1=\sigma_0^2\).  Testing the Liouville equations
with \(1\) shows that
\[
 \sigma_t^j(H^{-\kappa})=\sigma_0(H^{-\kappa}),
 \qquad |t|\le T,
 \quad j\in\{1,2\}.
\]
If this common mass is zero, then \(\sigma_t^1=\sigma_t^2=0\).
Otherwise, divide all measures by this mass.
It is therefore enough to consider the case in which \(\sigma_0\) and
all measures below are probability measures.

Proposition~\ref{prop:common-superposition} gives probability measures
\(\eta_1,\eta_2\) on the Polish path space \(\mathscr X_T\) such that
\[
 (e_t)_\#\eta_j=\sigma_t^j,
 \qquad |t|\le T.
\]
Taking regular conditional
probabilities \(u\mapsto\eta_{j,u}\) with respect to \(e_0\), we obtain
\[
 \eta_j(\dd\gamma)
 =\int\eta_{j,u}(\dd\gamma)\,\dd\sigma_0(u),
 \qquad
 \eta_{j,u}\{\gamma:e_0(\gamma)=u\}=1
\]
for \(\sigma_0\)-almost every \(u\).  Their conditional product defines
the coupling
\[
 \overline\eta
 :=\int\eta_{1,u}\otimes\eta_{2,u}\,\dd\sigma_0(u)
 \quad\text{on }\mathscr X_T\times\mathscr X_T.
\]
Its coordinate marginals are \(\eta_1,\eta_2\), and it is concentrated
on pairs with the same initial value:
\[
 (\operatorname{pr}_j)_\#\overline\eta=\eta_j,
 \qquad j\in\{1,2\},
 \qquad
 \overline\eta\{(\gamma_1,\gamma_2):
 e_0(\gamma_1)=e_0(\gamma_2)\}=1,
\]
and
\[
 (e_t\circ\operatorname{pr}_j)_\#\overline\eta
 =(e_t)_\#\eta_j
 =\sigma_t^j
 \le C_j\nu.
\]
Using Proposition~\ref{prop:common-superposition},
Lemma~\ref{lem:prescribed-potential-propagators} and
Theorem~\ref{thm:crossed-product-DNS}, we obtain, for $|t|\le T$,
\[
 \gamma_1(t)=\gamma_2(t)
 \quad\text{for }\overline\eta\text{-almost every pair.}
\]
Consequently,
\[
 \sigma_t^1
 =(e_t\circ\operatorname{pr}_1)_\#\overline\eta
 =(e_t\circ\operatorname{pr}_2)_\#\overline\eta
 =\sigma_t^2,
 \qquad |t|\le T.
\]
 Define
\[
 \tau_t:=(S_t)_\#\sigma_0.
\]
Theorem~\ref{thm:canonical-flow} implies \(S_0u=u\) for \(\nu\)-almost
every \(u\), \((S_t)_\#\nu=\nu\), and
\(t\mapsto S_tu\) is continuous in \(H^{-\kappa}\) for
\(\nu\)-almost every \(u\).  Hence \(\tau_0=\sigma_0\), and
\[
 \tau_t(E)=\sigma_0(S_t^{-1}E)
 \le C_0\nu(S_t^{-1}E)=C_0\nu(E)
\]
for every Borel set \(E\) and $C_0=\min\{C_1,C_2\}$. It is easy to see that \((\tau_t)\) is a dominated Liouville curve with
\(\tau_0=\sigma_0\). As a result, we have
\(\sigma_t^1=\sigma_t^2=\tau_t=(S_t)_\#\sigma_0\). Under the second
continuity assumption, approximation in \(L^2(\nu)\), together with
\(0\le\sigma_t^j\le C_j\nu\), extends the continuity from
\(\mathscr D\) to every bounded continuous test function, for each
\(j\in\{1,2\}\). Hence both curves are weakly continuous.  
\end{proof}

Theorem~\ref{thm:DLU} proves statement~\textup{(U)} of
Subsection~\ref{subsec:proof-outline}.  Hence
Subsection~\ref{subsec:Liouville-transfer} applies to the
renormalized Hartree NLS dynamics in dimensions two and three.

\begin{proof}[Proof of Proposition~\ref{prop:bounded-correlations}]
By Theorem~\ref{thm:DLU}, statement~\textup{(U)} holds.
The conclusion follows from Theorem~\ref{thm:bounded-multipoint}.
\end{proof}

\begin{proof}[Proof of Theorem~\ref{thm:main-identified-dynamics}]
Part~\textup{(i)} is Theorem~\ref{thm:one-generator}, including its
uniform bound.  Part~\textup{(ii)} follows from
Corollary~\ref{cor:strong-identified-dynamics}, since statement~\textup{(U)} has now
been proved.
\end{proof}

\section{Wick observables on finitely many Fourier modes}
\label{sec:Wick-observables}

Proposition~\ref{prop:bounded-correlations} gives the dynamical limits for bounded
cylinder observables. We now turn to ordinary Wick polynomials depending
on finitely many Fourier modes and prove Theorem~\ref{thm:main-Wick}.
Since their quantizations are unbounded, we cannot apply the bounded
multiplication argument of Theorem~\ref{thm:bounded-multipoint} directly.

Throughout the section, $J$, the finite Fourier spaces $E_j$, and the
polynomials $P_j$ are fixed independently of $\lambda$. We write
\[
 X_{\lambda,j}=\Op_{\lambda,E_j}^{\rm W}(P_j),\qquad
 X_j=(P_j)_{E_j},\qquad d_j=\deg P_j,\qquad D=\sum_{j=1}^Jd_j,
\]
and, for $\boldsymbol t=(t_1,\ldots,t_J)$, set
\[
 Y_\lambda(\boldsymbol t)=\prod_{j=1}^J\Ubb_\lambda(t_j)(X_{\lambda,j}),
 \qquad Y(\boldsymbol t,u)=\prod_{j=1}^JX_j(S_{t_j}u).
\]
Their traces are denoted by
\begin{equation}\label{eq:Wick-correlation-notation}
 \Cscr_\lambda(\boldsymbol t)=\Tr(Y_\lambda(\boldsymbol t)\Gamma_\lambda),
 \qquad \Cscr(\boldsymbol t)=\int Y(\boldsymbol t,u)\,\dd\nu(u).
\end{equation}
By H\"older's inequality and invariance of $\nu$,
$Y\in C(\mathbb R^J;L^2(\nu))$.

Subsection~\ref{subsec:Wick-static-Hilbert} derives the two-point
limit from equilibrium convergence and convergence of the Heisenberg
dynamics, and proves the cutoff estimates needed for longer products.
In Subsection~\ref{subsec:Gaussian-time-Wick}, these estimates give
convergence of the trace with a factor $\Gamma_\lambda^{1/J}$ between
observables. Lemma~\ref{lem:directional-Gaussian-Wick}
then gives convergence of a Gaussian time average of $\Cscr_\lambda$
of width $c\lambda$, for each fixed $c>0$.
Subsection~\ref{subsec:Gaussian-removal-Wick} removes this average.
In dimension two, the first time-derivative bound and the logarithmic
particle-number estimate suffice. For smooth interactions in dimension
three, higher time derivatives are controlled by the Gibbs estimates of
Appendix~\ref{app:Wick-product-estimates}, and
higher-order Gaussian differences then give the limit. This completes
the proof of Theorem~\ref{thm:main-Wick}.

\subsection{Finite-mode estimates and two-point correlations}
\label{subsec:Wick-static-Hilbert}

We first prove the two-point limits in both dimensions. We then prove the Schatten and radial-cutoff
estimates needed for longer products.

\begin{proposition}
\label{prop:finite-mode-Wick-Hilbert}
Assume \eqref{eq:finite-potential-regularity}, with $d\in\{2,3\}$.
Let $P_1:E_1\to\mathbb C$ and $P_2:E_2\to\mathbb C$ be polynomials
on fixed finite Fourier spaces $E_1,E_2$. We use the coordinate
convention of \eqref{eq:classical-Wick-symbol}, so that
$(P_j)_{E_j}(u)=P_j(z_{E_j}(u),\overline{z_{E_j}(u)})$.
Then, as $\lambda\downarrow0$,
\[
 \Op_{\lambda,E_1}^{\rm W}(P_1)\to_{\rm id}(P_1)_{E_1},
 \qquad
 \Op_{\lambda,E_2}^{\rm W}(P_2)\to_{\rm id}(P_2)_{E_2}
.
\]
Consequently, for arbitrary fixed $s,t\in\mathbb R$, as $\lambda\downarrow0$,
\begin{equation}\label{eq:finite-mode-Wick-two-point}
\begin{aligned}
 &\Tr\!\left(
   \bigl(\Ubb_\lambda(s)\Op_{\lambda,E_1}^{\rm W}(P_1)\bigr)^*
   \Ubb_\lambda(t)\bigl(\Op_{\lambda,E_2}^{\rm W}(P_2)\bigr)
   \Gamma_\lambda\right)\\
 &\qquad\longrightarrow
 \int\overline{(P_1)_{E_1}(S_su)}(P_2)_{E_2}(S_tu)\,\dd\nu(u)
.
\end{aligned}
\end{equation}\end{proposition}
\begin{proof}
We first apply Lemma~\ref{lem:finite-mode-mixed-calculus} to obtain
\[
 X_{\lambda,j}=\Op_{\lambda,E_j}^{\rm W}(P_j)
       \to_{\rm id}X_j=(P_j)_{E_j}.
\]
For a fixed polynomial $P$ on $E$, with
$X_\lambda=\Op_{\lambda,E}^{\rm W}(P)$ and $X=(P)_E$, we also have
\begin{equation}\label{eq:finite-mode-norm-approximation}
 \|X_\lambda-\Op_{\lambda,\gH}^{\rm A}F\|_{2,\lambda}^2
 \longrightarrow\|X-F\|_{L^2(\nu)}^2,
 \qquad F\in\Acal_{\rm cyl}.
\end{equation}
Indeed, we expand the square and apply the norm and inner-product
convergence in Lemma~\ref{lem:identified-calculus}, together with
$\Op_{\lambda,\gH}^{\rm A}F\to_{\rm id}F$.

We then apply Corollary~\ref{cor:strong-identified-dynamics}, in the
form \eqref{eq:identified-dynamics-strong}, to obtain
\[
 \Ubb_\lambda(s)X_{\lambda,1}\to_{\rm id}\Ucal(s)X_1,
 \qquad
 \Ubb_\lambda(t)X_{\lambda,2}\to_{\rm id}\Ucal(t)X_2.
\]
The trace in \eqref{eq:finite-mode-Wick-two-point} is the weighted
inner product of these two vectors. By
Lemma~\ref{lem:identified-calculus}\textup{(ii)}, it converges to
\[
 \langle\Ucal(s)X_1,\Ucal(t)X_2\rangle_{L^2(\nu)}
 =\int\overline{X_1(S_su)}X_2(S_tu)\,\dd\nu(u),
\]
since $\Ucal(t)F=F\circ S_t$. Thus
\eqref{eq:finite-mode-Wick-two-point} follows.
\end{proof}

We next prove the cutoff estimates used for longer products. Fix
$\chi\in C_c^\infty([0,\infty);[0,1])$, equal to one on $[0,1]$ and
zero on $[2,\infty)$, and put $\chi_K(r)=\chi(r/K)$ for $K\ge1$.
For a fixed finite Fourier space $E$ and a polynomial $P$ on $E$, set
\[
 \mathcal N_E=\dG(\Pi_E),\qquad
 X_\lambda=\Op_{\lambda,E}^{\rm W}(P),\qquad
 X_\lambda^{[K]}=\chi_K(\lambda\mathcal N_E)
                  X_\lambda\chi_K(\lambda\mathcal N_E).
\]

\begin{lemma}
\label{lem:static-Wick-Schatten}
For every fixed $2\le p<\infty$,
\begin{equation}\label{eq:static-Wick-Schatten}
 \sup_{0<\lambda\le\lambda_0}
 \|\Gamma_\lambda^{1/p}X_\lambda\|_{\mathfrak S^p}<\infty,
\end{equation}
and
\begin{equation}\label{eq:radial-Wick-tail}
 \lim_{K\to\infty}\sup_{0<\lambda\le\lambda_0}
 \|\Gamma_\lambda^{1/p}(X_\lambda-X_\lambda^{[K]})\|_{\mathfrak S^p}=0.
\end{equation}
\end{lemma}
\begin{proof}
 Assume first that $P$ is real-valued,
let $d_P=\deg P$, and choose an integer $r\ge1$ such that $2r\ge p$.
As in Lemma~\ref{lem:app-field-sector-bounds}, write
$\mathsf E_n^E=\1_{\{\mathcal N_E=n\}}$.
The sector bounds give
\[
 \mathsf E_m^E X_\lambda \mathsf E_n^E=0 \quad\text{if }|m-n|>d_P,
 \qquad
 \|\mathsf E_m^E X_\lambda \mathsf E_n^E\|_{\mathrm{op}}
 \le C(1+\lambda n)^{d_P/2}.
\]

For $L\ge K\ge1$, the definition of the cutoffs implies
$\mathsf E_m^E(X_\lambda^{[L]}-X_\lambda^{[K]})\mathsf E_n^E=0$ whenever
$\lambda\max\{m,n\}\le K$. Thus a nonzero block satisfies
$|m-n|\le d_P$ and
\[
 K<\lambda\max\{m,n\}
 \le\lambda(n+d_P)
 \le(1+\lambda_0d_P)(1+\lambda n).
\]
Consequently,
\[
 \|\mathsf E_m^E X_\lambda^{[L]}\mathsf E_n^E\|_{\mathrm{op}}
 \le C(1+\lambda n)^{d_P/2},
 \qquad
 \|\mathsf E_m^E(X_\lambda^{[L]}-X_\lambda^{[K]})\mathsf E_n^E\|_{\mathrm{op}}
 \le CK^{-1}(1+\lambda n)^{d_P/2+1},
\]
with constants independent of $\lambda,K,L$.
Each factor changes the particle number by at most $d_P$,
so in a product of $r$ factors all intermediate weights
are comparable to the initial weight $1+\lambda n$.
Multiplying the block bounds and summing over sectors by
Cauchy--Schwarz inequality, we obtain
\begin{align*}
	\langle\psi,|X_\lambda^{[L]}|^{2r}\psi\rangle
	&=\|(X_\lambda^{[L]})^r\psi\|^2
	\le C_r
	\|(\1+\lambda\mathcal N_E)^{rd_P/2}\psi\|^2,\\
	\langle\psi,|X_\lambda^{[L]}-X_\lambda^{[K]}|^{2r}\psi\rangle
	&=\|(X_\lambda^{[L]}-X_\lambda^{[K]})^r\psi\|^2
	\le C_rK^{-2r}
	\|(\1+\lambda\mathcal N_E)^{r(d_P+2)/2}\psi\|^2
\end{align*}
for every finite-particle vector $\psi$,
where the equalities use self-adjointness.
Taking expectations in the state $\omega_\lambda$ and using
\eqref{eq:finite-mode-number-moment-bound}, we obtain
\[
 \omega_\lambda(|X_\lambda^{[L]}|^{2r})\le C_r,
 \qquad
 \omega_\lambda(|X_\lambda^{[L]}-X_\lambda^{[K]}|^{2r})
 \le C_rK^{-2r},
\]
uniformly in $\lambda,K,L$.
The Araki--Lieb--Thirring inequality \cite{ArakiALT} and
H\"older's inequality imply
\[
 \|\Gamma_\lambda^{1/p}X_\lambda^{[L]}\|_{\mathfrak S^p}
 \le C_p,
 \qquad
 \|\Gamma_\lambda^{1/p}
   (X_\lambda^{[L]}-X_\lambda^{[K]})\|_{\mathfrak S^p}
 \le C_pK^{-1}.
\]

Thus, for each fixed $\lambda$,
$\Gamma_\lambda^{1/p}X_\lambda^{[L]}$ converges in
$\mathfrak S^p$ as $L\to\infty$. For every finite-particle vector
$\psi$, one has $X_\lambda^{[L]}\psi=X_\lambda\psi$ for all
sufficiently large $L$, so the limit is $\Gamma_\lambda^{1/p}X_\lambda$.
The previous estimates imply
\eqref{eq:static-Wick-Schatten} and \eqref{eq:radial-Wick-tail}
after letting $L\to\infty$.
For complex-valued $P$, apply the result separately to its real
and imaginary parts and use the triangle inequality.
\end{proof}

Write $\mathsf E_n=\1_{\{\mathcal N=n\}}$ and
$\mathsf P_N=\sum_{n=0}^N\mathsf E_n$, as in the proof of
Theorem~\ref{thm:one-generator}.
Lemma~\ref{lem:app-field-sector-bounds} implies
\[
 \|\Ubb_\lambda(t_j)X_{\lambda,j}\mathsf E_n\|_{\mathrm{op}}
 \le C_j(1+\lambda n)^{d_j/2}.
\]
Each factor changes the particle number by at most $d_j$. Thus
\begin{equation}\label{eq:Wick-product-number-bound}
 \sup_{\boldsymbol t}
 \|Y_\lambda(\boldsymbol t)(\1+\lambda\mathcal N)^{-D/2}\|_{\mathrm{op}}
 \le C.
\end{equation}
The same bound holds for $Y_{\lambda,N}$, obtained by replacing each
$X_{\lambda,j}$ with $\mathsf P_NX_{\lambda,j}\mathsf P_N$.
Since $\mathcal N$ commutes with $\Gamma_\lambda$,
\begin{align*}
 \|Y_\lambda\Gamma_\lambda\|_{\mathfrak S^1}
 &\le C\omega_\lambda((\1+\lambda\mathcal N)^{D/2}),\\
 \|Y_\lambda\Gamma_\lambda^{1/2}\|_{\mathfrak S^2}
 &\le C\omega_\lambda((\1+\lambda\mathcal N)^D)^{1/2}.
\end{align*}
Both moments are finite by \eqref{eq:fixed-lambda-number-moments}.
Moreover, $(Y_{\lambda,N}-Y_\lambda)\mathsf E_n=0$ for $n\le N-D$.
Thus, for $\theta\in\{1,1/2\}$,
\[
 \sup_{\boldsymbol t}
 \|(Y_{\lambda,N}-Y_\lambda)\Gamma_\lambda^\theta\|_{\mathfrak S^{1/\theta}}
 \le C\Bigl[
 \omega_\lambda\!\left((\1+\lambda\mathcal N)^{D/(2\theta)}
                 \1_{\{\mathcal N>N-D\}}\right)\Bigr]^\theta
 \longrightarrow0
\]
as $N\to\infty$ for each fixed $\lambda>0$.

\subsection{Gibbs-weighted traces and Gaussian averages}
\label{subsec:Gaussian-time-Wick}

We first prove a commutator estimate for bounded cylinder operators.
Together with the radial cutoffs from
Subsection~\ref{subsec:Wick-static-Hilbert}, it gives convergence of
the trace with equal Gibbs powers. We then use the Gaussian identity
to obtain convergence of the time-averaged correlation in both
dimensions. The subsection ends with the Gaussian difference formula
used to remove the average.

\begin{lemma}
\label{lem:bounded-cylinder-Gibbs-commutator}
Let $F\in\Acal_{\rm cyl}$ and $B_\lambda=\Op_{\lambda,\gH}^{\rm A}F$.
For every fixed $2\le p<\infty$, as $\lambda\downarrow0$,
\[
 \|[\Gamma_\lambda^{1/p},B_\lambda]\|_{\mathfrak S^p}
       \longrightarrow0.
\]
Moreover, with
$E_\lambda(t,s)=\Ubb_\lambda(t)B_\lambda-\Ubb_\lambda(s)B_\lambda$,
\[
 \|E_\lambda(t,s)\Gamma_\lambda^{1/p}\|_{\mathfrak S^p}
 +\|\Gamma_\lambda^{1/p}E_\lambda(t,s)\|_{\mathfrak S^p}
 \le C_{F,p}|t-s|^{2/p}
,\qquad t,s\in\mathbb R.
\]
\end{lemma}
\begin{proof}
For a bounded operator $A$, applying the Araki--Lieb--Thirring inequality
\cite{ArakiALT,AudenaertALT}, we obtain
\[
 \|A\Gamma_\lambda^{1/p}\|_{\mathfrak S^p}^p
 \le\omega_\lambda(|A|^p)
 \le\|A\|_{\mathrm{op}}^{p-2}\|A\|_{2,\lambda}^2.
\]
Applying this also to $A^*$ implies
\[
 \|A\Gamma_\lambda^{1/p}\|_{\mathfrak S^p}
 \le\|A\|_{\mathrm{op}}^{1-2/p}\|A\|_{2,\lambda}^{2/p},
 \qquad
 \|\Gamma_\lambda^{1/p}A\|_{\mathfrak S^p}
 \le\|A\|_{\mathrm{op}}^{1-2/p}\|A^*\|_{2,\lambda}^{2/p}.
\]
Using Theorem~\ref{thm:one-generator}, we obtain
\[
 E_\lambda(t,s)=\int_s^t\Ubb_\lambda(u)\Lbb_\lambda B_\lambda\,\dd u,
 \quad
 \|E_\lambda(t,s)\|_{2,\lambda}
 +\|E_\lambda(t,s)^*\|_{2,\lambda}\le C_F|t-s|.
\]
Together with $\|E_\lambda(t,s)\|_{\mathrm{op}}\le2\|F\|_\infty$,
this proves the time-increment estimate.

Choose $0\le\vartheta\in C_c^\infty(\mathbb R)$ with
$\int_{\mathbb R}\vartheta=1$, and, for $\delta>0$, set
\[
 \vartheta_\delta(s)=\delta^{-1}\vartheta(s/\delta),
 \qquad
 B_{\lambda,\delta}
 =\int_{\mathbb R}\vartheta_\delta(s)
       \Ubb_\lambda(s)B_\lambda\,\dd s. 
\]
Differentiating the translated kernel
$\vartheta_\delta(\,\cdot-t)$ in $L^1(\mathbb R)$ shows that
\[
 \Lbb_\lambda B_{\lambda,\delta}
 =-\int_{\mathbb R}\vartheta_\delta'(s)
       \Ubb_\lambda(s)B_\lambda\,\dd s,
 \qquad
 \|\Lbb_\lambda B_{\lambda,\delta}\|_{\mathrm{op}}
 \le C_F\delta^{-1}.
\]
Since
$B_\lambda-B_{\lambda,\delta}
=-\int_{\mathbb R}\vartheta_\delta(s)E_\lambda(s,0)\,\dd s$, we use
Minkowski's inequality and the time-increment estimate to derive
\begin{align*}
	&\|(B_\lambda-B_{\lambda,\delta})
	\Gamma_\lambda^{1/p}\|_{\mathfrak S^p}
	+\|\Gamma_\lambda^{1/p}
	(B_\lambda-B_{\lambda,\delta})\|_{\mathfrak S^p}\le C_{F,p}\int_{\mathbb R}
	\vartheta_\delta(s)|s|^{2/p}\,\dd s
	\le C_{F,p}\delta^{2/p}.
\end{align*}
Let $\alpha=1/p$. The commutator formula implies
\[
 [\Gamma_\lambda^\alpha,B_{\lambda,\delta}]
 =\ii\lambda\int_0^\alpha
   \Gamma_\lambda^{\alpha-u}
   (\Lbb_\lambda B_{\lambda,\delta})
   \Gamma_\lambda^u\,\dd u,
\]
 which also yields
\[
 \|[\Gamma_\lambda^\alpha,B_{\lambda,\delta}]
       \|_{\mathfrak S^p}
 \le\alpha\lambda
       \|\Lbb_\lambda B_{\lambda,\delta}\|_{\mathrm{op}}
 \le C_{F,p}\lambda\delta^{-1}.
\]
Consequently,
\[
 \|[\Gamma_\lambda^{1/p},B_\lambda]\|_{\mathfrak S^p}
 \le C_{F,p}\bigl(\delta^{2/p}+\lambda\delta^{-1}\bigr).
\]
Taking $\delta=\sqrt\lambda$ proves the claim.
\end{proof}

\begin{lemma}
\label{lem:positive-gap-correlations}
Let $J\ge2$ and define
\[
 F_\lambda(\boldsymbol t)
 =\Tr\!\left(\Gamma_\lambda^{1/J}\Ubb_\lambda(t_1)(X_{\lambda,1})
       \cdots\Gamma_\lambda^{1/J}\Ubb_\lambda(t_J)(X_{\lambda,J})\right).
\]
Under \eqref{eq:finite-potential-regularity},
$F_\lambda\to \Cscr$ locally uniformly on $\mathbb R^J$
as $\lambda\downarrow0$, and
\begin{equation}\label{eq:positive-gap-uniform-bound}
 \sup_{0<\lambda\le\lambda_0}\sup_{\boldsymbol t\in\mathbb R^J}|F_\lambda(\boldsymbol t)|<\infty.
\end{equation}
\end{lemma}

\begin{proof}
Let $\alpha=1/J$. We use Lemma~\ref{lem:static-Wick-Schatten} and
the Schatten--H\"older inequality to obtain
\[
 |F_\lambda(\boldsymbol t)|
 \le\prod_{j=1}^J\|\Gamma_\lambda^\alpha
                   \Ubb_\lambda(t_j)X_{\lambda,j}\|_{\mathfrak S^J}
 =\prod_{j=1}^J\|\Gamma_\lambda^\alpha X_{\lambda,j}\|_{\mathfrak S^J}
 \le C.
\]
For $M\ge1$, let
\[
 F_{j,M}(u)=\chi_M(\|\Pi_{E_j}u\|^2)^2X_j(u),\qquad
 B_{\lambda,j,M}=\Op_{\lambda,\gH}^{\rm A}F_{j,M},\qquad
 A_j=\Ubb_\lambda(t_j)B_{\lambda,j,M},
\]
and set
\[
 F_{\lambda,M}(\boldsymbol t)=\Tr\Big(\prod_{j=1}^J(\Gamma_\lambda^\alpha A_j)\Big),
 \qquad \Cscr_M(\boldsymbol t)=\int\prod_{j=1}^JF_{j,M}(S_{t_j}u)\,\dd\nu.
\]
By \eqref{eq:radial-Wick-tail}, Lemma~\ref{lem:domain-safe-radial-AW},
and the Schatten--H\"older inequality, we obtain
\[
 \sup_{\boldsymbol t}|F_\lambda-F_{\lambda,M}|
 \le C\varepsilon_M+C_M'\sqrt\lambda,\qquad
 \varepsilon_M=\sum_j\sup_{0<\lambda\le\lambda_0}
 \|\Gamma_\lambda^\alpha(X_{\lambda,j}-X_{\lambda,j}^{[M]})\|_{\mathfrak S^J}
 \xrightarrow[M\to\infty]{}0.
\]

Fix $M$, write $G=\Gamma_\lambda^{1/J}$, and set
\[
 T_\ell=\Tr\!\Big(
    \Big(\prod_{j=1}^\ell A_j\Big)G^\ell
                       \prod_{j=\ell+1}^J(GA_j)\Big),\quad0\le\ell\le J.
\]
We have $T_0=F_{\lambda,M}$, $T_J=\Tr(A_1\cdots A_J\Gamma_\lambda)$, and
\begin{align*}
 T_{\ell-1}-T_\ell
 &=\Tr\!\Big(\Big(\prod_{j<\ell}A_j\Big)
        [G^\ell,A_\ell]\prod_{j>\ell}(GA_j)\Big),\qquad
 [G^\ell,A_\ell]=\sum_{a=0}^{\ell-1}G^a[G,A_\ell]G^{\ell-1-a}.
\end{align*}
Each summand contains one commutator and $J-1$ other factors $G$ in
$\mathfrak S^J$. Thus Lemma~\ref{lem:bounded-cylinder-Gibbs-commutator}
implies, uniformly in the times,
\[
 |T_{\ell-1}-T_\ell|
 \le\ell\|[G,B_{\lambda,\ell,M}]\|_{\mathfrak S^J}
                       \prod_{j\ne\ell}\|F_{j,M}\|_\infty
 \longrightarrow0.
\]
Proposition~\ref{prop:bounded-correlations} implies $T_J\to \Cscr_M$ pointwise.
Applying the time-increment estimate of Lemma~\ref{lem:bounded-cylinder-Gibbs-commutator} with $p=J$
to one factor at a time, and using the Schatten--H\"older inequality, we find
\[
 |F_{\lambda,M}(\boldsymbol t+\boldsymbol h)-F_{\lambda,M}(\boldsymbol t)|
 \le C_M'\sum_j|h_j|^{2/J}.
\]
Hence $F_{\lambda,M}\to\Cscr_M$ locally uniformly.
By dominated convergence, $F_{j,M}\to X_j$ in every finite $L^q(\nu)$,
since $|F_{j,M}|\le|X_j|$. H\"older's inequality and invariance of $\nu$ imply
\[
 \|\Cscr_M-\Cscr\|_\infty
 \le\sum_j\|F_{j,M}-X_j\|_{L^J(\nu)}
                   \prod_{k\ne j}\|X_k\|_{L^J(\nu)}\longrightarrow0.
\]
For each compact $\mathcal K$,
$\limsup_{\lambda\downarrow0}\|F_\lambda-\Cscr\|_{L^\infty(\mathcal K)}
\le C\varepsilon_M+\|\Cscr_M-\Cscr\|_\infty$.
The result follows by letting $M\to\infty$.
\end{proof}

For $J\ge2$, fix
\[
 r_j=-\frac{j-1}{J},\qquad \boldsymbol r=(r_1,\ldots,r_J),\qquad
 D_{\boldsymbol r}=\sum_{j=1}^Jr_j\partial_{t_j}.
\]
For $\delta>0$ and $f\in C(\mathbb R^J;\mathbb C)$ of polynomial
growth, define the Gaussian average along $\boldsymbol r$ by
\[
 G_\delta f(\boldsymbol t)
 =\int_{\mathbb R}\varphi_\delta(s)f(\boldsymbol t+s\boldsymbol r)\,\dd s,
 \qquad
 \varphi_\delta(s)=\frac{\ee^{-s^2/(2\delta^2)}}{\sqrt{2\pi}\,\delta}.
\]
We next transfer the convergence of $F_\lambda$ to a Gaussian
time average of $\Cscr_\lambda$. At a finite spectral cutoff,
redistributing the Gibbs weight amounts to an imaginary-time shift in the direction $\boldsymbol r$. Gaussian averaging
transfers this shift to the kernel, giving a representation involving only real-time values of $F_\lambda$.

\begin{lemma}
\label{lem:directional-Gaussian-Wick}
Under the hypotheses of Lemma~\ref{lem:positive-gap-correlations},
for every fixed $c>0$,
\begin{equation}\label{eq:directional-Gaussian-identity}
 G_{c\lambda}\Cscr_\lambda(\boldsymbol t)
 =\ee^{1/(2c^2)}\int_{\mathbb R}\varphi_c(z)\ee^{\ii z/c^2}
                   F_\lambda(\boldsymbol t+\lambda z\boldsymbol r)\,\dd z.
\end{equation}
In particular,
\begin{equation}\label{eq:directional-Gaussian-limit}
 G_{c\lambda}\Cscr_\lambda\longrightarrow \Cscr
 \quad\text{locally uniformly}.
\end{equation}
\end{lemma}

\begin{proof}
Fix $\lambda>0$. Choose a joint eigenbasis $(\psi_a)_{a\ge1}$ of
$\mathbb H_\lambda$ and $\mathcal N$, with
$\mathbb H_\lambda\psi_a=E_a\psi_a$ and
$\gamma_a=\mathcal Z_\lambda^{-1}\ee^{-\lambda E_a}$.
Set
\[
 \mathsf S_L=\sum_{a=1}^L|\psi_a\rangle\langle\psi_a|,
 \qquad
 X_{\lambda,j}^{(L)}=\mathsf S_LX_{\lambda,j}\mathsf S_L.
\]
Define the truncated traces, with the original normalization of
$\Gamma_\lambda$, by
\begin{align*}
 \Cscr_\lambda^{(L)}(\boldsymbol t)
 =\Tr\!\Big(
    \Bigl[\prod_{j=1}^J\Ubb_\lambda(t_j)(X_{\lambda,j}^{(L)})\Bigr]
    \Gamma_\lambda\Big),\quad
 F_\lambda^{(L)}(\boldsymbol t)
 =\Tr\!\Big(\prod_{j=1}^J
    \Bigl[\Gamma_\lambda^{1/J}
          \Ubb_\lambda(t_j)(X_{\lambda,j}^{(L)})\Bigr]\Big).
\end{align*}
For $\boldsymbol a=(a_1,\ldots,a_J)\in\{1,\ldots,L\}^J$,
with $a_{J+1}=a_1$, let
\[
 \omega_{\boldsymbol a,j}=E_{a_j}-E_{a_{j+1}},\qquad
 \boldsymbol\omega_{\boldsymbol a}
       =(\omega_{\boldsymbol a,1},\ldots,\omega_{\boldsymbol a,J}),\qquad
 x_{\boldsymbol a}=\prod_{j=1}^J
       \langle\psi_{a_j},X_{\lambda,j}\psi_{a_{j+1}}\rangle.
\]
Expanding the traces in this basis, we obtain
\begin{align*}
 \Cscr_\lambda^{(L)}(\boldsymbol t)
 =\sum_{\boldsymbol a\in\{1,\ldots,L\}^J}
       \gamma_{a_1}x_{\boldsymbol a}
       \ee^{\ii\boldsymbol t\cdot\boldsymbol\omega_{\boldsymbol a}},\quad
 F_\lambda^{(L)}(\boldsymbol t)
 =\sum_{\boldsymbol a\in\{1,\ldots,L\}^J}
       \Bigl(\prod_{j=1}^J\gamma_{a_j}^{1/J}\Bigr)
       x_{\boldsymbol a}
       \ee^{\ii\boldsymbol t\cdot\boldsymbol\omega_{\boldsymbol a}}.
\end{align*}
Let $\xi_{\boldsymbol a}
       =\lambda\boldsymbol r\cdot\boldsymbol\omega_{\boldsymbol a}$.
The choice $r_j=-(j-1)/J$ implies
\[
 \boldsymbol r\cdot\boldsymbol\omega_{\boldsymbol a}
 =E_{a_1}-\frac1J\sum_{j=1}^JE_{a_j},\qquad
 \gamma_{a_1}
 =\Bigl(\prod_{j=1}^J\gamma_{a_j}^{1/J}\Bigr)\ee^{-\xi_{\boldsymbol a}}.
\]
We have for every real $\xi$,
\begin{align*}
 \int_{\mathbb R}\varphi_c(z)\ee^{\ii\xi z}\,\dd z
 &=\ee^{-c^2\xi^2/2},\\
 \ee^{1/(2c^2)}\int_{\mathbb R}
       \varphi_c(z)\ee^{\ii z/c^2}\ee^{\ii\xi z}\,\dd z
 &=\ee^{1/(2c^2)-c^2(\xi+c^{-2})^2/2}
  =\ee^{-\xi-c^2\xi^2/2}.
\end{align*}
Changing variables $s=\lambda z$ and using these two formulas we obtain
\begin{align*}
 G_{c\lambda}\Cscr_\lambda^{(L)}(\boldsymbol t)
 &=\int_{\mathbb R}\varphi_c(z)
       \Cscr_\lambda^{(L)}(\boldsymbol t+\lambda z\boldsymbol r)\,\dd z=\sum_{\boldsymbol a\in\{1,\ldots,L\}^J}
       \gamma_{a_1}x_{\boldsymbol a}
       \ee^{\ii\boldsymbol t\cdot\boldsymbol\omega_{\boldsymbol a}}
       \ee^{-c^2\xi_{\boldsymbol a}^2/2}\\
 &=\sum_{\boldsymbol a\in\{1,\ldots,L\}^J}
       \Bigl(\prod_{j=1}^J\gamma_{a_j}^{1/J}\Bigr)
       x_{\boldsymbol a}
       \ee^{\ii\boldsymbol t\cdot\boldsymbol\omega_{\boldsymbol a}}
       \ee^{-\xi_{\boldsymbol a}-c^2\xi_{\boldsymbol a}^2/2}=\ee^{1/(2c^2)}\int_{\mathbb R}\varphi_c(z)\ee^{\ii z/c^2}
       F_\lambda^{(L)}(\boldsymbol t+\lambda z\boldsymbol r)\,\dd z.
\end{align*}
All sums are finite, so the interchanges with the integrals are
immediate. This is \eqref{eq:directional-Gaussian-identity} for the
truncated traces.

We now let $L\to\infty$ with $\lambda$ fixed.
Since $\mathsf S_L$ commutes with $\mathcal N$ and
$\mathbb H_\lambda$, each factor still changes the
particle number by at most $d_j$, and
\[
 \|\Ubb_\lambda(t_j)(X_{\lambda,j}^{(L)})\mathsf E_n\|_{\mathrm{op}}
 =\|\mathsf S_L(\Ubb_\lambda(t_j)X_{\lambda,j}\mathsf E_n)
              \mathsf S_L\|_{\mathrm{op}}
 \le C_j(1+\lambda n)^{d_j/2}.
\]
The ordered products are therefore bounded on $\mathsf E_n\Fock$
by $C(1+\lambda n)^{D/2}$, uniformly in $L$ and $\boldsymbol t$,
and converge strongly there since $\mathsf S_L\to\1$ strongly.
Thus \eqref{eq:fixed-lambda-number-moments}
and dominated convergence in the joint eigenbasis imply
\[
 \Cscr_\lambda^{(L)}(\boldsymbol t)\longrightarrow
       \Cscr_\lambda(\boldsymbol t),\qquad
 \sup_{L,\boldsymbol t}|\Cscr_\lambda^{(L)}(\boldsymbol t)|
 \le C\omega_\lambda\bigl((\1+\lambda\mathcal N)^{D/2}\bigr)<\infty.
\]
For the equally weighted traces, the Araki--Lieb--Thirring inequality
and \eqref{eq:fixed-lambda-number-moments} imply
$\Gamma_\lambda^{1/J}X_{\lambda,j}\in\mathfrak S^J$,
as in the proof of Lemma~\ref{lem:static-Wick-Schatten}.
Since $\mathsf S_L$ commutes with $\Gamma_\lambda$,
\[
 \Gamma_\lambda^{1/J}\Ubb_\lambda(t_j)(X_{\lambda,j}^{(L)})
 =\mathsf S_L\bigl(\Gamma_\lambda^{1/J}
                    \Ubb_\lambda(t_j)X_{\lambda,j}\bigr)\mathsf S_L
 \longrightarrow\Gamma_\lambda^{1/J}\Ubb_\lambda(t_j)X_{\lambda,j}
 \quad\text{in }\mathfrak S^J.
\]
The Schatten--H\"older inequality then implies
\[
 F_\lambda^{(L)}(\boldsymbol t)\longrightarrow F_\lambda(\boldsymbol t),
 \qquad
 \sup_{L,\boldsymbol t}|F_\lambda^{(L)}(\boldsymbol t)|
 \le\prod_{j=1}^J
       \|\Gamma_\lambda^{1/J}X_{\lambda,j}\|_{\mathfrak S^J}.
\]
These two bounds allow dominated convergence in both Gaussian
integrals and yield \eqref{eq:directional-Gaussian-identity}.

Finally, since
\(\ee^{1/(2c^2)}
\int_{\mathbb R}\varphi_c(z)\ee^{\ii z/c^2}\,\dd z=1,
\)
we use \eqref{eq:directional-Gaussian-identity} to have, for every compact
time set $\mathcal K$ and $R>0$,
\begin{align*}
 \|G_{c\lambda}\Cscr_\lambda-\Cscr\|_{L^\infty(\mathcal K)}
 &\le \ee^{1/(2c^2)}
   \sup_{\substack{\boldsymbol t\in\mathcal K\\|z|\le R}}
      |F_\lambda(\boldsymbol t+\lambda z\boldsymbol r)-\Cscr(\boldsymbol t)|
 +C_c\int_{|z|>R}\varphi_c(z)\,\dd z.
\end{align*}
The first term tends to zero by Lemma~\ref{lem:positive-gap-correlations}
and continuity of $\Cscr$. Letting $R\to\infty$, we obtain
\eqref{eq:directional-Gaussian-limit}.
\end{proof}

\subsubsection*{Gaussian differences}

For $f\in C(\mathbb R^J;\mathbb C)$ and $N\in\mathbb N_0$, define
\[
 D_{\boldsymbol r}^N f(\boldsymbol t)
 :=\left.\frac{\dd^N}{\dd s^N}
          f(\boldsymbol t+s\boldsymbol r)\right|_{s=0},
 \qquad D_{\boldsymbol r}^0f=f,
\]
whenever the derivative exists. For $f\in C^N(\mathbb R^J;\mathbb C)$,
this equals $(\sum_{j=1}^Jr_j\partial_{t_j})^Nf$.

Set $G_0=\1$. Let $K\ge1$ be an integer, and suppose that $f$ and its
first $2K$ derivatives along $\boldsymbol r$ are continuous and of
polynomial growth. The variance of the kernel in
$G_{\lambda\sqrt\theta}$ is $\lambda^2\theta$. By direct calculation, we have
for $\theta>0$,
\[
 \partial_\theta\varphi_{\lambda\sqrt\theta}(s)
       =\frac{\lambda^2}{2}\partial_s^2\varphi_{\lambda\sqrt\theta}(s).
\]
Differentiating the Gaussian integral and integrating by parts twice
therefore imply
\begin{align*}
 \frac{\dd}{\dd\theta}(G_{\lambda\sqrt\theta}f)(\boldsymbol t)
 &=\frac{\lambda^2}{2}\int_{\mathbb R}
       \partial_s^2\varphi_{\lambda\sqrt\theta}(s)
                  f(\boldsymbol t+s\boldsymbol r)\,\dd s\\
 &=\frac{\lambda^2}{2}\int_{\mathbb R}
       \varphi_{\lambda\sqrt\theta}(s)
       (D_{\boldsymbol r}^2f)(\boldsymbol t+s\boldsymbol r)\,\dd s
 =\frac{\lambda^2}{2}(G_{\lambda\sqrt\theta}D_{\boldsymbol r}^2f)(\boldsymbol t).
\end{align*}
Polynomial growth makes the boundary terms vanish, and
$G_{\lambda\sqrt\theta}f\to f$ locally uniformly as $\theta\downarrow0$.
Integrating over $\theta\in[0,1]$ we obtain
\[
 (\1-G_\lambda)f=-\frac{\lambda^2}{2}
       \int_0^1G_{\lambda\sqrt\theta}D_{\boldsymbol r}^2f\,\dd\theta.
\]
The Gaussian convolution identity gives
$G_{\delta_1}G_{\delta_2}=G_{\sqrt{\delta_1^2+\delta_2^2}}$,
and $D_{\boldsymbol r}$ commutes with these averages.
Iterating the last identity thus yields
\begin{equation}\label{eq:iterated-Gaussian-remainder}
 (\1-G_\lambda)^Kf
 =\left(-\frac{\lambda^2}{2}\right)^K
   \int_{[0,1]^K}G_{\lambda\sqrt{\theta_1+\cdots+\theta_K}}
                   D_{\boldsymbol r}^{2K}f\,\dd\boldsymbol\theta.
\end{equation}
Since $G_\lambda^j=G_{\sqrt j\,\lambda}$, we use the binomial formula to deduce
\[
 (\1-G_\lambda)^K
 =\1+\sum_{j=1}^K(-1)^j\binom KjG_{\sqrt j\,\lambda}.
\]
For the correlations in \eqref{eq:Wick-correlation-notation}, define
\begin{align}\label{def:Hlambdak}
 H_{\lambda,K}
 =\sum_{j=1}^K(-1)^{j+1}\binom KjG_{\sqrt j\,\lambda}\Cscr_\lambda.
\end{align}
Then $\Cscr_\lambda-H_{\lambda,K}=(\1-G_\lambda)^K\Cscr_\lambda$.

By Lemma~\ref{lem:directional-Gaussian-Wick} and
$\sum_{j=1}^K(-1)^{j+1}\binom Kj=1$, for every fixed $K$,
\[
 H_{\lambda,K}\longrightarrow\Cscr
 \quad\text{locally uniformly on }\mathbb R^J.
\]

\subsection{Removal of the Gaussian average}
\label{subsec:Gaussian-removal-Wick}

We first estimate the correlation and its first time derivatives under
\eqref{eq:finite-potential-regularity}. This removes the Gaussian
average in dimension two. For the three-dimensional limit, we then
assume smoothness of the interaction and prove higher time-derivative
bounds. Recall $g_d$ from \eqref{def:gd} and the projections
$\mathsf E_n,\mathsf P_N$ from
Subsection~\ref{subsec:Wick-static-Hilbert}.

\begin{lemma}
\label{lem:ordered-trace-number-bounds}
Under \eqref{eq:finite-potential-regularity},
$\Cscr_\lambda\in C^1(\mathbb R^J)$ for each fixed $\lambda>0$. For $0<\lambda\le\lambda_0$,
\[
 \sup_{\boldsymbol t}|\Cscr_\lambda(\boldsymbol t)|
       \le Cg_d(\lambda)^{D/2},\qquad
 \sup_{\boldsymbol t}|\partial_{t_j}\Cscr_\lambda(\boldsymbol t)|
       \le Cg_d(\lambda)^{1+D/2}.
\]
\end{lemma}
\begin{proof}
Let $c_\lambda=\lambda\Tr(\gamma_{0,\lambda})$.
Recall \eqref{eq:exact-quantum-generator-formula} and let
\[
 R_{\lambda,j}^{(1)}
 =\ii[\dG(h),X_{\lambda,j}]
 +\frac\ii2\sum_{\eta\in\mathfrak I}w_\eta
   \bigl(B_\lambda(f_\eta)\adop{f_\eta}X_{\lambda,j}
          +\adop{f_\eta}X_{\lambda,j}B_\lambda(f_\eta)\bigr),
\]
where $\adop{b}X=[\dG(b),X]$ as in
\eqref{eq:quantum-density-derivation}. Lemma~\ref{lem:app-field-sector-bounds},
$\|f_\eta\|_\infty\le1$, and
\[
 \|B_\lambda(f_\eta)\mathsf E_n\|_{\mathrm{op}}
 \le\lambda n+\lambda\Tr(\gamma_{0,\lambda})\le C(1+c_\lambda+\lambda n)
\]
imply, for fixed $\lambda>0$ and with constants uniform for
$0<\lambda\le\lambda_0$,
\begin{align*}
 \|X_{\lambda,j}\mathsf E_n\|_{\mathrm{op}}
   &\le C_j(1+\lambda n)^{d_j/2},\qquad
 \|R_{\lambda,j}^{(1)}\mathsf E_n\|_{\mathrm{op}}
   \le C_j(1+\lambda n)^{d_j/2}(1+c_\lambda+\lambda n).
\end{align*}
 As in
Theorem~\ref{thm:one-generator}, their cutoffs by $\mathsf P_N$
satisfy the bounded commutator identity and converge in
$\|\cdot\|_{2,\lambda}$ by \eqref{eq:fixed-lambda-number-moments}.
Thus
$X_{\lambda,j}\in\Dom(\Lbb_\lambda)$ and
$\Lbb_\lambda X_{\lambda,j}=R_{\lambda,j}^{(1)}$.

Set $A_{\lambda,j}(t)=\Ubb_\lambda(t)X_{\lambda,j}$ and
\[
 Y'_{\lambda,j}(\boldsymbol t)
 =\Big(\prod_{k<j}A_{\lambda,k}(t_k)\Big)
       \Ubb_\lambda(t_j)R_{\lambda,j}^{(1)}
   \Big(\prod_{k>j}A_{\lambda,k}(t_k)\Big).
\]
Since $\ee^{\ii t\mathbb H_\lambda}$ commutes with every $\mathsf E_n$,
\(
\|(\Ubb_\lambda(t)T)\mathsf E_n\|_{\mathrm{op}}
=\|T\mathsf E_n\|_{\mathrm{op}}\),
for \(T\in\{X_{\lambda,j},R_{\lambda,j}^{(1)}\}\).
Every intermediate particle number differs from the input $n$ by at
most $D$. Multiplying the previous bounds therefore implies, uniformly
in time,
\[
 \|Y_\lambda(\boldsymbol t)\mathsf E_n\|_{\mathrm{op}}
 \le C(1+\lambda n)^{D/2},\qquad
 \|Y'_{\lambda,j}(\boldsymbol t)\mathsf E_n\|_{\mathrm{op}}
 \le C(1+\lambda n)^{D/2}(1+c_\lambda+\lambda n).
\]
Let $\Cscr_{\lambda,N}$ be the trace with each $X_{\lambda,k}$ replaced
by $\mathsf P_NX_{\lambda,k}\mathsf P_N$. Its derivative inserts
$\mathsf P_NR_{\lambda,j}^{(1)}\mathsf P_N$ at position $j$.  For input particle number $n\le N-D$, all intermediate numbers
are at most $N$, so the cutoffs leave both products unchanged.
With $\tau_{\lambda,n}=\Tr(\mathsf E_n\Gamma_\lambda)$,
summing the previous bounds over $n>N-D$ implies
\begin{align*}
 &\sup_{\boldsymbol t}\bigl(
 |\Cscr_{\lambda,N}-\Cscr_\lambda|
 +|\partial_{t_j}\Cscr_{\lambda,N}-\Tr(Y'_{\lambda,j}\Gamma_\lambda)|\bigr)\le C\sum_{n>N-D}(1+\lambda n)^{D/2}
                         (1+c_\lambda+\lambda n)\tau_{\lambda,n}
 \longrightarrow0.
\end{align*}
Here the tail tends to zero as $N\to\infty$ for fixed $\lambda$
by \eqref{eq:fixed-lambda-number-moments}.
This proves $\Cscr_\lambda\in C^1$ and
$\partial_{t_j}\Cscr_\lambda=\Tr(Y'_{\lambda,j}\Gamma_\lambda)$.
Finally, $1+c_\lambda\le Cg_d(\lambda)$ for $0<\lambda\le\lambda_0$, so
\eqref{eq:perturbed-number-exponential} at $\tau=0$ and
$(1+x)^b\le C_{b,\sigma}\ee^{\sigma x}$ for $x,b\ge0$ imply
\begin{align*}
 |\Cscr_\lambda|&\le C\omega_\lambda((\1+\lambda\mathcal N)^{D/2})
                  \le Cg_d(\lambda)^{D/2},\\
 |\partial_{t_j}\Cscr_\lambda|
 &\le C\bigl[(1+c_\lambda)\omega_\lambda((\1+\lambda\mathcal N)^{D/2})
       +\omega_\lambda((\1+\lambda\mathcal N)^{D/2+1})\bigr]
 \le Cg_d(\lambda)^{1+D/2}.
\end{align*}
\end{proof}

\subsubsection*{Dimension two}

We now apply Lemma~\ref{lem:ordered-trace-number-bounds} to remove
the Gaussian average in dimension two. The first time-derivative bound implies that the difference between
$\Cscr_\lambda$ and its Gaussian average $G_\lambda\Cscr_\lambda$
is bounded by $C\lambda(1+|\log\lambda|)^{1+D/2}$, uniformly in
$\boldsymbol t\in\mathbb R^J$. This bound tends to zero as
$\lambda\downarrow0$.

\begin{theorem}
\label{thm:two-dimensional-Wick}
Assume $d=2$ and \eqref{eq:finite-potential-regularity}.
For every fixed finite collection of finite-mode Wick polynomials,
\begin{equation}\label{eq:finite-mode-Wick-multipoint}
 \Cscr_\lambda(\boldsymbol t)\longrightarrow \Cscr(\boldsymbol t)
 \qquad\text{locally uniformly on }\mathbb R^J.
\end{equation}
\end{theorem}
\begin{proof}
Since
\[
 \Cscr_\lambda(\boldsymbol t+s\boldsymbol r)-\Cscr_\lambda(\boldsymbol t)
 =s\int_0^1D_{\boldsymbol r}\Cscr_\lambda
                    (\boldsymbol t+\theta s\boldsymbol r)\,\dd\theta,
\]
we obtain by Lemma~\ref{lem:ordered-trace-number-bounds}
\begin{align*}
 \|\Cscr_\lambda-G_\lambda \Cscr_\lambda\|_\infty
 &\le\|D_{\boldsymbol r}\Cscr_\lambda\|_\infty
                     \int|s|\varphi_\lambda(s)\,\dd s\le C_J\sqrt{2/\pi}\,\lambda(1+|\log\lambda|)^{1+D/2}
 \longrightarrow0.
\end{align*}
For every compact $\mathcal K$, Lemma~\ref{lem:directional-Gaussian-Wick}
now implies
\[
 \|\Cscr_\lambda-\Cscr\|_{L^\infty(\mathcal K)}
 \le\|\Cscr_\lambda-G_\lambda \Cscr_\lambda\|_\infty
      +\|G_\lambda \Cscr_\lambda-\Cscr\|_{L^\infty(\mathcal K)}\longrightarrow0.
\]
\end{proof}

\subsubsection*{Smooth interactions in dimension three}

For the three-dimensional argument below, assume $d=3$,
\eqref{eq:finite-potential-regularity}, and
\begin{equation}\label{eq:smooth-potential-regularity}
 \|v\|_{\mathcal W^L}
 =\sum_{k\in\mathbb Z^3}\langle k\rangle^L|\widehat v(k)|<\infty
 \qquad\text{for every }L\ge0.
\end{equation}
Thus $v\in C^\infty(\mathbb T^3)$, and $g_3(\lambda)=\lambda^{-1/2}$.
For $J\ge2$, Lemma~\ref{lem:directional-Gaussian-Wick} already implies the
convergence of $G_{c\lambda}\Cscr_\lambda$ for each fixed $c>0$. It remains to
remove this average.

We first bound the $N$-th derivative along $\boldsymbol r$ by
$C_N(1+|\boldsymbol t|)^{R_N}\lambda^{-N/2-D/4}$, where
$D=\sum_j\deg P_j$. With $N=2K$, a $K$-fold Gaussian difference then
has an error of order $\lambda^{K-D/4}$, which tends to zero when
$K>D/4$. The required Gibbs estimate for products of commutators is
proved in Appendix~\ref{app:Wick-product-estimates}.

To this end, define
\[
 R_{\lambda,j}^{(0)}=X_{\lambda,j},\qquad
 R_{\lambda,j}^{(r+1)}=\ii[\mathbb H_\lambda,R_{\lambda,j}^{(r)}]
\]
on smooth finite-particle vectors.
At order one, this agrees with
$R_{\lambda,j}^{(1)}=\Lbb_\lambda X_{\lambda,j}$ from
Lemma~\ref{lem:ordered-trace-number-bounds}.

\begin{proposition}
\label{prop:smooth-raw-derivatives}
Let $J\ge2$ and $D=\sum_j\deg P_j$. For every integer $N\ge0$,
$D_{\boldsymbol r}^N\Cscr_\lambda$ exists, is continuous, and satisfies
\begin{equation}\label{eq:raw-derivatives-for-Gaussian-removal}
 |D_{\boldsymbol r}^N\Cscr_\lambda(\boldsymbol t)|
 \le C_N(1+|\boldsymbol t|)^{R_N}\lambda^{-N/2-D/4},
 \qquad 0<\lambda\le\lambda_0.
\end{equation}
The constants $C_N,R_N$ may depend on the fixed observables and $N$, but not on
$\lambda$ or $\boldsymbol t$.
\end{proposition}
\begin{proof}
We have by the Leibniz rule
\[
 D_{\boldsymbol r}^N Y_\lambda(\boldsymbol t)
 =\sum_{\substack{\boldsymbol\alpha\in\mathbb N_0^J\\
                   |\boldsymbol\alpha|=N}}
      \frac{N!}{\boldsymbol\alpha!}\boldsymbol r^{\boldsymbol\alpha}
      \prod_{j=1}^J\Ubb_\lambda(t_j)
           \bigl(R_{\lambda,j}^{(\alpha_j)}\bigr).
\]
Lemma~\ref{lem:Wick-commutator-Gibbs-products} shows that each term,
after right multiplication by $\Gamma_\lambda^{1/2}$, is continuous
in $\mathfrak S^2$ and has norm at most
$C_N(1+|\boldsymbol t|)^{(J+1)N}g_3(\lambda)^{N+D/2}$.

Since $\|\Gamma_\lambda^{1/2}\|_{\mathfrak S^2}=1$, we obtain
\[
 |D_{\boldsymbol r}^N\Cscr_\lambda(\boldsymbol t)|
 \le \|(D_{\boldsymbol r}^NY_\lambda(\boldsymbol t))
                               \Gamma_\lambda^{1/2}\|_{\mathfrak S^2}
 \le C_N(1+|\boldsymbol t|)^{(J+1)N}g_3(\lambda)^{N+D/2}.
\]
This implies continuity and
\eqref{eq:raw-derivatives-for-Gaussian-removal}, with $R_N=(J+1)N$.
\end{proof}

\begin{theorem}
\label{thm:smooth-3d-Wick}
Assume $d=3$, \eqref{eq:finite-potential-regularity}, and
\eqref{eq:smooth-potential-regularity}. For every fixed finite family
of finite-mode Wick polynomials, the correlations in
\eqref{eq:Wick-correlation-notation} satisfy
\[
 \lim_{\lambda\downarrow0}\sup_{\boldsymbol t\in[-T,T]^J}
       |\Cscr_\lambda(\boldsymbol t)-\Cscr(\boldsymbol t)|=0,
 \qquad T>0.
\]
\end{theorem}

\begin{proof}
Choose an integer $K>D/4$.
For every $0<\delta\le\lambda\sqrt K$, a change of
variable in the Gaussian integral and
Proposition~\ref{prop:smooth-raw-derivatives} imply
\begin{align*}
 \bigl|G_\delta D_{\boldsymbol r}^{2K}\Cscr_\lambda(\boldsymbol t)\bigr|
 &\le C_{2K}\lambda^{-K-D/4}
   \int_{\mathbb R}\varphi_1(z)
          (1+|\boldsymbol t+\delta z\boldsymbol r|)^{R_{2K}}\,\dd z
 \le C_{K,T}\lambda^{-K-D/4}.
\end{align*}
 In \eqref{eq:iterated-Gaussian-remainder},
$\delta=\lambda\sqrt{\theta_1+\cdots+\theta_K}$, we then have
\[
 \sup_{|\boldsymbol t|\le T}
       |\Cscr_\lambda(\boldsymbol t)-H_{\lambda,K}(\boldsymbol t)|
 \le C_{K,T}\left(\frac{\lambda^2}{2}\right)^K\lambda^{-K-D/4}
 \le C_{K,T}\lambda^{K-D/4}\longrightarrow0,
\]
where we recall  $H_{\lambda,K}$ from \eqref{def:Hlambdak}.
The locally uniform convergence $H_{\lambda,K}\to\Cscr$ established
in Subsection~\ref{subsec:Gaussian-time-Wick} completes the proof.
\end{proof}

\begin{proof}[Proof of Theorem~\ref{thm:main-Wick}]
Part~\textup{(i)} is Proposition~\ref{prop:finite-mode-Wick-Hilbert}.
Part~\textup{(ii)} follows from Theorems~\ref{thm:two-dimensional-Wick}
and~\ref{thm:smooth-3d-Wick}.
\end{proof}

\appendix

\section{Finite-dimensional Fock-space and anti-Wick estimates}
\label{app:Toeplitz}

This appendix collects the finite-dimensional Fock-space and quantization
estimates used in the generator and radial-cutoff arguments.  We first
give in Subsection~\ref{appsubsec:field-sector-bounds} the
creation--annihilation bounds.
Lemma~\ref{lem:app-radial-AW} proves the exact radial anti-Wick formula.
Subsection~\ref{appsubsec:finite-mode-mixed-convergence} proves the mixed
Wick--anti-Wick convergence used for the fixed spectral cutoff in
Section~\ref{sec:higher-density-generator} and for the finite-mode
argument in Section~\ref{sec:Wick-observables}.  Subsection~\ref{appsubsec:finite-rank-mixed-words} proves the finite-rank mixed-word
estimate used in Section~\ref{sec:higher-density-generator}.
Subsection~\ref{appsubsec:radial-number-functional-calculus} proves the
smooth finite-mode-number cutoff comparison needed in
Section~\ref{sec:Wick-observables} by diagonalizing the radial anti-Wick
operator and using a one-dimensional Gamma-distribution estimate.  It
also contains the radially truncated Wick-to-anti-Wick comparison used
in the multipoint argument.

\subsection{Bounds for creation and annihilation operators}
\label{appsubsec:field-sector-bounds}

The following standard estimates are used several times.

\begin{lemma}
\label{lem:app-field-sector-bounds}
Let \(E\subset\gH\) be finite dimensional, and let
\(\mathcal N_E=\dG(\Pi_E)\),
\(\mathsf E_n^E=\1_{\{\mathcal N_E=n\}}\).  For
\(f_1,\ldots,f_r,g_1,\ldots,g_s\in E\), set
\[
 \mathsf W_{r,s}(\boldsymbol f;\boldsymbol g)
 =a_\lambda^*(f_1)\cdots a_\lambda^*(f_r)
   a_\lambda(g_1)\cdots a_\lambda(g_s).
\]
For \(n\ge0\),
\begin{equation}
 \|\mathsf W_{r,s}(\boldsymbol f;\boldsymbol g)\mathsf E_n^E\|_{\mathrm{op}}
 \le C_{r,s}(1+\lambda n)^{(r+s)/2}
 \prod_{i=1}^r\|f_i\|_{\gH}
 \prod_{j=1}^s\|g_j\|_{\gH},
 \qquad 0<\lambda\le1.                                    \label{eq:app-field-sector-bound}
\end{equation}

More generally, if \(V\) is a fixed monomial of degree \(d_V\) in scaled
creation and annihilation operators with vectors in \(E\), then
\begin{equation}\label{eq:app-field-word-bound}
 \|V\psi\|
 \le C_V\|(\1+\lambda\mathcal N_E)^{d_V/2}\psi\|.
\end{equation}
\end{lemma}

\begin{proof}
Successive application of the one-field bounds in
\cite[Section~6.1]{LNR15} proves \eqref{eq:app-field-sector-bound} and
the sector shift.  The same argument in the given order of $V$ proves
\eqref{eq:app-field-word-bound}.
\end{proof}

The next computation is the finite-dimensional radial analog of the
coherent-state Gamma representation used in
\cite[Lemma~11.1]{NZZ25}.  We include the short proof because the exact
operator identity, including the dimension shift, is used below.

\begin{lemma}
\label{lem:app-radial-AW}
Let \(E\simeq\mathbb C^{r_E}\) be fixed, where
\(r_E=\dim_{\mathbb C}E\), and let $\mathcal N_E$ be the
finite-mode number operator from Lemma~\ref{lem:app-field-sector-bounds}.
For every integer \(M\ge1\), 
\begin{equation}\label{eq:radial-anti-Wick-moments}
 \Op_{\lambda,E}^{\rm A}(\|z\|^{2M})
 =\prod_{\ell=0}^{M-1}
  \bigl(\lambda\mathcal N_E+\lambda(r_E+\ell)\bigr).
\end{equation}
\end{lemma}

\begin{proof}
The symbol \(\|z\|^{2M}\) is invariant under the unitary group of
\(E\).  Its anti-Wick quantization therefore commutes with the second-quantized
unitary action and is scalar on each irreducible sector
\(E^{\otimes_s n}\).  Evaluate that scalar on
\(\psi_1^{\otimes n}\), where
\((\psi_j)_{j=1}^{r_E}\) is an orthonormal basis.
The coherent-state formula gives
\[
 \frac{1}{(\pi\lambda)^{r_E}\lambda^n n!}
 \int_{\mathbb C^{r_E}}
  \|z\|^{2M}|z_1|^{2n}
  \ee^{-\|z\|^2/\lambda}\,\dd z
 =\lambda^M\prod_{\ell=0}^{M-1}(n+r_E+\ell).
\]
The result then follows.
\end{proof}

\subsection{Convergence of finite-mode Wick--anti-Wick products}
\label{appsubsec:finite-mode-mixed-convergence}

For a finite-dimensional complex Hilbert space \(E\), let
\(\mathscr P_D(E)\) denote the polynomials in \((z,\bar z)\) of total
degree at most \(D\).  This is a finite-dimensional space; convergence
in \(\mathscr P_D(E)\) means convergence of the coefficients in any,
and hence every, fixed linear coordinate system.  For a monomial in which all creation operators stand to the left of
all annihilation operators,
\[
 W=a_\lambda^*(u_1)\cdots a_\lambda^*(u_r)
   a_\lambda(v_1)\cdots a_\lambda(v_s),
 \qquad u_\alpha,v_\beta\in E,
\]
we associate the polynomial
\[
 p_W(z)=
 \prod_{\alpha=1}^r\langle z,u_\alpha\rangle
 \prod_{\beta=1}^s\langle v_\beta,z\rangle.
\]
Then \(W=\Op_{\lambda,E}^{\rm W}(p_W)\) by
\eqref{eq:ordinary-Wick-definition}; the correspondence is extended
linearly to finite sums.

\begin{lemma}
\label{lem:finite-mode-mixed-calculus}
Let \(E_0\subset E\subset\gH\) be fixed finite-dimensional Fourier
spaces, and fix \(J\in\mathbb N\) and \(D\in\mathbb N_0\).  For
\(1\le\ell\le J\), let
\[
 p_{\lambda,\ell}^{\rm L},p_{\lambda,\ell}^{\rm R}
 \in\mathscr P_D(E),
 \qquad
 g_{\lambda,\ell}\in C_c^\infty(E_0),
 \qquad
 \deg p_{\lambda,\ell}^{\rm L}
 +\deg p_{\lambda,\ell}^{\rm R}\le D.
\]
Assume that there is a compact set \(\mathscr K\Subset E_0\) such that
\(\supp g_{\lambda,\ell}\subset\mathscr K\) for all \(\lambda,\ell\),
and, as \(\lambda\downarrow0\),
\[
 p_{\lambda,\ell}^{\rm L}\longrightarrow p_\ell^{\rm L},
 \qquad
 p_{\lambda,\ell}^{\rm R}\longrightarrow p_\ell^{\rm R}
 \quad\text{in }\mathscr P_D(E),
 \qquad
 g_{\lambda,\ell}\longrightarrow g_\ell
 \quad\text{in }C_c^\infty(E_0).
\]
Define
\[
 A_{\lambda,\ell}
 :=\Op_{\lambda,E_0}^{\rm A}(g_{\lambda,\ell})
   \otimes\1_{\Fock(E\cap (E_0)^\perp)},
\]
and
\begin{equation}
\label{eq:finite-mode-mixed-word}
 \mathfrak R_\lambda
 :=\sum_{\ell=1}^J
 \Op_{\lambda,E}^{\rm W}(p_{\lambda,\ell}^{\rm L})
 A_{\lambda,\ell}
 \Op_{\lambda,E}^{\rm W}(p_{\lambda,\ell}^{\rm R}).
\end{equation}
All operators are extended by the identity on \(\Fock(E^\perp)\).  Set
\[
 r(z):=\sum_{\ell=1}^J
 p_\ell^{\rm L}(z)g_\ell(\Pi_{E_0}z)p_\ell^{\rm R}(z),
 \qquad z\in E,
\]
and extend \(r\) to \(\Xcal\) as the cylinder function
\(r\circ\Pi_E\).  Then \(\mathfrak R_\lambda\) defines a vector of
\(\Hcal_\lambda\) for every fixed \(\lambda>0\), and
\begin{equation}
\label{eq:finite-mode-mixed-convergence}
 \mathfrak R_\lambda\to_{\mathrm{id}}r
 \qquad(\lambda\downarrow0).
\end{equation}
\end{lemma}

\begin{proof}
Put
\[
 E_1:=E\cap (E_0)^\perp,
 \qquad
 \mathcal N_{E_1}:=\dG(\Pi_{E_1}).
\]
The middle anti-Wick symbols have a common compact support in \(E_0\).
The only difficulty is therefore the polynomial growth of the Wick
factors in the remaining finite-dimensional directions \(E_1\).
The coefficient convergence also shows that all polynomial
coefficients are uniformly bounded.  Hence the limiting function
\(r\) has polynomial growth on \(E\), and
\(r\in L^p(\nu)\) for every finite \(p\).

Choose \(\vartheta\in C_c^\infty([0,\infty))\), with
\(0\le\vartheta\le1\) and \(\vartheta=1\) near zero, and put
\[
 \vartheta_S(z_1):=\vartheta(S^{-1}\|z_1\|^2),
 \qquad z_1\in E_1, S>0.
\]
In \eqref{eq:finite-mode-mixed-word}, replace
\(A_{\lambda,\ell}\) by
\[
 A_{\lambda,\ell}^S
 :=\Op_{\lambda,E}^{\rm A}\!\left(
     g_{\lambda,\ell}(\Pi_{E_0}z)\,
     \vartheta_S(\Pi_{E_1}z)
   \right),
\]
and denote the resulting sum by \(\mathfrak R_\lambda^S\).  For fixed
\(S\), all variables are restricted to a common compact subset of
\(E\).  Repeated use of
\eqref{eq:finite-mode-AW-eigen}--
\eqref{eq:finite-mode-AW-integration} therefore implies
\[
 \mathfrak R_\lambda^S
 =\Op_{\lambda,\gH}^{\rm A}(r_{\lambda,S})
\]
with \(r_{\lambda,S}\in C_c^\infty(E)\).  The assumed convergence of
the coefficients implies
\[
 r_{\lambda,S}\longrightarrow
 (\vartheta_S\circ\Pi_{E_1})r
 \quad\text{in }C_c^\infty(E).
\]
Indeed, the leading terms give the ordinary product in the definition
of \(r\), while every correction produced by \eqref{eq:finite-mode-AW-integration} contains an explicit factor \(\lambda\).  Hence as  $\lambda\downarrow0$
\begin{equation}
\label{eq:finite-mode-fixed-S-convergence}
 \left\|\mathfrak R_\lambda^S
 -\Op_{\lambda,\gH}^{\rm A}
   \bigl((\vartheta_S\circ\Pi_{E_1})r\bigr)
 \right\|_{2,\lambda}
 \longrightarrow0
.
\end{equation}

We next remove the cutoff.  Let
$ X_{\lambda,S}:=\Op_{\lambda,E_1}^{\rm A}(\vartheta_S).$
Factorization of the coherent-state positive operator-valued measure
implies
\[
 A_{\lambda,\ell}-A_{\lambda,\ell}^S
 =\Op_{\lambda,E_0}^{\rm A}(g_{\lambda,\ell})
  \otimes(\1-X_{\lambda,S}).
\]
For every integer \(M\ge1\),
\(0\le1-\vartheta(t)\le C_Mt^M\).  Anti-Wick positivity and
Lemma~\ref{lem:app-radial-AW} imply
\[
 0\le\1-X_{\lambda,S}
 \le C_{M,E}S^{-M}
      (\1+\lambda\mathcal N_{E_1})^M.
\]
Move the fields from \(E_0\) through the middle anti-Wick operator by
\eqref{eq:finite-mode-AW-eigen}--
\eqref{eq:finite-mode-AW-integration}.  The remaining fields lie in
\(E_1\), and their total number is at most \(D\).  The previous
inequality and Lemma~\ref{lem:app-field-sector-bounds} then imply, as a
quadratic-form inequality 
\[
 (\mathfrak R_\lambda-\mathfrak R_\lambda^S)^*
 (\mathfrak R_\lambda-\mathfrak R_\lambda^S)
 \le C_{M,E,D}S^{-2M}
 (\1+\lambda\mathcal N_{E_1})^{2M+D},
\]
which yields
\begin{equation}\label{eq:mixed-radial-tail-bound}
 \|\mathfrak R_\lambda-\mathfrak R_\lambda^S\|_{2,\lambda}^2
 \le C_{M,E,D}S^{-2M}
 \Tr\!\left[(\1+\lambda\mathcal N_{E_1})^{2M+D}
             \Gamma_\lambda\right].
\end{equation}
Since \(E_1\) is fixed and finite-dimensional,
we obtain by \eqref{eq:finite-mode-number-moment-bound}
\begin{equation}
\label{eq:finite-mode-tail-removal}
 \lim_{S\to\infty}\limsup_{\lambda\downarrow0}
 \|\mathfrak R_\lambda-\mathfrak R_\lambda^S\|_{2,\lambda}=0.
\end{equation}

Finally,
\((\vartheta_S\circ\Pi_{E_1})r\to r\) in \(L^2(\nu)\) by dominated
convergence.  First choose \(S\) so that this classical error and the
error in \eqref{eq:finite-mode-tail-removal} are small, and then use
\eqref{eq:finite-mode-fixed-S-convergence} in
Definition~\ref{def:identified}.  Thus
\eqref{eq:finite-mode-mixed-convergence} follows.
\end{proof}

\begin{corollary}
\label{cor:finite-mode-mixed-pairing}
Let $E\subset\gH$ be a fixed finite Fourier space and let $P$ be a
polynomial on $E$. Then, for every $\Phi\in\Acal_{\rm cyl}$,
\[
 \ip{\Op_{\lambda,\gH}^{\rm A}\Phi}
             {\Op_{\lambda,E}^{\rm W}(P)}_\lambda
 \longrightarrow \int\overline{\Phi(u)}\,(P)_E(u)\,\dd\nu(u)
 \qquad(\lambda\downarrow0).
\]
\end{corollary}
\begin{proof}
By finite linearity it is enough to take $\Phi=G\circ\Pi_{E_0}$,
with $E_0$ a finite Fourier space and $G\in C_c^\infty(E_0)$.
Enlarge the finite Fourier space so that it contains both $E$ and $E_0$.
Lemma~\ref{lem:finite-mode-mixed-calculus}, with left Wick polynomial
$1$, anti-Wick coefficient $\Op_{\lambda,\gH}^{\rm A}\overline\Phi$,
and right Wick polynomial $\Op_{\lambda,E}^{\rm W}(P)$, implies
\[
 \Op_{\lambda,\gH}^{\rm A}\overline\Phi\,
       \Op_{\lambda,E}^{\rm W}(P)
 \to_{\rm id}\overline\Phi(P)_E.
\]
Since $\Op_{\lambda,\gH}^{\rm A}1=1\to_{\rm id}1$, using
Lemma~\ref{lem:identified-calculus}\textup{(ii)} we obtain convergence
of the pairing with the unit vector $1$, which is the required result.
\end{proof}

\subsection{Bounds for products with finite-rank one-particle operators}
\label{appsubsec:finite-rank-mixed-words}

The following estimate is used in the summation of the interaction modes in
Section~\ref{sec:higher-density-generator}.

\begin{lemma}
\label{lem:finite-rank-mixed-word}
Fix a finite Fourier projection \(\Pi\), 
\(G\in C_c^\infty(\Pi\gH)\), and \(d_0,r_0,K_0<\infty\).  Let
\(b=b^*\) satisfy \(\|b\|_{\mathrm{op}}\le1\), and suppose that a Fourier
space \(E\supset \Pi\gH+b\Pi\gH\) has dimension at most \(d_0\).  Set
\[
 C_{\lambda,b}:=
 \adop{b}\Op_{\lambda,\gH}^{\rm A}G.
\]
Let \(R=R^*\) obey
\[
 R=\Pi_ER\Pi_E,\qquad
 \operatorname{rank}R\le r_0,\qquad
 \|R\|_{\mathrm{op}}+\|R\|_{\mathfrak S^2(\gH)}\le K_0.
\]
Then
\begin{equation}\label{eq:abstract-finite-rank-mixed-word}
 \|B_\lambda(R)C_{\lambda,b}\|_{2,\lambda}^2
 +\|C_{\lambda,b}B_\lambda(R)\|_{2,\lambda}^2
 \le C\Big(1+\sum_{j=1}^{3}
  \|\lambda^j\Gamma_\lambda^{(j)}\|_{\mathfrak S^2(\gH^{\otimes_s j})}\Big),
\end{equation}
where \(C\) depends only on
\(\Pi,G,d_0,r_0,K_0,\lambda_0\).  In particular, the right side is bounded
uniformly for \(0<\lambda\le\lambda_0\).
\end{lemma}

\begin{proof}
By Lemma~\ref{lem:finite-mode-commutator}, we may write
\[
 C_{\lambda,b}
 =\sum_{\alpha=1}^{N_0}
 A_{\lambda,\alpha}V_{\lambda,\alpha},
 \qquad
 N_0\le 1+2\operatorname{rank}\Pi,
\]
where each \(A_{\lambda,\alpha}\) is the anti-Wick quantization of a
bounded smooth function on \(\Pi\gH\), and \(V_{\lambda,\alpha}\) is
either \(\1\), one creation operator, or one annihilation operator with
a vector in \(E\cap(\Pi\gH)^\perp\). Moreover,
\[
 \sup_{0<\lambda\le\lambda_0}
 \max_{1\le\alpha\le N_0}
 \|A_{\lambda,\alpha}\|_{\mathrm{op}}
 \le C,
\]
where \(C\) depends only on \(\Pi\) and finitely many seminorms of
\(G\).

Since \(R=R^*\), \(R=\Pi_ER\Pi_E\), and
\(\operatorname{rank}R\le r_0\), there are orthonormal vectors
\(u_1,\ldots,u_r\in E\) and real numbers
\(\rho_1,\ldots,\rho_r\), with \(r\le r_0\), such that
\[
 R
 =\sum_{\ell=1}^r
 \rho_\ell |u_\ell\rangle\langle u_\ell|,
 \qquad
 \sum_{\ell=1}^r|\rho_\ell|
 =\|R\|_{\mathfrak S^1(\gH)}
 \le r_0^{1/2}\|R\|_{\mathfrak S^2(\gH)}
 \le r_0^{1/2}K_0.
\]
Consequently,
\[
 B_\lambda(R)
 =\sum_{\ell=1}^r
 \rho_\ell
 a_\lambda^*(u_\ell)a_\lambda(u_\ell)
 -\lambda\Tr(R\gamma_{0,\lambda})\1,
\]
and
\[
 |\lambda\Tr(R\gamma_{0,\lambda})|
 \le
 \|R\|_{\mathfrak S^2(\gH)}
 \|\lambda\gamma_{0,\lambda}\|_{\mathfrak S^2(\gH)}
 \le CK_0.
\]

For \(C_{\lambda,b}B_\lambda(R)\), the bounded anti-Wick factor is
already on the left. To treat \(B_\lambda(R)C_{\lambda,b}\), decompose
$u_\ell=\Pi u_\ell+\Pi^\perp u_\ell.$
The fields with vectors in \(E\cap(\Pi\gH)^\perp\) commute with
\(A_{\lambda,\alpha}\), whereas the fields with vectors in
\(\Pi\gH\) are moved through it by
\eqref{eq:finite-mode-AW-eigen}--\eqref{eq:finite-mode-AW-integration}.
The resulting anti-Wick factors involve derivatives of \(G\),
multiplied by at most two coordinate functions on \(\Pi\gH\), and are
therefore uniformly bounded because \(G\) is compactly supported.
The terms arising from the derivative parts of
\eqref{eq:finite-mode-AW-integration} carry an additional factor
\(\lambda\), which is harmless for \(0<\lambda\le\lambda_0\).

It follows that each product in
\eqref{eq:abstract-finite-rank-mixed-word} is a sum of a uniformly
bounded number of terms
 $\widetilde A_{\lambda,\beta}W_{\lambda,\beta},$
where
 $\|\widetilde A_{\lambda,\beta}\|_{\mathrm{op}}\le C$
and \(W_{\lambda,\beta}\) is a monomial of degree at most three in
scaled creation and annihilation operators with vectors in \(E\) of
norm at most one.

Recall $\mathcal N_E=\dG(\Pi_E)$.
By Lemma~\ref{lem:app-field-sector-bounds}, 
\[
 W_{\lambda,\beta}^*W_{\lambda,\beta}
 \le C(\1+\lambda\mathcal N_E)^3.
\]
We therefore obtain
\begin{align*}
 \|\widetilde A_{\lambda,\beta}
   W_{\lambda,\beta}\|_{2,\lambda}^2
 &\le
 C\Tr\bigl((\1+\lambda\mathcal N_E)^3\Gamma_\lambda\bigr)
 \\
 &\le
 C\Big(
  1+\sum_{j=1}^3
  \|\Pi_E^{\otimes j}\|_{\mathfrak S^2(\gH^{\otimes_s j})}
  \|\lambda^j\Gamma_\lambda^{(j)}
    \|_{\mathfrak S^2(\gH^{\otimes_s j})}
 \Big)
 \\
 &\le
 C\Big(
  1+\sum_{j=1}^3
  \|\lambda^j\Gamma_\lambda^{(j)}
    \|_{\mathfrak S^2(\gH^{\otimes_s j})}
 \Big).
\end{align*}
Summing the finitely many terms proves
\eqref{eq:abstract-finite-rank-mixed-word}. The
total-particle-number cutoff is removed using
\eqref{eq:fixed-lambda-number-moments}, and the uniform bound in
\(\lambda\) follows from \eqref{eq:LNR-RDM-uniform} for
\(j=1,2,3\).
\end{proof}

\subsection{Radial cutoffs for the finite-mode number operator}
\label{appsubsec:radial-number-functional-calculus}

We shall use only one estimate involving a smooth function of the
finite-mode number operator.  On the Fock space over a fixed
finite-dimensional space \(E\), we compare
\(
\chi(\lambda\mathcal N_E)\)
with \(
\Op^{\rm A}_{\lambda,E}\bigl(\chi(\|z\|^2)\bigr).
\)
Because \(\chi(\|z\|^2)\) depends only on \(\|z\|\), both operators are
diagonal with respect to the particle-number decomposition.  It is
therefore enough to compare their eigenvalues on the \(n\)-particle
sector.  The eigenvalue of the second operator is given by the
coherent-state integral, which becomes a one-dimensional Gamma
integral after passing to polar coordinates.  The same calculation
appears in \cite[Lemma~11.1]{NZZ25}.  We repeat it here to obtain a
bound that is uniform in \(n\), and hence an operator-norm estimate.

\begin{lemma}
\label{lem:app-radial-number-functional-calculus}
Let \(E\simeq\mathbb C^{r_E}\) be fixed. Recall $\mathcal N_E$ from Lemma~\ref{lem:app-radial-AW} and set \(
 n_E(z)=\|z\|^2.
\)
For every \(\chi\in C_c^2(\mathbb R;\mathbb C)\), there is a constant
\(C_{E,\chi}<\infty\) such that
\begin{equation}\label{eq:app-radial-number-functional-calculus}
 \left\|
 \chi(\lambda\mathcal N_E)
 -\Op_{\lambda,E}^{\rm A}(\chi\circ n_E)
 \right\|_{\mathrm{op}}
 \le C_{E,\chi}\lambda,
 \qquad 0<\lambda\le1.
\end{equation}
\end{lemma}

\begin{proof}
Let \(\mathsf E_n^E=\1_{\{\mathcal N_E=n\}}\). The operator \(\chi(\lambda\mathcal N_E)\) acts on the \(n\)-particle sector
as multiplication by \(\chi(\lambda n)\).  Since
\(z\mapsto\chi(\|z\|^2)\) is invariant under every unitary
transformation of \(E\), its anti-Wick quantization commutes with the
corresponding second-quantized unitary operators.  Schur's lemma
therefore shows that it also acts as scalar multiplication on
\(E^{\otimes_s n}\).  The operator
\(\chi(\lambda\mathcal N_E)\) has eigenvalue \(\chi(\lambda n)\).
Using polar coordinates in the coherent-state integral, as in
Lemma~\ref{lem:app-radial-AW}, gives
\[
 \Op_{\lambda,E}^{\rm A}(\chi\circ n_E)\mathsf E_n^E
 =b_{n,\lambda}\mathsf E_n^E,
\]
where
\begin{align*}
 b_{n,\lambda}
 &=\frac{1}{\lambda^{n+r_E}(n+r_E-1)!}
   \int_0^\infty
   \chi(s)\ee^{-s/\lambda}s^{n+r_E-1}\,\dd s\\
 &=\frac{1}{(n+r_E-1)!}
   \int_0^\infty
   \chi(\lambda r)\ee^{-r}r^{n+r_E-1}\,\dd r.
\end{align*}
For \(k>0\), let \(Y_k\) be the positive random variable with density
\[
 \frac{1}{\Gamma(k)}r^{k-1}\ee^{-r}\1_{(0,\infty)}(r)\,\dd r.
\]
Direct integration gives, for \(\tau\ge0\),
\[
 \mathbb EY_k=k,\qquad
 \operatorname{Var}(Y_k)=k,\qquad
 \mathbb E\ee^{-\tau Y_k}=(1+\tau)^{-k}.
\]
Taking \(k=n+r_E\), the second integral above is
\[
 b_{n,\lambda}=\mathbb E\bigl[\chi(\lambda Y_{n+r_E})\bigr].
\]

Choose \(R\ge1\) such that
\(\supp\chi\cap[0,\infty)\subset[0,R]\).  Suppose first that
\(\lambda n\le4R\), and put \(Z=\lambda Y_{n+r_E}\).  Then
\[
 \mathbb EZ=\lambda(n+r_E),
 \qquad
 \operatorname{Var}(Z)=\lambda^2(n+r_E).
\]
Put \(\mu=\mathbb EZ\).  Taylor's formula with integral remainder,
applied pointwise at \(x=Z\), reads
\[
 \chi(Z)=\chi(\mu)+\chi'(\mu)(Z-\mu)
 +(Z-\mu)^2\int_0^1(1-t)
   \chi''\bigl(\mu+t(Z-\mu)\bigr)\,\dd t.
\]
Taking expectations removes the linear term because
\(\mathbb E(Z-\mu)=0\).  Hence
\[
 \left|\mathbb E\chi(Z)-\chi(\mathbb EZ)\right|
 \le \|\chi''\|_\infty\mathbb E|Z-\mu|^2
       \int_0^1(1-t)\,\dd t
 =\frac12\|\chi''\|_\infty\operatorname{Var}(Z).
\]
The mean-value theorem also gives
\[
 |\chi(\mathbb EZ)-\chi(\lambda n)|
 \le r_E\lambda\|\chi'\|_\infty.
\]
Since
\(\lambda^2(n+r_E)\le(4R+r_E)\lambda\), we obtain
\begin{equation}\label{eq:app-radial-number-small-sector}
 |b_{n,\lambda}-\chi(\lambda n)|
 \le C_{r_E,\chi}\lambda
 \qquad(\lambda n\le4R).
\end{equation}

Suppose next that \(\lambda n>4R\).  Then
\(\chi(\lambda n)=0\), while
\[
 |b_{n,\lambda}|
 \le\|\chi\|_\infty
 \mathbb P\!\left(Y_{n+r_E}\le\frac R\lambda\right)
 \le\|\chi\|_\infty
 \mathbb P\!\left(Y_{n+r_E}\le\frac{n+r_E}{4}\right).
\]
Markov's inequality applied to \(\ee^{-3Y_k}\), together with the formula
above at \(\tau=3\), gives
\[
 \mathbb P(Y_k\le k/4)
 \le \ee^{3k/4}\mathbb E\ee^{-3Y_k}
 =\ee^{-(\log4-3/4)k}.
\]
Since \(n>4R/\lambda\), there is \(c>0\) such that
\[
 |b_{n,\lambda}|\le\|\chi\|_\infty\ee^{-cR/\lambda}
 \le C_{R,\chi}\lambda,
 \qquad 0<\lambda\le1.
\]
Together with \eqref{eq:app-radial-number-small-sector}, this proves
\[
 \sup_{n\ge0}|b_{n,\lambda}-\chi(\lambda n)|
 \le C_{E,\chi}\lambda.
\]
Taking the supremum over the particle-number sectors proves
\eqref{eq:app-radial-number-functional-calculus}.
\end{proof}

\begin{lemma}
\label{lem:domain-safe-radial-AW}
Let \(E\subset\gH\) be a fixed finite Fourier space, with orthogonal
projection \(\Pi_E\), and recall 
\(
 \mathcal N_E
 \),  \(n_E(z)
\) from Lemma~\ref{lem:app-radial-number-functional-calculus}.
Let \(P\) be a fixed polynomial on \(E\), and set
\(X_\lambda:=\Op_{\lambda,E}^{\rm W}(P)\).  Choose
\(\chi\in C_c^\infty([0,\infty);[0,1])\), with \(\chi=1\) on
\([0,1]\), and set \(\chi_K(r)=\chi(r/K)\).  For every fixed
\(K<\infty\), the operator
\[
 X_\lambda^{[K]}
 :=\chi_K(\lambda\mathcal N_E)X_\lambda
   \chi_K(\lambda\mathcal N_E)
\]
extends from the finite-particle core to a bounded operator.  With
\[
 f_K(z)=\chi_K(n_E(z))^2P(z),
\]
one has \(f_K\in C_c^\infty(E;\mathbb C)\), independently of
\(\lambda\), and
\begin{equation}\label{eq:domain-safe-radial-AW}
 \|X_\lambda^{[K]}-\Op_{\lambda,E}^{\rm A}f_K\|_{\mathrm{op}}
 \le C_K\lambda^{1/2},
 \qquad 0<\lambda\le1.
\end{equation}
If \(P\) is real-valued, then \(f_K\) is real-valued.
\end{lemma}

\begin{proof}
Put
\[
 C_{\lambda,K}:=\chi_K(\lambda\mathcal N_E),
 \qquad g_K:=\chi_K\circ n_E,
 \qquad G_{\lambda,K}:=\Op_{\lambda,E}^{\rm A}(g_K).
\]
Lemma~\ref{lem:app-radial-number-functional-calculus} gives
\begin{equation}\label{eq:radial-cutoff-functional-comparison}
 \|C_{\lambda,K}-G_{\lambda,K}\|_{\mathrm{op}}\le C_K\lambda.
\end{equation}
Since \(X_\lambda\) changes the finite-mode particle number by at most
a fixed integer, the sector estimate in
Lemma~\ref{lem:app-field-sector-bounds} gives
\begin{equation*}
 \|X_\lambda C_{\lambda,K}\|_{\mathrm{op}}
 +\|C_{\lambda,K}X_\lambda\|_{\mathrm{op}}\le C_K.
\end{equation*}

Replacing the left cutoff first, we obtain
\[
 C_{\lambda,K}X_\lambda C_{\lambda,K}
 =G_{\lambda,K}X_\lambda C_{\lambda,K}
  +O_{\mathrm{op},K}(\lambda).
\]
Here and below, \(O_{\mathrm{op},K}(\lambda)\) may denote a different
operator at each occurrence, with operator norm at most \(C_K\lambda\)
for \(0<\lambda\le1\). 
Write
\(P(z,\bar z)=\sum_{\alpha,\beta}
c_{\alpha\beta}\bar z^\alpha z^\beta\).  As in the proof of
Lemma~\ref{lem:finite-mode-mixed-calculus}, successive use of
\eqref{eq:finite-mode-AW-eigen} and
\eqref{eq:finite-mode-AW-integration} implies
\[
 G_{\lambda,K}X_\lambda
 =\sum_{\alpha,\beta}\sum_{\gamma\le\beta}
 c_{\alpha\beta}(-\lambda)^{|\gamma|}
 \binom{\beta}{\gamma}
 \Op_{\lambda,E}^{\rm A}\!\left(
 z^{\beta-\gamma}
 \partial_{\bar z}^{\gamma}(\bar z^\alpha g_K)
 \right).
\]
The term \(\gamma=0\) is
\(\Op_{\lambda,E}^{\rm A}(g_KP)\).  Every other term carries an
explicit factor \(\lambda\) and is the anti-Wick quantization of a
smooth function supported in \(\supp g_K\).  Hence
\begin{equation*}
 G_{\lambda,K}X_\lambda
 =\Op_{\lambda,E}^{\rm A}(g_KP)+O_{\mathrm{op},K}(\lambda).
\end{equation*}
Replacing the right cutoff and using
\eqref{eq:radial-cutoff-functional-comparison} gives
\[
 C_{\lambda,K}X_\lambda C_{\lambda,K}
 =\Op_{\lambda,E}^{\rm A}(g_KP)
  \Op_{\lambda,E}^{\rm A}(g_K)+O_{\mathrm{op},K}(\lambda).
\]
Finally, Lemma~\ref{lem:AW-product} yields
\[
 \left\|
 \Op_{\lambda,E}^{\rm A}(g_KP)
 \Op_{\lambda,E}^{\rm A}(g_K)
 -\Op_{\lambda,E}^{\rm A}(g_K^2P)
 \right\|_{\mathrm{op}}
 \le C_K\lambda^{1/2}.
\]
This proves \eqref{eq:domain-safe-radial-AW}.
\end{proof}

\section{Proof of Proposition~\ref{prop:transported-density-comparison}}
\label{app:transported-density-calculus}

We first explain the formal calculation.  Let
\(A,B\in L^1([-T,T];\mathcal W^1)\) be real-valued,
\(t\in[-T,T]\), and put
\(D=A-B\) and \(A_\theta=B+\theta D\), $\theta\in [0,1]$. For later use, we first introduce the
interaction propagator
\[
 V_\theta(t,s)
 :=\ee^{\ii th}U_{A_\theta}(t,s)\ee^{-\ii sh}.
\]
For a one-particle operator \(T\), write
\(\alpha_t(T):=\ee^{\ii th}T\ee^{-\ii th}\).
Suppose first that \(u\) is smooth, so that all quadratic forms below
are well defined.  Differentiating in \(\theta\) implies for \(q\in\mathbb Z^d\setminus\{0\}\), 
\[
 \langle U_A(t,0)u,M_qU_A(t,0)u\rangle
 -\langle U_B(t,0)u,M_qU_B(t,0)u\rangle
 =\ii\int_0^1\!\int_0^t
 \langle u,[\mathcal D_\theta(r),Y_{\theta,q}(t)]u\rangle
 \,\dd r\,\dd\theta,
\]
where
\[
 \mathcal D_\theta(r)
 :=V_\theta(r,0)^*\alpha_r(D(r))V_\theta(r,0),
 \qquad
 Y_{\theta,q}(t)
 :=V_\theta(t,0)^*\alpha_t(M_q)V_\theta(t,0).
\]

For \(u\in H^{-\kappa}\), the field does not belong to \(\gH\), so the unregularized
quadratic form may be undefined. Thus the proof reduces to giving a meaning to the quadratic form on
the right and estimating it uniformly in \(\theta\) and \(r\).  

The main difficulty is that the exceptional set must be fixed before
the prescribed potentials are chosen.  In the application below, the
potentials are generated by the paths themselves and therefore depend
on the initial field.  If the transported quadratic forms were
constructed separately for each fixed potential, the exceptional set
could also depend on that potential.  Such a construction would not
allow us to compare two paths starting from the same rough initial
field.

We therefore begin with the elementary operators
\[
W_{\ell,y}f(x)=e^{i\ell\cdot x}f(x+y),
\qquad \ell\in\mathbb Z^d,\quad y\in\mathbb T^d.
\]
For each \(\ell\), we construct the corresponding centered quadratic
field as a continuous function of \(y\), on a single set of full
\(\mu_0\)-measure.  This set is fixed independently of the prescribed
potentials.

The propagator is expanded by iterating the Duhamel formula.  Each term
is an integral over ordered time variables.  For fixed values of these
variables, the product of the elementary operators again has the form
\[
f(x)\longmapsto c\,e^{i\ell\cdot x}f(x+y),
\qquad |c|=1.
\]
Here the Fourier index \(\ell\) is determined by the Fourier modes
appearing in the product, whereas the translation \(y\) depends on all
the time variables.  As these variables are integrated out, we
therefore obtain, for each fixed \(\ell\), not a single translation
but a weighted integral over different translations.  This weighted
integral is described by a finite complex measure in \(y\), which we
call the coefficient measure.
The centered quadratic expression corresponding to such a term is then
defined by integrating the elementary centered quadratic fields with
respect to its coefficient measures.  Since all the elementary fields
were constructed on one common set of full \(\mu_0\)-measure,
independently of the prescribed potentials, every expression obtained
from the propagator is defined on the same set.

A small technical issue arises when we pass from the Duhamel expansion
to the coefficient measures.  Under free conjugation, the coefficient measure of a Fourier mode is a Dirac mass whose location depends
continuously on \(t\).  Such moving Dirac masses are not continuous in
total variation.  Consequently, the time integrals in the
Duhamel formula cannot be defined directly as integrals in the
coefficient-measure norm.
Instead, we first obtain the Volterra expansion and prove that it converges in operator norm.
For each term of the expansion, we then define its coefficient measure
by pushing forward the measure on the corresponding time simplex under
the map that assigns the resulting translation parameter. The factorial bounds on these measures yield convergence of the full series, the exponential propagator estimate, and the bounds needed for the commutator.

\subsection{Coefficient measures for translated quadratic fields}
\label{appsubsec:translated-quadratic-fields}
\label{appsubsec:renormalized-commutator-pairings}

This subsection prepares the quadratic expressions used in the
comparison argument.  
To this end, we first introduce the elementary operator identities used in these
constructions.
 For
\(\ell\in\mathbb Z^d\) and \(y\in\mathbb T^d\), let
\(\tau_yf(x)=f(x+y)\) and set \(W_{\ell,y}=M_\ell\tau_y\).  Then
\begin{equation}\label{eq:modulation-translation-identities}
 W_{\ell,y}W_{k,z}
 =\ee^{\ii k\cdot y}W_{\ell+k,y+z},
 \qquad
 \alpha_t(W_{\ell,y})
 =\ee^{\ii t|\ell|^2}W_{\ell,y+2t\ell}.
\end{equation}
Set \(h_p=|p|^2+m\).  Since
\(W_{\ell,y}\mathrm e_p=\ee^{\ii p\cdot y}\mathrm e_{p+\ell}\), for
\(j\geq0\) define
\[
 \begin{aligned}
 Z_{\ell,j}(y;u)
 &:=\langle \Pi_{2^j}u,W_{\ell,y}\Pi_{2^j}u\rangle
 -\Tr(h^{-1}\Pi_{2^j}W_{\ell,y}\Pi_{2^j})\\
 &=\sum_{\substack{|p|\leq2^j\\|p+\ell|\leq2^j}}
 \left(
 \overline{u_{p+\ell}}\,u_p
 -\frac{\delta_{\ell0}}{h_p}
 \right)\ee^{\ii p\cdot y},
 \qquad y\in\mathbb T^d.
 \end{aligned}
\]
At \(y=0\), this is the mode defined in
\eqref{eq:classical-density-mode}, i.e. 
$Z_{\ell,j}(0;u)=\rho_{2^j,\ell}(u).$

\begin{lemma}
\label{lem:quadratic-translation-field}
Let \(\kappa\) be fixed by
\eqref{eq:path-Sobolev-exponent}.  There exists a \(\mu_0\)-full Borel
set \(\Omega_*\subset H^{-\kappa}(\mathbb T^d)\), such that for every \(u\in\Omega_*\),
\begin{equation}\label{eq:universal-Z-limit}
 Z_\ell(\,\cdot\,;u)
 :=\lim_{j\to\infty}Z_{\ell,j}(\,\cdot\,;u)
 \quad\text{in }C(\mathbb T^d)
\end{equation}
for all \(\ell\in\mathbb Z^d\).  Moreover, there
is \(\gamma>0\) such that
\begin{equation}\label{eq:universal-Z-weighted}
 \sup_{\ell\in\mathbb Z^d}
 \langle\ell\rangle^{1/4}
 \|Z_\ell(\,\cdot\,;u)\|_{C^\gamma(\mathbb T^d)}
 <\infty.
\end{equation}
After setting \(Z_\ell=0\) outside \(\Omega_*\), the map
\((u,y)\mapsto Z_\ell(y;u)\) is Borel.
\end{lemma}

\begin{proof}
Wick's rule gives, for every \(j\geq0\), \(\ell\in\mathbb Z^d\), and
\(y\in\mathbb T^d\),
\[
 \mathbb E_{\mu_0}|Z_{\ell,j}(y)|^2
 =\sum_{\substack{|p|\leq2^j\\|p+\ell|\leq2^j}}
 \frac1{h_ph_{p+\ell}}.
\]
Fix \(0<\beta<1/2\).  Since
\[
 |\ee^{\ii p\cdot y}-\ee^{\ii p\cdot z}|
 \leq C_\beta |p|^\beta |y-z|^\beta,
\]
the same calculation and the standard discrete convolution estimate
imply
\[
 \mathbb E_{\mu_0}
 |Z_{\ell,j}(y)-Z_{\ell,j}(z)|^2
 \leq C_\beta |y-z|^{2\beta}
 \begin{cases}
 \langle\ell\rangle^{-2+2\beta}
       \log(2+\langle\ell\rangle),&d=2,\\[1mm]
 \langle\ell\rangle^{-1+2\beta},&d=3.
 \end{cases}
\]
Using Gaussian hypercontractivity and Kolmogorov's criterion, we obtain, whenever \(r>d/\beta\) and \(0<\gamma<\beta-d/r\),
\[
 \sup_{j\geq0}
 \|Z_{\ell,j}\|_{L^r(\mu_0;C^\gamma(\mathbb T^d))}
 \leq C_{r,\beta,\gamma}
 \langle\ell\rangle^{-1/2+\beta}.
\]
Applying the same argument to
\(Z_{\ell,j+1}-Z_{\ell,j}\), with the sum restricted to
\(2^j<\max\{|p|,|p+\ell|\}\leq2^{j+1}\), implies
\[
 \|Z_{\ell,j+1}-Z_{\ell,j}\|_{L^r(\mu_0;C^\gamma(\mathbb T^d))}
 \leq C_{r,\ell,\beta,\gamma}2^{-j(1/2-\beta)}.
\]
Hence, for every fixed \(\ell\), the dyadic sequence converges almost
surely in \(C^\gamma(\mathbb T^d)\).  Taking the intersection over the
countable set \(\mathbb Z^d\) gives a Borel set
\(\Omega_1\) of full measure on which all these limits exist.  Extend
each limit by zero outside \(\Omega_1\).   The uniform
estimate passes to the limit and gives
\[
 \|Z_\ell\|_{L^r(\mu_0;C^\gamma(\mathbb T^d))}
 \leq C_{r,\beta,\gamma}
 \langle\ell\rangle^{-1/2+\beta}.
\]

Choose \(\beta=1/8\), \(r>8d\), and
\(0<\gamma<1/8-d/r\).  Then
\[
 \mathbb E_{\mu_0}
 \sum_{\ell\in\mathbb Z^d}
 \left(
 \langle\ell\rangle^{1/4}
 \|Z_\ell\|_{C^\gamma(\mathbb T^d)}
 \right)^r
 \leq C\sum_{\ell\in\mathbb Z^d}
 \langle\ell\rangle^{-r/8}<\infty.
\]
Let \(\Omega_2\) be the Borel event on which this series is finite and
set $\Omega_*
 :=\Omega_1\cap\Omega_2\cap H^{-\kappa}(\mathbb T^d).$
This implies \eqref{eq:universal-Z-limit} and
\eqref{eq:universal-Z-weighted}. The map \(u\mapsto Z_\ell(\cdot;u)\), extended by zero outside
\(\Omega_1\), is Borel as a \(C(\mathbb T^d)\)-valued map.  Since
\((f,y)\mapsto f(y)\) is continuous on
\(C(\mathbb T^d)\times\mathbb T^d\), it follows that
\((u,y)\mapsto Z_\ell(y;u)\) is jointly Borel.
\end{proof}

For \(u\in\Omega_*\), let
\[
 \mathcal K(u):=\sup_{\ell\in\mathbb Z^d}
 \sup_{y\in\mathbb T^d}|Z_\ell(y;u)|<\infty,
\]
and put \(\mathcal K(u)=0\) outside \(\Omega_*\).

We now introduce the coefficient measures described at the beginning of this appendix.
For a finite complex Borel measure \(\mu\) on \(\mathbb T^d\), define its
Fourier--Stieltjes coefficients by
\[
 (\mathcal F_{\mathrm{FS}}\mu)(p)
 :=\int_{\mathbb T^d}\ee^{-\ii p\cdot y}\,\dd\mu(y),
 \qquad p\in\mathbb Z^d.
\]
Let \(\mathcal M(\mathbb T^d)\) denote the finite complex Borel
measures on \(\mathbb T^d\), and define the weighted coefficient space
\[
 \mathfrak M_1
 :=\Big\{
 \boldsymbol{\mathfrak m}
   =(\mathfrak m_\ell)_{\ell\in\mathbb Z^d}:
 \mathfrak m_\ell\in\mathcal M(\mathbb T^d),\quad
 \|\boldsymbol{\mathfrak m}\|_{\mathfrak M_1}
 :=\sum_{\ell\in\mathbb Z^d}
 \langle\ell\rangle\|\mathfrak m_\ell\|_{\mathrm{TV}}<\infty
 \Big\}.
\]
For \(\boldsymbol{\mathfrak m}\in\mathfrak M_1\) and each
\(\ell\in\mathbb Z^d\), define \(T_\ell\in\mathcal B(\gH)\) by the
weak operator integral
\[
 \langle g,T_\ell f\rangle
 :=\int_{\mathbb T^d}\langle g,W_{\ell,y}f\rangle
   \,\dd\mathfrak m_\ell(y),
 \qquad f,g\in\gH.
\]
Thus \(T_\ell\) is the superposition of the operators
\(W_{\ell,y}\) over the translation variable \(y\).  Equivalently, in
orthonormal-basis coordinates,
\[
 (T_\ell f)_{p+\ell}
 =(\mathcal F_{\mathrm{FS}}\mathfrak m_\ell)(-p)f_p,
 \qquad
 \|T_\ell\|_{\mathrm{op}}\leq\|\mathfrak m_\ell\|_{\mathrm{TV}}.
\]
Define
\[
 \mathscr J(\boldsymbol{\mathfrak m})
 :=\sum_{\ell\in\mathbb Z^d}T_\ell.
\]
The series converges absolutely in operator norm, and we set
\(\mathscr A_1:=\mathscr J(\mathfrak M_1)\).  The next lemma
proves uniqueness of the coefficient family and 
estimates used below.

\begin{lemma}
\label{lem:crossed-product-renormalization}
The following assertions hold.
\begin{enumerate}[label=\textup{(\roman*)},leftmargin=2.2em]
\item The map \(\mathscr J:\mathfrak M_1\to\mathcal B(\gH)\) is
injective.  Hence every \(T\in\mathscr A_1\) has a unique
coefficient family \((\mathfrak m_\ell^T)_\ell\).  We equip
\(\mathscr A_1\) with the norm
\[
 \|T\|_{\mathscr A_1}
 :=\sum_{\ell\in\mathbb Z^d}
  \langle\ell\rangle\|\mathfrak m_\ell^T\|_{\mathrm{TV}}.
\]
We also use the auxiliary unweighted coefficient norm on \(\mathscr A_1\),
\[
\|T\|_{\mathscr A_0}
:=\sum_{\ell\in\mathbb Z^d}\|\mathfrak m_\ell^T\|_{\mathrm{TV}}
\leq\|T\|_{\mathscr A_1}.
\]
\item It holds that
\begin{equation}\label{eq:crossed-algebra-product}
 \|TS\|_{\mathscr A_1}
 \leq C\|T\|_{\mathscr A_1}
        \|S\|_{\mathscr A_1},
 \qquad
 \|T^*\|_{\mathscr A_1}
 =\|T\|_{\mathscr A_1}.
\end{equation}
The unweighted norm satisfies
\[
\|TS\|_{\mathscr A_0}\leq\|T\|_{\mathscr A_0}\|S\|_{\mathscr A_0},
\qquad
\|T^*\|_{\mathscr A_0}=\|T\|_{\mathscr A_0}.
\]
Moreover,
\begin{equation}\label{eq:free-conjugation-coefficient-measures}
 \mathfrak m_\ell^{\alpha_t(T)}
 =\ee^{\ii t|\ell|^2}
 (\tau_{2t\ell})_\#\mathfrak m_\ell^T,
 \qquad
 \|\alpha_t(T)\|_{\mathscr A_1}
 =\|T\|_{\mathscr A_1},
\end{equation}
where \(\tau_\xi:\mathbb T^d\to\mathbb T^d\) is the translation
\(\tau_\xi(y)=y+\xi\).
\item For \(u\in\Omega_*\) and
\(T\in\mathscr A_1\), define
\begin{equation}\label{eq:quadratic-functional-definition}
 \mathcal Q_u(T)
 :=\sum_{\ell\in\mathbb Z^d}
 \int_{\mathbb T^d}Z_\ell(y;u)\,
 \dd\mathfrak m_\ell^T(y).
\end{equation}
  In particular,
\begin{equation}\label{eq:quadratic-functional-bound}
 |\mathcal Q_u(T)|
 \leq\mathcal K(u)\|T\|_{\mathscr A_0},
 \qquad
 \mathcal Q_u(W_{\ell,y})=Z_\ell(y;u).
\end{equation}
For fixed \(T\), the map \(u\mapsto\mathcal Q_u(T)\), extended by zero
outside \(\Omega_*\), is Borel.
\end{enumerate}
\end{lemma}

\begin{proof}
For \(p,\ell\in\mathbb Z^d\),
\begin{equation*}
 \langle\mathrm e_{p+\ell},
 \mathscr J(\boldsymbol{\mathfrak m})\mathrm e_p\rangle
 =\int_{\mathbb T^d}\ee^{\ii p\cdot y}
   \,\dd\mathfrak m_\ell(y)
 =(\mathcal F_{\mathrm{FS}}\mathfrak m_\ell)(-p).
\end{equation*}
If \(\mathscr J(\boldsymbol{\mathfrak m})=0\), all Fourier--Stieltjes
coefficients of every \(\mathfrak m_\ell\) vanish, and hence
\(\mathfrak m_\ell=0\) for every \(\ell\).  This proves
\textup{(i)}.

For \(T,S\in\mathscr A_1\),
\eqref{eq:modulation-translation-identities} implies
\[
 \int f(r)\,\dd\mathfrak m_j^{TS}(r)
 =\sum_{\ell+k=j}\iint
 f(y+z)\ee^{\ii k\cdot y}
 \,\dd\mathfrak m_\ell^T(y)
 \,\dd\mathfrak m_k^S(z),
 \qquad f\in C(\mathbb T^d).
\]
The series converges absolutely in total variation and
\[
 \|\mathfrak m_j^{TS}\|_{\mathrm{TV}}
 \leq\sum_{\ell+k=j}
 \|\mathfrak m_\ell^T\|_{\mathrm{TV}}
 \|\mathfrak m_k^S\|_{\mathrm{TV}}.
\]
Summing without weights gives the \(\mathscr A_0\) product bound.
Thus, using
\(\langle\ell+k\rangle\leq\sqrt2
  \langle\ell\rangle\langle k\rangle\),
\[
 \|TS\|_{\mathscr A_1}
 \leq \sqrt2\|T\|_{\mathscr A_1}
          \|S\|_{\mathscr A_1}.
\]
If \(R(y)=-y\), then
\[
 W_{\ell,y}^*=\ee^{\ii\ell\cdot y}W_{-\ell,-y},
 \qquad
 \mathfrak m_{-\ell}^{T^*}
 =R_\#\!\left(
   \ee^{\ii\ell\cdot y}\overline{\mathfrak m_\ell^T}
  \right),
\]
which implies the adjoint equalities in both coefficient norms, including
\eqref{eq:crossed-algebra-product}.  The second identity in
\eqref{eq:modulation-translation-identities} implies
\eqref{eq:free-conjugation-coefficient-measures}.
Thus \textup{(ii)} holds and 
\textup{(iii)} follows easily.
\end{proof}

We finish this subsection by defining the trace correction for
commutators.  To this end, we introduce the following notation. 
Let
\[
 G_\delta(y)
 :=\sum_{p\in\mathbb Z^d}
 \frac{\ee^{-\delta h_p}}{h_p}\ee^{\ii p\cdot y},
 \qquad \delta>0,
\]
and, for \(y\neq0\), let
\(G_0(y)=\lim_{\delta\downarrow0}G_\delta(y)\).  Here \(|y|\)
denotes the distance from \(y\) to \(0\) on the torus.  It is well
known that
\begin{equation}\label{eq:truncated-green-bound}
 |G_\delta(y)|
 \leq C\left(1+\frac{1}{\sqrt\delta+|y|}\right).
\end{equation}
\begin{proposition}
\label{prop:renormalized-crossed-commutator}

For \(X,Y\in\mathscr A_1\), the following limit exists:
\begin{equation}\label{eq:trace-commutator-bound}
 \Tr(h^{-1}[X,Y])
 :=\lim_{\delta\downarrow0}
 \Tr(\ee^{-\delta h}h^{-1}[X,Y]),
 \qquad
 |\Tr(h^{-1}[X,Y])|
 \leq C\|X\|_{\mathscr A_1}\|Y\|_{\mathscr A_0}.
\end{equation}
For \(u\in\Omega_*\), define
\[
 \mathcal C_u([X,Y])
 :=\mathcal Q_u([X,Y])+\Tr(h^{-1}[X,Y]).
\]
  Then
\begin{equation}\label{eq:raw-commutator-bound}
 |\mathcal C_u([X,Y])|
 \leq C(1+\mathcal K(u))
 \|X\|_{\mathscr A_1}\|Y\|_{\mathscr A_0}.
\end{equation}
In particular, if \(V\in\mathscr A_1\) is unitary on \(\gH\),
\(X\in\mathscr A_1\), \(q\in\mathbb Z^d\setminus\{0\}\), and
\(y\in\mathbb T^d\), then
\begin{equation}\label{eq:conjugated-mode-commutator-bound}
 |\mathcal C_u([X,V^*W_{q,y}V])|
 \leq C(1+\mathcal K(u))
 \|X\|_{\mathscr A_1}\|V\|_{\mathscr A_0}^2.
\end{equation}
\end{proposition}

\begin{proof}
We first take \(X=W_{\ell,y}\) and \(Y=W_{k,z}\).
For every \(r\in\mathbb Z^d\) and \(\xi\in\mathbb T^d\),
\begin{equation}\label{eq:zero-mode-trace-identity}
 \Tr(\ee^{-\delta h}h^{-1}W_{r,\xi})
 =\1_{\{r=0\}}G_\delta(\xi).
\end{equation}
Combining this identity with
\eqref{eq:modulation-translation-identities}, we obtain
\begin{equation*}
 \Tr\!\left(
 \ee^{-\delta h}h^{-1}[W_{\ell,y},W_{k,z}]
 \right)
 =\left(\ee^{-\ii\ell\cdot y}-\ee^{\ii\ell\cdot z}\right)
 G_\delta(y+z)\1_{\{\ell+k=0\}}
 =\ee^{-\ii\ell\cdot y}
 \left(1-\ee^{\ii\ell\cdot(y+z)}\right)
 G_\delta(y+z)\1_{\{\ell+k=0\}}.
\end{equation*}
Since
\(|1-\ee^{\ii\ell\cdot\xi}|
 \leq C\min\{1,|\ell||\xi|\}\),
using \eqref{eq:truncated-green-bound} we obtain
\[
 \sup_{\delta>0,\,\xi\in\mathbb T^d}
 |(1-\ee^{\ii\ell\cdot\xi})G_\delta(\xi)|
 \leq C\langle\ell\rangle.
\]
For \(\ell\in\mathbb Z^d\), define
\[
 F_{\ell,\delta}(\xi)
 :=(1-\ee^{\ii\ell\cdot\xi})G_\delta(\xi),
 \qquad
 F_{\ell,0}(\xi)
 :=\begin{cases}
 (1-\ee^{\ii\ell\cdot\xi})G_0(\xi),&\xi\neq0,\\
 0,&\xi=0.
 \end{cases}
\]
Then, for every \(\xi\in\mathbb T^d\),
\[
 F_{\ell,\delta}(\xi)\longrightarrow F_{\ell,0}(\xi),
 \qquad
 \sup_{\delta>0,\,\xi\in\mathbb T^d}
 |F_{\ell,\delta}(\xi)|\leq C\langle\ell\rangle.
\]
If \((\mathfrak m_\ell^X)_\ell\) and
\((\mathfrak m_k^Y)_k\) are the coefficient families of \(X\) and
\(Y\), respectively, then
\[
 \Tr(\ee^{-\delta h}h^{-1}[X,Y])
 =\sum_{\ell\in\mathbb Z^d}\iint
 \ee^{-\ii\ell\cdot y}F_{\ell,\delta}(y+z)
 \,\dd\mathfrak m_\ell^X(y)\,\dd\mathfrak m_{-\ell}^Y(z).
\]
Moreover,
\[
 \sum_{\ell\in\mathbb Z^d}\iint
 |F_{\ell,\delta}(y+z)|
 \,\dd|\mathfrak m_\ell^X|(y)\,\dd|\mathfrak m_{-\ell}^Y|(z)
 \leq C\sum_{\ell\in\mathbb Z^d}\langle\ell\rangle
 \|\mathfrak m_\ell^X\|_{\mathrm{TV}}\|\mathfrak m_{-\ell}^Y\|_{\mathrm{TV}}
 \leq C\|X\|_{\mathscr A_1}\|Y\|_{\mathscr A_0}.
\]
Using dominated convergence we obtain
\[
 \Tr(h^{-1}[X,Y])
 =\sum_{\ell\in\mathbb Z^d}\iint
 \ee^{-\ii\ell\cdot y}F_{\ell,0}(y+z)
 \,\dd\mathfrak m_\ell^X(y)
 \,\dd\mathfrak m_{-\ell}^Y(z),
\]
which implies \eqref{eq:trace-commutator-bound}.
For the centered part, the unweighted product bound gives
\[
|\mathcal Q_u([X,Y])|
\leq 2\mathcal K(u)\|X\|_{\mathscr A_0}\|Y\|_{\mathscr A_0}
\leq 2\mathcal K(u)\|X\|_{\mathscr A_1}\|Y\|_{\mathscr A_0}.
\]
Together with \eqref{eq:trace-commutator-bound}, this yields
\eqref{eq:raw-commutator-bound}.

The coefficient family of \(W_{q,y}\) is
\(\mathfrak m_q=\delta_y\) and \(\mathfrak m_k=0\) for \(k\neq q\),
so
\[
 \|W_{q,y}\|_{\mathscr A_0}=1.
\]
Thus Lemma~\ref{lem:crossed-product-renormalization} implies
\[
 \|V^*W_{q,y}V\|_{\mathscr A_0}
 \leq C\|V\|_{\mathscr A_0}^2.
\]
Substituting this into \eqref{eq:raw-commutator-bound}, we obtain
\eqref{eq:conjugated-mode-commutator-bound}.
\end{proof}

\subsection{Removing the heat regularization}
\label{appsubsec:heat-regularization}

In this subsection we regularize $Y_{\theta,q}$ using the heat
multiplier \(\mathsf H_\varepsilon:=\ee^{-\varepsilon h/2}\):
\[
 Y_{\theta,q,\varepsilon}(t)
 :=V_\theta(t,0)^*\mathsf H_\varepsilon
   \alpha_t(M_q)\mathsf H_\varepsilon V_\theta(t,0).
\]
 The heat multipliers make
the commutator in the formal identity sufficiently smoothing: it maps
\(H^{-\kappa}\) to \(H^\kappa\) and is trace class.  Its pairing with
the rough initial field is therefore a Sobolev-duality
pairing.  A cutoff calculation identifies this pairing with the
centered quadratic expression constructed above, together with a
trace correction.  The centered part is controlled by the total
variations of the coefficient measures.  In the trace correction, the
two multiplication orders in the commutator produce a factor that
vanishes at the singular point of the Green function (see the proof of Proposition~\ref{prop:renormalized-crossed-commutator}).  This
cancellation removes the singularity and gives a finite limit when
the heat regularization is removed.

\begin{lemma}
\label{lem:heat-smoothed-raw-form}
Let
\(X,V\in\mathscr A_1\), fix \(q\in\mathbb Z^d\) and
\(y\in\mathbb T^d\), and set
\[
 Y_\varepsilon
 =V^*\mathsf H_\varepsilon W_{q,y}\mathsf H_\varepsilon V,
 \qquad
 L_\varepsilon=[X,Y_\varepsilon].
\]
Then \(Y_\varepsilon\) and \(L_\varepsilon\) are contained in
\[
 \mathcal B(H^{-\kappa},H^{\kappa})
 \cap\mathfrak S^1(\gH).
\]
For every \(u\in\Omega_*\) and
\(T\in\{Y_\varepsilon,L_\varepsilon\}\),
\begin{equation}\label{eq:heat-smoothed-duality}
 \langle u,Tu\rangle_{H^{-\kappa},H^{\kappa}}
 =\mathcal Q_u(T)+\Tr(h^{-1}T).
\end{equation}
Moreover,
\begin{align}
 \|Y_\varepsilon\|_{H^{-\kappa}\to H^{\kappa}}
 &\leq C_{\varepsilon,\kappa}\langle q\rangle
 \|V\|_{\mathscr A_1}^2,\quad
 \|L_\varepsilon\|_{H^{-\kappa}\to H^{\kappa}}
 \leq C_{\varepsilon,\kappa}\langle q\rangle
 \|X\|_{\mathscr A_1}\|V\|_{\mathscr A_1}^2.
 \label{eq:heat-smoothed-Sobolev-bound}
\end{align}
\end{lemma}

\begin{proof}
 For \(|\sigma|\leq\kappa\), the Fourier transform of
\(W_{\ell,y}\) and
\(
 \langle p+\ell\rangle^{|\sigma|}
 \leq2^{|\sigma|/2}\langle p\rangle^{|\sigma|}
       \langle\ell\rangle^{|\sigma|}
\)
imply
\[
 \|T\|_{H^\sigma\to H^\sigma}
 \leq C_{\kappa}\sum_\ell\langle\ell\rangle
 \|\mathfrak m_\ell^T\|_{\mathrm{TV}}
 =C_{\kappa}\|T\|_{\mathscr A_1}.
\]
For \(\mathsf H_\varepsilon=\ee^{-\varepsilon h/2}\),
\begin{equation*}
 \|\mathsf H_\varepsilon\|_{H^{-\kappa}\to H^{\kappa}}
 =\sup_{p\in\mathbb Z^d}
 \langle p\rangle^{2\kappa}\ee^{-\varepsilon h_p/2}<\infty,
 \qquad
 \|\mathsf H_\varepsilon\|_{\mathfrak S^2(\gH)}^2
 =\sum_{p\in\mathbb Z^d}\ee^{-\varepsilon h_p}<\infty.
\end{equation*}
  Since
\(\|\mathsf H_\varepsilon\|_{H^\kappa\to H^\kappa}\leq1\), the above
 bound implies
\[
 \begin{aligned}
 \|Y_\varepsilon\|_{H^{-\kappa}\to H^\kappa}
 &\leq
 \|V^*\|_{H^\kappa\to H^\kappa}
 \|\mathsf H_\varepsilon\|_{H^\kappa\to H^\kappa}
 \|W_{q,y}\|_{H^\kappa\to H^\kappa}
 \|\mathsf H_\varepsilon\|_{H^{-\kappa}\to H^\kappa}
 \|V\|_{H^{-\kappa}\to H^{-\kappa}}\\
 &\leq C_{\varepsilon,\kappa}\langle q\rangle
 \|V\|_{\mathscr A_1}^2.
 \end{aligned}
\]
Applying the same bound to
\(L_\varepsilon=XY_\varepsilon-Y_\varepsilon X\) gives
\[
 \|L_\varepsilon\|_{H^{-\kappa}\to H^\kappa}
 \leq C_{\varepsilon,\kappa}\langle q\rangle
 \|X\|_{\mathscr A_1}
 \|V\|_{\mathscr A_1}^2,
\]
which proves \eqref{eq:heat-smoothed-Sobolev-bound}.

For the trace-class assertion, both
\(V^*\mathsf H_\varepsilon\) and
\(\mathsf H_\varepsilon V\) are Hilbert--Schmidt.  Hence
we obtain
\[
 \|Y_\varepsilon\|_{\mathfrak S^1(\gH)}
 \leq \|V\|_{\mathrm{op}}^2
       \|\mathsf H_\varepsilon\|_{\mathfrak S^2(\gH)}^2
       \|W_{q,y}\|_{\mathrm{op}}<\infty,
\]
and
\[
 \|L_\varepsilon\|_{\mathfrak S^1(\gH)}
 \leq2\|X\|_{\mathrm{op}}\|Y_\varepsilon\|_{\mathfrak S^1(\gH)}<\infty.
\]
The same conclusions hold for finite linear combinations.

Define
\[
 \mathcal Q_{u,M}(T)
 :=\langle \Pi_Mu,T\Pi_Mu\rangle
   -\Tr(h^{-1}\Pi_MT\Pi_M).
\]
Since  $T\in\mathcal B(H^{-\kappa},H^{\kappa})
 \cap\mathfrak S^1(\gH)$ we obtain
\[
\lim_{M\to\infty} \mathcal Q_{u,M}(T)= \langle u,Tu\rangle_{H^{-\kappa},H^{\kappa}}
 -\Tr(h^{-1}T).
\]
It remains to identify this limit with \(\mathcal Q_u(T)\).

Assume first that the coefficient families of \(X\) and \(V\) have
finite Fourier-mode support.  Thus either choice of \(T\)
made above has
finite mode support.  At the dyadic cutoff \(M=2^j\), linearity and the
definition of \(Z_{\ell,j}\) imply
\[
 \begin{aligned}
 \mathcal Q_{u,2^j}(T)
 &=\sum_\ell\int_{\mathbb T^d}
   \Bigl[
    \langle\Pi_{2^j}u,W_{\ell,z}\Pi_{2^j}u\rangle
    -\Tr(h^{-1}\Pi_{2^j}W_{\ell,z}\Pi_{2^j})
   \Bigr]\,\dd\mathfrak m_\ell^T(z)\\
 &=\sum_\ell\int_{\mathbb T^d}
   Z_{\ell,j}(z;u)\,\dd\mathfrak m_\ell^T(z).
 \end{aligned}
\]
The sum over \(\ell\) is finite, so
\eqref{eq:universal-Z-limit} may be applied term by term.  It follows
that the right-hand side converges to \(\mathcal Q_u(T)\).  Since the
full integer limit of \(\mathcal Q_{u,M}(T)\) already exists, it has the
same value as this dyadic subsequential limit.

For general \(X,V\), truncate their Fourier modes:
\[
 \mathfrak m_\ell^{X_n}
 :=\1_{\{|\ell|\leq n\}}\mathfrak m_\ell^X,
 \qquad
 \mathfrak m_\ell^{V_n}
 :=\1_{\{|\ell|\leq n\}}\mathfrak m_\ell^V.
\]
Let \(X_n,V_n\in\mathscr A_1\) be the corresponding operators.
Then
\[
 X_n\to X,
 \qquad V_n\to V
 \quad\text{in }\mathscr A_1\quad(n\to\infty).
\]
Define
\[
 Y_{\varepsilon,n}
 :=V_n^*\mathsf H_\varepsilon W_{q,y}\mathsf H_\varepsilon V_n,
 \qquad
 L_{\varepsilon,n}:=[X_n,Y_{\varepsilon,n}].
\]
If the operator denoted by \(T\) is \(Y_\varepsilon\), set
\(T_n=Y_{\varepsilon,n}\); if it is \(L_\varepsilon\), set
\(T_n=L_{\varepsilon,n}\).  The relevant differences are
\begin{align*}
 Y_{\varepsilon,n}-Y_\varepsilon
 &=(V_n^*-V^*)\mathsf H_\varepsilon W_{q,y}
   \mathsf H_\varepsilon V_n
   +V^*\mathsf H_\varepsilon W_{q,y}\mathsf H_\varepsilon(V_n-V),\\
 L_{\varepsilon,n}-L_\varepsilon
 &=[X_n-X,Y_{\varepsilon,n}]
   +[X,Y_{\varepsilon,n}-Y_\varepsilon].
\end{align*}
Then we have
\begin{align*}
 &\|T_n-T\|_{\mathscr A_1}
 +\|T_n-T\|_{H^{-\kappa}\to H^\kappa}
 +\|T_n-T\|_{\mathfrak S^1(\gH)}\\
 &\qquad\leq
 C_{\varepsilon,\kappa,q,X,V}
 \bigl(\|X_n-X\|_{\mathscr A_1}
       +\|V_n-V\|_{\mathscr A_1}\bigr)
 \longrightarrow0,
\end{align*}
which implies that
\[
 |\mathcal Q_{u,M}(T)-\mathcal Q_{u,M}(T_n)|
 \leq\|u\|_{H^{-\kappa}}^2
 \|T-T_n\|_{H^{-\kappa}\to H^{\kappa}}
 +\|h^{-1}\|_{\mathrm{op}}\|T-T_n\|_{\mathfrak S^1(\gH)}\to0.
\]
Moreover, \eqref{eq:quadratic-functional-bound} gives
\[
 |\mathcal Q_u(T)-\mathcal Q_u(T_n)|
 \leq\mathcal K(u)
 \|T-T_n\|_{\mathscr A_1}.
\]
For each fixed \(n\), the finite-mode argument gives
\(\mathcal Q_{u,M}(T_n)\to\mathcal Q_u(T_n)\).  Therefore
taking the upper limit in \(M\) implies
\begin{align*}
 \limsup_{M\to\infty}
 |\mathcal Q_{u,M}(T)-\mathcal Q_u(T)|
 &\leq
 \sup_M|\mathcal Q_{u,M}(T)-\mathcal Q_{u,M}(T_n)|
 +|\mathcal Q_u(T_n)-\mathcal Q_u(T)|\\
 &\longrightarrow0\qquad(n\to\infty).
\end{align*}
We have therefore proved the truncation limit
\[
 \lim_{M\to\infty}\mathcal Q_{u,M}(T)=\mathcal Q_u(T).
\]
Hence \eqref{eq:heat-smoothed-duality} follows.
\end{proof}

\begin{lemma}
\label{lem:heat-regularization-removal}
Let
\[
 Y=V^*W_{q,y_0}V,
 \qquad
 Y_\varepsilon
 =V^*\mathsf H_\varepsilon W_{q,y_0}
  \mathsf H_\varepsilon V,
 \qquad q\in\mathbb Z^d\setminus\{0\},
\]
where \(V\in\mathscr A_1\) is unitary on \(\gH\), and let
\(X\in\mathscr A_1\).  For \(u\in\Omega_*\),
\begin{equation}\label{eq:heat-regularization-limit}
 \lim_{\varepsilon\downarrow0}
 \left\{
 \mathcal Q_u([X,Y_\varepsilon])
 +\Tr(h^{-1}[X,Y_\varepsilon])
 \right\}
 =\mathcal C_u([X,Y]).
\end{equation}
Moreover, uniformly in \(\varepsilon>0\),
\begin{equation}\label{eq:heat-regularization-uniform}
 \left|
 \mathcal Q_u([X,Y_\varepsilon])
 +\Tr(h^{-1}[X,Y_\varepsilon])
 \right|
 \leq C(1+\mathcal K(u))
 \|X\|_{\mathscr A_1}
 \|V\|_{\mathscr A_0}^2.
\end{equation}
\end{lemma}

\begin{proof}
We first establish the heat-kernel estimate used below.  With respect to
Lebesgue measure on \(\mathbb T^d\), let
\[
 p_\varepsilon(\xi)
 :=(2\pi)^{-d}\sum_{p\in\mathbb Z^d}
 \ee^{-\varepsilon|p|^2/2}\ee^{\ii p\cdot\xi}.
\]
Then \(p_\varepsilon\) is the normalized heat kernel,
\(\int_{\mathbb T^d}p_\varepsilon(\xi)\,\dd\xi=1\), and its Fourier
multiplier is \(\ee^{-\varepsilon|p|^2/2}\).  Thus, for every
\(f\in C(\mathbb T^d\times\mathbb T^d)\),
\begin{equation}\label{eq:double-heat-approximate-identity}
 \lim_{\varepsilon\downarrow0}\iint
 p_\varepsilon(\xi)p_\varepsilon(\zeta)f(\xi,\zeta)\,\dd \xi\,\dd \zeta=f(0,0).
\end{equation}
The smoothing operator has the representation
\[
 \mathsf H_\varepsilon
 =\ee^{-\varepsilon m/2}
 \int_{\mathbb T^d}W_{0,\xi}p_\varepsilon(\xi)\,\dd \xi,
 \qquad
 \|\mathsf H_\varepsilon\|_{\mathscr A_1}\leq1.
\]
Since \(\mathsf H_\varepsilon\in\mathfrak S^2(\gH)\), one has
\(Y_\varepsilon,[X,Y_\varepsilon]\in\mathfrak S^1(\gH)\).  We obtain
\[
 \Tr(h^{-1}[X,Y_\varepsilon])
 =\lim_{\delta\downarrow0}
  \Tr(\ee^{-\delta h}h^{-1}[X,Y_\varepsilon]).
\]
We have the coefficient representations
\[
 X=\sum_{\ell}\int W_{\ell,x}\,\dd\mathfrak m_\ell^X(x),\qquad
 V=\sum_c\int W_{c,z}\,\dd\mathfrak m_c^V(z),\qquad
 V^*=\sum_b\int W_{b,y}^*\,\dd\overline{\mathfrak m_b^V}(y).
\]
For heat variables
\(\xi,\zeta\), repeated use of
\eqref{eq:modulation-translation-identities} gives
\begin{equation}\label{eq:heat-mode-product}
 \begin{aligned}
 &W_{b,y}^*W_{0,\xi}W_{q,y_0}W_{0,\zeta}W_{c,z}
   =\Phi_{b,c}^{q}(\xi,\zeta)W_{j,r},\\
 &\Phi_{b,c}^{q}(\xi,\zeta)
   :=\ee^{\ii b\cdot y}
      \ee^{\ii q\cdot(\xi-y)}
      \ee^{\ii c\cdot(\xi-y+y_0+\zeta)},\\
 &j:=q+c-b,
 \qquad
 r:=-y+y_0+z+\xi+\zeta.
 \end{aligned}
\end{equation}
In particular,
\[
 |\Phi_{b,c}^{q}(\xi,\zeta)|=1,
 \qquad
 |\Phi_{b,c}^{q}(\xi,\zeta)-\Phi_{b,c}^{q}(0,0)|
 \leq C(|q|+|c|)(|\xi|+|\zeta|).
\]

Put
\[
 w:=x-y+y_0+z,
 \qquad x+r=w+\xi+\zeta.
\]
Using \eqref{eq:modulation-translation-identities}, we obtain
\[
 [W_{\ell,x},W_{j,r}]
 =\bigl(\ee^{\ii j\cdot x}-\ee^{\ii \ell\cdot r}\bigr)
  W_{\ell+j,x+r}.
\]
Together with \eqref{eq:zero-mode-trace-identity}, this implies
\[
 \Tr\!\left(
  \ee^{-\delta h}h^{-1}
  [W_{\ell,x},\Phi_{b,c}^{q}(\xi,\zeta)W_{j,r}]
  \right)
 =
 \1_{\{\ell+j=0\}}\ee^{\ii j\cdot x}
 \Phi_{b,c}^{q}(\xi,\zeta)
 [1-\ee^{-\ii j\cdot(w+\xi+\zeta)}]G_\delta(w+\xi+\zeta).
\]
We now justify the heat-kernel limit in the preceding expression,
including the case in which the argument of the Green function
converges to its singular point.  For \(j\in\mathbb Z^d\), define
\[
\Psi_j(\eta)
:=
\begin{cases}
	\bigl(1-\ee^{-\ii j\cdot\eta}\bigr)G_0(\eta),
	& \eta\neq0,\\
	0,&\eta=0.
\end{cases}
\]
The estimate used in the proof of
Proposition~\ref{prop:renormalized-crossed-commutator} gives
\begin{equation*}
	\sup_{\eta\in\mathbb T^d}\abs{\Psi_j(\eta)}
	\leq C\langle j\rangle,
\end{equation*}
and, for every \(\eta\in\mathbb T^d\),
\[
\bigl(1-\ee^{-\ii j\cdot\eta}\bigr)G_\delta(\eta)
\longrightarrow \Psi_j(\eta)
\qquad(\delta\downarrow0).
\]
We claim that, for every fixed \(w\in\mathbb T^d\),
\begin{equation}\label{eq:special-heat-kernel-limit}
	\begin{aligned}
		&\lim_{\varepsilon\downarrow0}
		\ee^{-\varepsilon m}
		\iint p_\varepsilon(\xi)p_\varepsilon(\zeta)
		\Phi_{b,c}^{q}(\xi,\zeta)
		\Psi_j(w+\xi+\zeta)\,\dd\xi\,\dd\zeta
		=\Phi_{b,c}^{q}(0,0)\Psi_j(w).
	\end{aligned}
\end{equation}
When \(w\neq0\), the integrand is continuous in a neighborhood of
\((\xi,\zeta)=(0,0)\).  Its contribution outside this neighborhood
tends to zero because the heat kernels concentrate at the origin and
\(\Psi_j\) is bounded.  Hence
\eqref{eq:double-heat-approximate-identity} gives
\eqref{eq:special-heat-kernel-limit} in this case.

It remains to consider \(w=0\).  Using
\eqref{eq:heat-mode-product}, we obtain
\begin{align*}
	&\left|
	\iint p_\varepsilon(\xi)p_\varepsilon(\zeta)
	\bigl(
	\Phi_{b,c}^{q}(\xi,\zeta)
	-\Phi_{b,c}^{q}(0,0)
	\bigr)
	\Psi_j(\xi+\zeta)\,\dd\xi\,\dd\zeta
	\right|\\
	&\qquad\leq
	C(|q|+|c|)\langle j\rangle
	\iint p_\varepsilon(\xi)p_\varepsilon(\zeta)
	(|\xi|+|\zeta|)\,\dd\xi\,\dd\zeta
	\longrightarrow0.
\end{align*}
The functions \(p_{2\varepsilon}\) and \(G_0\) are even.  For
\(\eta\neq0\),
\[
\Psi_j(\eta)
=
\bigl(1-\cos(j\cdot\eta)\bigr)G_0(\eta)
+\ii\sin(j\cdot\eta)G_0(\eta).
\]
The second term is odd, and therefore
\[
\int_{\mathbb T^d}p_{2\varepsilon}(\eta)
\sin(j\cdot\eta)G_0(\eta)\,\dd\eta=0.
\]
On the other hand, \eqref{eq:truncated-green-bound} and
\(\abs{1-\cos(j\cdot\eta)}
\leq C\min\{1,|j|^2|\eta|^2\}\)
give
\[
\abs{1-\cos(j\cdot\eta)}\,\abs{G_0(\eta)}
\leq C\langle j\rangle^2|\eta|.
\]
Consequently,
\[
\left|
\int_{\mathbb T^d}p_{2\varepsilon}(\eta)
\bigl(1-\cos(j\cdot\eta)\bigr)G_0(\eta)\,\dd\eta
\right|
\leq C\langle j\rangle^2\sqrt{\varepsilon}
\longrightarrow0.
\]
Thus we obtain
\eqref{eq:special-heat-kernel-limit} also when \(w=0\).

It follows from \eqref{eq:heat-mode-product} that
\[
\begin{aligned}
	&\lim_{\varepsilon\downarrow0}\lim_{\delta\downarrow0}
	\ee^{-\varepsilon m}
	\iint p_\varepsilon(\xi)p_\varepsilon(\zeta)
	\Tr\!\left(
	\ee^{-\delta h}h^{-1}
	[W_{\ell,x},
	\Phi_{b,c}^{q}(\xi,\zeta)W_{j,r}]
	\right)\,\dd\xi\,\dd\zeta\\
	&\qquad=
	\Tr\!\left(
	h^{-1}
	[W_{\ell,x},W_{b,y}^*W_{q,y_0}W_{c,z}]
	\right).
\end{aligned}
\]
Moreover, the elementary integrands are bounded by
\(
C\1_{\{\ell+j=0\}}\langle j\rangle
=
C\1_{\{\ell+j=0\}}\langle\ell\rangle.
\)
This is summable with respect to the coefficient measures, since
\[
\sum_{\ell,b,c}
\langle\ell\rangle
\|\mathfrak m_\ell^X\|_{\mathrm{TV}}
\|\mathfrak m_b^V\|_{\mathrm{TV}}
\|\mathfrak m_c^V\|_{\mathrm{TV}}
\leq
\|X\|_{\mathscr A_1}\|V\|_{\mathscr A_0}^2.
\]
Hence Fubini's theorem, dominated convergence, and
Proposition~\ref{prop:renormalized-crossed-commutator} imply
\[
 \lim_{\varepsilon\downarrow0}
 \Tr(h^{-1}[X,Y_\varepsilon])
 =\Tr(h^{-1}[X,Y]),
\]
and, uniformly in \(\varepsilon>0\),
\[
 |\Tr(h^{-1}[X,Y_\varepsilon])|
 \leq C
 \|X\|_{\mathscr A_1}
 \|V\|_{\mathscr A_0}^2.
\]
For the centered quadratic term, \eqref{eq:heat-mode-product} and
\eqref{eq:quadratic-functional-definition} give
\begin{equation}\label{eq:centered-mode-commutator}
 \mathcal Q_u([W_{\ell,x},\Phi_{b,c}^{q}(\xi,\zeta)W_{j,r}])
 =\Phi_{b,c}^{q}(\xi,\zeta)
 \bigl(\ee^{\ii j\cdot x}-\ee^{\ii \ell\cdot r}\bigr)
 Z_{\ell+j}(x+r;u).
\end{equation}
Integrating \eqref{eq:centered-mode-commutator} with respect to
\(\mathfrak m_\ell^X(x)\),
\(\overline{\mathfrak m_b^V}(y)\), and
\(\mathfrak m_c^V(z)\), and then summing over \(\ell,b,c\), gives
\(\mathcal Q_u([X,Y_\varepsilon])\).  Using dominated
convergence we obtain
\[
 \mathcal Q_u([X,Y_\varepsilon])
 \longrightarrow \mathcal Q_u([X,Y])
 \qquad\varepsilon\downarrow0.
\]
In \eqref{eq:centered-mode-commutator}, \(|\Phi_{b,c}^{q}|=1\)
and the phase difference has modulus at most \(2\).
Integrating against the coefficient measures and the normalized heat
kernels therefore gives, uniformly in \(\varepsilon\) and \(q\),
\[
|\mathcal Q_u([X,Y_\varepsilon])|
\leq 2\mathcal K(u)\|X\|_{\mathscr A_0}\|V\|_{\mathscr A_0}^2
\leq 2\mathcal K(u)\|X\|_{\mathscr A_1}\|V\|_{\mathscr A_0}^2.
\]
Combining the two previous limits and the definition of \(\mathcal C_u\) in
Proposition~\ref{prop:renormalized-crossed-commutator}, we obtain
\[
 \lim_{\varepsilon\downarrow0}
 \left\{\mathcal Q_u([X,Y_\varepsilon])
 +\Tr(h^{-1}[X,Y_\varepsilon])\right\}
 =\mathcal C_u([X,Y]).
\]
This proves \eqref{eq:heat-regularization-limit}.  The two uniform
bounds above prove \eqref{eq:heat-regularization-uniform}.
\end{proof}

\subsection{Volterra series and propagator bounds}
\label{appsubsec:prescribed-propagators}
This subsection completes the proof of Proposition~\ref{prop:transported-density-comparison}.  We first do the Volterra expansion.  Since a freely transported Fourier mode is not continuous in the coefficient norm, we construct the coefficient measures separately for each Volterra term by pushing forward measures on the corresponding time simplex.  The resulting total-variation estimates give an exponential bound for the propagator and justify differentiation with respect to an interpolation between two prescribed potentials.  We then apply this derivative formula to the heat-regularized transported density.  Together with the commutator estimates and the removal of the heat regularization proved above, this gives the existence of the transported density modes and the comparison estimate in Proposition~\ref{prop:transported-density-comparison}.

The coefficient measures in the time integrals below are defined weakly, by pairing with continuous test functions; the corresponding operator identities are understood in the weak operator sense.

\begin{proposition}
\label{prop:measure-valued-Dyson}
Let \(A,B\in L^1([-T,T];\mathcal W^1)\) be real-valued, set
\(D:=A-B\) and \(A_\theta:=B+\theta D\) for \(\theta\in[0,1]\), and
let \(s,t\in[-T,T]\).  Put
\(I_{s,t}:=[\min\{s,t\},\max\{s,t\}]\) and define
 \[
 V_\theta(t,s)
 :=\ee^{\ii th}U_{A_\theta}(t,s)\ee^{-\ii sh},
 \]
where \(U_{A_\theta}\) is obtained from
Lemma~\ref{lem:prescribed-potential-propagators}.
Then
\(V_\theta(t,s)\in\mathscr A_1\) and
\begin{equation}\label{eq:adapted-V-A1}
 \|V_\theta(t,s)\|_{\mathscr A_1}
 \leq C\ee^{C\int_{I_{s,t}}\|A_\theta(r)\|_{\mathcal W^1}\,\dd r}.
\end{equation}
Moreover, \(\theta\mapsto V_\theta(t,s)\) is \(C^1\) in
\(\mathscr A_1\), and
\[
 \partial_\theta V_\theta(t,s)
 =-\ii\int_s^t
 V_\theta(t,r)\alpha_r(D(r))V_\theta(r,s)\,\dd r.
\]
In particular,
\begin{equation}\label{eq:adapted-theta-TV}
 \|\partial_\theta V_\theta(t,s)\|_{\mathscr A_1}
 \leq C\ee^{C\mathfrak a_T(A,B)}
 \int_{I_{s,t}}\|D(r)\|_{\mathcal W^1}\,\dd r,
\end{equation}
for $\mathfrak a_T(A,B)
=\int_{-T}^T
\bigl(\|A(r)\|_{\mathcal W^1}
+\|B(r)\|_{\mathcal W^1}\bigr)\,\dd r$. 
\end{proposition}

\begin{proof}
We first assume \(s\leq t\).  In the weak operator sense, the interaction
Duhamel formula is
\[
 V_\theta(t,s)
 =\1-\ii\int_s^t
 \alpha_r(A_\theta(r))V_\theta(r,s)\,\dd r.
\]
For \(n\geq1\), let
\[
 \Delta_n(s,t)
 :=\{(r_1,\ldots,r_n)\in(s,t)^n:
       s<r_n<\cdots<r_1<t\},
\]
and define the order-\(n\) Volterra term by
\[
 V_{\theta,n}(t,s)
 :=(-\ii)^n\int_{\Delta_n(s,t)}
 \alpha_{r_1}(A_\theta(r_1))\cdots\alpha_{r_n}(A_\theta(r_n))
 \,\dd r_1\cdots\dd r_n,
 \qquad V_{\theta,0}(t,s):=\1.
\]
Iterating the integral equation \(N\) times implies
\[
 V_\theta(t,s)
 =\sum_{n=0}^{N-1}V_{\theta,n}(t,s)
 +\mathcal R_{\theta,N}(t,s),
\]
where
\[
 \mathcal R_{\theta,N}(t,s)
 =(-\ii)^N\int_{\Delta_N(s,t)}
 \alpha_{r_1}(A_\theta(r_1))\cdots\alpha_{r_N}(A_\theta(r_N))
 V_\theta(r_N,s)\,\dd\mathbf r.
\]
Since
\(
 \|\alpha_r(A_\theta(r))\|_{\mathrm{op}}
 \leq \|A_\theta(r)\|_{\mathcal W^1},
\)
the exponential estimate in
Lemma~\ref{lem:prescribed-potential-propagators} yields
\[
 \|\mathcal R_{\theta,N}(t,s)\|_{\mathrm{op}}
 \leq C\ee^{C\int_s^t
       \|A_\theta(r)\|_{\mathcal W^1}\,\dd r}
 \frac{\left(C\int_s^t
       \|A_\theta(r)\|_{\mathcal W^1}\,\dd r\right)^N}{N!}
 \longrightarrow0.
\]
Here the scalar majorant
\(\prod_{j=1}^N\|A_\theta(r_j)\|_{\mathcal W^1}\) is symmetric in
\((r_1,\ldots,r_N)\), and \(\Delta_N(s,t)\) occupies \(1/N!\) of
\((s,t)^N\).
Consequently,
\[
 V_\theta(t,s)=\sum_{n=0}^\infty V_{\theta,n}(t,s)
 \quad\text{in }\mathcal B(\gH).
\]

We next represent each Volterra term by its coefficient measures.
Since
$ A_\theta(r)
 =\sum_{k\in\mathbb Z^d}
   \widehat{A_\theta(r)}(k)M_k$,
and
\(\alpha_r(M_k)=\ee^{\ii r|k|^2}W_{k,2rk}\), we obtain
\[
 \alpha_r(A_\theta(r))
 =\sum_{k\in\mathbb Z^d}
 \widehat{A_\theta(r)}(k)\ee^{\ii r|k|^2}W_{k,2rk},
 \qquad
 \|\alpha_r(A_\theta(r))\|_{\mathscr A_1}
 =\|A_\theta(r)\|_{\mathcal W^1}.
\]
Fix \(n\geq1\) and
\(\mathbf k=(k_1,\ldots,k_n)\in(\mathbb Z^d)^n\).  Repeated use of
\eqref{eq:modulation-translation-identities} implies that
\[
 \begin{aligned}
 \prod_{j=1}^nW_{k_j,2r_jk_j}
 &=\ee^{2\ii\sum_{1\leq i<j\leq n}r_i k_i\cdot k_j}
 W_{\sum_{j=1}^nk_j,y_{\mathbf k}(\mathbf r)},\qquad
 y_{\mathbf k}(\mathbf r):=2\sum_{j=1}^nr_jk_j.
 \end{aligned}
\]
On \(\Delta_n(s,t)\), define the finite complex measure
\[
 \dd\mu_{\theta,\mathbf k}^{(n)}(\mathbf r):=(-\ii)^n
 \prod_{j=1}^n\!\left[
   \widehat{A_\theta(r_j)}(k_j)\ee^{\ii r_j|k_j|^2}
 \right]
 \ee^{2\ii\sum_{1\leq i<j\leq n}r_i k_i\cdot k_j}
 \,\dd r_1\cdots\dd r_n.
\]
For \(\ell\in\mathbb Z^d\) and a Borel set
\(\mathscr E\subset\mathbb T^d\), set
\[
 \mathfrak m_{\theta,n,\ell}^{t,s}(\mathscr E)
 :=\sum_{\substack{\mathbf k\in(\mathbb Z^d)^n\\
                    \sum_{j=1}^nk_j=\ell}}
 \int_{\Delta_n(s,t)}
 \1_{\mathscr E}\!\left(y_{\mathbf k}(\mathbf r)\right)
 \,\dd\mu_{\theta,\mathbf k}^{(n)}(\mathbf r).
\]
Taking total variations we obtain
\begin{align*}
 \sum_{\ell\in\mathbb Z^d}\langle\ell\rangle
   \|\mathfrak m_{\theta,n,\ell}^{t,s}\|_{\mathrm{TV}}&\leq
 \sum_{\mathbf k\in(\mathbb Z^d)^n}
 \Big\langle\sum_{j=1}^nk_j\Big\rangle
 \int_{\Delta_n(s,t)}
 \prod_{j=1}^n
   \bigl|\widehat{A_\theta(r_j)}(k_j)\bigr|\,\dd\mathbf r\\
 &\leq C^n\int_{\Delta_n(s,t)}
 \prod_{j=1}^n\Big(
   \sum_{k_j\in\mathbb Z^d}
   \langle k_j\rangle
   \bigl|\widehat{A_\theta(r_j)}(k_j)\bigr|
 \Big)\dd\mathbf r
 =\frac{C^n}{n!}
 \left(\int_s^t
 \|A_\theta(r)\|_{\mathcal W^1}\,\dd r\right)^n.
\end{align*}
 Substituting the coefficient
expansion into the Volterra term, applying Fubini's theorem, and using
\eqref{eq:modulation-translation-identities}, we obtain
\[
 V_{\theta,n}(t,s)
 =\mathscr J\!\left(
   (\mathfrak m_{\theta,n,\ell}^{t,s})_{\ell\in\mathbb Z^d}
  \right).
\]
Therefore,
\begin{equation}\label{eq:dyson-pushforward-TV}
 \|V_{\theta,n}(t,s)\|_{\mathscr A_1}
 \leq\frac{C^n}{n!}
 \left(
  \int_s^t\|A_\theta(r)\|_{\mathcal W^1}\,\dd r
 \right)^n.
\end{equation}
Thus the Volterra series converges also in
\(\mathscr A_1\).  Summing
\eqref{eq:dyson-pushforward-TV} we obtain
\eqref{eq:adapted-V-A1} for \(s\leq t\).

It remains to differentiate in \(\theta\).  
We have
\[
 \begin{aligned}
 \frac{\dd}{\dd r}
 \bigl(V_{\theta'}(t,r)V_\theta(r,s)f\bigr)
 &=
 \ii(\theta'-\theta)
 V_{\theta'}(t,r)\alpha_r(D(r))V_\theta(r,s)f.
 \end{aligned}
\]
Since
\[
 V_{\theta'}(t,t)V_\theta(t,s)=V_\theta(t,s),
 \qquad
 V_{\theta'}(t,s)V_\theta(s,s)=V_{\theta'}(t,s),
\]
integration in \(r\) yields 
\begin{equation}\label{eq:adapted-propagator-difference}
 V_{\theta'}(t,s)-V_\theta(t,s)
 =-\ii(\theta'-\theta)\int_s^t
 V_{\theta'}(t,r)\alpha_r(D(r))V_\theta(r,s)\,\dd r.
\end{equation}
  Moreover,
\begin{align*}
 &\left\|\int_s^t
 V_{\theta'}(t,r)\alpha_r(D(r))V_\theta(r,s)\,\dd r
 \right\|_{\mathscr A_1}\leq C\int_s^t
 \|V_{\theta'}(t,r)\|_{\mathscr A_1}
 \|D(r)\|_{\mathcal W^1}
 \|V_\theta(r,s)\|_{\mathscr A_1}\,\dd r.
\end{align*}
Since
\(
 \|\alpha_r(D(r))\|_{\mathscr A_1}
 =\|D(r)\|_{\mathcal W^1}
\) by \eqref{eq:free-conjugation-coefficient-measures}, 
using \eqref{eq:crossed-algebra-product} and
\eqref{eq:adapted-V-A1} we obtain
\begin{equation}\label{eq:adapted-propagator-Lipschitz}
 \|V_{\theta'}(t,s)-V_\theta(t,s)\|_{\mathscr A_1}
 \leq C|\theta'-\theta|\ee^{C\mathfrak a_T(A,B)}
 \int_s^t\|D(r)\|_{\mathcal W^1}\,\dd r.
\end{equation}
Dividing \eqref{eq:adapted-propagator-difference} by
\(\theta'-\theta\) and using
\eqref{eq:adapted-propagator-Lipschitz} on each interval \([r,t]\), we
get
\begin{align*}
 &\left\|
 \frac{V_{\theta'}(t,s)-V_\theta(t,s)}{\theta'-\theta}
 +\ii\int_s^t
 V_\theta(t,r)\alpha_r(D(r))V_\theta(r,s)\,\dd r
 \right\|_{\mathscr A_1}\\
 &\quad\leq
 C\ee^{C\mathfrak a_T(A,B)}
 \int_s^t
 \|V_{\theta'}(t,r)-V_\theta(t,r)\|_{\mathscr A_1}
 \|D(r)\|_{\mathcal W^1}\,\dd r\\
 &\quad\leq
 C|\theta'-\theta|\ee^{C\mathfrak a_T(A,B)}
 \left(\int_s^t\|D(r)\|_{\mathcal W^1}\,\dd r\right)^2
 \longrightarrow0\qquad(\theta'\to\theta).
\end{align*}
Consequently,
\[
 \partial_\theta V_\theta(t,s)
 =-\ii\int_s^t
 V_\theta(t,r)\alpha_r(D(r))V_\theta(r,s)\,\dd r.
\]
The same estimate, applied to the derivative formula at two parameters,
proves continuity in \(\theta\), and
\eqref{eq:adapted-theta-TV} follows from
\eqref{eq:crossed-algebra-product} and
\eqref{eq:adapted-V-A1}.

For \(t<s\), we have
\(V_\theta(t,s)=V_\theta(s,t)^*\) and
\(\alpha_r(D(r))^*=\alpha_r(D(r))\).  
The estimates follow in the same way with \([s,t]\) replaced by
\(I_{s,t}\).

\end{proof}

We now prove Proposition~\ref{prop:transported-density-comparison} by
combining the coefficient estimates with the argument at the beginning of this section.

\begin{proof}[Proof of Proposition~\ref{prop:transported-density-comparison}]
Fix \(u\in\Omega_*\), \(q\in\mathbb Z^d\setminus\{0\}\), and
\(t\in[-T,T]\).    
By \eqref{eq:free-conjugation-coefficient-measures}, Proposition~\ref{prop:measure-valued-Dyson},
and
\eqref{eq:crossed-algebra-product}, we obtain
\begin{equation}\label{eq:adapted-X-A1}
 \|\mathcal D_\theta(r)\|_{\mathscr A_1}
 \leq C\|V_\theta(r,0)\|_{\mathscr A_1}^2
 \|\alpha_r(D(r))\|_{\mathscr A_1}
 \leq C\ee^{C\mathfrak a_T(A,B)}\|D(r)\|_{\mathcal W^1}.
\end{equation}
 Proposition~\ref{prop:measure-valued-Dyson}
and the cocycle identity
\(V_\theta(t,r)=V_\theta(t,0)V_\theta(r,0)^*\) imply
\[
 \partial_\theta V_\theta(t,0)
 =-\ii V_\theta(t,0)\int_0^t\mathcal D_\theta(r)\,\dd r.
\]
Differentiating \(Y_{\theta,q,\varepsilon}(t)\) therefore yields
\begin{equation*}
 \partial_\theta Y_{\theta,q,\varepsilon}(t)
 =\ii\int_0^t
 [\mathcal D_\theta(r),Y_{\theta,q,\varepsilon}(t)]\,\dd r.
\end{equation*}
Since $\rho_{\varepsilon,q}(U_{A_\theta}(t,0)u)
 =\langle u,Y_{\theta,q,\varepsilon}(t)u\rangle_{H^{-\kappa},H^\kappa}$, by
\eqref{eq:heat-smoothed-duality} we obtain
\begin{align}
 &\rho_{\varepsilon,q}(U_A(t,0)u)
  -\rho_{\varepsilon,q}(U_B(t,0)u)
 \notag\\
 &\quad=\ii\int_0^1\!\int_0^t
 \Bigl\{
 \mathcal Q_u([\mathcal D_\theta(r),
                  Y_{\theta,q,\varepsilon}(t)])
 +\Tr\!\left(
 h^{-1}[\mathcal D_\theta(r),Y_{\theta,q,\varepsilon}(t)]
 \right)
 \Bigr\}\,\dd r\,\dd\theta.
 \label{eq:fixed-heat-interpolation}
\end{align}
We now let \(\varepsilon\downarrow0\).  Since
\(\alpha_t(M_q)=\ee^{\ii t|q|^2}W_{q,2tq}\),
using Lemma~\ref{lem:heat-regularization-removal},
Proposition~\ref{prop:measure-valued-Dyson}, and \eqref{eq:adapted-X-A1}, we obtain
\begin{align*}
 &\left|
 \mathcal Q_u([\mathcal D_\theta(r),
                 Y_{\theta,q,\varepsilon}(t)])
 +\Tr\!\left(
 h^{-1}[\mathcal D_\theta(r),Y_{\theta,q,\varepsilon}(t)]
 \right)
 \right|\leq
 C(1+\mathcal K(u))
 \ee^{C\mathfrak a_T(A,B)}\|D(r)\|_{\mathcal W^1}.
\end{align*}
 Thus
Lemma~\ref{lem:heat-regularization-removal} and dominated convergence
imply
\begin{equation}\label{eq:unregularized-density-interpolation}
 \lim_{\varepsilon\downarrow0}\!\left[
 \rho_{\varepsilon,q}(U_A(t,0)u)-
 \rho_{\varepsilon,q}(U_B(t,0)u)\right]
 =\ii\int_0^1\!\int_0^t
 \mathcal C_u([\mathcal D_\theta(r),Y_{\theta,q}(t)])
 \,\dd r\,\dd\theta,
\end{equation}
where
\(Y_{\theta,q}(t)=V_\theta(t,0)^*\alpha_t(M_q)V_\theta(t,0)\).
Consequently,
\begin{equation}\label{eq:transported-density-difference-bound}
 \left|
 \lim_{\varepsilon\downarrow0}
 \bigl[
 \rho_{\varepsilon,q}(U_A(t,0)u)
 -\rho_{\varepsilon,q}(U_B(t,0)u)
 \bigr]
 \right|
 \leq C(1+\mathcal K(u))
 \ee^{C\mathfrak a_T(A,B)}
 \int_{I_t}\|D(r)\|_{\mathcal W^1}\,\dd r.
\end{equation}

It remains to prove that the two limits exist separately.  We apply the previous
 argument once to the pair \((A,0)\) and once to
the pair \((B,0)\).  For the zero potential,
\[
 U_0(t,0)=\ee^{-\ii th},
 \qquad
 V^{(0)}(t,s)=\1.
\]
  Since \(q\neq0\),
\[
 \Tr(h^{-1}\mathsf H_\varepsilon M_q\mathsf H_\varepsilon)=0,
 \qquad
 \Tr(h^{-1}\mathsf H_\varepsilon
                    \alpha_t(M_q)\mathsf H_\varepsilon)=0.
\]
Lemma~\ref{lem:heat-smoothed-raw-form} therefore implies
\begin{align*}
 \rho_{\varepsilon,q}(\ee^{-\ii th}u)
 =&\mathcal Q_u\!\left(
   \mathsf H_\varepsilon\alpha_t(M_q)\mathsf H_\varepsilon
  \right)
 \\=&\ee^{-\varepsilon m}\ee^{\ii t|q|^2}
 \iint p_\varepsilon(\xi)p_\varepsilon(\zeta)\ee^{\ii q\cdot\xi}
 Z_q(2tq+\xi+\zeta;u)\,\dd\xi\,\dd\zeta.
\end{align*}
Since \(Z_q(\,\cdot\,;u)\) is continuous,
using \eqref{eq:double-heat-approximate-identity} we obtain
\[
 \rho_{\varepsilon,q}(\ee^{-\ii th}u)
 \longrightarrow
 \mathcal Q_u(\alpha_t(M_q))
 =\ee^{\ii t|q|^2}Z_q(2tq;u)
 \qquad(\varepsilon\downarrow0).
\]
Applying \eqref{eq:unregularized-density-interpolation} to \((A,0)\)
shows that
\(
 \lim_{\varepsilon\downarrow0}\bigl[
  \rho_{\varepsilon,q}(U_A(t,0)u)
  -\rho_{\varepsilon,q}(\ee^{-\ii th}u)
 \bigr]
\)
exists.  Together with the free limit, this proves the existence of
\(\widetilde\rho_q^A(t;u)\).  Applying the same argument to \((B,0)\)
proves the existence of \(\widetilde\rho_q^B(t;u)\). Thus
\eqref{eq:transported-density-difference-bound} is exactly
\eqref{eq:transported-density-comparison}.
\end{proof}

\section{Gibbs estimates for Wick products}
\label{app:Wick-product-estimates}

Throughout this appendix, assume $d=3$,
\eqref{eq:finite-potential-regularity}, and
\eqref{eq:smooth-potential-regularity}, with $0<\lambda\le\lambda_0\le1$.
We use $g_3(\lambda)=\lambda^{-1/2}$ from \eqref{def:gd} and the
iterated commutators $R_{\lambda,j}^{(r)}$ from
Subsection~\ref{subsec:Gaussian-removal-Wick}.
The following estimate is used in
Proposition~\ref{prop:smooth-raw-derivatives}.

\begin{lemma}
\label{lem:Wick-commutator-Gibbs-products}
Let $X_{\lambda,j}=\Op_{\lambda,E_j}^{\rm W}(P_j)$, $1\le j\le J$,
for fixed finite Fourier spaces $E_j$ and polynomials $P_j$.
For $\boldsymbol\alpha\in\mathbb N_0^J$, set
$N=|\boldsymbol\alpha|$, $d_j=\deg P_j$, and $D=\sum_{j=1}^Jd_j$.
Then
\begin{equation}\label{eq:Wick-commutator-Gibbs-products}
 \Big\|\Big(\prod_{j=1}^J\Ubb_\lambda(t_j)
       \bigl(R_{\lambda,j}^{(\alpha_j)}\bigr)\Big)
       \Gamma_\lambda^{1/2}\Big\|_{\mathfrak S^2}
 \le C_{\boldsymbol\alpha}(1+|\boldsymbol t|)^{(J+1)N}
       g_3(\lambda)^{N+D/2}.
\end{equation}
The product, initially defined on smooth finite-particle vectors,
extends after right multiplication by $\Gamma_\lambda^{1/2}$ to a
continuous $\mathfrak S^2$-valued function of $\boldsymbol t\in\mathbb R^J$.
The constant is independent of $\lambda$ and $\boldsymbol t$.
\end{lemma}
\begin{proof}
For $r,M\in\mathbb N_0$, we first  define the following operators
\[
 b_\lambda=g_3(\lambda)\1+\lambda\mathcal N,\qquad
 \mathsf T_r=\Big(\frac{\1+\dG(h^r)}{\1+\mathcal N}\Big)^{1/2},\qquad
 Q_{\lambda,M}=\Big(\sum_{r=0}^M
                  b_\lambda^{-2r}\mathsf T_r^2\Big)^{1/2}.
\]
These operators commute. With $\mathsf E_n=\1_{\{\mathcal N=n\}}$, we have
\[
 b_{\lambda,n}:=g_3(\lambda)+\lambda n,\qquad
 b_\lambda\mathsf E_n=b_{\lambda,n}\mathsf E_n,\qquad
 \mathsf T_0=Q_{\lambda,0}=\1\le Q_{\lambda,M}.
\]

\smallskip
\noindent\emph{Step 1.}
We first prove
\begin{equation}\label{eq:Fock-Sobolev-Gibbs-bound}
 \|b_\lambda^{q/2}Q_{\lambda,M}\Gamma_\lambda^{1/2}\|_{\mathfrak S^2}^2
 \le C_{M,q}g_3(\lambda)^q,
 \qquad M\in\mathbb N_0,\quad q\ge0.
\end{equation}
Following the kernel comparison in \cite[Lemma~11.3]{LNR}, extend the
$n$-particle Hamiltonian to $\gH^{\otimes n}$ by
\[
 \mathbb H_{\lambda,n}=\sum_{i=1}^n h_i+\lambda^{-1}W_{\lambda,n},
 \qquad h_i=-\Delta_{x_i}+m.
\]
Here $W_{\lambda,n}$ is the multiplication operator obtained from
$\mathbb W_\lambda$ in Subsection~\ref{subsec:equilibrium-density-tails}.
With $c_\lambda=\lambda\Tr(\gamma_{0,\lambda})$ as in
Lemma~\ref{lem:ordered-trace-number-bounds},
\begin{align*}
 W_{\lambda,n}(\boldsymbol x)
 &=\frac12\sum_{k\in\mathbb Z^3}\widehat v(k)
   \Big|\lambda\sum_{j=1}^n\ee^{\ii k\cdot x_j}
             -c_\lambda\1_{\{k=0\}}\Big|^2\\
 &=\lambda^2\sum_{i<j}v(x_i-x_j)
   +\frac{\lambda^2n}{2}v(0)-\lambda n c_\lambda\widehat v(0)
   +\frac{c_\lambda^2}{2}\widehat v(0)\ge0.
\end{align*}
Hence, for $\alpha\in\mathbb N_0^3\setminus\{0\}$,
\[
 \partial_{x_i}^{\alpha}W_{\lambda,n}
 =\lambda^2\sum_{j\ne i}(\partial^\alpha v)(x_i-x_j),\qquad
 \|\partial_{x_i}^{\alpha}W_{\lambda,n}\|_\infty
 \le\lambda^2n\|\partial^\alpha v\|_\infty.
\]
 The Trotter product formula \cite{Trotter59} implies
\[
 \Big(\ee^{-\lambda\sum_i h_i/\ell}
          \ee^{-W_{\lambda,n}/\ell}\Big)^\ell f
 \xrightarrow[\ell\to\infty]{}\ee^{-\lambda\mathbb H_{\lambda,n}}f
 \quad\text{in }\gH^{\otimes n},\qquad f\in\gH^{\otimes n}.
\]
For $\ell\ge2$, denote the kernel of the product by
$K_{\lambda,n}^{(\ell)}$, and the kernel of
$\ee^{-\lambda\mathbb H_{\lambda,n}}$ by $K_{\lambda,n}$.
For the Euclidean and periodic kernels of $\ee^{t\Delta/2}$, write
\[
 \widetilde p_t(z)=(2\pi t)^{-3/2}\ee^{-|z|^2/(2t)},\qquad
 p_t(z)=\sum_{k\in\mathbb Z^3}\widetilde p_t(z+2\pi k).
\]
Expanding the periodic kernels, we obtain
\begin{align*}
 K_{\lambda,n}^{(\ell)}(\boldsymbol x,\boldsymbol y)
 &=\ee^{-\lambda mn}\sum_{\boldsymbol k\in(\mathbb Z^3)^n}
   \int_{(\mathbb R^3)^{n(\ell-1)}}
      \prod_{j=1}^\ell\Big(\ee^{-W_{\lambda,n}(\boldsymbol z_j)/\ell}
        \prod_{i=1}^n\widetilde p_{2\lambda/\ell}
                            (z_{j,i}-z_{j-1,i})\Big)
      \prod_{j=1}^{\ell-1}\dd\boldsymbol z_j,\\
 &\hspace{20mm}\boldsymbol z_0=\boldsymbol x,\qquad
              \boldsymbol z_\ell=\boldsymbol y+2\pi\boldsymbol k.
\end{align*}
Let $\boldsymbol z_j=(1-j/\ell)\boldsymbol x
       +(j/\ell)(\boldsymbol y+2\pi\boldsymbol k)
       +\boldsymbol\zeta_j,
 \boldsymbol\zeta_0=\boldsymbol\zeta_\ell=0.$
Then
\[
 \sum_{j=1}^\ell|\boldsymbol z_j-\boldsymbol z_{j-1}|^2
 =\frac{|\boldsymbol y-\boldsymbol x+2\pi\boldsymbol k|^2}{\ell}
       +\sum_{j=1}^\ell|\boldsymbol\zeta_j-\boldsymbol\zeta_{j-1}|^2.
\]
Define the probability measure
\[
 \dd\mathbb P_{\lambda,n,\ell}(\boldsymbol\zeta)
 =\frac{\displaystyle\prod_{j=1}^\ell\prod_{i=1}^n
          \widetilde p_{2\lambda/\ell}(\zeta_{j,i}-\zeta_{j-1,i})}
        {\widetilde p_{2\lambda}(0)^n}
      \prod_{j=1}^{\ell-1}\dd\boldsymbol\zeta_j,
 \qquad \int\dd\mathbb P_{\lambda,n,\ell}=1.
\]
Thus we obtain
\begin{align*}
 K_{\lambda,n}^{(\ell)}(\boldsymbol x,\boldsymbol y)
 &=\ee^{-\lambda mn}\sum_{\boldsymbol k\in(\mathbb Z^3)^n}
   \prod_{i=1}^n\widetilde p_{2\lambda}(y_i-x_i+2\pi k_i)
   \int\ee^{-\Phi_\ell}\,\dd\mathbb P_{\lambda,n,\ell},\\
 \Phi_\ell&=\frac1\ell\sum_{j=1}^\ell W_{\lambda,n}(\boldsymbol z_j)\ge0,\\
 \partial_{x_i}^\alpha\Phi_\ell
 &=\frac1\ell\sum_{j=1}^\ell(1-j/\ell)^{|\alpha|}
       (\partial_{x_i}^\alpha W_{\lambda,n})(\boldsymbol z_j),
 \qquad
 |\partial_{x_i}^\alpha\Phi_\ell|\le C_\alpha\lambda^2n.
\end{align*}
The chain rule and $\Phi_\ell\ge0$ imply, for $\alpha\ne0$,
\[
 |\partial_{x_i}^\alpha\ee^{-\Phi_\ell}|
 \le C_\alpha\sum_{a=1}^{|\alpha|}(\lambda^2n)^a
 \le C_\alpha(1+\lambda^2n)^{|\alpha|}.
\]
Together with
$|\partial^\alpha\widetilde p_{2\lambda}(z)|
 \le C_\alpha\lambda^{-|\alpha|/2}\widetilde p_{4\lambda}(z)$,
Leibniz's rule implies
\[
 |(h_i^r)_{\boldsymbol x}K_{\lambda,n}^{(\ell)}(\boldsymbol x,\boldsymbol y)|
 \le C_r\lambda^{-r}(1+\lambda^2n)^{2r}\ee^{-\lambda mn}
       p_{4\lambda}(y_i-x_i)\prod_{j\ne i}p_{2\lambda}(y_j-x_j).
\]
For fixed $n\ge1$ and $\lambda>0$, the Trotter convergence
\cite[Theorem~VIII.31]{ReedSimon1}, together with the uniform kernel
bound, implies that $K_{\lambda,n}^{(\ell)}$ converges to
$K_{\lambda,n}$ in the sense of distributions. The displayed
derivative bound therefore passes to
$(h_i^r)_{\boldsymbol x}K_{\lambda,n}$ in the weak-* topology of
$L^\infty$, initially almost everywhere. Since
$\mathbb H_{\lambda,n}$ has smooth coefficients on the torus, its
heat kernel is smooth, and the same bound holds pointwise.

Let $\mathsf P_{\mathrm{sym},n}$ be the orthogonal projection from
$\gH^{\otimes n}$ onto $\gH^{\otimes_s n}$, and set
$\mathcal Z_{0,\lambda,n}
 :=\Tr_{\gH^{\otimes_s n}}\ee^{-\lambda\sum_jh_j}.$
For $n\ge1$, we have
\begin{align*}
 &0\le\Tr_{\gH^{\otimes n}}
       (\mathsf P_{\mathrm{sym},n}h_i^r\ee^{-\lambda\mathbb H_{\lambda,n}})=\frac1{n!}\sum_{\pi\in\mathcal S_n}
   \int_{(\mathbb T^3)^n}
      \big[(h_i^r)_{\boldsymbol x}K_{\lambda,n}(\boldsymbol x,\boldsymbol y)
           \big]_{\boldsymbol y=\pi\boldsymbol x}\,\dd\boldsymbol x\\
 &\quad\le C_r\lambda^{-r}(1+\lambda^2n)^{2r}
       \Tr_{\gH^{\otimes n}}\!\big(
          \mathsf P_{\mathrm{sym},n}\ee^{-\lambda\sum_jh_j}
                               \ee^{\lambda\Delta_{x_i}}\big)
 \le C_r\lambda^{-r}(1+\lambda^2n)^{2r}\mathcal Z_{0,\lambda,n}.
\end{align*}
Here $\mathcal S_n$ is the permutation group and
$(\pi\boldsymbol x)_i=x_{\pi(i)}$ and in the last inequality we used
$0\le\ee^{\lambda\Delta_{x_i}}\le\1$.
Thus we obtain, for $n\ge0$,
\begin{align*}
 \Tr_{\gH^{\otimes_s n}}
       (\mathsf T_r^2\ee^{-\lambda\mathbb H_{\lambda,n}})
 &=\frac1{n+1}\Tr_{\gH^{\otimes n}}\!\Big(
       \mathsf P_{\mathrm{sym},n}\Big(\1+\sum_{i=1}^nh_i^r\Big)
                    \ee^{-\lambda\mathbb H_{\lambda,n}}\Big)\\
 &\le C_r\lambda^{-r}(1+\lambda^2n)^{2r}\mathcal Z_{0,\lambda,n}.
\end{align*}
Since
 $\lambda^{-1/2}(1+\lambda^2n)
 =g_3(\lambda)+\lambda^{3/2}n\le b_{\lambda,n},$
summing the nonnegative sector traces yields
\begin{align*}
 \|b_\lambda^{q/2}Q_{\lambda,M}\Gamma_\lambda^{1/2}\|_{\mathfrak S^2}^2
 &=\frac1{\mathcal Z_\lambda}\sum_{r=0}^M\sum_{n\ge0}
       b_{\lambda,n}^{q-2r}
       \Tr_{\gH^{\otimes_s n}}
           (\mathsf T_r^2\ee^{-\lambda\mathbb H_{\lambda,n}})\\
 &\le\frac{C_M}{\mathcal Z_\lambda}\sum_{n\ge0}
                b_{\lambda,n}^{q}\mathcal Z_{0,\lambda,n}
 =C_M\frac{\mathcal Z_{0,\lambda}}{\mathcal Z_\lambda}
                \Tr(b_\lambda^q\Gamma_{0,\lambda})
 \le C_{M,q}g_3(\lambda)^q.
\end{align*}
Here we used $\mathcal Z_{0,\lambda}/\mathcal Z_\lambda\le C$ from
\cite[Lemma~5.14]{LNR}, and the quasi-free number-moment bound
$\Tr(b_\lambda^q\Gamma_{0,\lambda})
 \le C_q(g_3(\lambda)+\lambda\Tr\gamma_{0,\lambda})^q
 \le C_qg_3(\lambda)^q$.

\smallskip
\noindent\emph{Step 2.}
Let $\boldsymbol p=-\ii\nabla$ on $\gH$, with
$\boldsymbol p^\beta=\prod_{\nu=1}^3(-\ii\partial_\nu)^{\beta_\nu}$.
For $\varphi,f\in C^\infty(\mathbb T^3)$, $\beta\in\mathbb N_0^3$, and
smooth finite-particle vectors $\psi$, we prove
\begin{equation}\label{eq:Fock-Sobolev-operator-bounds}
 \begin{aligned}
 \|Q_{\lambda,M}\ee^{-\ii t\mathbb H_\lambda}\psi\|
   &\le C_M(1+|t|)^M\|Q_{\lambda,M}\psi\|,\\
 \|Q_{\lambda,M}\lambda\dG(\varphi\boldsymbol p^\beta)\psi\|
   &\le C_{M,\varphi,\beta}
           \|b_\lambda^{1+|\beta|}Q_{\lambda,M+|\beta|}\psi\|,\\
 \|Q_{\lambda,M}a_\lambda^\#(f)\psi\|
   &\le C_{M,f}\|b_\lambda^{1/2}Q_{\lambda,M}\psi\|.
 \end{aligned}
\end{equation}
For $\psi^{(n)}\in C^\infty_{\mathrm{sym}}((\mathbb T^3)^n)$, we have
\begin{align*}
 \|\mathsf T_r\psi^{(n)}\|^2
 &=\frac1{n+1}\Big(\|\psi^{(n)}\|^2+
          \sum_{i=1}^n\|h_i^{r/2}\psi^{(n)}\|^2\Big)\\
 &\asymp_{r,m}\frac1{n+1}\Big(\|\psi^{(n)}\|^2+
       \sum_{i=1}^n\sum_{|\alpha|\le r}
          \|\partial_{x_i}^\alpha\psi^{(n)}\|^2\Big).
\end{align*}
Write $\psi_t^{(n)}=\ee^{-\ii t\mathbb H_{\lambda,n}}\psi^{(n)}$.
For $r\ge1$ and $|\alpha|\le r$,
\begin{align*}
 [\partial_{x_i}^\alpha,\mathbb H_{\lambda,n}]
 &=\lambda\sum_{j\ne i}\sum_{0<\gamma\le\alpha}
   \binom\alpha\gamma(\partial^\gamma v)(x_i-x_j)
                   \partial_{x_i}^{\alpha-\gamma},\\
 \partial_{x_i}^\alpha\psi_t^{(n)}
 &=\ee^{-\ii t\mathbb H_{\lambda,n}}\partial_{x_i}^\alpha\psi^{(n)}
   -\ii\int_0^t\ee^{-\ii(t-s)\mathbb H_{\lambda,n}}
       [\partial_{x_i}^\alpha,\mathbb H_{\lambda,n}]
                               \psi_s^{(n)}\,\dd s.
\end{align*}
Thus, for $t\ge0$,
\[
 \|\mathsf T_r\psi_t^{(n)}\|
 \le C_r\|\mathsf T_r\psi^{(n)}\|
       +C_r\lambda n\int_0^t\|\mathsf T_{r-1}\psi_s^{(n)}\|\,\dd s.
\]
Iteration, starting with $\|\psi_t^{(n)}\|=\|\psi^{(n)}\|$,
and the same argument for negative times imply
\[
 b_{\lambda,n}^{-r}\|\mathsf T_r\psi_t^{(n)}\|
 \le C_r\sum_{j=0}^r|t|^j
       \Big(\frac{\lambda n}{b_{\lambda,n}}\Big)^j
       b_{\lambda,n}^{-(r-j)}\|\mathsf T_{r-j}\psi^{(n)}\|
 \le C_r(1+|t|)^r\|Q_{\lambda,r}\psi^{(n)}\|.
\]
Squaring and summing over $r\le M$ and $n$ proves the first estimate in
\eqref{eq:Fock-Sobolev-operator-bounds}.

For the second estimate, Leibniz's rule and permutation symmetry imply,
for $n\ge1$,
\begin{align*}
 \|\mathsf T_r\lambda\dG(\varphi\boldsymbol p^\beta)\psi^{(n)}\|
 &\le C_{r,\varphi,\beta}\lambda n
                              \|\mathsf T_{r+|\beta|}\psi^{(n)}\|,\\
 \|Q_{\lambda,M}\lambda\dG(\varphi\boldsymbol p^\beta)\psi^{(n)}\|^2
 &\le C_{M,\varphi,\beta}(\lambda n)^2
       \sum_{r=0}^M b_{\lambda,n}^{-2r}
                              \|\mathsf T_{r+|\beta|}\psi^{(n)}\|^2\\
 &\le C_{M,\varphi,\beta}b_{\lambda,n}^{2+2|\beta|}
                              \|Q_{\lambda,M+|\beta|}\psi^{(n)}\|^2.
\end{align*}
Summation over $n$ proves the second estimate in
\eqref{eq:Fock-Sobolev-operator-bounds}.

For creation and annihilation, the formulas in
Subsection~\ref{subsec:quantizations} and the Cauchy-Schwarz inequality for annihilation imply
\begin{align*}
 \|\mathsf T_ra_\lambda^*(f)\psi^{(n)}\|^2
 &\le\frac{\lambda(n+1)}{n+2}
   \Big(\|h^{r/2}f\|^2\|\psi^{(n)}\|^2
         +(n+1)\|f\|^2\|\mathsf T_r\psi^{(n)}\|^2\Big),\\
 \|\mathsf T_ra_\lambda(f)\psi^{(n)}\|^2
 &\le\lambda\|f\|^2\Big(\|\psi^{(n)}\|^2
                    +\sum_{i=2}^n\|h_i^{r/2}\psi^{(n)}\|^2\Big)
       \quad(n\ge1).
\end{align*}
Thus, for $a_\lambda^\#=a_\lambda$ or $a_\lambda^*$,
\[
 \|\mathsf T_ra_\lambda^\#(f)\psi^{(n)}\|^2
 \le C_{r,f}\lambda(n+1)\|\mathsf T_r\psi^{(n)}\|^2.
\]
Since $\lambda(n+1)\le b_{\lambda,n}$ and
$b_{\lambda,n\pm1}\asymp b_{\lambda,n}$ whenever $n\pm1\ge0$,
summing with the weights $b_{\lambda,n\pm1}^{-2r}$ proves the last
estimate in \eqref{eq:Fock-Sobolev-operator-bounds}.

\smallskip
\noindent\emph{Step 3.}
By linearity, it suffices to treat monomials $P_j$.
For $b=M_k\boldsymbol p^\beta$, \eqref{eq:real-density-square} implies
\begin{align*}
 [\mathbb H_\lambda,B_\lambda(f_\eta)]
   &=\lambda\dG([h,f_\eta]),\qquad
 [\dG(h),\lambda\dG(b)]=\lambda\dG([h,b]),\\
 [\lambda^{-1}\mathbb W_\lambda,\lambda\dG(b)]
   &=\frac12\sum_{\eta\in\mathfrak I}w_\eta
       \bigl(B_\lambda(f_\eta)\lambda\dG([f_\eta,b])
           +\lambda\dG([f_\eta,b])B_\lambda(f_\eta)\bigr),
\end{align*}
where
\begin{align*}
 [h,M_k\boldsymbol p^\beta]
 &=M_k(2k\cdot\boldsymbol p+|k|^2)\boldsymbol p^\beta,\qquad
 [M_{k'},M_k\boldsymbol p^\beta]
 =-M_{k+k'}\sum_{0<\zeta\le\beta}
                  \binom\beta\zeta(k')^\zeta\boldsymbol p^{\beta-\zeta}.
\end{align*}
We also have
\begin{align*}
 [\dG(h),a_\lambda^\#(f)]&=\pm a_\lambda^\#(hf),\\
 [\lambda^{-1}\mathbb W_\lambda,a_\lambda^\#(f)]
   &=\pm\frac12\sum_{\eta\in\mathfrak I}w_\eta
        \bigl(B_\lambda(f_\eta)a_\lambda^\#(f_\eta f)
             +a_\lambda^\#(f_\eta f)B_\lambda(f_\eta)\bigr),
\end{align*}
with $+$ for creation and $-$ for annihilation.

We apply these identities one factor at a time, using
\[
 [\mathbb H_\lambda,A_1\cdots A_q]
 =\sum_{a=1}^q A_1\cdots A_{a-1}
             [\mathbb H_\lambda,A_a]A_{a+1}\cdots A_q.
\]
Each resulting term is an ordered product of factors of the form
$a_\lambda^\#(f)$, where $a_\lambda^\#$ denotes either $a_\lambda$
or $a_\lambda^*$, and factors of the form $B_\lambda(f_\eta)$ or
$\lambda\dG(M_k\boldsymbol p^\beta)$.
The latter factors are kept unexpanded. Let $s$ be their number, assign
order zero to $B_\lambda(f_\eta)$ and order $|\beta|$ to
$\lambda\dG(M_k\boldsymbol p^\beta)$, and let $\ell$ be the sum of
these orders.

For a monomial $P_j$ of degree $d_j$, the initial product contains
exactly $d_j$ factors of the form $a_\lambda^\#(f)$ and $s=\ell=0$.
A commutator with $\dG(h)$ does not add a factor and increases $\ell$
by at most one. When it acts on $a_\lambda^\#(f)$, it only replaces $f$ by $hf$.

A commutator of $\lambda^{-1}\mathbb W_\lambda$ with
$a_\lambda^\#(f)$ adds one factor $B_\lambda(f_\eta)$,
increasing $s+\ell$ by one.
When it acts on $\lambda\dG(M_k\boldsymbol p^\beta)$, each nonzero
term replaces $\beta$ by $\beta-\zeta$, with $|\zeta|\ge1$, and adds
one factor $B_\lambda(f_\eta)$. Thus
\[
 (s+1)+(\ell-|\zeta|)\le s+\ell.
\]
The interaction commutator vanishes on factors of order zero.
The commutator identities above preserve the number of explicit
$a_\lambda^\#(f)$ factors. Hence every term of $R_{\lambda,j}^{(r)}$
contains exactly $d_j$ such factors, and induction gives
\(
 s+\ell\le r.
\)
The centered factors also satisfy the order-zero bound. Indeed,
$\Tr(f_\eta\gamma_{0,\lambda})
 =\1_{\{\eta=0\}}\Tr\gamma_{0,\lambda}$, so
\[
 B_\lambda(f_\eta)=\lambda\dG(f_\eta)
                    -c_\lambda\1_{\{\eta=0\}}\1,\qquad
 \|Q_{\lambda,M}B_\lambda(f_\eta)\psi\|
 \le C_M\langle k_\eta\rangle^M\|b_\lambda Q_{\lambda,M}\psi\|,
\]
by \eqref{eq:Fock-Sobolev-operator-bounds} and
$c_\lambda\le Cg_3(\lambda)$.

The factors $B_\lambda(f_\eta)$ and
$\lambda\dG(M_k\boldsymbol p^\beta)$ preserve particle number,
whereas each creation or annihilation operator changes it by $+1$ or $-1$.
In a product containing at most $D$ factors of the form
$a_\lambda^\#(f)$, every nonzero intermediate sector $n'$ satisfies
$|n'-n|\le D$, where $n$ is the input sector.
Since $\lambda\le1\le b_{\lambda,n}$,
\[
 |b_{\lambda,n'}-b_{\lambda,n}|\le\lambda D,
 \qquad
 (1+D)^{-1}b_{\lambda,n}\le b_{\lambda,n'}
                                  \le(1+D)b_{\lambda,n}.
\]
Consequently, applying \eqref{eq:Fock-Sobolev-operator-bounds}
 in the original order implies, for each term
$A_1\cdots A_{d_j+s}$,
\[
 \|Q_{\lambda,M}A_1\cdots A_{d_j+s}\psi\|
 \le C\|b_\lambda^{d_j/2+s+\ell}Q_{\lambda,M+\ell}\psi\|
 \le C\|b_\lambda^{d_j/2+r}Q_{\lambda,M+r}\psi\|.
\]
Here each factor $a_\lambda^\#(f)$ contributes a weight
$b_\lambda^{1/2}$, and a factor of order $|\beta|$ contributes
$b_\lambda^{1+|\beta|}$ and requires $|\beta|$ additional Sobolev orders.
Each term of an $r$-fold commutator has at most $r$ interaction
indices. Its frequencies are sums of these and
fixed input frequencies. The constants in the previous bounds grow only
polynomially in these frequencies. Thus, for some finite $L$, their sum is
bounded by
\[
 C\sum_{a=0}^r\sum_{\eta_1,\ldots,\eta_a\in\mathfrak I}
       \prod_{\nu=1}^a w_{\eta_\nu}\langle k_{\eta_\nu}\rangle^L
 =C\sum_{a=0}^r\Big(\sum_{\eta\in\mathfrak I}
                       w_\eta\langle k_\eta\rangle^L\Big)^a<\infty
\]
by \eqref{eq:smooth-potential-regularity}. In particular, the
interaction sums converge absolutely in the weighted operator norm
\[
 \big\|Q_{\lambda,M}R_{\lambda,j}^{(r)}
      \big(b_\lambda^{r+d_j/2}Q_{\lambda,M+r}\big)^{-1}
       \big\|_{\mathrm{op}}\le C_{M,r,j}.
\]
The same bound holds for a general $P_j$, since its monomials have
degree at most $d_j$ and $b_\lambda\ge\1$.

Set $t_0=0$. The ordered product is
\[
 Y_{\lambda,\boldsymbol\alpha}(\boldsymbol t)
 :=\prod_{j=1}^J\Ubb_\lambda(t_j)
                 \bigl(R_{\lambda,j}^{(\alpha_j)}\bigr)
 =\Big(\prod_{j=1}^J
       \ee^{\ii(t_j-t_{j-1})\mathbb H_\lambda}
                         R_{\lambda,j}^{(\alpha_j)}\Big)
                         \ee^{-\ii t_J\mathbb H_\lambda}.
\]
There are $J+1$ propagators, each requiring a Sobolev order at most $N$.
The previous bounds therefore imply
\begin{equation}\label{eq:Wick-product-weighted-operator}
 \|Y_{\lambda,\boldsymbol\alpha}(\boldsymbol t)\psi\|
 \le C_{\boldsymbol\alpha}(1+|\boldsymbol t|)^{(J+1)N}
       \|b_\lambda^{N+D/2}Q_{\lambda,N}\psi\|.
\end{equation}
Using \eqref{eq:Fock-Sobolev-Gibbs-bound} with $M=N$ and $q=2N+D$
we obtain \eqref{eq:Wick-commutator-Gibbs-products}.

For fixed $\lambda$, let $\mathsf S_L$ be as in
Lemma~\ref{lem:directional-Gaussian-Wick}.
The maps
$Y_{\lambda,\boldsymbol\alpha}(\cdot)
\Gamma_\lambda^{1/2}\mathsf S_L$
are smooth, and
\eqref{eq:Wick-product-weighted-operator} together with
\eqref{eq:Fock-Sobolev-Gibbs-bound} implies
\[
Y_{\lambda,\boldsymbol\alpha}(\cdot)
\Gamma_\lambda^{1/2}\mathsf S_L
\longrightarrow
Y_{\lambda,\boldsymbol\alpha}(\cdot)
\Gamma_\lambda^{1/2}
\quad\text{in }C_{\mathrm{loc}}(\mathbb R^J;\mathfrak S^2).
\]
The same convergence at order $N+1$ allows passage to the limit
in the integrated derivative identities, and hence justifies
time differentiation in $\mathfrak S^2$.
\end{proof}

\end{document}